\documentclass[a4paper,11pt]{article}
\usepackage{amsmath}
\usepackage{amsthm}
\usepackage{amsfonts}
\usepackage{mathrsfs}
\usepackage{amssymb}
\usepackage{mathpazo}
\usepackage[dvipsnames]{xcolor}
\usepackage{bm}
\usepackage[no-weekday]{eukdate}
\usepackage[bb=boondox]{mathalfa}
\usepackage{paralist}
\usepackage{natbib}
\usepackage{url}
\usepackage[textwidth=6.1in,textheight = 10in]{geometry}
\usepackage{placeins}
\usepackage[hidelinks]{hyperref}
\hypersetup{
	pdftitle = {{Cross-temporal probabilistic forecast reconciliation: methodological and practical issues}}, 
	pdfauthor = {{Daniele Girolimetto and George Athanasopoulos and Tommaso Di Fonzo and Rob J Hyndman}}
}

\usepackage{multirow, float}
\usepackage{makecell}
\usepackage{calc}
\usepackage{graphicx}
\usepackage{enumitem}
\usepackage{microtype}
\usepackage{todonotes}
\usepackage{caption}
\DeclareCaptionStyle{italic}[justification=centering]
 {labelfont={bf},textfont={it},labelsep=colon}
\usepackage[compact]{titlesec}
\titlespacing{\section}{0pt}{2ex}{1ex}
\titlespacing{\subsection}{0pt}{1ex}{0ex}
\titlespacing{\subsubsection}{0pt}{0.5ex}{0ex}
\usepackage{titling}\pretitle{\begin{center}\LARGE\bfseries}
\usepackage[bottom,hang,flushmargin]{footmisc}
\usepackage{adjustbox}

\allowdisplaybreaks
\newcommand{\blind}{1}

\usepackage{caption}
\usepackage{subcaption}

\usepackage{longtable}
\usepackage{booktabs}
\usepackage{array}
\newcolumntype{M}[1]{>{\centering\arraybackslash}m{#1}}
\newcolumntype{L}[1]{>{\raggedright\arraybackslash}m{#1}}
\newcolumntype{R}[1]{>{\raggedleft\arraybackslash}m{#1}}

\newcommand{\Unovet}{\bm{1}}

\newcommand{\bvet}{\bm{b}}

\newcommand{\evet}{\bm{e}}

\newcommand{\uvet}{\bm{u}}

\newcommand{\xvet}{\bm{x}}
\newcommand{\yvet}{\bm{y}}

\newcommand{\Avet}{\bm{A}}
\newcommand{\Bvet}{\bm{B}}
\newcommand{\Cvet}{\bm{C}}

\newcommand{\Evet}{\bm{E}}

\newcommand{\Gvet}{\bm{G}}

\newcommand{\Ivet}{\bm{I}}

\newcommand{\Mvet}{\bm{M}}

\newcommand{\Pvet}{\bm{P}}

\newcommand{\Svet}{\bm{S}}

\newcommand{\Uvet}{\bm{U}}

\newcommand{\Xvet}{\bm{X}}

\newcommand{\Zerovet}{\bm{0}}
\newcommand{\Omegavet}{\bm{\Omega}}

\definecolor{mybluehl}{HTML}{cbd3ff}

\makeatletter
\newcommand{\githuburl}{\begingroup%
\if0\blind
{
$\dots$
}\fi
\if1\blind
{
\url{https://github.com/danigiro/ctprob}.
}\fi
\endgroup}
\makeatother

\makeatletter
\def\@endtheorem{\endtrivlist}
\makeatother

\usepackage{tikz}
\usetikzlibrary{arrows,shapes,positioning,shadows,trees}
\usetikzlibrary{matrix, decorations.pathreplacing, arrows, calc, fit, arrows.meta, decorations.pathmorphing, decorations.markings, decorations,calligraphy}

\tikzset{
  basic/.style  = {draw, text width=2cm, drop shadow, font=\sffamily, rectangle},
  root/.style   = {basic, rounded corners=2pt, thin, align=center,
                   fill=green!30},
  level 2/.style = {basic, rounded corners=6pt, thin,align=center, fill=green!60,
                   text width=4em},
  level 3/.style = {basic, thin, align=left, fill=pink!60, text width=1.5em}
}
\newcommand{\relation}[3]
{
	\draw (#3.south) -- +(0,-#1) -| ($ (#2.north) $)
}

\pgfdeclareimage[height=0.85cm]{ngreen}{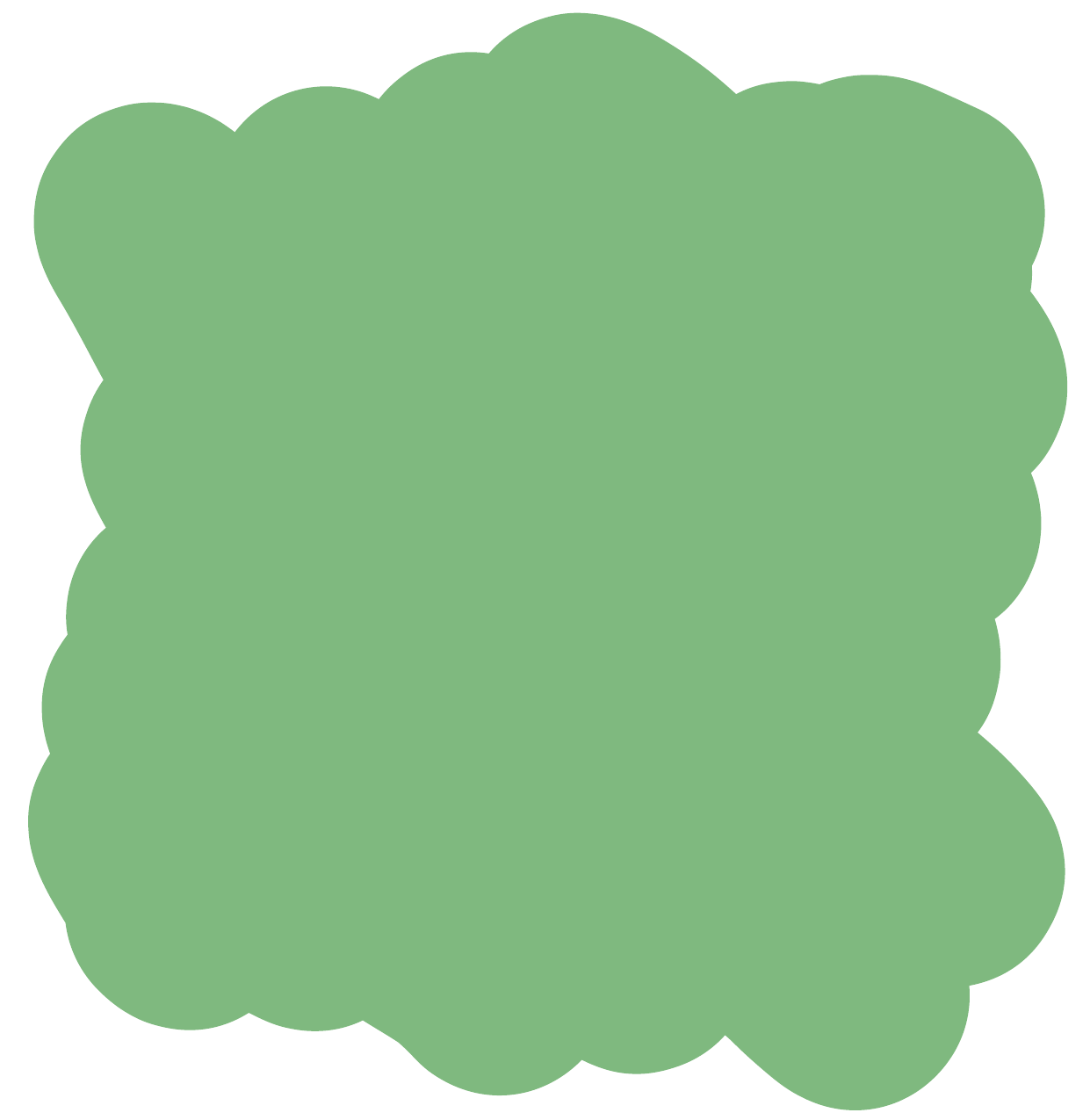}
\pgfdeclareimage[height=0.85cm]{nblue}{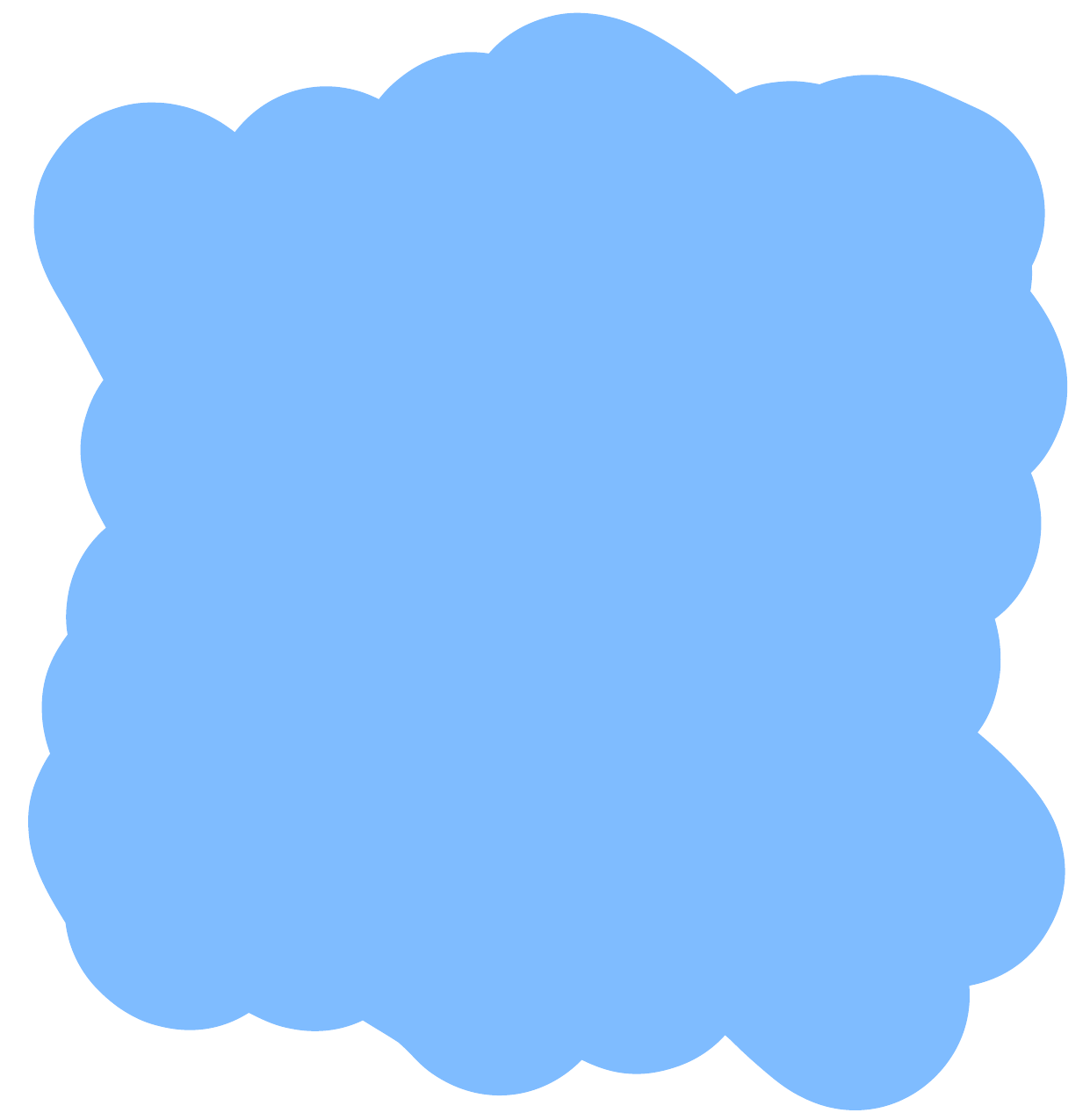}
\pgfdeclareimage[height=0.85cm]{nred}{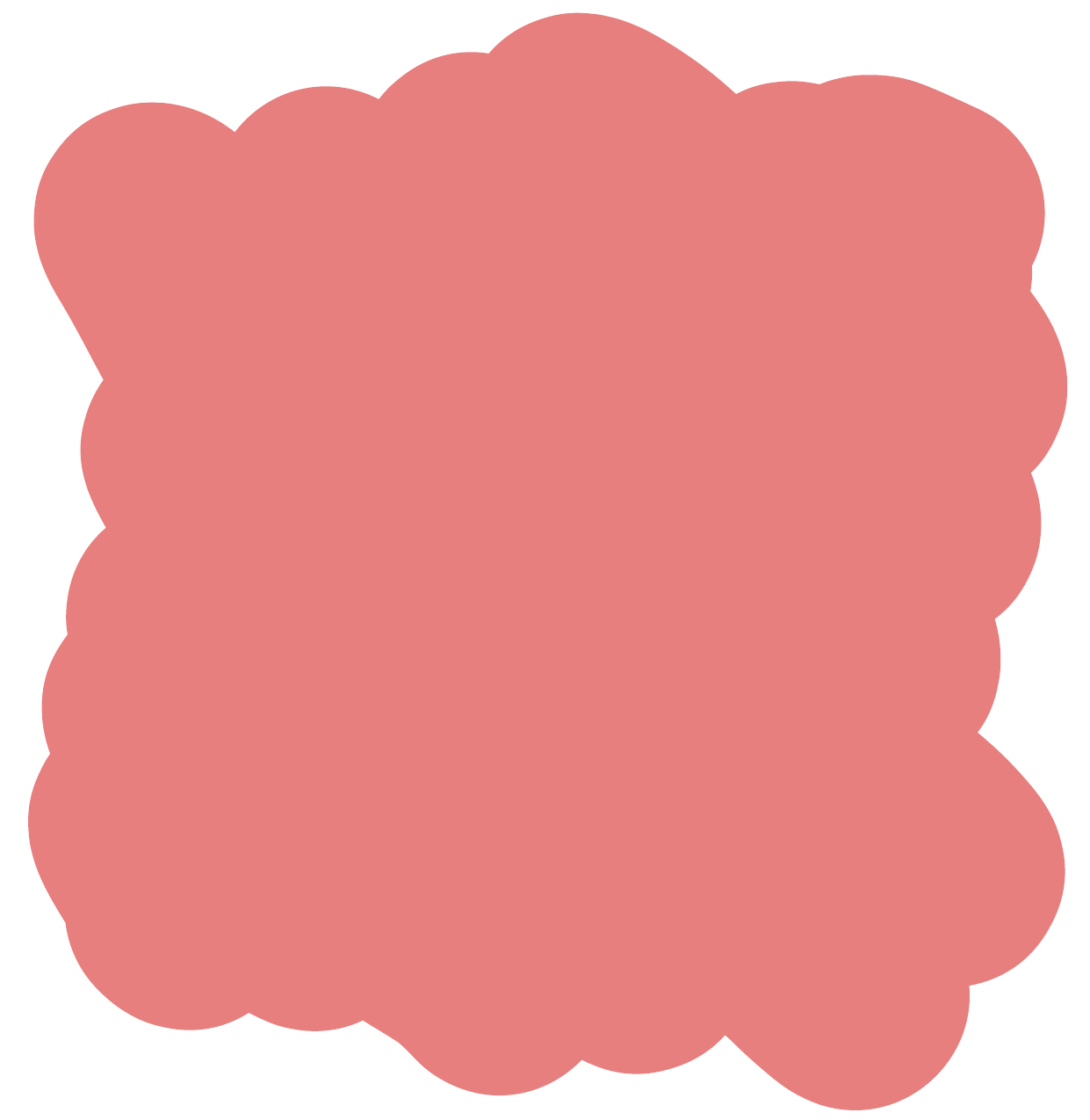}
\pgfdeclareimage[height=0.85cm]{nblack}{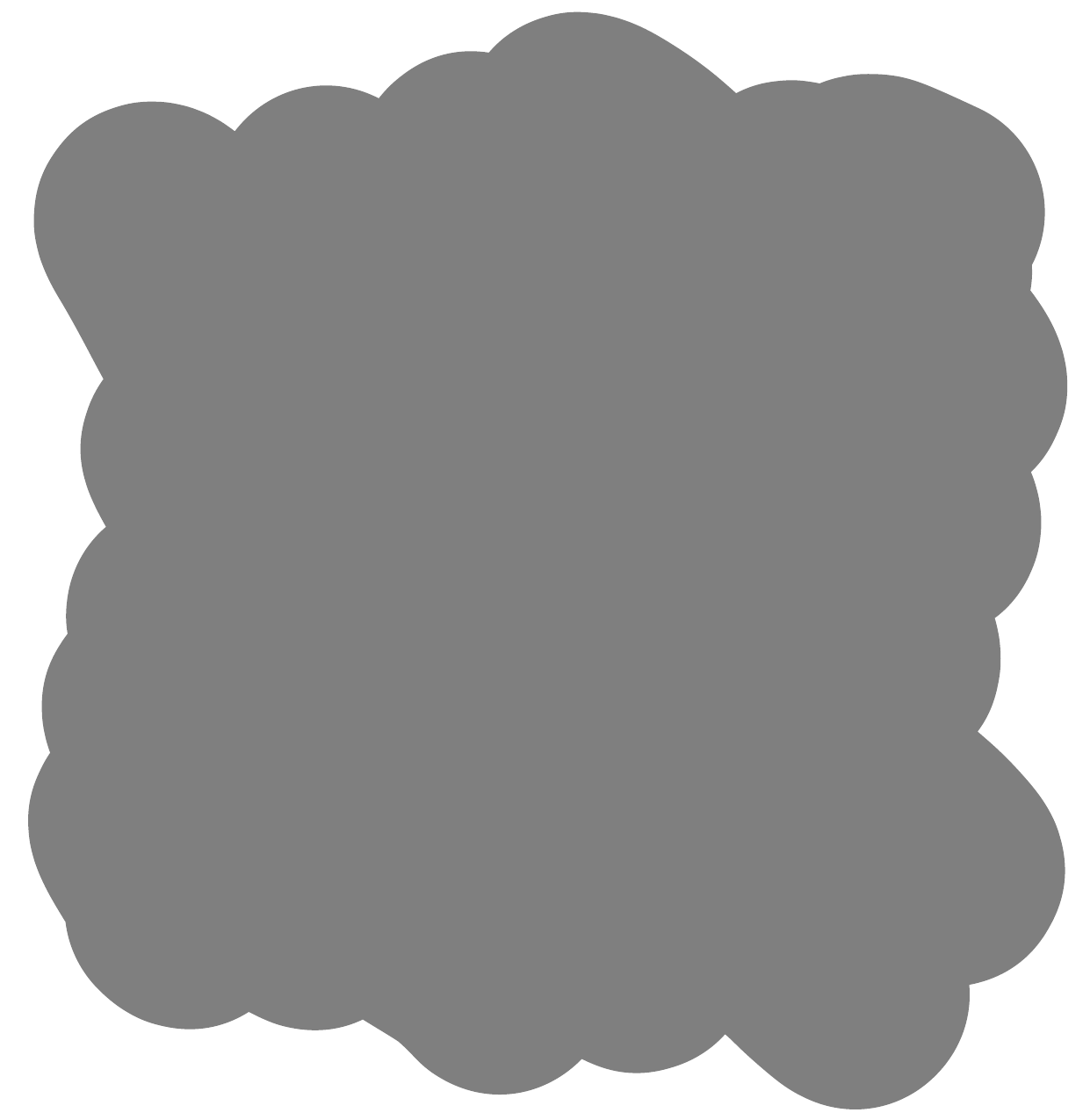}

\pgfdeclareimage[height=0.4cm]{ngreen2}{fig/boot/ngreen.pdf}
\pgfdeclareimage[height=0.4cm]{nblue2}{fig/boot/nblue.pdf}
\pgfdeclareimage[height=0.4cm]{nred2}{fig/boot/nred.pdf}
\pgfdeclareimage[height=0.4cm]{nblack2}{fig/boot/nblack.pdf}

\theoremstyle{definition}
\newtheorem{definition}{Definition}[section]
\newtheorem{theorem}{Theorem}[section]

\makeatletter
\newcommand{\maketitleblind}{\begingroup%
\if1\blind
{
\clearpage\maketitle
\thispagestyle{empty}
}\fi

\if0\blind
{
\begin{center}%
  \let \footnote \thanks
    {\LARGE \@title \par}%
    \vskip 1.5em%
    {\large \@date}%
  \end{center}
  \bigskip
} \fi
\endgroup}
\makeatother

\usepackage[affil-it]{authblk}
\makeatletter
\renewenvironment{abstract}{%
    \if@twocolumn
      \section*{\abstractname}%
    \else %
      \begin{center}%
        {\bfseries \large\abstractname\vspace{\z@}}%
      \end{center}%
      \quotation
    \fi}
    {\if@twocolumn\else\endquotation\fi}
\makeatother

\title{\normalfont \textbf{Cross-temporal probabilistic forecast reconciliation}: \\ methodological and practical issues}

\author[1]{Daniele Girolimetto}
\author[2]{George Athanasopoulos}
\author[1]{Tommaso Di Fonzo}
\affil[1]{\small Department of Statistical Sciences, University of Padua, Padova 35121, Italy}
\author[2]{Rob J Hyndman}
\affil[2]{\small Department of Econometrics \& Business Statistics, Monash University, Clayton VIC 3800, Australia}
\date{27 October 2023}

\begin{document}

\def\spacingset#1{\renewcommand{\baselinestretch}{#1}\small\normalsize}
\spacingset{1.1}

\thispagestyle{empty} \clearpage\maketitleblind

\begingroup
\let\thefootnote\relax\footnotetext{\raggedright Email: \href{mailto:daniele.girolimetto@phd.unipd.it}{daniele.girolimetto@unipd.it} (DG), \href{mailto:george.athanasopoulos@monash.edu}{george.athanasopoulos@monash.edu} (GA), \\ \phantom{Email: }\href{mailto:tommaso.difonzo@unipd.it}{tommaso.difonzo@unipd.it} (TDF), and \href{mailto:rob.hyndman@monash.edu}{rob.hyndman@monash.edu} (RJH)}
\endgroup

\begin{abstract}
\noindent Forecast reconciliation is a post-forecasting process that involves transforming a set of incoherent forecasts into coherent forecasts which satisfy a given set of linear constraints for a multivariate time series. In this paper we extend the current state-of-the-art cross-sectional probabilistic forecast reconciliation approach to encompass a cross-temporal framework, where temporal constraints are also applied. Our proposed methodology employs both parametric Gaussian and non-parametric bootstrap approaches to draw samples from an incoherent cross-temporal distribution.
	To improve the estimation of the forecast error covariance matrix, we propose using multi-step residuals, especially in the time dimension where the usual one-step residuals fail.
	To address high-dimensionality issues, we present four alternatives for the covariance matrix, where we exploit the two-fold nature (cross-sectional and temporal) of the cross-temporal structure, and introduce the idea of overlapping residuals.
	We assess the effectiveness of the proposed cross-temporal reconciliation approaches through a simulation study that investigates their theoretical and empirical properties and two forecasting experiments, using the Australian GDP and the Australian Tourism Demand datasets. For both applications, the optimal cross-temporal reconciliation approaches significantly outperform the incoherent base forecasts in terms of the Continuous Ranked Probability Score and the Energy Score.
	 Overall, the results highlight the potential of the proposed methods to improve the accuracy of probabilistic forecasts and to address the challenge of integrating disparate scenarios while coherently taking into account short-term operational, medium-term tactical, and long-term strategic planning.
\end{abstract}

\begin{itemize}[nosep, align=left, leftmargin = !]
	\item[\textbf{Keywords}] \textit{Forecast reconciliation,~~Linearly constrained multiple time series,~~Cross-temporal, \\ Probabilistic forecasting,~~GDP,~~Tourism flows}
\end{itemize}
%{\it Keywords: Forecast reconciliation, Linearly constrained multiple time series, Cross-temporal, Probabilistic forecasting, GDP, Tourism flows}
\vfill

\newpage
\spacingset{1.3}
%\tableofcontents

%\newpage

\section{Introduction}

Forecast reconciliation is a post-forecasting process intended to improve the quality of forecasts for a system of linearly constrained multiple time series \citep{hyndman2011, panagiotelis2021}. There are many fields where forecast reconciliation is useful, such as when forecasting demand in supply chains with product categories \citep{punia2020, kourentzes2021}, electricity demand and power generation \citep{spiliotis2020, bentaieb2021}, GDP and its components \citep{athanasopoulos2020}, tourist flows across geographic regions and travel purpose \citep{kourentzes2019}, and more. Moreover, effective decision-making depends on the support of accurate and coherent forecasts, making the use of forecast reconciliation methods increasingly popular in recent years \citep{athanasopoulos2023}. 

Temporal reconciliation is another important aspect of forecast reconciliation that can help organizations to better align their forecasting efforts. This approach consists in reconciling forecasts that are generated at different time horizons, such monthly, quarterly or annual. For example, a retail company may need to reconcile monthly forecasts of sales with quarterly forecasts of revenue to ensure that they are aligned and consistent. %\citep{kourentzes2022}

Classical reconciliation approaches (bottom-up, top-down, middle-out, see \citealp{dunn1976}, \citealp{gross1990}, \citealp{athanasopoulos2009}, respectively) addressed the issue of incoherent forecasts in a cross-sectional hierarchy by forecasting only one level and using these to generate forecasts for the remaining series. All of these approaches ignore useful information available at other levels \citep{pennings2017}. Recently, hierarchical forecasting \citep{fliedner2001} has significantly evolved to include modern least squares-based reconciliation techniques in the cross-sectional framework \citep{hyndman2011, wickramasuriya2019, panagiotelis2021}, later extended to temporal hierarchies \citep{athanasopoulos2017, nystrup2020}.
Obtaining coherent forecasts across both the cross-sectional and temporal dimensions (known as \textit{cross-temporal coherence}) has been limited to sequential approaches that address each dimension separately \citep{kourentzes2019, yagli2019, punia2020, spiliotis2020}. Recently, \citet{difonzo2023} suggested a unified reconciliation step that takes into account both the cross-sectional and temporal dimensions, instead of dealing with them separately, utilizing the entire cross-temporal hierarchy.

However, these cross-temporal works focus on point forecasting, and do not consider distributional or probabilistic forecasts \citep{gneiting2014}. In the cross-sectional and temporal frameworks, there have been some developments towards probabilistic forecasting including  \cite{bentaieb2017}, \cite{panamtash2018}, \cite{jeon2019}, \cite{yang2020}, \cite{yagli2020},
\cite{bentaieb2021}, \cite{corani2021}, \cite{corani2022}, \cite{zambon2022} and \cite{wickramasuriya2021b}. \cite{panagiotelis2023} made a significant contribution by formalizing cross-sectional probabilistic reconciliation using the geometric framework for point forecast reconciliation of \cite{panagiotelis2021}. They show how a reconciled forecast can be constructed from an arbitrary base forecast when its density is available and when only a sample can be drawn. They also show that in the case of elliptical distributions, the correct predictive distribution can be recovered via linear reconciliation, regardless of the base forecast location and scale parameters, and derive conditions for this to hold in the special case of reconciliation via projection.

In this paper, we extend cross-sectional probabilistic reconciliation to the cross-temporal case, working on issues related to the two-fold nature of this framework. First, we revise and develop the notation proposed by \cite{difonzo2023} to generalize the work of \cite{panagiotelis2023}. This allows us to move from cross-temporal point reconciliation to a probabilistic setting through the generalization of definitions and theorems well-established in the cross-sectional framework. Second, we propose solutions to draw a sample from the base forecast distribution according to either a parametric approach that assumes Gaussianity or a non-parametric approach that bootstraps the base model residuals. Third, we propose some solutions to specific problems that arise when combining the cross-sectional and temporal dimensions. We propose using multi-step residuals to estimate the relationships between different forecast horizons when we deal with temporal levels, since one-step residuals are not suitable for this purpose. To solve high-dimensionality issues we introduce the idea of overlapping residuals and consider alternative forms for constructing the covariance matrix. Fourth, we propose new shrinkage procedures for reconciliation that aim to identify a feasible cross-temporal structure. The algorithms described in this paper are implemented in the \texttt{FoReco} package \citep{foreco2023} for R \citep{rcoreteam2022}. Furthermore, the online appendix contains complementary materials on methodological and practical issues, and supplementary tables and graphs related to the empirical applications.

The remainder of the paper is structured as follows. In \autoref{sec:not}, we provide a unified notation for the cross-sectional, temporal and cross-temporal point reconciliation. We generalize the cross-sectional definitions and theorems developed by \cite{panagiotelis2023} in \autoref{sec:prob}, and propose both a parametric Gaussian and a non-parametric bootstrap approach to draw a sample from the base forecast distribution. In \autoref{sec:shrtech}, we analyze the structure of the cross-temporal covariance matrix, proposing four alternative forms, and propose shrinkage approaches for reconciliation. In addition, we explore cross-temporal residuals (overlapping and multi-step) looking at their advantages and limitations. %A simulation study is performed in \autoref{sec:mcsim}, to better understand the properties of the methodology.
Two empirical applications using the Australian GDP and the Australian Tourism Demand datasets are considered in Sections \ref{sec:ausgdp} and \ref{sec:vn525}, respectively\footnote{A complete set of results is available at the GitHub repository \githuburl}. Finally, \autoref{sec:conclusion} presents conclusions and a future research agenda on this and other related topics.

\section{Notation and definitions}\label{sec:not}

Let $\yvet_t = [y_{1,t},\dots,y_{i,t},\dots,y_{n,t}]'$ be an $n$-variate linearly constrained time series observed at the most temporally disaggregated level, with a seasonality of period $m$ (e.g., $m = 12$ for monthly data, $m = 4$ for quarterly data, $m = 24$ for hourly data). Suppose that the constraints are expressed by linear equations such that \citep{difonzo2023}
\begin{equation}
	\label{eq:cs_con}
	\Cvet_{cs}\yvet_t = \Zerovet_{(n_a \times 1)}, \qquad t = 1, \;\dots, \;T,
\end{equation}
where $\Cvet_{cs}$ is the $(n_a \times n)$ zero constraints cross-sectional matrix, that can be seen as the coefficient matrix of a linear system with $n_a$ equations and $n$ variables\footnote{\cite{blogH2022} and \cite{giro2022} show that this ‘zero-contrained representation' is more general and computationally efficient.}.

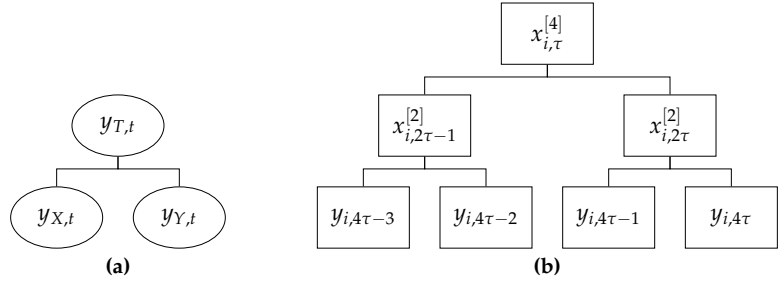
\begin{figure}[t]
\centering
		\resizebox{0.7\linewidth}{!}{\begin{tikzpicture}[baseline=(current  bounding  box.center),
			rel/.append style={shape=ellipse,
				draw=black,
			minimum width=1.5cm,
			minimum height=1cm},
			rel2/.append style={shape=rectangle,
				draw=black,
			minimum width=1.5cm,
			minimum height=1cm},
			connection/.style ={inner sep =0, outer sep =0}]
			
			\node[rel] at (0, 0) (X){$y_{X,t}$};
			\node[rel] at (2, 0) (Y){$y_{Y,t}$};
			\node[rel] at (1, 1.5) (Tt){$y_{T,t}$};
			
			\relation{0.2}{X}{Tt};
			\relation{0.2}{Y}{Tt};

			\node[rel2] at (5, 0) (AA){$y_{i,4\tau-3}$};
			\node[rel2] at (7, 0) (AB){$y_{i,4\tau-2}$};
			\node[rel2] at (9, 0) (BA){$y_{i,4\tau-1}$};
			\node[rel2] at (11, 0) (BB){$y_{i,4\tau}$};
			
			\node[rel2] at (6, 1.5) (A){$x_{i,2\tau-1}^{[2]}$};
			\node[rel2] at (10, 1.5) (B){$x_{i,2\tau}^{[2]}$};
			\node[rel2] at (8, 3) (T){$x_{i,\tau}^{[4]}$};
			
			\node at (1, -0.8) (){\small \textbf{(a)}};
			\node at (8, -0.8) (){\small \textbf{(b)}};
			
			\relation{0.2}{AA}{A};
			\relation{0.2}{AB}{A};
			\relation{0.2}{BA}{B};
			\relation{0.2}{BB}{B};
			
			\relation{0.2}{A}{T};
			\relation{0.2}{B}{T};
		\end{tikzpicture}}
		\caption{\textbf{(a)} A simple two-level cross-sectional hierarchy for 3 time series with $n_a = 1$ and $n_b = 2$. \textbf{(b)} A temporal hierarchy for a quarterly series ($m = 4$ and $\mathcal{K} = \{4,2,1\}$).}
		\label{fig:hierS}
		 \vspace*{-0.5\baselineskip}
\end{figure}

An example is a hierarchical time series where series at upper levels can be expressed by appropriately summing part or all of the series at the bottom level. \autoref{fig:hierS}(a) shows the two-level hierarchical structure for three linearly constrained time series such that $y_{T,t} = y_{X,t} + y_{Y,t}$, $\forall t = 1,...,T$. Now let $\yvet_t = \left[\uvet_t' \quad \bvet_t'\right]'$, where $\uvet_t = [y_{1,t}, \dots, y_{n_a,t}]'$ is the $n_a$-vector of upper levels time series and $\bvet_t = \left[y_{(n_a+1),t}\quad \dots \quad y_{n,t}\right]'$ is the $n_b$-vector of bottom level time series with $n = n_a+n_b$. The upper and lower level time series are connected by the cross-sectional aggregation matrix $\Avet_{cs}$ such that $\uvet_t = \Avet_{cs}\bvet_t$. Following \cite{giro2022}, we can always construct a zero-constraints cross-sectional matrix from the aggregation matrix, $\Cvet_{cs}=\left[\Ivet_{n_a} \quad {-\Avet_{cs}}\right]$, where $\Ivet_{n_a}$ is an identity matrix of dimension $n_a$. Finally, the cross-sectional structural matrix is given by $\Svet_{cs} = \left[\begin{array}{c}
	\Avet_{cs}\\
	\Ivet_{n_b}
\end{array}\right]$, providing the structural representation \citep{hyndman2011} $\yvet_t = \Svet_{cs} \bvet_t$. Considering the hierarchical example in \autoref{fig:hierS}(a), we have
$$
	\Avet_{cs} = \begin{bmatrix} 1 & 1 \end{bmatrix}, \quad \Cvet_{cs} = \begin{bmatrix}1 & -1 & -1 \end{bmatrix} \quad \text{and} \quad \Svet_{cs} = \begin{bmatrix}
		1 & 1 \\
		1 & 0 \\
		0 & 1
	\end{bmatrix}.
$$
In general there is no reason for $\uvet_t$ to be restricted to simple sums of $\bvet_t$; therefore $\Avet_{cs} \in \mathbb{R}^{n_a\times n_b}$ may contain any real values, and not only 0s and 1s.

Considering now the temporal framework, we denote as $\mathcal{K} = \{ k_p , k_{p-1}, \dots, k_2, k_1 \}$ the set of $p$ factors of $m$, in descending order, where $k_1= 1$ and $k_p= m$ \citep{athanasopoulos2017}. For example, for quarterly time series $m = 4$, $p = 3$, and $\mathcal{K} = \{4,2,1\}$.
Given a factor $k$ of $m$, and assuming that $T = N m$ (where $N$ is the length of the most temporally aggregated version of the series), we can construct a temporally aggregated version of the time series of a single variable $\{\yvet_{i,t}\}_{t = 1, \dots, T}$, through the non-overlapping sums of its $k$ successive values, which has a seasonal period equal to $M_k= \displaystyle\frac{m}{k}$: $x_{i,j}^{[k]} = \displaystyle\sum_{t=(j-1)k+1}^{jk} y_{i,t}$, where $j = 1,\dots, N_k$, $i = 1,\dots,n$, $N_k = \displaystyle\frac{T}{k}$ and $x_{i,j}^{[1]}=y_{i,t}$. Define $\tau$ as the observation index of the most aggregate level $k_p$. For a fixed temporal aggregation order $k \in \mathcal{K}$, we stack the observations in the column vector $\xvet_{i,\tau}^{[k]} = \left[x_{i,M_k(\tau-1)+1}^{[k]} \quad x_{i,M_k(\tau-1)+2}^{[k]} \quad \dots \quad x_{i,M_k\tau}^{[k]}\right]',$ and obtain the vector for all the temporal aggregation orders $\xvet_{i,\tau} = \left[x_{i,\tau}^{[k_p]} \quad \xvet_{i,\tau}^{[k_{p-1}]\prime} \quad \dots \quad \xvet_{i,\tau}^{[1]\prime} \right]'$, $\tau = 1,\dots,N$. The structural representation of the temporal hierarchy \citep{athanasopoulos2017} is then $\xvet_{i,\tau} = \Svet_{te}\xvet_{i,\tau}^{[1]}$, where $\Svet_{te} = \left[\begin{array}{c}
	\Avet_{te} \\
	\Ivet_{m}
\end{array}\right]$ is the $[(m+k^\ast) \times m]$ temporal structural matrix, $\Avet_{te} = \left[\Unovet_{k_p} \quad \Ivet_{\frac{m}{k_{p-1}}} \otimes \Unovet_{k_{p-1}} \quad \dots \quad \Ivet_{\frac{m}{k_{2}}}  \otimes \Unovet_{k_2} \right]'$
%$$
%	\Avet_{te} = \begin{bmatrix}
%		\multicolumn{3}{c}{\Unovet_{k_p}'}                       \\
%		\Ivet_{\frac{m}{k_{p-1}}} & \otimes & \Unovet_{k_{p-1}}' \\
%		                          & \vdots  &                    \\
%		\Ivet_{\frac{m}{k_{2}}}   & \otimes & \Unovet_{k_2}'
%	\end{bmatrix}
%$$
is the $(k^\ast \times m)$ temporal aggregation matrix with $k^\ast = \displaystyle\sum_{k \in \mathcal{K}\setminus\{k_1\}} M_k$, the number of upper time series of the temporal hierarchy, $\Unovet_{k_p}$ is a $(k_p \times 1)$ vector of all ones,
and $\otimes$ is the Kronecker product. For each series $x_{i,\tau}$, $i = 1,\dots,n$, we have also the zero-constrained representation
\begin{equation}
	\label{eq:te_con}
	\Cvet_{te}\xvet_{i,\tau} = \Zerovet_{[k^\ast \times (m+k^\ast)]}, \qquad \tau = 1,\dots,N, \qquad i = 1,\dots, n
\end{equation}
where $\Cvet_{te} = [\Ivet_{k^\ast} ~~ {-\Avet_{te}}]$ is the $[k^\ast \times (m+k^\ast)]$ zero constraints temporal matrix. \autoref{fig:hierS}(b) shows the hierarchical representation of a quarterly time series, for which $m = 4$, $\mathcal{K} = \{4,2,1\}$ and
$$
	\Avet_{te} = \begin{bmatrix}
		1 & 1 & 1 & 1 \\
		1 & 1 & 0 & 0 \\
		0 & 0 & 1 & 1 \\
	\end{bmatrix}, \quad \Cvet_{te} = \begin{bmatrix}
		1 & 0 & 0 & -1 & -1 & -1 & -1 \\
		0 & 1 & 0 & -1 & -1 & 0  & 0  \\
		0 & 0 & 1 & 0  & 0  & -1 & -1 \\
	\end{bmatrix} \quad \mathrm{and} \quad \Svet_{te} = \begin{bmatrix}
		\Avet_{te} \\
		\Ivet_4
	\end{bmatrix}.
$$
When we temporally aggregate each series, the cross-sectional constraints for the most temporally disaggregated series \eqref{eq:cs_con} hold for all the temporal aggregation orders such that $\Cvet_{cs}\xvet^{[k]}_j = \Zerovet_{(n_a \times 1)}$, for $k \in \mathcal{K}$ and $j = 1, \dots, N_k$, where $\xvet_j^{[k]} = \left[\uvet_j^{[k]\prime}\quad \bvet_j^{[k]\prime}\right]'$ with $\uvet^{[k]}_j = \left[ x^{[k]}_{1,\;j}\quad \dots\quad x^{[k]}_{n_a,\;j}\right]'$ is the $n_a$-vector of upper time series and $\bvet^{[k]}_j = \left[x^{[k]}_{(n_a+1),\;j}\quad\dots\quad x^{[k]}_{n,\;j}\right]'$ is the $n_b$-vector of bottom time series in the temporal hierarchy.

To include both cross-sectional and temporal constraints at the same time in a unified framework, we stack the series into a $[n \times (m+k^\ast)]$ matrix $\Xvet_\tau$, where we recall that $n$, $m$, and $k^*$ represent respectively the total number of time series, the seasonal period, and the number of upper time series of the temporal hierarchy. The rows and columns represent, respectively, the cross-sectional and the temporal dimension:
$$
	\Xvet_\tau = \begin{bmatrix}
		\xvet_{1,\tau}' \\[-0.1cm]
		\vdots          \\[-0.2cm]
		\xvet_{n,\tau}'
	\end{bmatrix} = \begin{bmatrix}
	\Uvet_{\tau}^{[k_p]} & \Uvet_{\tau}^{[k_p-1]} & \dots & \Uvet_{\tau}^{[1]}\\[0.25cm]
	\Bvet_{\tau}^{[k_p]} & \Bvet_{\tau}^{[k_p-1]} & \dots & \Bvet_{\tau}^{[1]}
	\end{bmatrix},
$$
where for any fixed $k$,
$\Uvet_{\tau}^{[k]}$ is the ($n_a\times N_k$) matrix grouping the upper time series, $\Bvet_{\tau}^{[k]}$ is the ($n_b\times N_k$) matrix grouping the bottom time series. For example, for the cross-temporal structure of \autoref{fig:hierS}, we have
$$
\Xvet_\tau = \left[\begin{array}{c|cc|cccc}
x^{[4]}_{T,\tau} & x^{[2]}_{T,2\tau-1} & x^{[2]}_{T,2\tau} & y_{T,4\tau-3} & y_{T,4\tau-2} & y_{T,4\tau-1} & y_{T,4\tau}\\
\midrule
x^{[4]}_{X,\tau} & x^{[2]}_{X,2\tau-1} & x^{[2]}_{X,2\tau} & y_{X,4\tau-3} & y_{X,4\tau-2} & y_{X,4\tau-1} & y_{X,4\tau}\\
x^{[4]}_{Y,\tau} & x^{[2]}_{Y,2\tau-1} & x^{[2]}_{Y,2\tau} & y_{Y,4\tau-3} & y_{Y,4\tau-2} & y_{Y,4\tau-1} & y_{Y,4\tau}\\
\end{array}\right].
$$
Further, $\Cvet_{cs}\Xvet_\tau = \Zerovet_{\left[n_a \times (m+k^\ast)\right]}$ and $\Cvet_{te}\Xvet_\tau' = \Zerovet_{(k^\ast \times n)} $. We can consider the cross-temporal framework as a generalization of the cross-sectional and temporal frameworks, that simultaneously takes into account both types of constraints. The cross-sectional reconciliation approach proposed by \cite{hyndman2011} can be obtained by assuming $m = 1$, while the temporal one \citep{athanasopoulos2017} is obtained when $n = 1$ (with $n_a = 0$ and $n_b = 1$).

\begin{figure}[!t]
  \centering
\begin{tikzpicture}
\tikzset{square matrix/.style={
    matrix of nodes,
    column sep=-\pgflinewidth, row sep=-\pgflinewidth,
    nodes={draw,
      minimum height=#1,
      anchor=center,
      text width=#1,
      align=center,
      inner sep=0pt,
      font=\tiny
    },
  },
  square matrix/.default=0.5cm
}

\matrix[square matrix, label={[font=\normalsize,text=black]left:{$\Cvet_{ct} = $}}](St)
{
%|[draw=none]|$x^{[4]}_{T,\tau}$ & |[draw=none]|$x^{[2]}_{T,2\tau-1}$ & |[draw=none]|$x^{[2]}_{T,2\tau}$ & |[draw=none]|$x^{[1]}_{T,4\tau-3}$ & |[draw=none]|$x^{[1]}_{T,4\tau-2}$ & |[draw=none]|$x^{[1]}_{T,4\tau-1}$ & |[draw=none]|$x^{[1]}_{T,4\tau}$ & |[draw=none]|$x^{[4]}_{X,\tau}$ & |[draw=none]|$x^{[2]}_{X,2\tau-1}$ & |[draw=none]|$x^{[2]}_{X,2\tau}$ & |[draw=none]|$x^{[1]}_{X,4\tau-3}$ & |[draw=none]|$x^{[1]}_{X,4\tau-2}$ & |[draw=none]|$x^{[1]}_{X,4\tau-1}$ & |[draw=none]|$x^{[1]}_{X,4\tau}$ & |[draw=none]|$x^{[4]}_{Y,\tau}$ & |[draw=none]|$x^{[2]}_{Y,2\tau-1}$ & |[draw=none]|$x^{[2]}_{Y,2\tau}$ & |[draw=none]|$x^{[1]}_{Y,4\tau-3}$ & |[draw=none]|$x^{[1]}_{Y,4\tau-2}$ & |[draw=none]|$x^{[1]}_{Y,4\tau-1}$ & |[draw=none]|$x^{[1]}_{Y,4\tau}$\\
|[fill=white]| & |[fill=white]| & |[fill=white]| & |[fill=black]| & |[fill=white]| & |[fill=white]| & |[fill=white]| & |[fill=white]| & |[fill=white]| &  |[fill=white]| & |[fill=red]| &  |[fill=white]| &  |[fill=white]| &  |[fill=white]| &  |[fill=white]| &  |[fill=white]| &  |[fill=white]| & |[fill=red]| &  |[fill=white]| &  |[fill=white]| &  |[fill=white]| \\
|[fill=white]| & |[fill=white]| & |[fill=white]| & |[fill=white]| & |[fill=black]| & |[fill=white]| & |[fill=white]| & |[fill=white]| & |[fill=white]| &  |[fill=white]| &  |[fill=white]| & |[fill=red]| &  |[fill=white]| &  |[fill=white]| &  |[fill=white]| &  |[fill=white]| &  |[fill=white]| &  |[fill=white]| & |[fill=red]| &  |[fill=white]| &  |[fill=white]| \\
|[fill=white]| & |[fill=white]| & |[fill=white]| & |[fill=white]| & |[fill=white]| & |[fill=black]| & |[fill=white]| & |[fill=white]| & |[fill=white]| &  |[fill=white]| &  |[fill=white]| &  |[fill=white]| & |[fill=red]| &  |[fill=white]| &  |[fill=white]| &  |[fill=white]| &  |[fill=white]| &  |[fill=white]| &  |[fill=white]| & |[fill=red]| &  |[fill=white]| \\
|[fill=white]| & |[fill=white]| & |[fill=white]| & |[fill=white]| & |[fill=white]| & |[fill=white]| & |[fill=black]| & |[fill=white]| & |[fill=white]| &  |[fill=white]| &  |[fill=white]| &  |[fill=white]| &  |[fill=white]| & |[fill=red]| &  |[fill=white]| &  |[fill=white]| &  |[fill=white]| &  |[fill=white]| &  |[fill=white]| &  |[fill=white]| & |[fill=red]| \\
|[fill=black]| & |[fill=white]| & |[fill=white]| & |[fill=red]| & |[fill=red]| & |[fill=red]| & |[fill=red]| & |[fill=white]| & |[fill=white]| &  |[fill=white]| &  |[fill=white]| &  |[fill=white]| &  |[fill=white]| &  |[fill=white]| &  |[fill=white]| &  |[fill=white]| &  |[fill=white]| &  |[fill=white]| &  |[fill=white]| &  |[fill=white]| &  |[fill=white]| \\
|[fill=white]| & |[fill=black]| & |[fill=white]| & |[fill=red]| & |[fill=red]| & |[fill=white]| & |[fill=white]| & |[fill=white]| & |[fill=white]| &  |[fill=white]| &  |[fill=white]| &  |[fill=white]| &  |[fill=white]| &  |[fill=white]| &  |[fill=white]| &  |[fill=white]| &  |[fill=white]| &  |[fill=white]| &  |[fill=white]| &  |[fill=white]| &  |[fill=white]| \\
|[fill=white]| & |[fill=white]| & |[fill=black]| & |[fill=white]| & |[fill=white]| & |[fill=red]| & |[fill=red]| & |[fill=white]| & |[fill=white]| &  |[fill=white]| &  |[fill=white]| &  |[fill=white]| &  |[fill=white]| &  |[fill=white]| &  |[fill=white]| &  |[fill=white]| &  |[fill=white]| &  |[fill=white]| &  |[fill=white]| &  |[fill=white]| &  |[fill=white]| \\
|[fill=white]| & |[fill=white]| & |[fill=white]| & |[fill=white]| & |[fill=white]| & |[fill=white]| & |[fill=white]| & |[fill=black]| & |[fill=white]| &  |[fill=white]| & |[fill=red]| & |[fill=red]| & |[fill=red]| & |[fill=red]| &  |[fill=white]| &  |[fill=white]| &  |[fill=white]| &  |[fill=white]| &  |[fill=white]| &  |[fill=white]| &  |[fill=white]| \\
|[fill=white]| & |[fill=white]| & |[fill=white]| & |[fill=white]| & |[fill=white]| & |[fill=white]| & |[fill=white]| & |[fill=white]| & |[fill=black]| &  |[fill=white]| & |[fill=red]| & |[fill=red]| &  |[fill=white]| &  |[fill=white]| &  |[fill=white]| &  |[fill=white]| &  |[fill=white]| &  |[fill=white]| &  |[fill=white]| &  |[fill=white]| &  |[fill=white]| \\
|[fill=white]| & |[fill=white]| & |[fill=white]| & |[fill=white]| & |[fill=white]| & |[fill=white]| & |[fill=white]| & |[fill=white]| & |[fill=white]| &  |[fill=black]| &  |[fill=white]| &  |[fill=white]| & |[fill=red]| & |[fill=red]| &  |[fill=white]| &  |[fill=white]| &  |[fill=white]| &  |[fill=white]| &  |[fill=white]| &  |[fill=white]| &  |[fill=white]| \\
|[fill=white]| & |[fill=white]| & |[fill=white]| & |[fill=white]| & |[fill=white]| & |[fill=white]| & |[fill=white]| & |[fill=white]| & |[fill=white]| &  |[fill=white]| &  |[fill=white]| &  |[fill=white]| &  |[fill=white]| &  |[fill=white]| &  |[fill=black]| &  |[fill=white]| &  |[fill=white]| & |[fill=red]| & |[fill=red]| & |[fill=red]| & |[fill=red]| \\
|[fill=white]| & |[fill=white]| & |[fill=white]| & |[fill=white]| & |[fill=white]| & |[fill=white]| & |[fill=white]| & |[fill=white]| & |[fill=white]| &  |[fill=white]| &  |[fill=white]| &  |[fill=white]| &  |[fill=white]| &  |[fill=white]| &  |[fill=white]| &  |[fill=black]| &  |[fill=white]| & |[fill=red]| & |[fill=red]| &  |[fill=white]| &  |[fill=white]| \\
|[fill=white]| & |[fill=white]| & |[fill=white]| & |[fill=white]| & |[fill=white]| & |[fill=white]| & |[fill=white]| & |[fill=white]| & |[fill=white]| &  |[fill=white]| &  |[fill=white]| &  |[fill=white]| &  |[fill=white]| &  |[fill=white]| &  |[fill=white]| &  |[fill=white]| &  |[fill=black]| &  |[fill=white]| &  |[fill=white]| & |[fill=red]| & |[fill=red]| \\
};

\node[rectangle,above delimiter=\{, yshift = -0.25cm, label={[font=\footnotesize,text=black, yshift = 0.25cm]above:{$\xvet_{\tau}-$ordered column}}] at (St.north) {\tikz{\path (St-1-1.west) rectangle (St-1-21.east);}};

%\node[fit=(St-22-1)(St-22-8), font=\scriptsize, yshift = -0.25cm, text width=6cm]{$\bvet_{\tau}^{[1]}$};

%\node[fit=(St-1-9)(St-21-9), draw = blue, xshift = 0.6cm, minimum width = 1.3cm, label={[font=\scriptsize,text=blue]above:{$\xvet_{\tau}$}}]{};

%\node[fit=(St-1-1.north west)(St-2-8.south east), label={[font=\normalsize,text=black]above:{$\Svet_{ct}$}}]{};

\end{tikzpicture}
\caption{Visual representation of the zero constraints cross-temporal matrix $\Cvet_{ct}$ defined in (\ref{eq:Cct}) for a system of 3 linearly constrained quarterly time series (see \autoref{fig:hierS}). The four upper rows describe the cross-sectional constraints (one for each quarter), the remaining rows the temporal constraints (one for each of the three time series). Colours legend: 0s in white, 1s in black, -1s in red.}
  \label{fig:Cmatvis}
\end{figure}
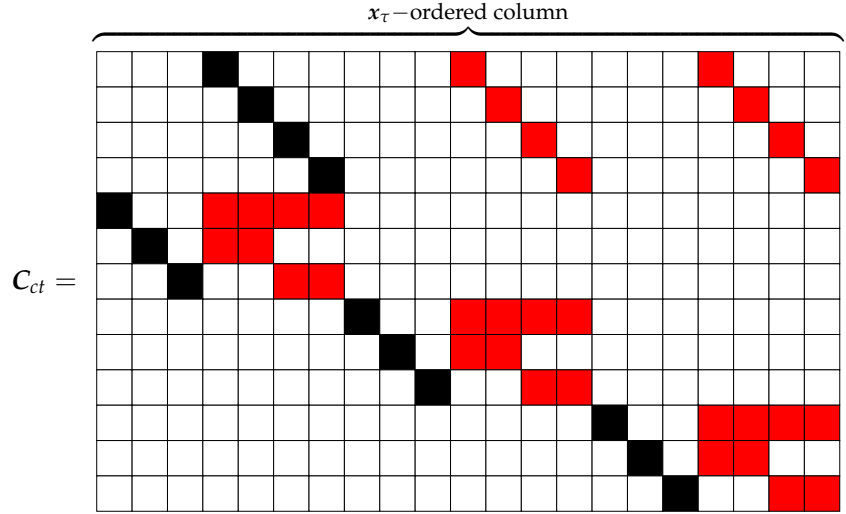

\cite{difonzo2023} show that the cross-temporal constraints working on the complete set of observations corresponding to time period $\tau$ can be expressed in a zero-constrained representation through the full rank $\left[(n_am+nk^\ast)\times n(m+k^\ast)\right]$ zero constraints cross-temporal matrix $\Cvet_{ct}$ such that
\begin{equation}
	\label{eq:Cct}
	\Cvet_{ct} = \begin{bmatrix}
		\Cvet_\ast \\
		\Ivet_n \otimes \Cvet_{te}
	\end{bmatrix} \quad \Longrightarrow \quad
	\Cvet_{ct} \xvet_{\tau} = \Zerovet_{[(n_am+nk^\ast)\times1]} \quad \mathrm{for} \quad \tau = 1,\dots,N,
\end{equation}
where $\xvet_{\tau} = \mathrm{vec}(\Xvet_{\tau}') = [\xvet_{1, \tau}',~ 	\dots, ~ \xvet_{n, \tau}']'$, $\Cvet_\ast = [\Zerovet_{(n_a m\times nk^\ast)} ~~ \Ivet_m \otimes \Cvet_{cs}]\Pvet'$, $\Pvet$ is the commutation matrix \citep[][p. 54]{magnus2019} such that $\Pvet \mathrm{vec}(\Xvet_{\tau}) = \mathrm{vec}(\Xvet_{\tau}')$, and the operator $\mathrm{vec}(\cdot)$ converts a matrix into a vector. \autoref{fig:Cmatvis} shows a visual example for the zero constraints cross-temporal matrix.
A structural representation can be considered as well: $\xvet_\tau = \Svet_{ct}\bvet^{[1]}_\tau = s(\bvet_{\tau}^{[1]})$, where \vspace{-0.1cm}
\begin{equation}
	\label{eq:Sct}
	\Svet_{ct} = \Svet_{cs} \otimes \Svet_{te}
	\vspace{-0.1cm}
\end{equation}
is the $\left[n(k^\ast+m)\times n_b m\right]$ cross-temporal summation matrix, $s: \mathbb{R}^{n_b m} \rightarrow \mathbb{R}^{n(m+k^\ast)}$ is the operator describing the pre-multiplication by $\Svet_{ct}$, and $\bvet^{[1]}_\tau = \mathrm{vec}(\Bvet^{[1]\prime}_{\tau})$. In \autoref{fig:Stilde}, we have represented $\Svet_{ct}$ for a system of 3 linearly constrained quarterly time series (see \autoref{fig:hierS}).
In agreement with \cite{panagiotelis2021}, $\xvet_{\tau}$ lies in an $(n_b m)$-dimensional subspace $\mathfrak{s}_{ct}$ of $\mathbb{R}^{n(k^\ast+m)}$, which we refer to as the \textit{cross-temporal coherent subspace}, spanned by the columns of $\Svet_{ct}$.

\begin{figure}[!t]
  \centering
\begin{tikzpicture}
\tikzset{square matrix/.style={
    matrix of nodes,
    column sep=-\pgflinewidth, row sep=-\pgflinewidth,
    nodes={draw,
      minimum height=#1,
      anchor=center,
      text width=#1,
      align=center,
      inner sep=0pt,
      font=\scriptsize
    },
  },
  square matrix/.default=0.5cm
}

\matrix[square matrix](S)
{
|[fill=black]| & |[fill=black]| & |[draw=none]|$T$\\
|[fill=black]| & |[fill=white]| & |[draw=none]|$X$\\
|[fill=white]| & |[fill=black]| & |[draw=none]|$Y$\\
|[draw=none]|$X$ & |[draw=none]|$Y$ & |[draw=none]|$i$\\
};

\matrix[square matrix, right=of S](K)
{
|[fill=black]| & |[fill=black]| & |[fill=black]| & |[fill=black]| & |[draw=none]|$\quad x^{[4]}_{i,\tau}$\\
|[fill=black]| & |[fill=black]| & |[fill=white]| & |[fill=white]| & |[draw=none]|$\quad x^{[2]}_{i,2\tau-1}$\\
|[fill=white]| & |[fill=white]| & |[fill=black]| & |[fill=black]| & |[draw=none]|$\quad x^{[2]}_{i,2\tau}$\\
|[fill=black]| & |[fill=white]| & |[fill=white]| & |[fill=white]| & |[draw=none]|$\quad y^{\phantom{[1]}}_{i,4\tau-3}$\\
|[fill=white]| & |[fill=black]| & |[fill=white]| & |[fill=white]| & |[draw=none]|$\quad y^{\phantom{[1]}}_{i,4\tau-2}$\\
|[fill=white]| & |[fill=white]| & |[fill=black]| & |[fill=white]| & |[draw=none]|$\quad y^{\phantom{[1]}}_{i,4\tau-1}$\\
|[fill=white]| & |[fill=white]| & |[fill=white]| & |[fill=black]| & |[draw=none]|$\quad y^{\phantom{[1]}}_{i,4\tau}$\\
|[draw=none]| & |[draw=none]| & |[draw=none]| & |[draw=none]| & & \\
};

\node[right of=S,font=\bfseries,xshift = 0.35cm, font = \large] (kro) {$\bigotimes$};
%\node[fit=(K-8-2)(K-8-3), font=\scriptsize, yshift = -0.25cm]{$\yvet_{(\tau)}^{i,[1]}$};

\node[fit=(K-8-2)(K-8-3), font=\scriptsize, yshift = -0.25cm, text width=6cm]{$\xvet_{i,\tau}^{[1]}$}; %\\ \rotatebox{90}{$=\;$} \\$\left[y^{\phantom{[1]}}_{i,4\tau-3}\; y^{\phantom{[1]}}_{i,4\tau-2}\; y^{\phantom{[1]}}_{i,4\tau-1}\; y^{\phantom{[1]}}_{i,4\tau}\right]'$};

\node[right of=S,font=\bfseries,xshift = 4.85cm, font = \huge] (ugu) {$\mathbf{=}$};

\matrix[square matrix, right=of K,xshift = 0.5cm](St)
{
|[fill=black]| & |[fill=black]| & |[fill=black]| & |[fill=black]| & |[fill=black]| & |[fill=black]| & |[fill=black]| & |[fill=black]| & |[draw=none]|$\quad x^{[4]}_{T,\tau}$\\
|[fill=black]| & |[fill=black]| & |[fill=white]| & |[fill=white]| & |[fill=black]| & |[fill=black]| & |[fill=white]| & |[fill=white]| & |[draw=none]|$\quad x^{[2]}_{T,2\tau-1}$\\
|[fill=white]| & |[fill=white]| & |[fill=black]| & |[fill=black]| & |[fill=white]| & |[fill=white]| & |[fill=black]| & |[fill=black]| & |[draw=none]|$\quad x^{[2]}_{T,2\tau}$\\
|[fill=black]| & |[fill=white]| & |[fill=white]| & |[fill=white]| & |[fill=black]| & |[fill=white]| & |[fill=white]| & |[fill=white]| & |[draw=none]|$\quad y^{\phantom{[1]}}_{T,4\tau-3}$\\
|[fill=white]| & |[fill=black]| & |[fill=white]| & |[fill=white]| & |[fill=white]| & |[fill=black]| & |[fill=white]| & |[fill=white]| & |[draw=none]|$\quad y^{\phantom{[1]}}_{T,4\tau-2}$\\
|[fill=white]| & |[fill=white]| & |[fill=black]| & |[fill=white]| & |[fill=white]| & |[fill=white]| & |[fill=black]| & |[fill=white]| & |[draw=none]|$\quad y^{\phantom{[1]}}_{T,4\tau-1}$\\
|[fill=white]| & |[fill=white]| & |[fill=white]| & |[fill=black]| & |[fill=white]| & |[fill=white]| & |[fill=white]| & |[fill=black]| & |[draw=none]|$\quad y^{\phantom{[1]}}_{T,4\tau}$\\
|[fill=black]| & |[fill=black]| & |[fill=black]| & |[fill=black]| & |[fill=white]| & |[fill=white]| & |[fill=white]| & |[fill=white]| & |[draw=none]|$\quad x^{[4]}_{X,\tau}$\\
|[fill=black]| & |[fill=black]| & |[fill=white]| & |[fill=white]| & |[fill=white]| & |[fill=white]| & |[fill=white]| & |[fill=white]| & |[draw=none]|$\quad x^{[2]}_{X,2\tau-1}$\\
|[fill=white]| & |[fill=white]| & |[fill=black]| & |[fill=black]| & |[fill=white]| & |[fill=white]| & |[fill=white]| & |[fill=white]| & |[draw=none]|$\quad x^{[2]}_{X,2\tau}$\\
|[fill=black]| & |[fill=white]| & |[fill=white]| & |[fill=white]| & |[fill=white]| & |[fill=white]| & |[fill=white]| & |[fill=white]| & |[draw=none]|$\quad y^{\phantom{[1]}}_{X,4\tau-3}$\\
|[fill=white]| & |[fill=black]| & |[fill=white]| & |[fill=white]| & |[fill=white]| & |[fill=white]| & |[fill=white]| & |[fill=white]| & |[draw=none]|$\quad y^{\phantom{[1]}}_{X,4\tau-2}$\\
|[fill=white]| & |[fill=white]| & |[fill=black]| & |[fill=white]| & |[fill=white]| & |[fill=white]| & |[fill=white]| & |[fill=white]| & |[draw=none]|$\quad y^{\phantom{[1]}}_{X,4\tau-1}$\\
|[fill=white]| & |[fill=white]| & |[fill=white]| & |[fill=black]| & |[fill=white]| & |[fill=white]| & |[fill=white]| & |[fill=white]| & |[draw=none]|$\quad y^{\phantom{[1]}}_{X,4\tau}$\\
|[fill=white]| & |[fill=white]| & |[fill=white]| & |[fill=white]| & |[fill=black]| & |[fill=black]| & |[fill=black]| & |[fill=black]| & |[draw=none]|$\quad x^{[4]}_{Y,\tau}$\\
|[fill=white]| & |[fill=white]| & |[fill=white]| & |[fill=white]| & |[fill=black]| & |[fill=black]| & |[fill=white]| & |[fill=white]| & |[draw=none]|$\quad x^{[2]}_{Y,2\tau-1}$\\
|[fill=white]| & |[fill=white]| & |[fill=white]| & |[fill=white]| & |[fill=white]| & |[fill=white]| & |[fill=black]| & |[fill=black]| & |[draw=none]|$\quad x^{[2]}_{Y,2\tau}$\\
|[fill=white]| & |[fill=white]| & |[fill=white]| & |[fill=white]| & |[fill=black]| & |[fill=white]| & |[fill=white]| & |[fill=white]| & |[draw=none]|$\quad y^{\phantom{[1]}}_{Y,4\tau-3}$\\
|[fill=white]| & |[fill=white]| & |[fill=white]| & |[fill=white]| & |[fill=white]| & |[fill=black]| & |[fill=white]| & |[fill=white]| & |[draw=none]|$\quad y^{\phantom{[1]}}_{Y,4\tau-2}$\\
|[fill=white]| & |[fill=white]| & |[fill=white]| & |[fill=white]| & |[fill=white]| & |[fill=white]| & |[fill=black]| & |[fill=white]| & |[draw=none]|$\quad y^{\phantom{[1]}}_{Y,4\tau-1}$\\
|[fill=white]| & |[fill=white]| & |[fill=white]| & |[fill=white]| & |[fill=white]| & |[fill=white]| & |[fill=white]| & |[fill=black]| & |[draw=none]|$\quad y^{\phantom{[1]}}_{Y,4\tau}$\\
|[draw=none]| & |[draw=none]| & |[draw=none]| & |[draw=none]| & |[draw=none]| & |[draw=none]| & |[draw=none]| & |[draw=none]| & \\
};

\node[fit=(St-22-1)(St-22-8), font=\scriptsize, yshift = -0.25cm, text width=6cm]{$\bvet_{\tau}^{[1]}$};

\node[fit=(St-1-9)(St-21-9), draw = black, xshift = 0.6cm, minimum width = 1.3cm, label={[font=\scriptsize,text=black]above:{$\xvet_{\tau}$}}]{};

\node[fit=(K-1-5)(K-7-5), draw = black, xshift = 0.5cm, minimum width = 1.2cm, label={[font=\scriptsize,text=black]above:{$\xvet_{i,\tau}$}}]{};

\node[fit=(K-1-1.north west)(K-2-4.south east), label={[font=\normalsize,text=black]above:{$\Svet_{te}$}}]{};

\node[fit=(St-1-1.north west)(St-2-8.south east), label={[font=\normalsize,text=black]above:{$\Svet_{ct}$}}]{};

\node[fit=(S-1-1.north west)(S-2-2.south east), label={[font=\normalsize,text=black]above:{$\Svet_{cs}$}}]{};

\draw[ultra thick,decorate,decoration={calligraphic brace,amplitude=7.5pt}] ($(St-1-9.north)+(1.35,-0.1)$) -- node[right, rotate = -90, anchor = center, align = center, text width = 2.5cm, yshift = 0.75cm, font = \footnotesize] {Temporal components of $T$} ($(St-7-9.south)+(1.35,0.05)$);

\draw[ultra thick,decorate,decoration={calligraphic brace,amplitude=7.5pt}] ($(St-8-9.north)+(1.35,-0.1)$) -- node[right, rotate = -90, anchor = center, align = center, text width = 2.5cm, yshift = 0.75cm, font = \footnotesize] {Temporal components of $X$} ($(St-14-9.south)+(1.35,0.05)$);

\draw[ultra thick,decorate,decoration={calligraphic brace,amplitude=7.5pt}] ($(St-15-9.north)+(1.35,-0.1)$) -- node[right, rotate = -90, anchor = center, align = center, text width = 2.5cm, yshift = 0.75cm, font = \footnotesize] {Temporal components of $Y$} ($(St-21-9.south)+(1.35,0.05)$);
\end{tikzpicture}
\vspace{-0.25cm}
  \caption{Visual representation of the cross-temporal summation matrix $\Svet_{ct} = \Svet_{cs} \otimes \Svet_{te}$ defined in (\ref{eq:Sct}) for a system of 3 linearly constrained quarterly time series (see \autoref{fig:hierS}). Colours legend: 0s in white, 1s in black.}
  \label{fig:Stilde}
\end{figure}
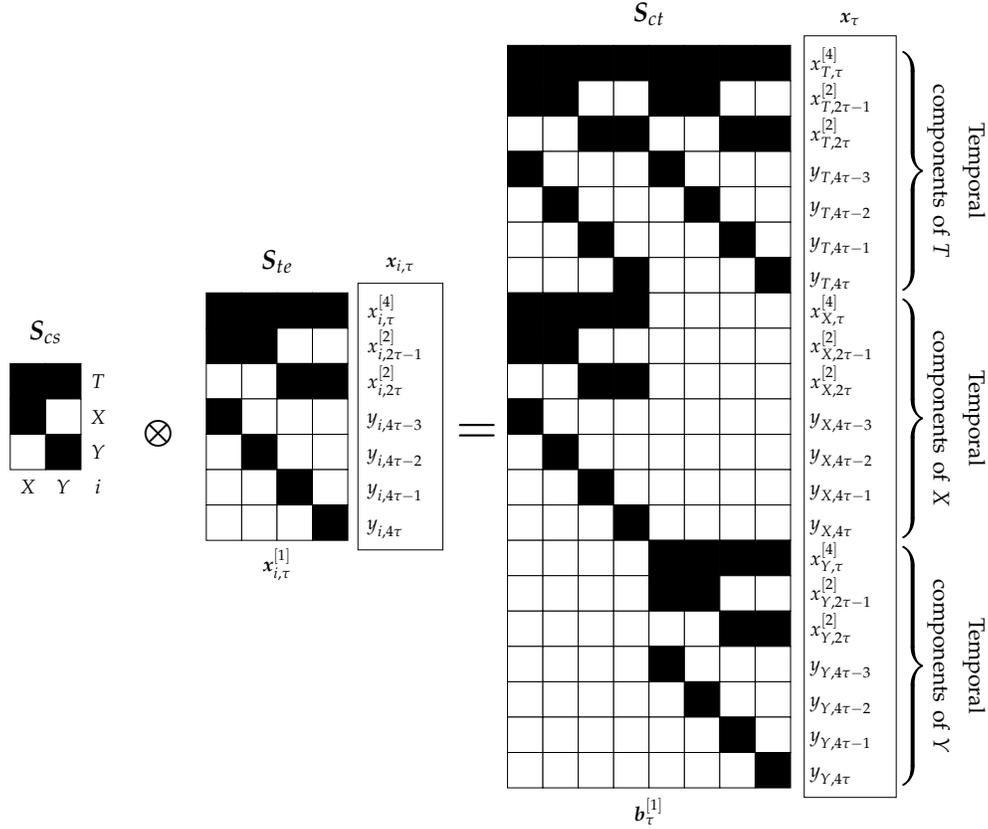

\subsection{Optimal point forecast reconciliation}\label{ssec:oct}
For $h = 1, \dots, H$, let
$$
	\widehat{\Xvet}_{h} = \begin{bmatrix}
		\widehat{\xvet}_{1,h}' \\[-0.1cm]
		\vdots                \\[-0.2cm]
		\widehat{\xvet}_{n,h}'
	\end{bmatrix} =\begin{bmatrix}
		\widehat{\Uvet}_{h}^{[m]} & \dots & \widehat{\Uvet}_{h}^{[k]} & \dots & \widehat{\Uvet}_{h}^{[1]} \\[0.25cm]
		\widehat{\Bvet}_{h}^{[m]} & \dots & \widehat{\Bvet}_{h}^{[k]} & \dots & \widehat{\Bvet}_{h}^{[1]} \\\end{bmatrix},
$$
be the $h$-step ahead base forecasts, where $\widehat{\Uvet}_{h}^{[k]}$ is the ($n_a\times M_k$) matrix grouping the upper time series, $\widehat{\Bvet}_{h}^{[k]}$ is the ($n_b\times M_k$) matrix grouping the bottom time series for a given temporal aggregation order $k$ and $H$ is the forecast horizon for the most temporally aggregated time series. Based on the example in \autoref{fig:hierS} for $H = 1$, we have that
$$
\widehat{\Xvet}_1 = \left[\begin{array}{c|cc|cccc}
\widehat{x}^{[4]}_{T,1} & \widehat{x}^{[2]}_{T,1} & \widehat{x}^{[2]}_{T,2} & \widehat{y}_{T,1} & \widehat{y}_{T,2} & \widehat{y}_{T,3} & \widehat{y}_{T,4}\\
\midrule
\widehat{x}^{[4]}_{X,1} & \widehat{x}^{[2]}_{X,1} & \widehat{x}^{[2]}_{X,2} & \widehat{y}_{X,1} & \widehat{y}_{X,2} & \widehat{y}_{X,3} & \widehat{y}_{X,4}\\
\widehat{x}^{[4]}_{Y,1} & \widehat{x}^{[2]}_{Y,1} & \widehat{x}^{[2]}_{Y,2} & \widehat{y}_{Y,1} & \widehat{y}_{Y,2} & \widehat{y}_{Y,3} & \widehat{y}_{Y,4}\\
\end{array}\right].
$$
The matrix $\widehat{\Xvet}_{h}$, %organized as ${\Xvet}_{\tau}$,
contains incoherent forecasts, such as $\Cvet_{ct} \widehat{\xvet}_{h} \neq \Zerovet_{[(n_am+nk^\ast)\times1]}$
with $h = 1, \dots, H$ and $\widehat{\xvet}_{h} = \mathrm{vec}(\widehat{\Xvet}_{h}')$. In this framework, the definition for forecast reconciliation in the cross-sectional framework given by \cite{panagiotelis2021} can be generalized as follows.

\begin{definition}
	Forecast reconciliation adjusts the base forecast $\widehat{\xvet}_{h}$ by finding a mapping $\psi: \mathbb{R}^{n(m+k^\ast)} \rightarrow \mathfrak{s}$ such that $\widetilde{\xvet}_{h} = \psi\left(\widehat{\xvet}_{h}\right)$, where $\widetilde{\xvet}_{h} \in \mathfrak{s}$ is the vector of the reconciled forecasts.
\end{definition}

For a given forecast horizon $h = 1,\dots, H$, the mapping $\psi$ may be defined as a projection onto $\mathfrak{s}$ given by \citep{panagiotelis2021, difonzo2023}
\begin{equation}
	\label{eq:Mvet}
	\widetilde{\xvet}_{h} = \psi\left(\widehat{\xvet}_h\right) = \Mvet \widehat{\xvet}_h,
\end{equation}
where $\Mvet = \Ivet_{n(m+ k^\ast)} - \Omegavet_{ct}\Cvet'_{ct}\left(\Cvet_{ct}\Omegavet_{ct}\Cvet'_{ct}\right)^{-1}\Cvet_{ct}$, for a positive definite matrix $\Omegavet_{ct}$, and $\widetilde{\xvet}_{h} = \mathrm{vec}(\widetilde{\Xvet}'_{h})$.
\citet{wickramasuriya2019} showed that the minimum variance linear unbiased reconciled forecasts, satisfying the unbiasedness condition $\text{E}(\widetilde{\xvet}_h -\xvet_h) = 0$, has solution \eqref{eq:Mvet} when $\Omegavet_{ct} = \text{Var}(\widehat{\xvet}_h -\xvet_h)$.

Alternatively, the cross-temporal reconciled forecasts $\widetilde{\Xvet}_{h}$ may be found according to the structural approach proposed by \cite{hyndman2011} for the cross-sectional framework, yielding $\widetilde{\xvet}_h = \Svet_{ct}\Gvet \widehat{\xvet}_h$ for some matrix $\Gvet$. \citet{wickramasuriya2019} showed that this leads to a solution equivalent to the cross-temporally reconciled forecasts in \eqref{eq:Mvet}, given by
\begin{equation}\label{eq:SGy}
	\widetilde{\xvet}_{h} = \psi\left(\widehat{\xvet}_h \right) = \left(s \circ g \right)\left(\widehat{\xvet}_h\right)=\Svet_{ct}\Gvet \widehat{\xvet}_{h},
\end{equation}
where $\Gvet = (\Svet_{ct}' \Omegavet_{ct}^{-1}\Svet_{ct})^{-1} \Svet_{ct}'\Omegavet_{ct}^{-1}$,~ and $\Mvet = \Svet_{ct} \Gvet$. In this case, $\psi$ is the composition of two transformations, say $s \circ g$, where $g: \mathbb{R}^{n(m+k^\ast)} \rightarrow \mathbb{R}^{n_b m}$ is a continuous function. In the online appendix A we report some cross-sectional, temporal and cross-temporal approximations for the covariance matrix to be used in \eqref{eq:Mvet} and \eqref{eq:SGy}.

\subsection{Cross-temporal bottom-up forecast reconciliation}\label{ssec:ctbu}

The classic bottom-up approach \citep{dunn1976, dangerfield1992} simply consists in summing-up the base forecasts of the most disaggregated level in the hierarchy to obtain forecasts of the upper-level series. To reduce the computational cost involved in optimal cross-temporal reconciliation, we may be interested in applying a reconciliation along only one dimension (cross-sectional or temporal) and reconstructing the cross-temporal structure using a partly bottom-up approach \citep{difonzo2022b, difonzo2023a, sanguri2022}.

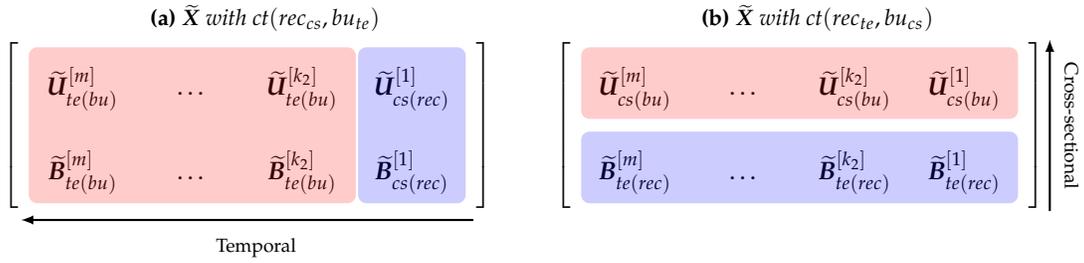
\begin{figure}[!t]
	\centering

	\begin{subfigure}[b]{0.49\textwidth}
		\centering
		\caption{$\widetilde{\Xvet}$ with ct$(rec_{cs}, bu_{te})$}
		\resizebox{\linewidth}{!}{
			\begin{tikzpicture}[>=latex, line width=1pt,
				Matrix/.style={
				matrix of nodes,
				font=\large,
				align=center,
				text width = 1.5cm,
				text height = 0.65cm,
				column sep=2pt,
				row sep=7pt,
				nodes in empty cells,
				left delimiter={[},
						right delimiter={]},
						ampersand replacement=\&
						}]
				\matrix[Matrix] (Mt){ % Matrix contents
				$\widetilde{\Uvet}_{te(bu)}^{[m]}$ \& \dots \& $\widetilde{\Uvet}_{te(bu)}^{[k_2]}$ \& $\widetilde{\Uvet}_{cs(rec)}^{[1]}$ \\
				$\widetilde{\Bvet}^{[m]}_{te(bu)}$ \& \dots \& $\widetilde{\Bvet}^{[k_2]}_{te(bu)}$ \& $\widetilde{\Bvet}^{[1]}_{cs(rec)}$ \\
				};
				\draw[<-, opacity = 0] (Mt.north east)++(0.4,0) coordinate (temp) -- (temp |- Mt.south) node [midway,label={[label distance=0.1cm,rotate=-90, xshift = 1.5mm, font=\footnotesize]Cross-sectional}]{};
				\draw[<-] (Mt.south west)++(0,-0.15) coordinate (temp) -- (temp -| Mt.east) node [midway,label={[label distance=0cm,xshift = 1.5mm, font=\footnotesize]below:Temporal}]{};
				\node[opacity=0.2,
					rounded corners,
					inner sep=0pt, fill = blue, fit=(Mt-1-4)(Mt-2-4)](Bt){};
				\node[opacity=0.2,
					rounded corners,
					inner sep=0pt, fill = red, fit=(Mt-1-1)(Mt-2-3)](At){};
			\end{tikzpicture}}
		\label{fig:tebu}
	\end{subfigure}
	\hfill
	\begin{subfigure}[b]{0.49\textwidth}
		\centering
		\caption{$\widetilde{\Xvet}$ with ct$(rec_{te}, bu_{cs})$}
		\resizebox{\linewidth}{!}{
			\begin{tikzpicture}[>=latex, line width=1pt,
				Matrix/.style={
				matrix of nodes,
				font=\large,
				align=center,
				text width = 1.5cm,
				text height = 0.65cm,
				column sep=2pt,
				row sep=7pt,
				nodes in empty cells,
				left delimiter={[},
						right delimiter={]},
						ampersand replacement=\&
						}]
				\matrix[Matrix] (Mcs){ % Matrix contents
				$\widetilde{\Uvet}_{cs(bu)}^{[m]}$ \& \dots \& $\widetilde{\Uvet}_{cs(bu)}^{[k_2]}$ \& $\widetilde{\Uvet}_{cs(bu)}^{[1]}$ \\
				$\widetilde{\Bvet}^{[m]}_{te(rec)}$ \& \dots \& $\widetilde{\Bvet}^{[k_2]}_{te(rec)}$ \& $\widetilde{\Bvet}^{[1]}_{te(rec)}$ \\
				};
				\draw[<-] (Mcs.north east)++(0.4,0) coordinate (temp) -- (temp |- Mcs.south) node [midway,label={[label distance=0.1cm,rotate=-90, xshift = 1.5mm, font=\footnotesize]Cross-sectional}]{};
				\draw[<-, opacity = 0] (Mcs.south west)++(0,-0.15) coordinate (temp) -- (temp -| Mcs.east) node [midway,label={[label distance=0cm,xshift = 1.5mm, font=\footnotesize]below:Temporal}]{};
				\node[opacity=0.2,
					rounded corners,
					inner sep=0pt, fill = blue, fit=(Mcs-2-1)(Mcs-2-4)](Bcs){};
				\node[opacity=0.2,
					rounded corners,
					inner sep=0pt, fill = red, fit=(Mcs-1-1)(Mcs-1-4)](Acs){};
			\end{tikzpicture}}
		\label{fig:csbu}
	\end{subfigure}
	\vspace{-0.25cm}
	\caption{A visual representation of partly bottom up starting from \eqref{fig:tebu} cross-sectionally reconciled forecasts for the temporal order 1 ($\widetilde{\Uvet}^{[1]}$ and $\widetilde{\Bvet}^{[1]}$) followed by temporal bottom-up, and \eqref{fig:csbu} temporally reconciled forecasts of the cross-sectional bottom time series $(\widetilde{\Bvet}^{[k]}, \, k\in \mathcal{K})$ followed by cross-sectional bottom-up. The \colorbox{mybluehl}{blue} background indicates generating reconciled forecasts along one dimension, while the \colorbox{pink}{pink} background indicates the forecasts obtained using bottom-up along the other.
		}
	\label{fig:bigBU}
	\vspace*{-0.5\baselineskip}
\end{figure}

\autoref{fig:bigBU} provides a visual representation of partly bottom-up in a two-step cross-temporal reconciliation approach. On the left (\autoref{fig:tebu}), we first compute the cross-sectionally reconciled forecasts at the highest frequency ($k = 1$), and then apply temporal bottom-up to obtain coherent cross-temporal forecasts. On the right (\autoref{fig:csbu}), we first compute temporally reconciled forecasts for the most disaggregated cross-sectional level, and then apply the cross-sectional bottom-up. We denote these two-step reconciliation approaches, respectively, as ct$(rec_{te},bu_{cs})$, and ct$(rec_{cs},bu_{te})$, where ‘$rec_{te}$’ and ‘$rec_{cs}$’ denote a forecast reconciliation approach in the temporal and cross-sectional dimensions and, ‘$bu_{cs}$’ and ‘$bu_{te}$’ denote using bottom-up in the cross-sectional and temporal dimensions, respectively. It is worth noting that the simple cross-temporal bottom-up approach corresponds to $\mathrm{ct}(bu_{cs}, bu_{te})=\mathrm{ct}(bu_{te}, bu_{cs})=\mathrm{ct}(bu)$.

\section{Probabilistic forecast reconciliation}\label{sec:prob}

To introduce the idea of coherence and probabilistic forecast reconciliation, we adapt the notations and the formal definitions introduced in \cite{wickramasuriya2021b} and \cite{panagiotelis2023} for the cross-sectional probabilistic case. These definitions can also be generalized to the cross-temporal framework by following the approach developed by \cite{corani2022} for count data. However, in this paper we only focus on the continuous case.

Our aim is to extend these definitions to \textit{cross-temporal coherent probabilistic forecasts} and \textit{cross-temporal probabilistic forecast reconciliation}. Let $(\mathbb{R}^{n_b m}, \mathcal{F}_{\mathbb{R}^{n_b m}}, \nu)$ be a probability space for the bottom time series $\bvet_{\tau}^{[1]}$, where $\mathcal{F}_{\mathbb{R}^{n_b m}}$ is the Borel $\sigma$-algebra on $\mathbb{R}^{n_b m}$. Then a $\sigma$-algebra $\mathcal{F}_{\mathfrak{s}}$ can be constructed from the collection of sets $s(\mathcal{B})$ for all $\mathcal{B} \in \mathcal{F}_{\mathbb{R}^{n_b m}}$.
\begin{definition}[Cross-temporal coherent probabilistic forecasts]
	Given the probability space $(\mathbb{R}^{n_b m}, \mathcal{F}_{\mathbb{R}^{n_b m}}, \nu)$, we define the coherent probability space as the triple $(\mathfrak{s}, \mathcal{F}_{\mathfrak{s}}, \breve{\nu})$ satisfying the following property:
	$\breve{\nu}(s(\mathcal{B}))=\nu(\mathcal{B})$, $\forall \mathcal{B} \in \mathcal{F}_{\mathbb{R}^{n_b m}}$.
\end{definition}
Let $(\mathbb{R}^{n(m+k^\ast)}, \mathcal{F}_{\mathbb{R}^{n(m+k^\ast)}}, \hat{\nu})$ be a probability space referring to the incoherent probabilistic forecast ($\widehat{\xvet}_{h}$) for all the $n$ series in the system at any temporal aggregation order $k \in \mathcal{K}$.
\begin{definition}[Cross-temporal probabilistic forecast reconciliation]\label{def:pfr}
	The reconciled probability measure of $\hat{\nu}$ with respect to $\psi$ is a probability measure $\tilde{\nu}$ on $\mathfrak{s}$ with $\sigma$-algebra $\mathcal{F}_{\mathfrak{s}}$ satisfying
	\begin{equation}\label{eq:pfr}
		\tilde{\nu}(\mathcal{A})=\hat{\nu}(\psi^{-1}(\mathcal{A})), \quad \forall \mathcal{A} \in \mathcal{F}_{\mathfrak{s}},
	\end{equation}
	where $\psi^{-1}(\mathcal{A})=\{x \in \mathbb{R}^{n(m+k^\ast)}: \psi(x) \in \mathcal{A}\}$ denotes the pre-image of $\mathcal{A}$.
\end{definition}
The map $\psi$ may be obtained as the composition $s \circ g$, as for the cross-temporal point reconciliation \eqref{eq:SGy}.

\begin{theorem}[Cross-temporal reconciled samples] \label{thm:rs}
	Suppose that $(\widehat{\xvet}_1, \dots, \widehat{\xvet}_L)$ is a sample drawn from a (cross-temporal) incoherent probability measure $\widehat{\nu}$. Then $(\widetilde{\xvet}_1, \dots, \widetilde{\xvet}_L)$, where $\widetilde{\xvet}_\ell=\psi(\widehat{\xvet}_\ell)$ and $\ell= 1, \dots, L$, is a sample drawn from the (cross-temporal) reconciled probability measure $\widetilde{\nu}$ defined in \eqref{eq:pfr}.
\end{theorem}
\begin{proof}
	See Theorem 4.5 in \cite{panagiotelis2023} using Definition \ref{def:pfr}.
\end{proof}
Theorem \ref{thm:rs} is the cross-temporal extension of Theorem 4.5 in \cite{panagiotelis2023}, valid only for the cross-sectional case. It means that a sample from the reconciled distribution can be obtained by reconciling each member of a sample from the incoherent distribution. With this result, we can separate the mechanism used to generate the base forecasts samples from the reconciliation phase.

\subsection{Parametric framework: Gaussian reconciliation}\label{ssec:prob_pf}

It is possible to obtain a reconciled probabilistic forecast analytically for some parametric distributions, such as the multivariate normal \citep{corani2021, eckert2021, panagiotelis2023, wickramasuriya2021b}. In the cross-sectional framework, \cite{panagiotelis2023} show that, starting from an elliptical distribution for the base forecasts, the reconciled forecast distribution is also elliptical. Using the results shown in \autoref{sec:not}, we extend\footnote{We assume $H =1$ and simplify the notation by removing the $h$ suffix without loss of generality} this results to the cross-temporal case. To obtain a reconciled forecast using the multivariate normal distribution, we start with a base forecast distributed as $\mathcal{N}(\widehat{\xvet}, \Omegavet)$, where $\widehat{\xvet}$ is the mean vector and $\Omegavet$ is the covariance matrix of the base forecasts. Using standard results for the Gaussian case, the reconciled forecast distribution is given by $\mathcal{N}(\widetilde{\xvet}, \widetilde{\Omegavet})$, where
\begin{equation}\label{eq:meanvar}
	\widetilde{\xvet} = \Mvet\widehat{\xvet} \quad \mbox{and} \quad \widetilde{\Omegavet} = \Mvet \Omegavet \Mvet',
\end{equation}
where $\Mvet$ is the projection matrix defined in \eqref{eq:Mvet}.
Note that if we assume that $\Omegavet = \Omegavet_{ct}$ (see the projection matrices in (\ref{eq:Mvet}) and (\ref{eq:SGy})), then the covariance matrix in \eqref{eq:meanvar} simplifies to $\widetilde{\Omegavet} = \Mvet \Omegavet_{ct}$. In the cross-temporal case, sensibly estimating the covariance matrix $\Omegavet$ can be difficult because we need to simultaneously consider both the temporal and cross-sectional structures. This requires many parameters to be estimated, which can be challenging in practice. Additionally, naively using one-step residuals to estimate the cross-temporal correlation structure can lead to an inappropriate estimate of the covariance matrix\footnote{In particular, some temporal covariances are fixed to zero (see the online appendix C for more details).}. These challenges will be explored in more depth in the following sections.

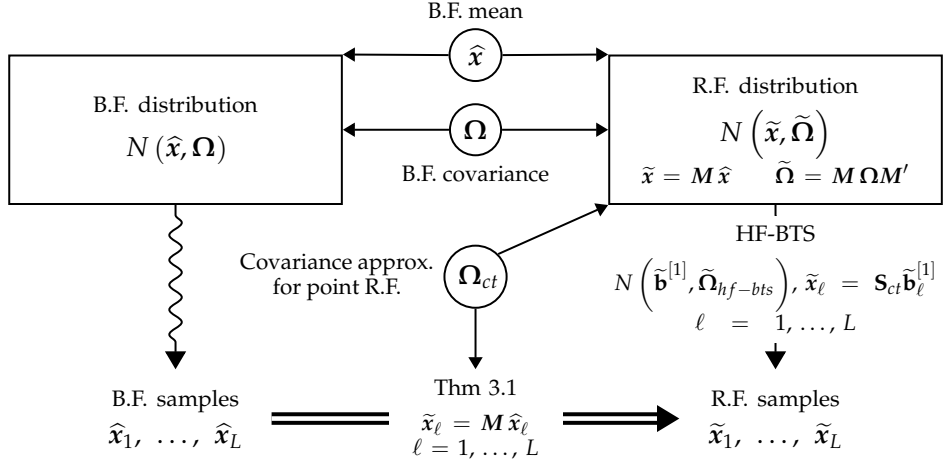
\begin{figure}[!t]
	\centering
	\begingroup
	\spacingset{0.9}
\begin{tikzpicture}

\node[draw, thick, inner sep=1mm, minimum width = 1cm, text width=4.5cm, minimum height = 2.1cm, text centered] (dist) {{\footnotesize R.F. distribution} \\[0.2cm]
%$\Mvet = f(\Omegavet)$\\[0.2cm]
$N\left(\widetilde{\xvet}, \widetilde{\Omegavet} \right)$\\
{\footnotesize $\widetilde{\xvet} = \Mvet\,\widehat{\xvet}\qquad \widetilde{\Omegavet} = \Mvet\,\Omegavet\Mvet'$}};

\node[draw, thick, circle, inner sep=1mm, text width=0.35cm, text centered, minimum width = 0.5cm, minimum height = 0.5cm, left= 1.5cm of dist, label={[align=center, font=\footnotesize]below:{B.F. covariance}}] (cov) {$\Omegavet$};

\node[draw, thick, inner sep=1mm, minimum width = 1cm, text width=4.5cm, minimum height = 2.1cm, text centered, left= 1.5cm of cov] (dist2) {{\footnotesize B.F. distribution} \\[0.2cm] $N\left(\widehat{\xvet}, \Omegavet\right)$};

\node[draw, thick, circle, inner sep=1mm, text width=0.5cm, text centered, minimum width = 0.5cm, minimum height = 0.5cm, below = 1.25cm of cov, label={[align=center, font=\footnotesize]left:{Covariance approx. \\for point R.F.}}] (covapprx) {$\Omegavet_{ct}$};

\node[draw, thick, circle, inner sep=1mm, text width=0.35cm, text centered, minimum width = 0.5cm, minimum height = 0.5cm, above = 0.3cm of cov, label=above:{\footnotesize B.F. mean}] (mean) {$\widehat{\xvet}$};

\node[inner sep=1mm, minimum width = 1cm, text width=2.5cm, text centered, below= 2.5cm of dist2] (sim_base) {{\footnotesize B.F. samples} \\[0.1cm] $\widehat{\xvet}_{1}, \;\dots, \;\widehat{\xvet}_{L}$};

\node[inner sep=1mm, minimum width = 1cm, text width=2.5cm, text centered, below= 2.5cm of dist] (sim_reco) {{\footnotesize R.F. samples} \\[0.1cm] $\widetilde{\xvet}_{1}, \;\dots, \;\widetilde{\xvet}_{L}$};

\draw[-{Triangle[scale=1]}, decorate, thick] (mean.east) -- (dist.north west);
\draw[-{Triangle[scale=1]}, decorate, thick] (cov) -- (dist);
\draw[-{Triangle[scale=1]}, decorate, thick] (mean.west) -- (dist2.north east);
\draw[-{Triangle[scale=1]}, decorate, thick] (cov) -- (dist2);

\draw[-{Triangle[scale=0.5]}, decorate, ultra thick, double, double distance=0.65mm] (sim_base) -- (sim_reco);

\draw[draw=none,fill=none] (sim_base) -- node[fill=white, anchor=center, pos=0.5,font = \footnotesize, text width = 2.2cm, text centered] (nodemid) {Thm \ref{thm:rs} \\[0.15cm] $\widetilde{\xvet}_{\ell} = \Mvet\,\widehat{\xvet}_{\ell}$ \\ $\ell = 1,\, \dots,\, L$} (sim_reco);
%\draw[-{Triangle[scale=1]}, decorate, thick] (sim_base) edge node[ fill=white, anchor=center, pos=0.5,font = \footnotesize, text width = 2.2cm, text centered] (nodemid) {Thm \ref{thm:rs} \\[0.15cm] $\widetilde{\xvet}_{\ell} = \Mvet\,\widehat{\xvet}_{\ell}$ \\ $\ell = 1,\, \dots,\, L$} (sim_reco);

\draw[-{Triangle[scale=1]}, decorate, thick] (covapprx.south) -- (nodemid.north);

\draw [-{Triangle[scale=1.5]}, decorate, decoration={snake,pre length=4pt,post length=15pt}, thick, shorten >= 5pt] (dist2) -- (sim_base);

%\draw [-{Triangle[scale=1.5]}, decorate, decoration={snake,pre length=4pt,post length=15pt}, thick, shorten >= 5pt] (dist) -- (sim_reco) node[midway, right, xshift = 0.25cm, font = \footnotesize, text width = 3cm, text centered] {$\widetilde{B}^{[1]}\sim N\left(\widetilde{\textbf{b}}^{[1]}, \widetilde{\Omegavet}^{[1]}_{b} \right)$ \\ $\widetilde{\textbf{b}}^{[1]}_{h,1}, \, \dots, \, \widetilde{\textbf{b}}^{[1]}_{h,L}$};

\draw [-{Triangle[scale=1.5]}, decorate, decoration={snake,pre length=4pt,post length=15pt}, thick, shorten >= 5pt] (dist) -- (sim_reco)  node[draw=none, midway, font = \footnotesize, text width = 6.25cm, text centered, fill = white, yshift = 0.2cm] {HF-BTS\\[0.2cm]$N\left(\widetilde{\textbf{b}}^{[1]}, \widetilde{\Omegavet}_{hf-bts} \right)$, $\widetilde{\xvet}_{\ell}=\textbf{S}_{ct} \widetilde{\textbf{b}}^{[1]}_{\ell}$\\ $\ell = 1,\, \dots,\, L$} ;
\draw[-{Triangle[scale=1]}, decorate, thick] (covapprx.north east) -- (dist.south west);
\end{tikzpicture}
	\endgroup
	\caption{Overview of cross-temporal forecast reconciliation in the Gaussian framework: two different but equivalent ways of obtaining reconciled forecast samples, as described in Section \ref{ssec:prob_pf}. The acronyms R.F and B.F. stand for Reconciled and Base Forecasts, respectively. HF-BTS stands for High Frequency Bottom Time Series.}
	
	%\caption{Visual {\color{red}\sout{description}} {\color{blue}overview} of cross-temporal forecast reconciliation in the Gaussian framework, as described in Section \ref{ssec:prob_pf}. The acronyms R.F and B.F. stand for Reconciled and Base Forecasts, respectively. HF-BTS stands for High Frequency Bottom Time Series.}
	\label{fig:gaussrel}
	\vspace*{-0.75\baselineskip}
\end{figure}

Focusing on the computational aspect\footnote{We use two R packages to sample from a the base forecast Gaussian distribution: \texttt{MASS} \citep{mass2002} and \texttt{Rfast} \citep{rfast2022} in Sections \ref{sec:ausgdp} and \ref{sec:vn525}, respectively.}, we can take several steps to reduce the time required to obtain simulations from the reconciled forecast distribution. For example when dealing with a genuine hierarchical structure, it is not necessary to simulate from a normal distribution with a defined covariance matrix for the entire structure. Instead, we can utilize the properties of elliptical distributions to simulate from the high frequency bottom time series and then obtain the complete simulation through the $\Svet_{ct}$ matrix. Furthermore, we do not need to calculate the reconciled mean and variance and generate a new sample if we already have a sample from the normal distribution of the base forecasts; we can simply apply the point forecast reconciliation \eqref{eq:Mvet} as outlined in Theorem \ref{thm:rs}. \autoref{fig:gaussrel} shows two different but equivalent ways of obtaining reconciled forecast samples: the former from the base distribution through the Theorem \ref{thm:rs}, and the latter from the reconciled distribution through the high frequency bottom time series forecasts $\widetilde{\bvet}^{[1]}$ only.
The two rectangles represent the base and reconciled forecast distributions, respectively. Enclosed within circles are the distribution parameters involved in the point forecast reconciliation process, transforming $\widehat{\xvet}$ into $\widetilde{\xvet}$ and $\Omegavet$ into $\widetilde{\Omegavet}$. The wave-like arrows represent the simulation processes, generating both base and reconciled forecast samples. Finally, the bold double arrow “$\Rightarrow$" illustrates the generation of the reconciled forecast distributions as described in Theorem \ref{thm:rs}.

\subsection{Non-parametric framework: bootstrap reconciliation}\label{ssec:boot}

Analytical expressions for the base and reconciled forecast distributions are sometimes challenging to obtain. Furthermore parametric assumptions can be restrictive and unrealistic. We propose a procedure called \textit{cross-temporal joint (block) bootstrap} (\textbf{ctjb}) to generate samples from the base forecast distributions that preserve cross-temporal relationships. This approach involves drawing samples of all series simultaneously from the most temporally aggregated level, and using the most temporally aggregated level to determine the corresponding time indices for the other levels.

Let $\widehat{\Evet}^{[k]}$ be the ($n \times N_k$) matrix of the residuals for $k \in \mathcal{K}$. \autoref{fig:res_boot} (on the left) provides a visualization of these matrices and how they are related to each other for the example in \autoref{fig:hierS}. It is assumed that the residuals cover four years ($N=4$): the green color corresponds to the first year, the blue to the second year, and so on. Further, let $\mathcal{M}_i$ be the model used to calculate the base forecasts and residuals for the $i^{th}$ series. %In this work, we assume $\mathcal{M}_i$ to be a univariate model, however nothing prevents the use of multivariate models, perhaps for different temporal levels or for groups of time series.
Assuming $H = 1$, $\tau$ is a random draw with replacement from $1,\dots, N$ and the $\ell^{th}$ bootstrap incoherent sample is
$\widehat{\xvet}_{i,\ell}^{[k]} = f_i(\mathcal{M}_i, \widehat{\evet}_{i}^{[k]})$,
where $f_i(\cdot)$ depends on the fitted %univariate
model $\mathcal{M}_i$. That is, $\widehat{\xvet}_{i,l}^{[k]}$ is a sample path simulated for the $i^{th}$ series with error approximated by the corresponding block bootstrapped sample residual $\widehat{\evet}_{i}^{[k]}$, the $i^{th}$ row of
$$
	\widehat{\Evet}^{[k]}_{\tau} = \begin{bmatrix}
		\widehat{e}^{[k]}_{1,M_k(\tau-1)+1} & \dots  & \widehat{e}^{[k]}_{1,M_k\tau}   \\
		\vdots                              & \ddots & \vdots                          \\
		\widehat{e}^{[k]}_{n,M_k(\tau-1)+1} & \dots  & \widehat{e}^{[k]}_{n,M_k\tau} \
	\end{bmatrix}\qquad k \in \mathcal{K}.
$$
\autoref{fig:res_boot} (on the right) shows $\widehat{\Evet}^{[k]}_{\tau}$ for the quarterly cross-temporal hierarchy in \autoref{fig:hierS}.

\begin{figure}[!t]
	\centering
	\begin{tikzpicture}
\matrix (e1) [matrix of nodes,row sep=0cm,column sep=0cm, nodes= {rectangle, fill=white, inner sep = 1pt, font = {\fontsize{7}{6}\selectfont}, minimum width=1em, minimum height=1em,anchor=center}, label={[xshift = 0.5em]above:{\footnotesize$\widehat{\textbf{E}}^{[1]}$}}]
{
$T$ & \pgfuseimage{ngreen2} & \pgfuseimage{ngreen2} & \pgfuseimage{ngreen2} & \pgfuseimage{ngreen2} & \pgfuseimage{nblue2} & \pgfuseimage{nblue2}  & \pgfuseimage{nblue2} & \pgfuseimage{nblue2} & \pgfuseimage{nred2} & \pgfuseimage{nred2} & \pgfuseimage{nred2} & \pgfuseimage{nred2} & \pgfuseimage{nblack2} & \pgfuseimage{nblack2} & \pgfuseimage{nblack2} & \pgfuseimage{nblack2}\\
$X$ & \pgfuseimage{ngreen2} & \pgfuseimage{ngreen2} & \pgfuseimage{ngreen2} & \pgfuseimage{ngreen2} & \pgfuseimage{nblue2} & \pgfuseimage{nblue2}  & \pgfuseimage{nblue2} & \pgfuseimage{nblue2} & \pgfuseimage{nred2} & \pgfuseimage{nred2} & \pgfuseimage{nred2} & \pgfuseimage{nred2} & \pgfuseimage{nblack2} & \pgfuseimage{nblack2} & \pgfuseimage{nblack2} & \pgfuseimage{nblack2}\\
$Y$ & \pgfuseimage{ngreen2} & \pgfuseimage{ngreen2} & \pgfuseimage{ngreen2} & \pgfuseimage{ngreen2} & \pgfuseimage{nblue2} & \pgfuseimage{nblue2}  & \pgfuseimage{nblue2} & \pgfuseimage{nblue2} & \pgfuseimage{nred2} & \pgfuseimage{nred2} & \pgfuseimage{nred2} & \pgfuseimage{nred2} & \pgfuseimage{nblack2} & \pgfuseimage{nblack2} & \pgfuseimage{nblack2} & \pgfuseimage{nblack2}\\
& 1 & 2 & 3 & 4 & 5 & 6 & 7 & 8 & 9 & 10 & 11 & 12 & 13 & 14 & 15 & 16 \\
};
\draw[decorate,thick, decoration={brace, mirror, amplitude=5pt,raise=-1pt}] (e1-4-2.south west) -- (e1-4-17.south east) node[midway, font = {\fontsize{7}{6}\selectfont}, yshift = -1em]{Quarterly, $t = 1,\dots,16$};

\matrix (ek) [above= 10mm of e1.north east, anchor=south east, matrix of nodes,row sep=0cm,column sep=0cm, nodes= {rectangle, fill=white, inner sep = 1pt, font = {\fontsize{7}{6}\selectfont}, minimum width=1em, minimum height=1em,anchor=center}, label={[xshift = 0.5em]above:{\footnotesize$\widehat{\textbf{E}}^{[2]}$}}]
{
$T$ & \pgfuseimage{ngreen2} & \pgfuseimage{ngreen2} & \pgfuseimage{nblue2} & \pgfuseimage{nblue2} & \pgfuseimage{nred2} & \pgfuseimage{nred2} & \pgfuseimage{nblack2} & \pgfuseimage{nblack2}\\
$X$ & \pgfuseimage{ngreen2} & \pgfuseimage{ngreen2} & \pgfuseimage{nblue2} & \pgfuseimage{nblue2} & \pgfuseimage{nred2} & \pgfuseimage{nred2} & \pgfuseimage{nblack2} & \pgfuseimage{nblack2}\\
$Y$ & \pgfuseimage{ngreen2} & \pgfuseimage{ngreen2} & \pgfuseimage{nblue2} & \pgfuseimage{nblue2} & \pgfuseimage{nred2} & \pgfuseimage{nred2} & \pgfuseimage{nblack2} & \pgfuseimage{nblack2}\\
& 1 & 2 & 3 & 4 & 5 & 6 & 7 & 8 \\
};
\draw[decorate,thick, decoration={brace, mirror, amplitude=5pt,raise=-1pt}] (ek-4-2.south west) -- (ek-4-9.south east) node[midway, font = {\fontsize{7}{6}\selectfont}, yshift = -1em]{Semi-annual, $j = 1,\dots,8$};

\matrix (em) [above= 10mm of e1.north west, anchor=south west, matrix of nodes,row sep=0cm,column sep=0cm, nodes= {rectangle, fill=white, inner sep = 1pt, font = {\fontsize{7}{6}\selectfont}, minimum width=1em, minimum height=1em,anchor=center}, label={[xshift = 0.5em]above:{\footnotesize$\widehat{\textbf{E}}^{[4]}$}}]
{
$T$ & \pgfuseimage{ngreen2} & \pgfuseimage{nblue2} & \pgfuseimage{nred2} & \pgfuseimage{nblack2}\\
$X$ & \pgfuseimage{ngreen2} & \pgfuseimage{nblue2} & \pgfuseimage{nred2} & \pgfuseimage{nblack2}\\
$Y$ & \pgfuseimage{ngreen2} & \pgfuseimage{nblue2} & \pgfuseimage{nred2} & \pgfuseimage{nblack2}\\
& 1 & 2 & 3 & 4 \\
};
\draw[decorate,thick, decoration={brace, mirror, amplitude=5pt,raise=-1pt}] (em-4-2.south west) -- (em-4-5.south east) node[midway, font = {\fontsize{7}{6}\selectfont}, yshift = -1em]{Annual, $\tau = 1,\dots,4$};

\node[right= -0.75em of e1.north east, yshift = 2em,
       anchor=north west, font = {\fontsize{9}{10}\selectfont}] {$\xrightarrow[\text{residuals}]{\text{bootstrapped}}$};
       
\matrix (e1b) [right= 16mm of e1.north east, anchor=north west, matrix of nodes,row sep=0cm,column sep=0cm, nodes= {rectangle, fill=white, inner sep = 1pt, font = {\fontsize{7}{6}\selectfont}, minimum width=1em, minimum height=1em,anchor=center}, label={[name=labe11] above:{\footnotesize$\widehat{\textbf{E}}_1^{[1]}$}}]
{
\pgfuseimage{nblue2} & \pgfuseimage{nblue2} & \pgfuseimage{nblue2} & \pgfuseimage{nblue2}\\
\pgfuseimage{nblue2} & \pgfuseimage{nblue2} & \pgfuseimage{nblue2} & \pgfuseimage{nblue2}\\
\pgfuseimage{nblue2} & \pgfuseimage{nblue2} & \pgfuseimage{nblue2} & \pgfuseimage{nblue2}\\
5 & 6 & 7 & 8 \\
};

\matrix (ekb) [above= 10mm of e1b.north east, anchor=south east, matrix of nodes,row sep=0cm,column sep=0cm, nodes= {rectangle, fill=white, inner sep = 1pt, font = {\fontsize{7}{6}\selectfont}, minimum width=1em, minimum height=1em,anchor=center}, label={[name=labek1] above:{\footnotesize$\widehat{\textbf{E}}_1^{[2]}$}}]
{
\pgfuseimage{nblue2} & \pgfuseimage{nblue2} \\
\pgfuseimage{nblue2} & \pgfuseimage{nblue2} \\
\pgfuseimage{nblue2} & \pgfuseimage{nblue2} \\
3 & 4 \\
};

\matrix (emb) [above= 10mm of e1b.north west, anchor=south west, matrix of nodes,row sep=0cm,column sep=0cm, nodes= {rectangle, fill=white, inner sep = 1pt, font = {\fontsize{7}{6}\selectfont}, minimum width=1em, minimum height=1em,anchor=center}, label={[name=labem1, xshift = 0.25em] above:{\footnotesize$\widehat{\textbf{e}}_1^{[4]}$}}]
{
\pgfuseimage{nblue2} \\
\pgfuseimage{nblue2} \\
\pgfuseimage{nblue2} \\
2 \\
};
\node[draw,inner sep=0mm,label={[name = boot1n, font = {\fontsize{7}{6}\selectfont}] below:{$1^{st}$ sample ($\tau = 2$)}},fit=(e1b) (ekb) (emb) (labek1) (labe11) (labem1)] (boot1) {};

\node[above right= 0.2em of e1b.north east, yshift = 1em,
       anchor=north west] {$\dots$};
       
\matrix (e1bl) [right= 7.5mm of e1b.north east, anchor=north west, matrix of nodes,row sep=0cm,column sep=0cm, nodes= {rectangle, fill=white, inner sep = 1pt, font = {\fontsize{7}{6}\selectfont}, minimum width=1em, minimum height=1em,anchor=center}, label={[name=labe1l] above:{\footnotesize$\widehat{\textbf{E}}_L^{[1]}$}}]
{
\pgfuseimage{nblack2} & \pgfuseimage{nblack2} & \pgfuseimage{nblack2} & \pgfuseimage{nblack2}\\
\pgfuseimage{nblack2} & \pgfuseimage{nblack2} & \pgfuseimage{nblack2} & \pgfuseimage{nblack2}\\
\pgfuseimage{nblack2} & \pgfuseimage{nblack2} & \pgfuseimage{nblack2} & \pgfuseimage{nblack2}\\
13 & 14 & 15 & 16 \\
};

\matrix (ekbl) [above= 10mm of e1bl.north east, anchor=south east, matrix of nodes,row sep=0cm,column sep=0cm, nodes= {rectangle, fill=white, inner sep = 1pt, font = {\fontsize{7}{6}\selectfont}, minimum width=1em, minimum height=1em,anchor=center}, label={[name=labekl] above:{\footnotesize$\widehat{\textbf{E}}_L^{[2]}$}}]
{
\pgfuseimage{nblack2} & \pgfuseimage{nblack2} \\
\pgfuseimage{nblack2} & \pgfuseimage{nblack2} \\
\pgfuseimage{nblack2} & \pgfuseimage{nblack2} \\
7 & 8 \\
};

\matrix (embl) [above= 10mm of e1bl.north west, anchor=south west, matrix of nodes,row sep=0cm,column sep=0cm, nodes= {rectangle, fill=white, inner sep = 1pt, font = {\fontsize{7}{6}\selectfont}, minimum width=1em, minimum height=1em,anchor=center}, label={[name=labeml, xshift = 0.25em] above:{\footnotesize$\widehat{\textbf{e}}_L^{[4]}$}}]
{
\pgfuseimage{nblack2} \\
\pgfuseimage{nblack2} \\
\pgfuseimage{nblack2} \\
4 \\
};
\node[draw,inner sep=0mm,label={[name = boot1n, font = {\fontsize{7}{6}\selectfont}] below:{$L^{th}$ sample ($\tau = 4$)}},fit=(e1bl) (ekbl) (embl) (labekl) (labe1l) (labeml)] (boot1) {};
\end{tikzpicture}
	\caption{Example of bootstrapped residuals for 3 linearly constrained quarterly time series (see \autoref{fig:hierS}). On the left there are the residual matrices with 4 years of data ($N=4$): the green, blue, red and black colors correspond, respectively, to years 1, 2, 3 and 4. On the right the bootstrapped residuals are represented.}
	\label{fig:res_boot}
	\vspace*{-0.75\baselineskip}
\end{figure}
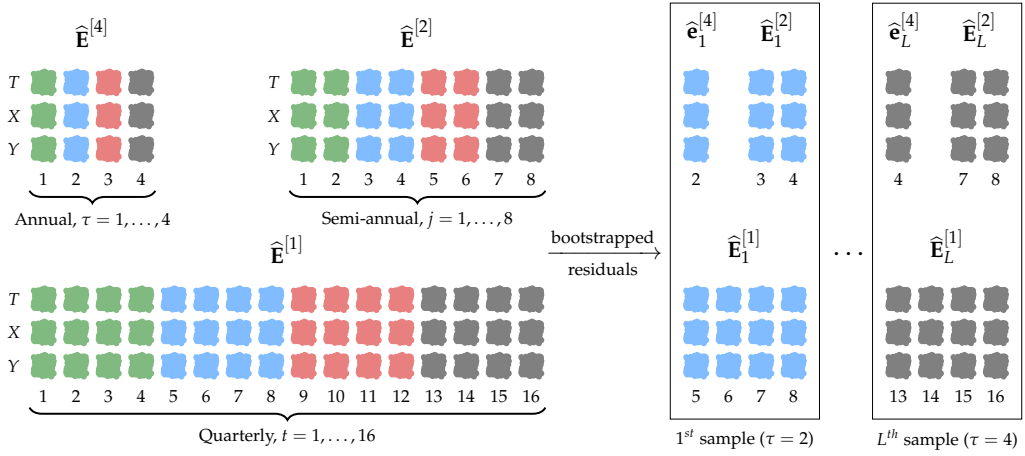

One of the main advantages of the cross-temporal joint bootstrap is that it allows us to accurately account for the dependence between the different levels of temporal aggregation and not only the cross-sectional dependencies. By sampling residuals from the most temporally aggregated level and using it to determine the indices for the other levels, we can ensure that the bootstrap sample reflects the underlying data distribution. Additionally, the cross-temporal joint bootstrap is easy to implement %in R \citep{rcoreteam2022} using the package \texttt{forecast} \citep{Rforecast}
for many forecasting models, making it a practical and efficient tool. Furthermore, this approach is easily scalable in order to utilize multiple computing power simultaneously for each individual series. This can be especially useful when dealing with large datasets or when trying to speed up the analysis process.

\section{Cross-temporal covariance matrix estimation}\label{sec:shrtech}

As the covariance matrix $\Omegavet$ is unknown in practice, a natural estimate is the empirical sample covariance matrix of the base forecasts $\widehat{\Omegavet}$. In this section, our focus will be exclusively on the cross-temporal framework., this means that we have to estimate $r = n(k^\ast+m)[n(k^\ast+m)-1]/2$ different parameters. A possible solution to estimating many parameters when we have fewer observations than $r$, is to construct a shrinkage estimator \citep{efron1975a,efron1975,efron1977}, using a convex combination of $\widehat{\Omegavet}$ and a diagonal target matrix $\widehat{\Omegavet}_D = \widehat{\Omegavet} \odot \Ivet_{n(k^\ast+m)}$, such that $\widehat{\Omegavet}_{G} = \lambda \widehat{\Omegavet}_D + (1-\lambda) \widehat{\Omegavet}$, where $\lambda \in [0,1]$ is the shrinkage intensity parameter that can be estimate using the unbiased estimator proposed by \cite{ledoit2004a} (see \citealp{schafer2005}). The linear combination involving these two matrices is referred to as \textit{Global shrinkage} (\textit{G}), where all off-diagonal elements are shrunk towards zero. $\widehat{\Omegavet}_{G}$ corresponds to the matrix used by the reconciliation approach oct$(shr)$ \citep{difonzo2023}. However, shrinking all off-diagonal elements to zero, when we know that the covariance matrix has a cross-sectional and/or temporal structure, results in information loss. Therefore, we propose to estimate a smaller matrix, and to use the cross-sectional and/or temporal structure to obtain a better estimator for the covariance matrix of the entire system. Given that $\Svet_{ct} = \Svet_{cs} \otimes \Svet_{te}$, it is possible to express the actual covariance matrix in terms of three smaller matrices such that
\begin{equation}\label{eq:OmSct}
\begin{aligned}
\widetilde{\Omegavet} &= \Svet_{ct}\Omegavet_{\textit{hf-bts}}\Svet_{ct}' \\
	&= \left(\Ivet_n \otimes \Svet_{te}\right)\Omegavet_{\textit{hf}}\left(\Ivet_n \otimes \Svet_{te}\right)' \\
	&= \left(\Svet_{cs} \otimes \Ivet_{m+k^\ast}\right)\Omegavet_{bts}\left(\Svet_{cs} \otimes \Ivet_{m+k^\ast}\right)',
\end{aligned}
\end{equation}
where $\Omegavet_{\textit{hf-bts}}$ is the $(n_b m\times n_b m)$ covariance matrix for the bottom time series at temporal aggregation level $k = 1$ (highest frequency bottom time series), $\Omegavet_{\textit{hf}}$ is the $(nm\times nm)$ covariance matrix related to all the high frequency time series and $\Omegavet_{bts}$ is the $[n_b(k^\ast + m)\times n_b(k^\ast + m)]$ covariance matrix related to bottom time series at any temporal aggregation. Equation (\ref{eq:OmSct}) offers three decompositions of the covariance matrix $\widetilde{\Omegavet}$, each characterized by well-defined structures: $\Svet_{ct}$ capturing cross-temporal, $\Ivet_n \otimes \Svet_{te}$ temporal, and $\Svet_{cs} \otimes \Ivet_{m+k^\ast}$ cross-sectional relationships. At the same time, each involves smaller covariance matrices %of reduced dimensions 
as $\Omegavet_{\textit{hf-bts}}$, $\Omegavet_{\textit{hf}}$, and $\Omegavet_{\textit{bts}}$. Starting from these representations, we propose three different approaches (\textit{HB}, \textit{H}, and \textit{B}, respectively) to approximate $\widetilde{\Omegavet}$.

Therefore, we can apply the idea of “Stein-type shrinkage" \citep{efron1977} to $\Omegavet_{\textit{hf-bts}}$, $\Omegavet_{\textit{hf}}$ and $\Omegavet_{\textit{bts}}$ by using the corresponding empirical base forecasts residuals estimation. We obtain the following expressions (see the online appendix B for details):
\begin{itemize}[nosep]
	\item \textit{High frequency Bottom time series shrinkage matrix} (\textit{HB}): \\
	$\widehat{\Omegavet}_{HB} = \lambda \Svet_{ct}\widehat{\Omegavet}_{\textit{hf-bts}, D}\Svet_{ct}'+ (1-\lambda) \Svet_{ct}\widehat{\Omegavet}_{\textit{hf-bts}}\Svet_{ct}'$;
	\item \textit{High frequency shrinkage matrix} (\textit{H}): \\ $\widehat{\Omegavet}_{H}  = \lambda (\Ivet_{n} \otimes \Svet_{te})\widehat{\Omegavet}_{hf, D}(\Ivet_{n} \otimes \Svet_{te})' + (1-\lambda) (\Ivet_{n} \otimes \Svet_{te})\widehat{\Omegavet}_{\textit{hf}}(\Ivet_{n} \otimes \Svet_{te})'$;
	\item \textit{Bottom time series shrinkage matrix} (\textit{B}): \\$\widehat{\Omegavet}_{B} = \lambda \left(\Svet_{cs} \otimes \Ivet_{m+k^\ast}\right)\widehat{\Omegavet}_{bts, D}\left(\Svet_{cs} \otimes \Ivet_{m+k^\ast}\right)' +  (1-\lambda) \left(\Svet_{cs} \otimes \Ivet_{m+k^\ast}\right)\widehat{\Omegavet}_{bts}\left(\Svet_{cs} \otimes \Ivet_{m+k^\ast}\right)'$,
\end{itemize}
where $\widehat{\Omegavet}_{l, D} = \Ivet_{n_b m}\odot\widehat{\Omegavet}_{j}$, $l = \{\textit{hf-bts}, \;\textit{hf}, \;\textit{bts}\}$, and $\lambda$ is the shrinkage parameter.
These matrices are not full rank, meaning their inverses, needed to compute the projection to the coherent subspace, do not exist. To address this, a ridge regularization of the form $\widehat{\Omegavet} + \omega \Ivet$ was used \citep{marquardt1970t}, where $\omega$ is chosen to make the matrix invertible without introducing excessive bias. \autoref{fig:shr_vis} gives some visual insights on the covariance matrices obtainable with $\lambda = 0$ and $\lambda = 1$, respectively, for a simple cross-temporal hierarchical structure with 3 time series and $\mathcal{K} = \{4,2,1\}$ (see \autoref{fig:hierS}).

\begin{table}[!t]
	\centering
	\begingroup
	\spacingset{1.1}
	\begin{tabular}{cccccc}
		\toprule
		\textbf{Method}            & \textbf{\# of different parameters}                             & \textbf{GDP} & \textbf{Tourism}\\
		\midrule
		\addlinespace[0.25cm]
		\textit{G}                          & $r = \displaystyle\frac{n(k^\ast+m)[n(k^\ast+m)-1]}{2}$         & $221\,445$          & $108\,052\,350$ \\
		\addlinespace[0.25cm]
		\textit{B}              & $r_{HB}<\displaystyle\frac{n_b(k^\ast+m)[n_b(k^\ast+m)-1]}{2}<r$ & \makecell{$94\,395$ \\[-0.1cm] {\footnotesize$(57\%)$}}           & \makecell{$36\,231\,328$ \\[-0.1cm] {\footnotesize$(66\%)$}}\\
		\addlinespace[0.25cm]
		\textit{H}               & $r_{HB}<\displaystyle\frac{nm[nm-1]}{2}<r$ & \makecell{$72\,390$ \\[-0.1cm] {\footnotesize$(67\%)$}}           & \makecell{$19\,848\,150$ \\[-0.1cm] {\footnotesize$(82\%)$}}\\
		\addlinespace[0.25cm]
		\textit{HB} & $r_{HB} = \displaystyle\frac{n_bm[n_bm-1]}{2}<r$       & \makecell{$30\,876$ \\[-0.1cm] {\footnotesize$(86\%)$}}           & \makecell{$6\,655\,776$ \\[-0.1cm] {\footnotesize$(94\%)$}}  \\
		\addlinespace[0.1cm]
		\bottomrule
	\end{tabular}
	\endgroup
	\caption{Number of different parameters that need to be estimated for %the Monte Carlo simulation (AR(2), see \autoref{sec:mcsim}),
	the Australian GDP (see \autoref{sec:ausgdp}) and the Australian Tourism Demand (see \autoref{sec:vn525}) forecasting experiments. %: %the first one has $3$ time series (one upper and two bottom) with temporal aggregation $\mathcal{K} = \{2, 1\}$;
	%the first one has $95$ quarterly ($m = 4$ and $k^\ast = 3$) time series ($62$ free and $33$ constraints, see \citealp{giro2022}); the second one has a total of 525 monthly ($m = 12$ and $k^\ast = 16$) time series ($304$ bottom and $221$ upper).
	The percentage reductions in the number of parameters compared to the global approach \textit{G} are reported in parentheses.}
	\label{tab:num_param}
	\vspace*{-0.5\baselineskip}
\end{table}

Another important aspect is the number of parameters to be estimated through the residuals of the base forecasts. In \autoref{tab:num_param} we report the number of different parameters %for the Monte Carlo simulation (AR2) and
for the two forecasting experiment: Australian GDP (see \autoref{sec:ausgdp}) and Australian Tourism Demand (see \autoref{sec:vn525}). In addition, we also calculate the percentage reductions in the number of parameters compared to the global approach. %, that is: $\% \text{ reduction} = 100(1-r_i/r_G)$ with $i \in \{\textit{HB}, H, B\}$.
As we can see, \textit{G} involves a considerably large number of parameters compared to other estimators. \textit{HB} leads to the largest decrease of around 85\%, whereas approaches \textit{H} and \textit{B}  lie somewhere between \textit{G} and \textit{HB}. In general, as $m$ and $n$ increase, using \textit{H} requires the estimation of less parameters than \textit{B}.

It is worth noting that when using the $HB$ covariance matrix, we make the assumption that the base error covariance matrix is coherent. This assumption is valid provided the base forecasts also approximately fulfil constraints (\ref{eq:Cct}), which is expected for any reasonable set of forecasts. %This improve the construction of base forecast distribution in the Gaussian framework that closely resembles the true distribution of the time series.
In addition, with this covariance matrix, the computational complexity of the reconciliation phase is reduced. Specifically, \autoref{thm:HBols} extends %and generalise
Theorem 1 in \cite{hyndman2011}, showing that reconciling using a coherent covariance matrix simplifies to the $ols$ approach.
\begin{theorem}\label{thm:HBols}
	Let $\widehat{\Omega}_{hf-bts}$ be a $[(n_bm)\times (n_bm)]$ p.d. matrix. Then, using $\Omegavet_{ct} = \Svet_{ct}\widehat{\Omega}_{hfbts}\Svet_{ct}'$ in the reconciliation formulae (\ref{eq:Mvet}) and (\ref{eq:SGy}) is equivalent to using $\Omegavet_{ct} = \Ivet_{n(m+k^\ast)}$ ($ols$ approach).
\end{theorem}
\begin{proof}
	See online appendix B.
\end{proof}

In the forecasting experiments that follow (and in the simulation in the online appendix C), we closely analyze these different constructions with a dual purpose. In particular, we use the full covariance matrix ($\lambda = 0$) of the base forecasts to obtain base forecast samples of the linearly constrained time series under Gaussianity. We also use the shrinkage versions as approximations of the covariance matrix to be used for reconciliation (excluding HB, see \autoref{thm:HBols}). This will allow us to better understand the properties and abilities of each parameterization.

\begin{figure}[!t]
	\centering
	\includegraphics[width = \linewidth]{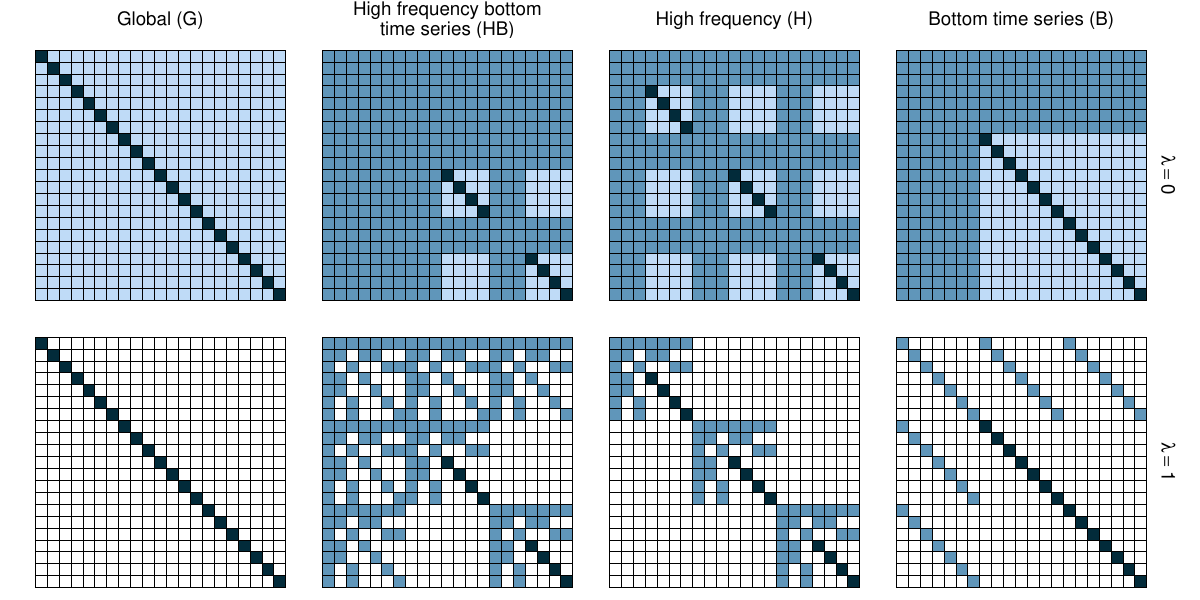}
	\caption{Representation of four types of covariance matrices that can be obtained from the cross-temporal hierarchical structure (example based on the quarterly series of Figure 1) for two different values of $\lambda \in \{0, 1\}$, the shrinkage parameter. The entries in black are not modified by shrinkage, the entries in light blue are those actively involved in the shrinkage phase, while the entries in darker blue are derived directly from the cross-sectional and/or temporal structure and hence not estimated. Additionally, for $\lambda = 1$, the white entries correspond to a zero value.}
	\label{fig:shr_vis}
\end{figure}

\subsection{Multi-step residuals} \label{ssec:multi_res}

Model residuals may be used to estimate the covariance matrix in cross-temporal forecast reconciliation. In time series analysis, it is common to use residuals corresponding to one-step ahead forecasts. However, due to the temporal dimension in our setting, residuals corresponding to different forecast horizons are required. Thus, we define \textit{multi-step residuals} as $e_{i,h,j}^{[k]} = x_{i,j+h}^{[k]} - \widehat{x}_{i,j+h|j}^{[k]}$, where $i = 1,\dots,n$, $j = 1,\dots,N_k$ and $\widehat{x}_{i,j+h|t}^{[k]}$ is the $h$-step fitted value, calculated as the $h$-step-ahead forecast using data up to time $j$. In general, these residuals will be autocorrelated except when $h=1$.

Following \cite{difonzo2023}, we use a matrix organization of the residuals similar to the one for the base forecasts in \autoref{ssec:oct}. Specifically, let $N$ be the total number of observations for the most temporally aggregate time series. Then, the $N_k$-vectors of multi-step residuals for the temporal aggregation $k$ and the series $i$, $\evet_{i,h}^{[k]} =  \Big[e_{i,h,1}^{[k]} \quad e_{i,h,2}^{[k]} \quad \dots \quad e_{i,h,N_k}^{[k]}\Big]'$ with $h = 1,\dots, M_k$, can be organized in matrix form as
$$
	\Evet_i^{[k]} = \begin{bmatrix}
		e_{i,1,1}^{[k]}                     & e_{i,2,2}^{[k]}                     & \dots & e_{i,M_k,M_k}^{[k]} \\
		\vdots                            & \vdots                            &       & \vdots                  \\
		e_{i,1,N_k - M_k + 1}^{[k]} & e_{i,2,N_k - M_k + 2}^{[k]} & \dots & e_{i,M_k,N_k}^{[k]}         \\
	\end{bmatrix}.
$$
Let $\Evet_i = \Big[\Evet_i^{[m]} \quad \Evet_i^{[k_p-1]} \quad \dots \quad \Evet_i^{[1]}  \Big]$. Then the $[N \times n(m+k^\ast)]$ cross-temporal residual matrix is given by $
	\Evet = \Big[\Evet_1 \quad \Evet_2 \quad \dots \quad \Evet_n \Big]$.

To better understand the properties of the proposed alternatives, a simulation study was performed (the results are shown in the online appendix C). We have studied the effect of combining cross-sectional and temporal aggregations using a simple hierarchy, where the small size and nature of the data generating process make it possible to exactly calculate the true cross-temporal covariance structure, thus providing insights into the nature of the time series data involved in the forecast reconciliation process. We find that simulating base forecasts from multi-step residuals allows for a more accurate estimation of the covariance matrix and that reconciliation further improves the forecast accuracy. %of these estimates: accurate base forecasts for $k = 1$ assist the good performance for bottom-up and optimal cross-temporal reconciliation approaches, such as oct(wlsv) and oct(bdshr), which perform well in terms of both CRPS and ES (for details, see Appendix C).

\subsection{Overlapping residuals}\label{ssec:over_res}

Another issue that arises in the case of cross-temporal reconciliation is the low number of available residuals, especially for the higher orders of temporal aggregation. A possible solution is to use residuals calculated using overlapping series by allowing the year to have a varying starting time. To better explain how to calculate overlapping residuals, assume we have a single series $\yvet = [y_1 \; y_2 \; y_3 \; \dots\; y_{T-1}\; y_{T}]'$. We can construct $k$ non overlapping series such that $\xvet^{[k], s} = \left\{x^{[k],s}_{j}\right\}_{j = 1}^{N_k-s}$ where $x^{[k],s}_{j} = \displaystyle\sum_{t = (j-1)k+s+1}^{jk-s} y_t$, with $s = 0, \dots, (k-1)$. For example, suppose we have a biannual series with $k = 2$ and $T = 6$, then we can construct two annual time series depending on which time is deemed the start of the year: $\xvet^{[2], 0} =  \Big[x_1^{[2], 0}\quad x_2^{[2], 0}\quad x_{3}^{[2], 0} \Big]' =\Big[y_1 + y_2\quad y_3 + y_4\quad y_5 + y_6\Big]'$ and $\xvet^{[2], 1} = \Big[x_1^{[2], 1}\quad x_2^{[2], 1} \Big]' = \Big[y_2 + y_3\quad  y_4 + y_5 \Big]'$. To calculate overlapping residuals, we propose the following steps:
\begin{enumerate}[nosep]
	\item Fit a model to $\xvet^{[k], 0}$ (i.e., select an appropriate model and estimate the model parameters using the available data) and calculate the residuals.
	\item Apply the same model in step 1 to $\xvet^{[k], s}$ for $s = 1, \dots, k-1$, without re-estimating the parameters, and calculate the residuals.
\end{enumerate}

The resulting residuals can be used to estimate the covariance matrix in cross-temporal forecast reconciliation. This increases the number of available residuals, particularly when working with higher frequency observations such as monthly or daily data.
It is important to note that this approach assumes that the model used in step 1 is appropriate for all the different series $\xvet^{[k], s}$. Some seasonal models will not be appropriate as the seasonal pattern will be shifted for different values of $s$. However, this will not affect seasonal ARIMA models as the seasonality is defined in terms of lags which are unaffected by the value of $s$.

\section{Forecasting Australian GDP}\label{sec:ausgdp}

The Australian Quarterly National Accounts (QNA) dataset has been widely studied in the literature on forecast reconciliation \citep{athanasopoulos2020, difonzo2023}. Building on these results, we now consider cross-temporally reconciled probabilistic forecasts.

We use univariate ARIMA models\footnote{We use the \texttt{auto.arima} function from the R package \texttt{forecast} \citep{Rforecast}.} to obtain quarterly base forecasts for the $n = 95$ QNA time series, spanning the period 1984:Q4 -- 2018:Q1, defining GDP from both the Income and Expenditure sides. We perform a rolling forecast experiment with an expanding window: the first training sample spans the period 1984:Q4 to 1994:Q3, and the last ends in 2017:Q1, for a total of 91 forecast origins. For the temporal aggregation dimension we aggregate the quarterly data to both semi-annual and annual. We obtain $4$-step, $2$-step and $1$-step ahead base forecasts respectively from the quarterly, semi-annual and annual frequencies, i.e., $\mathcal{K} = \{4,2,1\}$.

\begin{table}[!t]
	\centering
	%{\small
	\begin{tabular}{M{0.17\linewidth}|L{0.77\linewidth}}
		\toprule
		\textbf{Label} & \textbf{Description} \\
		\midrule
		ct$(shr_{cs}, bu_{te})$ & Partly bottom-up (\autoref{ssec:ctbu}) starting from cross-sectional reconciled forecasts using the $shr$ approach.\\
		\addlinespace[0.15cm]
		oct$(\;\cdot\;)$ & Optimal cross-temporal reconciliation for the $struc$, $wlsv$ and $bdshr$ approaches. One-step residuals were used with $wlsv$ and $bdshr$. \\
		\addlinespace[0.15cm]
		oct$_h(\;\cdot\;)$ & Optimal cross-temporal reconciliation with multi-step residuals (see \autoref{ssec:multi_res}) for the approaches presented in \autoref{sec:shrtech}: $hshr$ for \textit{High frequency shrinkage}, and $bshr$ for \textit{bottom time series shrinkage}.\\
		\addlinespace[0.15cm]
		oct$_o(\;\cdot\;)$ & Optimal cross-temporal reconciliation with overlapping residuals (see \autoref{ssec:over_res}) for the $wlsv$ and $bdshr$ approaches. \\
		\addlinespace[0.15cm]
		oct$_{oh}(hshr)$ & Optimal cross-temporal reconciliation with overlapping and multi-step residuals (see Section \ref{ssec:multi_res} and \ref{ssec:over_res}) for the $hshr$ (\textit{High frequency shrinkage}) approach presented in \autoref{sec:shrtech}.\\
		\bottomrule
	\end{tabular}%}
	\caption{Cross-temporal reconciliation approaches for %the Monte Carlo simulation (see \autoref{sec:mcsim}),
	the Australian GDP (see \autoref{sec:ausgdp}) and the Australian Tourism Demand (see \autoref{sec:vn525}) forecasting experiments. All the reconciliation procedures are available in \texttt{FoReco} \citep{foreco2023}.}
	\label{tab:notation}
	\vspace*{-0.5\baselineskip}
\end{table}

The base forecast samples in the Gaussian case are obtained using the sample covariance matrices with the \textit{Global} (G) and \textit{High frequency} (H) parameterization (\autoref{sec:shrtech}), since it is not possible to identify a unique representation for the other cases\footnote{When simultaneously considering Income and Expenditure sides hierarchies, the result is a general linearly constrained time series, where bottom and upper time series are not uniquely defined, making unfeasible the cross-sectional bottom-up reconciliation approach \citep{giro2022}.}. We compare the results obtained using multi-step residuals with and without overlapping, in order to measure the benefit of obtaining overlapping residuals. In the non-parametric case, we use the cross-temporal joint bootstrap (ctjb) presented in \autoref{ssec:boot}. Finally, to reconcile the resulting (1000) base forecasts samples, we have applied the following techniques\footnote{The results with shrunk covariance matrices are available in the online appendix D.2, where we also report the results obtained using other reconciliation approaches.} (see \autoref{tab:notation}): ct$(shr_{cs}, bu_{te})$, ct$(wls_{cs}, bu_{te})$, oct$_o(wlsv)$, oct$_o(bdshr)$, and oct$_{oh}(hshr)$.

The accuracy of the probabilistic forecasts is evaluated using the Continuous Ranked Probability Score (CRPS, \citealp{matheson1976ms, gneiting2014}), %given by
%\begin{equation}\label{eq:crps}
%	\operatorname{CRPS}(\widehat{P}_i, z_i)=\frac{1}{L} \sum_{l=1}^{L}\left|x_{i,l}-z_i\right|-\frac{1}{2 L^{2}} \sum_{l=1}^{L} \sum_{j=1}^{L}\left|x_{i,l}-x_{i,j}\right|, \quad i = 1,\dots,n,
%\end{equation}
%where $\widehat{P}_i(\omega)=\displaystyle\frac{1}{L} \sum_{l=1}^{L} \mathbf{1}\left(x_{i,l} \leq \omega\right)$, $\xvet_{1}, \xvet_{2}, \dots, \xvet_{L}\in \mathbb{R}^{n}$ is a collection of $L$ random draws from the predictive distribution and $\zvet \in \mathbb{R}^{n}$ is the observation vector. CRPS
which is an index that considers the single series and provides us a marginal evaluation of the approaches. In addition, we employ the Energy Score (ES, \citealp{gneiting2014}), that is the CRPS extension to the multivariate case, to evaluate the forecasting accuracy for the whole system \citep{panagiotelis2023, wickramasuriya2021b}.
%\begin{equation}\label{eq:es}
%	\operatorname{ES}(\widehat{P}, \zvet)=\frac{1}{L} \sum_{l=1}^{L}\left\|\xvet_{l}-\zvet\right\|_{2}-\frac{1}{2(L-1)} \sum_{i=1}^{L-1}\left\|\xvet_{l}-\xvet_{l+1}\right\|_{2}
%\end{equation}
%where $	\lVert \cdot \rVert_2$ is the L$_2$ norm.
In particular, we consider the geometric mean of the relative CRPS \citep{fleming1986}, and the relative ES:
\begin{equation}\label{eq:skill}
	\operatorname{\overline{RelCRPS}}_{j,s}^{[k]} = \left(\prod_{i = 1}^n \frac{CRPS^{[k]}_{i, j, s}}{CRPS^{[k]}_{i, 0, 0}}\right)^{\frac{1}{n}} \qquad \mathrm{and} \qquad \operatorname{RelES}_{j,s}^{[k]} = \frac{ES^{[k]}_{j, s}}{ES^{[k]}_{0, 0}},
\end{equation}
where $j$ denotes the reconciliation approach and $s$ indicates the approach used to simulate the base forecasts. As a reference approach ($s=0$ and $j=0$), we consider the base forecasts produced by the Bootstrap approach. If we consider all the temporal aggregation orders (i.e. $\forall k \in \mathcal{K}$), the overall accuracy indices are given by, respectively,
\begin{equation}\label{eq:skill_all}
	\operatorname{\overline{RelCRPS}}_{j,s} = \left(\prod_{\substack{i = 1, \dots, n \\ k \in \mathcal{K}}}\frac{CRPS^{[k]}_{i, j, s}}{CRPS^{[k]}_{i, 0, 0}}\right)^{\frac{1}{n(k^\ast+m)}}\mbox{and } \operatorname{\overline{RelES}}_{j,s}= \left(\prod_{k \in \mathcal{K}}\frac{ES^{[k]}_{j, s}}{ES^{[k]}_{0, 0}}\right)^{\frac{1}{(k^\ast+m)}}.
\end{equation}
%and
%\begin{equation}\label{eq:skillES_all}
%	\operatorname{AvgRelES}_{j,s}= \left(\prod_{k \in \mathcal{K}}\frac{ES^{[k]}_{j, s}}{ES^{[k]}_{0, 0}}\right)^{\frac{1}{(k^\ast+m)}}.
%\end{equation}

\subsection{Results}\label{ssec:ausresults}
\begin{table}[!t]
	\centering
	\begingroup
	\spacingset{1}
	\fontsize{9}{11}\selectfont
	
\begin{tabular}[t]{l|>{}cccc>{}c|ccccc}
\toprule
\multicolumn{1}{c}{\textbf{}} & \multicolumn{10}{c}{\textbf{Base forecasts' sample approach}} \\
\cmidrule(l{0pt}r{0pt}){2-11}
\multicolumn{1}{c}{\makecell[c]{\bfseries Reconciliation\\\bfseries approach}} & \multicolumn{1}{c}{ctjb} & \multicolumn{4}{c}{\makecell[c]{Gaussian approach\textsuperscript{*}}} & \multicolumn{1}{c}{ctjb} & \multicolumn{4}{c}{\makecell[c]{Gaussian approach\textsuperscript{*}}} \\
\multicolumn{1}{c}{} &  & G$_{h}$ & H$_{h}$ & G$_{oh}$ & \multicolumn{1}{c}{H$_{oh}$} &  & G$_{h}$ & H$_{h}$ & G$_{oh}$ & \multicolumn{1}{c}{H$_{oh}$}\\
\midrule
\addlinespace[0.3em]
\multicolumn{1}{c}{} & \multicolumn{10}{c}{\textit{$\overline{RelCRPS}$}}\\ \multicolumn{1}{c}{} & \multicolumn{5}{c}{\textbf{$\forall k \in \{4,2,1\}$}} & \multicolumn{5}{c}{\textbf{$k = 1$}}\\
base & \textcolor{black}{1.000} & \textcolor{black}{0.979} & \textcolor{black}{0.995} & \textcolor{black}{0.968} & \textcolor{black}{0.976} & \textcolor{black}{1.000} & \textcolor{black}{0.988} & \textcolor{black}{0.988} & \textcolor{black}{0.971} & \textcolor{black}{0.971}\\
ct$(shr_{cs}, bu_{te})$ & \textcolor{black}{0.937} & \textcolor{black}{0.956} & \textcolor{black}{0.956} & \textcolor{black}{0.976} & \textcolor{black}{0.976} & \textcolor{black}{0.992} & \textcolor{red}{1.008} & \textcolor{red}{1.008} & \textcolor{red}{1.029} & \textcolor{red}{1.029}\\
ct$(wls_{cs}, bu_{te})$ & \textcolor{black}{0.930} & \textcolor{black}{0.917} & \textcolor{black}{0.917} & \textcolor{black}{0.898} & \textcolor{black}{0.898} & \textcolor{black}{0.986} & \textcolor{black}{0.974} & \textcolor{black}{0.975} & \textcolor{black}{0.956} & \textcolor{black}{0.956}\\
oct$_o(wlsv)$ & \textcolor{black}{\textbf{0.926}} & \textcolor{black}{\textbf{0.911}} & \textcolor{black}{\textbf{0.912}} & \textcolor{black}{\textbf{0.896}} & \textcolor{blue}{\textbf{0.895}} & \textcolor{black}{\textbf{0.984}} & \textcolor{black}{\textbf{0.971}} & \textcolor{black}{\textbf{0.970}} & \textcolor{black}{\textbf{0.954}} & \textcolor{blue}{\textbf{0.954}}\\
oct$_o(bdshr)$ & \textcolor{black}{0.978} & \textcolor{black}{0.964} & \textcolor{black}{0.946} & \textcolor{black}{0.952} & \textcolor{black}{0.930} & \textcolor{red}{1.034} & \textcolor{red}{1.016} & \textcolor{red}{1.003} & \textcolor{red}{1.005} & \textcolor{black}{0.989}\\
oct$_{oh}(hshr)$ & \textcolor{red}{1.006} & \textcolor{black}{0.983} & \textcolor{red}{1.009} & \textcolor{black}{0.974} & \textcolor{red}{1.001} & \textcolor{red}{1.068} & \textcolor{red}{1.046} & \textcolor{red}{1.059} & \textcolor{red}{1.034} & \textcolor{red}{1.061}\\
\addlinespace[0.3em]
\multicolumn{1}{c}{} & \multicolumn{5}{c}{\textbf{$k = 2$}} & \multicolumn{5}{c}{\textbf{$k = 4$}}\\
base & \textcolor{black}{1.000} & \textcolor{black}{0.984} & \textcolor{black}{0.993} & \textcolor{black}{0.968} & \textcolor{black}{0.976} & \textcolor{black}{1.000} & \textcolor{black}{0.966} & \textcolor{red}{1.004} & \textcolor{black}{0.964} & \textcolor{black}{0.981}\\
ct$(shr_{cs}, bu_{te})$ & \textcolor{black}{0.949} & \textcolor{black}{0.966} & \textcolor{black}{0.966} & \textcolor{black}{0.987} & \textcolor{black}{0.987} & \textcolor{black}{0.874} & \textcolor{black}{0.896} & \textcolor{black}{0.896} & \textcolor{black}{0.914} & \textcolor{black}{0.914}\\
ct$(wls_{cs}, bu_{te})$ & \textcolor{black}{0.942} & \textcolor{black}{0.928} & \textcolor{black}{0.928} & \textcolor{black}{0.909} & \textcolor{black}{0.909} & \textcolor{black}{0.866} & \textcolor{black}{0.853} & \textcolor{black}{0.853} & \textcolor{black}{0.834} & \textcolor{black}{0.834}\\
oct$_o(wlsv)$ & \textcolor{black}{\textbf{0.938}} & \textcolor{black}{\textbf{0.921}} & \textcolor{black}{\textbf{0.923}} & \textcolor{black}{\textbf{0.907}} & \textcolor{blue}{\textbf{0.906}} & \textcolor{black}{\textbf{0.860}} & \textcolor{black}{\textbf{0.847}} & \textcolor{black}{\textbf{0.848}} & \textcolor{black}{\textbf{0.832}} & \textcolor{blue}{\textbf{0.830}}\\
oct$_o(bdshr)$ & \textcolor{black}{0.991} & \textcolor{black}{0.974} & \textcolor{black}{0.957} & \textcolor{black}{0.964} & \textcolor{black}{0.942} & \textcolor{black}{0.914} & \textcolor{black}{0.905} & \textcolor{black}{0.883} & \textcolor{black}{0.892} & \textcolor{black}{0.865}\\
oct$_{oh}(hshr)$ & \textcolor{red}{1.021} & \textcolor{black}{0.996} & \textcolor{red}{1.021} & \textcolor{black}{0.987} & \textcolor{red}{1.016} & \textcolor{black}{0.934} & \textcolor{black}{0.912} & \textcolor{black}{0.951} & \textcolor{black}{0.904} & \textcolor{black}{0.931}\\
\addlinespace[0.3em]
\multicolumn{1}{c}{} & \multicolumn{10}{c}{\textit{ES ratio indices}}\\ \multicolumn{1}{c}{} & \multicolumn{5}{c}{\textbf{$\forall k \in \{4,2,1\}$}} & \multicolumn{5}{c}{\textbf{$k = 1$}}\\
base & \textcolor{black}{1.000} & \textcolor{black}{0.970} & \textcolor{black}{0.988} & \textcolor{black}{0.960} & \textcolor{black}{0.970} & \textcolor{black}{1.000} & \textcolor{black}{0.977} & \textcolor{black}{0.977} & \textcolor{black}{0.965} & \textcolor{black}{0.965}\\
ct$(shr_{cs}, bu_{te})$ & \textcolor{black}{0.897} & \textcolor{black}{0.944} & \textcolor{black}{0.944} & \textcolor{black}{0.973} & \textcolor{black}{0.973} & \textcolor{black}{0.964} & \textcolor{red}{1.001} & \textcolor{red}{1.001} & \textcolor{red}{1.033} & \textcolor{red}{1.033}\\
ct$(wls_{cs}, bu_{te})$ & \textcolor{black}{\textbf{0.886}} & \textcolor{black}{0.880} & \textcolor{black}{\textbf{0.880}} & \textcolor{blue}{\textbf{0.860}} & \textcolor{black}{\textbf{0.860}} & \textcolor{black}{\textbf{0.954}} & \textcolor{black}{\textbf{0.944}} & \textcolor{black}{\textbf{0.945}} & \textcolor{blue}{\textbf{0.928}} & \textcolor{black}{\textbf{0.928}}\\
oct$_o(wlsv)$ & \textcolor{black}{0.891} & \textcolor{black}{\textbf{0.879}} & \textcolor{black}{0.881} & \textcolor{black}{0.864} & \textcolor{black}{0.864} & \textcolor{black}{0.958} & \textcolor{black}{0.945} & \textcolor{black}{0.945} & \textcolor{black}{0.931} & \textcolor{black}{0.931}\\
oct$_o(bdshr)$ & \textcolor{black}{0.940} & \textcolor{black}{0.928} & \textcolor{black}{0.910} & \textcolor{black}{0.918} & \textcolor{black}{0.895} & \textcolor{red}{1.004} & \textcolor{black}{0.986} & \textcolor{black}{0.971} & \textcolor{black}{0.980} & \textcolor{black}{0.961}\\
oct$_{oh}(hshr)$ & \textcolor{black}{0.986} & \textcolor{black}{0.968} & \textcolor{black}{0.999} & \textcolor{black}{0.959} & \textcolor{black}{0.992} & \textcolor{red}{1.053} & \textcolor{red}{1.034} & \textcolor{red}{1.049} & \textcolor{red}{1.024} & \textcolor{red}{1.055}\\
\addlinespace[0.3em]
\multicolumn{1}{c}{} & \multicolumn{5}{c}{\textbf{$k = 2$}} & \multicolumn{5}{c}{\textbf{$k = 4$}}\\
base & \textcolor{black}{1.000} & \textcolor{black}{0.972} & \textcolor{black}{0.985} & \textcolor{black}{0.959} & \textcolor{black}{0.969} & \textcolor{black}{1.000} & \textcolor{black}{0.959} & \textcolor{red}{1.000} & \textcolor{black}{0.957} & \textcolor{black}{0.976}\\
ct$(shr_{cs}, bu_{te})$ & \textcolor{black}{0.915} & \textcolor{black}{0.961} & \textcolor{black}{0.960} & \textcolor{black}{0.991} & \textcolor{black}{0.991} & \textcolor{black}{0.818} & \textcolor{black}{0.874} & \textcolor{black}{0.874} & \textcolor{black}{0.899} & \textcolor{black}{0.900}\\
ct$(wls_{cs}, bu_{te})$ & \textcolor{black}{\textbf{0.904}} & \textcolor{black}{0.896} & \textcolor{black}{\textbf{0.896}} & \textcolor{blue}{\textbf{0.877}} & \textcolor{black}{\textbf{0.877}} & \textcolor{black}{\textbf{0.807}} & \textcolor{black}{0.805} & \textcolor{black}{\textbf{0.805}} & \textcolor{blue}{\textbf{0.782}} & \textcolor{black}{\textbf{0.783}}\\
oct$_o(wlsv)$ & \textcolor{black}{0.908} & \textcolor{black}{\textbf{0.895}} & \textcolor{black}{0.898} & \textcolor{black}{0.881} & \textcolor{black}{0.882} & \textcolor{black}{0.812} & \textcolor{black}{\textbf{0.802}} & \textcolor{black}{0.806} & \textcolor{black}{0.786} & \textcolor{black}{0.786}\\
oct$_o(bdshr)$ & \textcolor{black}{0.960} & \textcolor{black}{0.947} & \textcolor{black}{0.929} & \textcolor{black}{0.938} & \textcolor{black}{0.915} & \textcolor{black}{0.860} & \textcolor{black}{0.856} & \textcolor{black}{0.836} & \textcolor{black}{0.841} & \textcolor{black}{0.816}\\
oct$_{oh}(hshr)$ & \textcolor{red}{1.007} & \textcolor{black}{0.988} & \textcolor{red}{1.017} & \textcolor{black}{0.979} & \textcolor{red}{1.014} & \textcolor{black}{0.904} & \textcolor{black}{0.888} & \textcolor{black}{0.934} & \textcolor{black}{0.881} & \textcolor{black}{0.913}\\
\bottomrule
\multicolumn{11}{l}{\rule{0pt}{1em}\rule{0pt}{1.75em}\makecell[l]{$^\ast$The Gaussian method employs a sample covariance matrix:\\G$_{h}$ and H$_{h}$ use multi-step residuals and G$_{oh}$ and H$_{oh}$ use overlapping and multi-step residuals.}}\\
\end{tabular}

	\endgroup
	\caption{$\overline{RelCRPS}$ and ES ratio indices defined in \eqref{eq:skill} and \eqref{eq:skill_all} for the Australian QNA dataset. Approaches performing worse than the benchmark (bootstrap base forecasts, ctjb) are highlighted in red, the best for each column is marked in bold, and the overall lowest value is highlighted in blue. The reconciliation approaches are described in \autoref{tab:notation}.}
	\label{tab:ausscore}
	\vspace*{-0.5\baselineskip}
\end{table}

Forecasting accuracy indices based on CRPS and ES are presented in \autoref{tab:ausscore}. As a benchmark approach, we use the base forecasts calculated using the bootstrap method. For base forecasts, we find that using a parametric approach with the normal distribution performs better than the non-parametric bootstrap approach. This is likely due to the limited number of residuals available for bootstrapping, which does not allow for sufficient exploration of the data. Directly specifying diagonal covariance matrices seems to be more effective than shrinking to a target covariance matrix. Among all the procedures, ct$(wls_{cs},bu_{te})$ and oct$_o(wlsv)$ show the greatest relative gains. In contrast, $oct_{oh}(hshr)$ does not show much improvement. Furthermore, the greatest improvements are observed for higher temporal aggregation levels.

\begin{figure}[t]
	\centering
	\includegraphics[width = 0.45\linewidth]{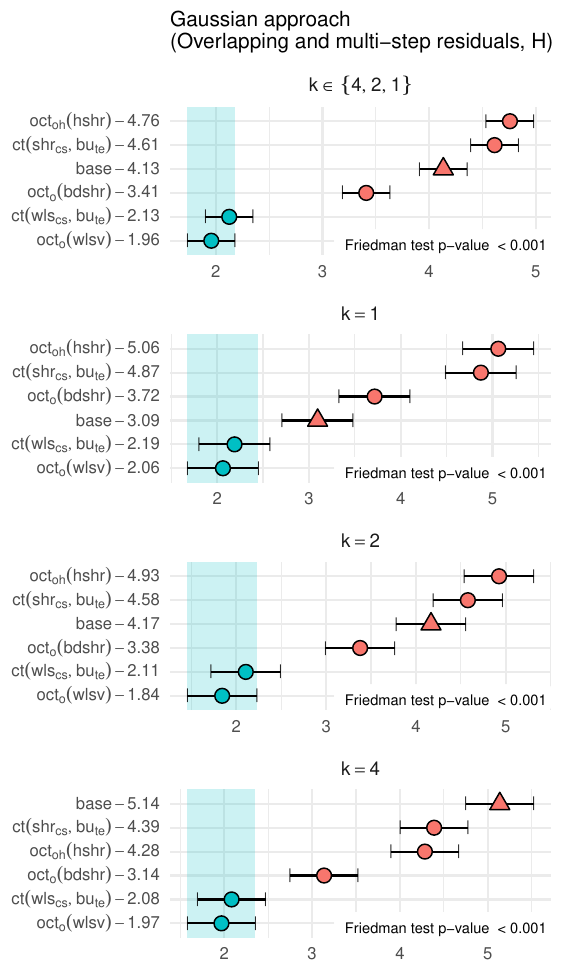}
	\includegraphics[width = 0.45\linewidth]{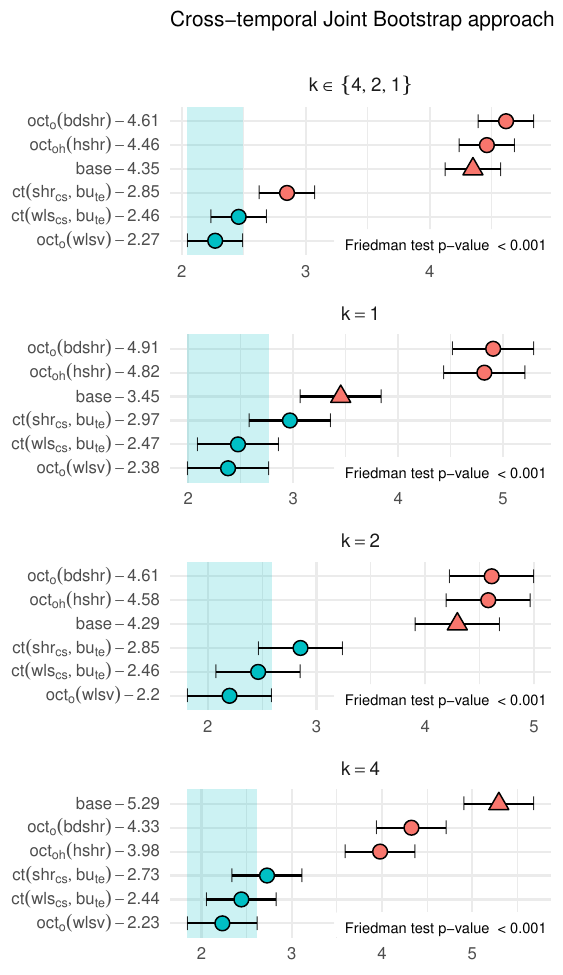}
	\caption{MCB Nemenyi test for the Australian QNA dataset using CRPS at different temporal aggregation levels for the Gaussian (using overlapping and multi-step residuals, H) and the non-parametric bootstrap approaches. In each panel, the Friedman test p-value is reported in the lower right corner. The mean rank of each approach is shown to the right of its name. Statistically significant differences in performance are indicated if the intervals of two forecast reconciliation procedures do not overlap. Thus, approaches that do not overlap with the blue interval are considered significantly worse than the best, and vice-versa.}
	\label{fig:mcb}
\end{figure}

We utilize the non-parametric Friedman test and the post hoc “Multiple Comparison with the Best" (MCB) Nemenyi test \citep{koning2005, kourentzes2019, makridakis2022, tsutils} to determine if the forecasting performances of the different techniques are significantly different from one another.
\autoref{fig:mcb} presents the MCB using the CRPS. The probabilistic forecasts from ct$(wls_{cs},bu_{te})$ and oct$_o(wlsv)$ are significantly better than the base forecasts at any level of aggregation. Unlike the application on the Australian Tourism Demand (see \autoref{sec:vn525}), in this case one of the partly bottom-up approaches is not significantly worse than the most performing optimal approach.

Overall, we find that using overlapping residuals almost always leads to a greater improvement in terms of both ES and CRPS. %, and this is also generally true for reconciliation.
Forecasts at the most aggregated level (year) seem to benefit the most from reconciliation, and using one-step overlapping residuals appears to be sufficient to improve forecasts if the generation of the base forecasts sample paths takes into account the multi-step structure.

\section{Forecasting Australian Tourism Demand}\label{sec:vn525}

The Australian Tourism Demand dataset \citep{wickramasuriya2019} measures the number of nights Australians spent away from home. It includes 228 monthly observations of Visitor Nights (VN) from January 1998 to December 2016, and has a cross-sectional grouped structure based on a geographic hierarchy crossed by purpose of travel. The geographic hierarchy comprises seven states, 27 zones, and 76 regions, for a total of 111 nested geographic divisions. Six of these zones are each formed by a single region, resulting in a total of 105 unique nodes in the hierarchy. The purpose of travel comprises four categories: holiday, visiting friends and relatives, business, and other.
To avoid redundancies \citep{difonzo2022a}, 24 nodes are not considered, resulting in an unbalanced hierarchy of 525 unique nodes instead of the theoretical 555 with duplicated nodes.
The dataset includes the 304 bottom series, which are aggregated into 221 upper time series. \autoref{tab:nseries} omits duplicated entries and updates the information in Table 7 from \cite{wickramasuriya2019}. This data can be temporally aggregated into 2, 3, 4, 6, or 12 months ($\mathcal{K} = \{12,4,3,2,1\}$).

\begin{table}[b]
	\spacingset{1.1}
	\setlength{\tabcolsep}{10pt}
	\centering
	\begin{tabular}{c|cc|c}
		\toprule
		& \multicolumn{3}{c}{\textbf{Number of series}}\\
		& \textbf{GD} & \textbf{PT} & \textbf{Tot}. \\
		\midrule
		Australia & 1 & 4 & 5 \\
		States & 7 & 28 & 35 \\
		Zones$^*$ & 21 & 84 & 105 \\
		Regions & 76 & 304 & 380 \\
		\bottomrule
		\textbf{Total} & \textbf{105}                                  & \textbf{420}   & \textbf{525} \\
		\bottomrule
		\addlinespace[0.3em]
	\end{tabular}\\
	{\footnotesize \textbf{*} 6 Zones with only one Region are included in Regions. GD: Geographic Division; PT: Purpose of Travel.}\\[0.1cm]
	\caption{\label{tab:nseries} Grouped time series for the Australian Tourism Demand dataset. }
	\vspace*{-0.5\baselineskip}
\end{table}

We perform a rolling forecast experiment with an expanding window. The process begins by using the first 10 years, from January 1998 to December 2008, to generate forecasts for the entire following year (2009). Then, the training set is increased by one month. This process is repeated until the last training set is used (January 1998 to December 2015) with a total of 85 different test sets. For the temporal aggregation dimension we aggregate the monthly data up to annual data. We obtain $12$-step, $6$-step, $4$-step, $3$-step, $2$-step and $1$-step ahead base forecasts respectively from the monthly data and the aggregation over 2, 3, 4, 6, and 12 months. ETS models selected by minimizing the AICc criterion \citep{Rforecast} %with the R package \texttt{forecast}
are fitted to the log-transformed data, with the resulting base forecasts being back-transformed to produce non-negative forecasts \citep{wickramasuriya2020}.

The (1000) base forecast samples are obtained using the Gaussian approach with sample\footnote{The results with shrunk covariance matrices are available in the online appendix E.2, where we also report the results obtained using other reconciliation approaches.} covariance matrices (\autoref{sec:shrtech}) using multi-step residuals\footnote{We do not include overlapping, as we are unable to correctly determine the residuals for the overlapping series using ETS models (see \autoref{ssec:over_res}).} and the bootstrap approach (\autoref{ssec:boot}). For reconciliation, 6 different approaches have been adopted (see \autoref{tab:notation}): ct$(shr_{cs}, bu_{te})$, oct$(struc)$, oct$(wlsv)$, oct$(bdshr)$, oct$_h(bshr)$, and oct$_h(hshr)$.

Negative forecasts may be produced during the reconciliation phase \citep{wickramasuriya2020, difonzo2022a, difonzo2023a} thus generating unreasonable values (e.g., a negative forecast for tourism demand makes no sense). To overcome this limitation, we applied the simple heuristic proposed by \cite{difonzo2022b, difonzo2023a}. Following Theorem \ref{thm:rs}, we are thus able to obtain reconciled samples respecting non-negativity constraints starting from an incoherent sample simulated from a Gaussian distribution. Finally, to evaluate the performance, we employ the Continuous Ranked Probability Score (CRPS), the Energy Score (ES), and the “Multiple Comparison with the Best" (MCB) Nemenyi test, introduced in Sections \ref{sec:ausgdp} and \ref{ssec:ausresults}.

\subsection{Results}

\begin{table}[!tb]
	\centering
	\begingroup
	\spacingset{1}
	\fontsize{9}{10}\selectfont
	
\begin{tabular}[t]{l|>{}cccc>{}c|ccccc}
\toprule
\multicolumn{1}{c}{\textbf{}} & \multicolumn{10}{c}{\textbf{Base forecasts' sample approach}} \\
\cmidrule(l{0pt}r{0pt}){2-11}
\multicolumn{1}{c}{\makecell[c]{\bfseries Reconciliation\\\bfseries approach}} & \multicolumn{1}{c}{ctjb} & \multicolumn{4}{c}{\makecell[c]{Gaussian approach\textsuperscript{*}}} & \multicolumn{1}{c}{ctjb} & \multicolumn{4}{c}{\makecell[c]{Gaussian approach\textsuperscript{*}}} \\
\multicolumn{1}{c}{} &  & G & B & H & \multicolumn{1}{c}{HB} &  & G & B & H & HB\\
\midrule
\addlinespace[0.3em]
\multicolumn{1}{c}{} & \multicolumn{10}{c}{\textit{$\overline{RelCRPS}$}}\\ \multicolumn{1}{c}{} & \multicolumn{5}{c}{\textbf{$\forall k \in \{12,6,4,3,2,1\}$}} & \multicolumn{5}{c}{\textbf{$k = 1$}}\\
base & \textcolor{black}{1.000} & \textcolor{black}{0.971} & \textcolor{black}{0.971} & \textcolor{black}{0.973} & \textcolor{black}{0.973} & \textcolor{black}{1.000} & \textcolor{black}{0.972} & \textcolor{black}{0.972} & \textcolor{black}{0.972} & \textcolor{black}{0.972}\\
ct$(shr_{cs}, bu_{te})$ & \textcolor{red}{1.057} & \textcolor{black}{0.974} & \textcolor{black}{0.969} & \textcolor{black}{0.974} & \textcolor{black}{0.969} & \textcolor{black}{0.976} & \textcolor{black}{0.963} & \textcolor{black}{0.962} & \textcolor{black}{0.963} & \textcolor{black}{0.962}\\
oct$(struc)$ & \textcolor{black}{0.982} & \textcolor{black}{0.962} & \textcolor{black}{0.961} & \textcolor{black}{0.961} & \textcolor{black}{0.959} & \textcolor{black}{0.970} & \textcolor{black}{0.963} & \textcolor{black}{0.963} & \textcolor{black}{0.963} & \textcolor{black}{0.963}\\
oct$(wlsv)$ & \textcolor{black}{0.987} & \textcolor{black}{0.959} & \textcolor{black}{0.959} & \textcolor{black}{0.958} & \textcolor{black}{0.957} & \textcolor{black}{0.952} & \textcolor{black}{0.957} & \textcolor{black}{0.957} & \textcolor{black}{0.957} & \textcolor{black}{0.957}\\
oct$(bdshr)$ & \textcolor{black}{0.975} & \textcolor{black}{\textbf{0.956}} & \textcolor{black}{\textbf{0.953}} & \textcolor{black}{\textbf{0.952}} & \textcolor{blue}{\textbf{0.951}} & \textcolor{blue}{\textbf{0.949}} & \textcolor{black}{\textbf{0.955}} & \textcolor{black}{\textbf{0.953}} & \textcolor{black}{\textbf{0.954}} & \textcolor{black}{\textbf{0.954}}\\
oct$_h(bshr)$ & \textcolor{black}{0.994} & \textcolor{red}{1.018} & \textcolor{red}{1.020} & \textcolor{red}{1.016} & \textcolor{red}{1.019} & \textcolor{black}{0.988} & \textcolor{red}{1.007} & \textcolor{red}{1.013} & \textcolor{red}{1.006} & \textcolor{red}{1.012}\\
oct$_h(hshr)$ & \textcolor{black}{\textbf{0.969}} & \textcolor{black}{0.993} & \textcolor{black}{0.993} & \textcolor{black}{0.990} & \textcolor{black}{0.991} & \textcolor{black}{0.953} & \textcolor{black}{0.977} & \textcolor{black}{0.977} & \textcolor{black}{0.979} & \textcolor{black}{0.979}\\
\addlinespace[0.3em]
\multicolumn{1}{c}{} & \multicolumn{5}{c}{\textbf{$k = 3$}} & \multicolumn{5}{c}{\textbf{$k = 12$}}\\
base & \textcolor{black}{1.000} & \textcolor{black}{0.971} & \textcolor{black}{0.971} & \textcolor{black}{0.972} & \textcolor{black}{0.973} & \textcolor{black}{1.000} & \textcolor{black}{0.968} & \textcolor{black}{0.967} & \textcolor{black}{0.969} & \textcolor{black}{0.969}\\
ct$(shr_{cs}, bu_{te})$ & \textcolor{red}{1.041} & \textcolor{black}{0.977} & \textcolor{black}{0.974} & \textcolor{black}{0.977} & \textcolor{black}{0.974} & \textcolor{red}{1.163} & \textcolor{black}{0.977} & \textcolor{black}{0.965} & \textcolor{black}{0.977} & \textcolor{black}{0.965}\\
oct$(struc)$ & \textcolor{black}{0.986} & \textcolor{black}{0.967} & \textcolor{black}{0.966} & \textcolor{black}{0.966} & \textcolor{black}{0.965} & \textcolor{black}{0.982} & \textcolor{black}{0.951} & \textcolor{black}{0.949} & \textcolor{black}{0.947} & \textcolor{black}{0.943}\\
oct$(wlsv)$ & \textcolor{black}{0.983} & \textcolor{black}{0.963} & \textcolor{black}{0.962} & \textcolor{black}{0.962} & \textcolor{black}{0.962} & \textcolor{red}{1.025} & \textcolor{black}{0.954} & \textcolor{black}{0.953} & \textcolor{black}{0.949} & \textcolor{black}{0.947}\\
oct$(bdshr)$ & \textcolor{black}{0.972} & \textcolor{black}{\textbf{0.960}} & \textcolor{black}{\textbf{0.958}} & \textcolor{black}{\textbf{0.957}} & \textcolor{blue}{\textbf{0.957}} & \textcolor{red}{1.002} & \textcolor{black}{\textbf{0.950}} & \textcolor{black}{\textbf{0.944}} & \textcolor{black}{\textbf{0.939}} & \textcolor{blue}{\textbf{0.935}}\\
oct$_h(bshr)$ & \textcolor{black}{0.999} & \textcolor{red}{1.021} & \textcolor{red}{1.022} & \textcolor{red}{1.018} & \textcolor{red}{1.022} & \textcolor{black}{0.987} & \textcolor{red}{1.024} & \textcolor{red}{1.021} & \textcolor{red}{1.021} & \textcolor{red}{1.019}\\
oct$_h(hshr)$ & \textcolor{black}{\textbf{0.971}} & \textcolor{black}{0.994} & \textcolor{black}{0.994} & \textcolor{black}{0.992} & \textcolor{black}{0.993} & \textcolor{black}{\textbf{0.978}} & \textcolor{red}{1.003} & \textcolor{red}{1.005} & \textcolor{black}{0.996} & \textcolor{black}{0.997}\\
\addlinespace[0.3em]
\multicolumn{1}{c}{} & \multicolumn{10}{c}{\textit{ES ratio indices}}\\ \multicolumn{1}{c}{} & \multicolumn{5}{c}{\textbf{$\forall k \in \{12,6,4,3,2,1\}$}} & \multicolumn{5}{c}{\textbf{$k = 1$}}\\
base & \textcolor{black}{1.000} & \textcolor{black}{0.956} & \textcolor{black}{0.955} & \textcolor{black}{0.958} & \textcolor{black}{0.951} & \textcolor{black}{1.000} & \textcolor{black}{0.952} & \textcolor{black}{0.950} & \textcolor{black}{0.952} & \textcolor{black}{0.950}\\
ct$(shr_{cs}, bu_{te})$ & \textcolor{red}{1.243} & \textcolor{black}{0.886} & \textcolor{black}{0.879} & \textcolor{black}{0.886} & \textcolor{black}{0.879} & \textcolor{red}{1.098} & \textcolor{black}{0.929} & \textcolor{black}{0.928} & \textcolor{black}{0.930} & \textcolor{black}{0.927}\\
oct$(struc)$ & \textcolor{red}{1.085} & \textcolor{black}{0.917} & \textcolor{black}{0.915} & \textcolor{black}{0.916} & \textcolor{black}{0.912} & \textcolor{red}{1.027} & \textcolor{black}{0.943} & \textcolor{black}{0.942} & \textcolor{black}{0.943} & \textcolor{black}{0.942}\\
oct$(wlsv)$ & \textcolor{red}{1.132} & \textcolor{black}{0.933} & \textcolor{black}{0.929} & \textcolor{black}{0.931} & \textcolor{black}{0.927} & \textcolor{red}{1.050} & \textcolor{black}{0.951} & \textcolor{black}{0.949} & \textcolor{black}{0.950} & \textcolor{black}{0.949}\\
oct$(bdshr)$ & \textcolor{red}{1.047} & \textcolor{black}{0.904} & \textcolor{black}{0.897} & \textcolor{black}{0.897} & \textcolor{black}{0.891} & \textcolor{red}{1.009} & \textcolor{black}{0.936} & \textcolor{black}{0.933} & \textcolor{black}{0.934} & \textcolor{black}{0.931}\\
oct$_h(bshr)$ & \textcolor{black}{\textbf{0.931}} & \textcolor{black}{\textbf{0.867}} & \textcolor{black}{\textbf{0.866}} & \textcolor{black}{\textbf{0.863}} & \textcolor{blue}{\textbf{0.860}} & \textcolor{black}{\textbf{0.965}} & \textcolor{black}{\textbf{0.927}} & \textcolor{black}{\textbf{0.927}} & \textcolor{black}{\textbf{0.925}} & \textcolor{blue}{\textbf{0.923}}\\
oct$_h(hshr)$ & \textcolor{red}{1.081} & \textcolor{black}{0.935} & \textcolor{black}{0.931} & \textcolor{black}{0.935} & \textcolor{black}{0.927} & \textcolor{red}{1.028} & \textcolor{black}{0.952} & \textcolor{black}{0.951} & \textcolor{black}{0.952} & \textcolor{black}{0.950}\\
\addlinespace[0.3em]
\multicolumn{1}{c}{} & \multicolumn{5}{c}{\textbf{$k = 3$}} & \multicolumn{5}{c}{\textbf{$k = 12$}}\\
base & \textcolor{black}{1.000} & \textcolor{black}{0.961} & \textcolor{black}{0.958} & \textcolor{black}{0.960} & \textcolor{black}{0.955} & \textcolor{black}{1.000} & \textcolor{black}{0.942} & \textcolor{black}{0.947} & \textcolor{black}{0.951} & \textcolor{black}{0.937}\\
ct$(shr_{cs}, bu_{te})$ & \textcolor{red}{1.245} & \textcolor{black}{0.911} & \textcolor{black}{0.904} & \textcolor{black}{0.911} & \textcolor{black}{0.904} & \textcolor{red}{1.326} & \textcolor{black}{0.779} & \textcolor{black}{0.767} & \textcolor{black}{0.777} & \textcolor{black}{0.766}\\
oct$(struc)$ & \textcolor{red}{1.096} & \textcolor{black}{0.939} & \textcolor{black}{0.936} & \textcolor{black}{0.938} & \textcolor{black}{0.933} & \textcolor{red}{1.077} & \textcolor{black}{0.826} & \textcolor{black}{0.822} & \textcolor{black}{0.823} & \textcolor{black}{0.818}\\
oct$(wlsv)$ & \textcolor{red}{1.142} & \textcolor{black}{0.953} & \textcolor{black}{0.949} & \textcolor{black}{0.951} & \textcolor{black}{0.946} & \textcolor{red}{1.149} & \textcolor{black}{0.851} & \textcolor{black}{0.845} & \textcolor{black}{0.847} & \textcolor{black}{0.840}\\
oct$(bdshr)$ & \textcolor{red}{1.060} & \textcolor{black}{0.926} & \textcolor{black}{0.920} & \textcolor{black}{0.921} & \textcolor{black}{0.915} & \textcolor{red}{1.021} & \textcolor{black}{0.808} & \textcolor{black}{0.796} & \textcolor{black}{0.796} & \textcolor{black}{0.787}\\
oct$_h(bshr)$ & \textcolor{black}{\textbf{0.954}} & \textcolor{black}{\textbf{0.895}} & \textcolor{black}{\textbf{0.895}} & \textcolor{black}{\textbf{0.892}} & \textcolor{blue}{\textbf{0.887}} & \textcolor{black}{\textbf{0.833}} & \textcolor{black}{\textbf{0.741}} & \textcolor{black}{\textbf{0.741}} & \textcolor{black}{\textbf{0.737}} & \textcolor{blue}{\textbf{0.735}}\\
oct$_h(hshr)$ & \textcolor{red}{1.093} & \textcolor{black}{0.955} & \textcolor{black}{0.951} & \textcolor{black}{0.956} & \textcolor{black}{0.949} & \textcolor{red}{1.066} & \textcolor{black}{0.851} & \textcolor{black}{0.846} & \textcolor{black}{0.848} & \textcolor{black}{0.838}\\
\bottomrule
\multicolumn{11}{l}{\rule{0pt}{1em}\rule{0pt}{1.75em}\makecell[l]{$^\ast$The Gaussian method employs a sample covariance matrix and includes four techniques\\ (G, B, H, HB) with multi-step residuals..}}\\
\end{tabular}

	\endgroup
	\caption{$\overline{RelCRPS}$ and ES ratio indices defined in \eqref{eq:skill} and \eqref{eq:skill_all} for the Australian Tourism Demand dataset. %A lower value, indicates a more accurate forecast.
	Approaches performing worse than the benchmark (bootstrap base forecasts, ctjb) are highlighted in red, the best for each column is marked in bold, and the overall lowest value is highlighted in blue. The reconciliation approaches are described in \autoref{tab:notation}.}
	\label{tab:vnscore}
	\vspace*{-0.5\baselineskip}
\end{table}

The CRPS and ES indices are shown, respectively, in \autoref{tab:vnscore} for monthly, quarterly and annual forecasts. These tables are divided by different temporal levels and each column uses a different approach to calculate the base forecasts, referred to as “base". The bootstrap method is used as a benchmark to calculate the accuracy indices.

%\footnote{The complete results for all temporal aggregation levels are reported in Appendix E.2.}

Base forecasts using a Gaussian approach are better in terms of both CRPS and ES compared to the bootstrap approach (the benchmark). Assumptions of truncated Gaussianity (Gaussian with negative values set to zero) may seem strict, but given the limited number of residuals, it can lead to improved forecasts in terms of CRPS and ES. Bootstrap forecasts suffer from the limited number of available residuals, leading in general to lower forecast accuracy. The Gaussian approach overcomes this limitation and provides better results. Regarding the different covariance matrix estimates for Gaussian base forecasts, there are no big differences. For this reason, using only the high frequency bottom time series can be useful to estimate fewer parameters and reduce the initial high dimensionality.

\begin{figure}[!t]
	\centering
	\includegraphics[width = 0.45\linewidth]{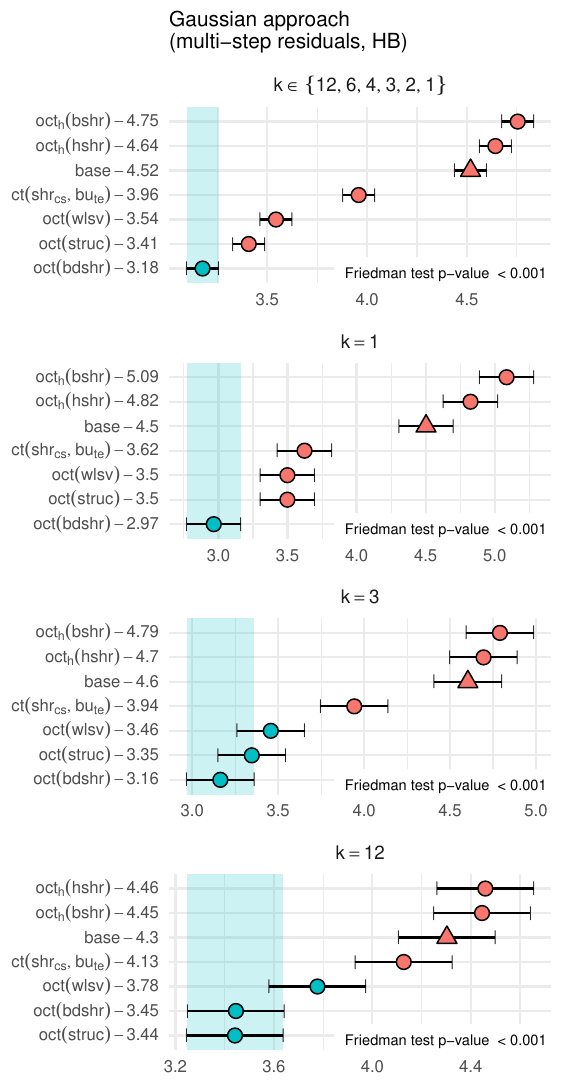}
	\includegraphics[width = 0.45\linewidth]{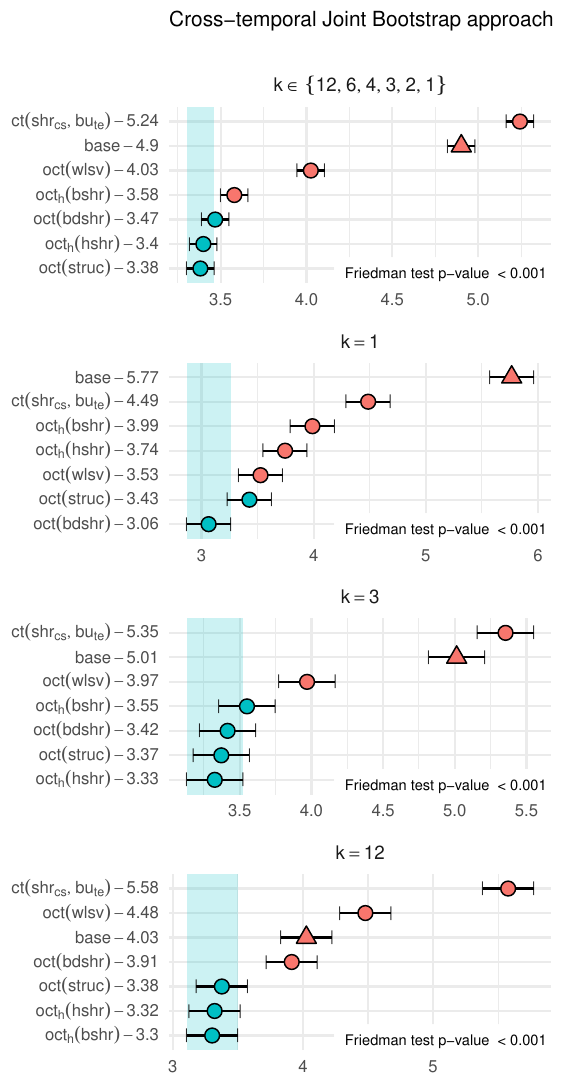}
	\caption{MCB Nemenyi test for the Australian Tourism Demand dataset using CRPS at different temporal aggregation levels for the Gaussian (multi-step residuals, HB) and the non-parametric bootstrap approaches. In each panel, the Friedman test p-value is reported in the lower right corner. The mean rank of each approach is shown to the right of its name. Statistically significant differences in performance are indicated if the intervals of two forecast reconciliation procedures do not overlap. Thus, approaches that do not overlap with the blue interval are considered significantly worse than the best, and vice-versa.}
	\label{fig:vnmcb}
\end{figure}

In the Gaussian case, partly bottom-up techniques like ct$(shr_{cs}, bu_{te})$ lead to better results than the benchmark (bootstrap base forecasts). However, it is not always guaranteed that the improvement is higher than the starting base forecasts (by comparing the value of each column). This is particularly true for higher levels of temporal aggregation. Overall, oct$(bdshr)$ in terms of CRPS is almost always the best. The shrinkage approach oct$_h(hshr)$ performs well in the bootstrap case: it is competitive with oct$(bdshr)$ at lower temporal frequency ($k \in \{2,1\}$) and it is able to improve for $k\ge 3$. In terms of ES, oct$(bdshr)$ is still competitive, although it does not always show the best relative performance, like oct$_h(bshr)$. 
It is also worth noting that oct$(struc)$, which does not make use of residuals, proves to be competitive by consistently improving on the base forecasts in terms of both CRPS and ES.

\autoref{fig:vnmcb} shows the MCB using the CRPS for the Gaussian approach using multi-step residuals (HB) and the non-parametric bootstrap approach. In general, the partly bottom-up procedure improves with respect to base forecasts at monthly level, but optimal cross-temporal procedures are always better. In the bootstrap framework, we can identify a group of three procedures, oct$(bdshr)$, oct$(hshr)$ and oct$(struc)$ that are almost always in the group of the best approaches (denoted by the blue dot). In the Gaussian framework, oct$(wlsv)$, oct$(struc)$, and oct$(bdshr)$ are always significantly better than base forecasts and equivalent in terms of results for temporal aggregation orders greater than 2. For monthly series, oct$(bdshr)$ is always significantly better than all other approaches.

\section{Conclusion}\label{sec:conclusion}

In this paper, we extend the probabilistic reconciliation setting developed by \cite{panagiotelis2023} for the cross-sectional case to the cross-temporal framework. Through appropriate notation, we show how theorems and definitions valid for the cross-sectional case can be reinterpreted and extended. The general notation proposed can help investigate extensions following different probabilistic approaches, such as those in \cite{jeon2019}, \cite{bentaieb2021} and \cite{corani2022}. We propose a Gaussian and a bootstrap approach to simulate the base forecasts able to take into account both cross-sectional and temporal relationships simultaneously, opening the way for further research into  cross-temporal probabilistic forecasting.

Moreover, we analyze the use of residuals, showing that one-step residuals fail to capture the temporal structure, and propose multi-step residuals that can fully capture the cross-temporal relationships. Due to the high-dimensionality of the cross-temporal setting when dealing with covariance matrices, we propose four alternative forms to reduce the number of parameters to be estimated, showing that the overlapping residuals may reduce the high-dimensionality burden by increasing the number of available residuals. These ideas are worth requiring further investigation in future works.

Finally, we perform empirical applications on two datasets commonly used in forecast reconciliation research: Australian GDP from Income and Expenditure sides and Australian Tourism Demand. We find that in both cases optimal cross-temporal reconciliation approaches significantly improve on base forecasts. We also compare these with partly bottom-up techniques that use uni-dimensional reconciliation (either cross-sectional or temporal) and confirm that simultaneously exploiting both dimensions in reconciliation produces better results, especially at higher levels of temporal aggregation.
This is more evident in the Australian Tourism Demand application, where the involved temporal hierarchies are richer, allowing the regression-based forecast reconciliation approaches to capture and exploit more features of the data through the available temporal aggregation levels \citep{kourentzes2014, kourentzes2016, kourentzes2017} compared to partly bottom-up. In these two datasets, oct$(wlsv)$ and oct$(bdshr)$ appear as the two best performing approaches, both in terms of improving forecast accuracy and computational efficiency (see the online appendix), thus corroborating the results of \cite{difonzo2023} for point forecast reconciliation.

%It is worth noting that oct$(wlsv)$ for it's reduced computational burder should be consider as sensible probabilistic reconciliation approach (similar results for point forecasting was found in \citealp{difonzo2023})}

In conclusion, cross-temporal forecast reconciliation is an important tool to improve the accuracy of forecasts while simultaneously ensuring their coherency both in space and time. Furthermore, these techniques can also be customized to suit the specific needs of an organization, allowing for the incorporation of relevant domain-specific knowledge (e.g., non negative constraints) and expertise, ensuring that the resulting forecasts are not only accurate but also coherent and more reliable for decision-making purposes.

\appendix
\setcounter{table}{0}
\renewcommand{\thetable}{\Alph{section}.\arabic{table}}

\if1\blind
{
  \phantomsection\addcontentsline{toc}{section}{Acknowledgments}
\section*{Acknowledgments}

\noindent The authors thank Casey Lichtendahl and all the participants to the 2023 IIF Workshop on Forecast Reconciliation in Prato (Italy). Tommaso Di Fonzo and Daniele Girolimetto acknowledge financial support from project PRIN2017 “HiDEA: Advanced Econometrics for High-frequency Data”, 2017RSMPZZ. Rob Hyndman acknowledges the support of the Australian Government through the Australian Research Council Industrial Transformation Training Centre in Optimisation Technologies, Integrated Methodologies, and Applications (OPTIMA), Project ID IC200100009.
} \fi

\phantomsection\addcontentsline{toc}{section}{References}

\bibliographystyle{agsm}
\bibliography{mybibfile.bib}

\end{document}

\def\spacingset#1{\renewcommand{\baselinestretch}{#1}\small\normalsize}
\spacingset{1.1}
  
\maketitleblind

\spacingset{1.3}
\tableofcontents
\newpage
%\listoftables
\appendix
\renewcommand{\thetable}{\Alph{section}.\arabic{table}}
\renewcommand{\thefigure}{\Alph{section}.\arabic{figure}}

\section{Cross-sectional, temporal and cross-temporal covariance approximations}\label{app:covapp}

\autoref{tab:cov_app} presents some approximations for the cross-sectional \citep{hyndman2011, hyndman2016, wickramasuriya2019} and the temporal \citep{athanasopoulos2017, nystrup2020} covariance matrices. \cite{difonzo2023} consider the following approximations for the cross-temporal covariance matrix.
\begin{itemize}[nosep, leftmargin=!, labelwidth=\widthof{ oct$(bdsam)$ -}, align=right]
	\item[oct$(ols)$ -] identity: $\Omegavet_{ct} = \Ivet_{n(k^*+m)}$.
	\item[oct$(struc)$ -] structural: $\Omegavet_{ct} = \mathrm{diag}(\Svet_{ct} \mathbf{1}_{mn_b})$.
	\item[oct$(wlsv)$ -] series variance scaling: $\Omegavet_{ct} = \widehat{\Omegavet}_{ct,wlsv}$, a straightforward extension of the series variance scaling matrix presented by \cite{athanasopoulos2017}.
	\item[oct$(bdshr)$ -] block-diagonal shrunk cross-covariance scaling: $\Omegavet_{ct} = \Pvet\widehat{\Wvet}^{BD}_{ct,shr}\Pvet'$, with $\widehat{\Wvet}^{BD}_{ct,shr}$ a block diagonal matrix where each $k-$block ($k = m,k_{p-1},\dots, 1$) is $\Ivet_{M_k} \otimes \widehat{\Wvet}^{[k]}_{shr}$, $\widehat{\Wvet}^{[k]}_{shr}$ is the shrunk estimate of the cross-sectional covariance matrix proposed by \cite{wickramasuriya2019}, and $\Pvet$ is the commutation matrix such that $\Pvet \mathrm{vec}(\Yvet_{\tau}) = \mathrm{vec}(\Yvet_{\tau}')$.
	\item[oct$(shr)$ -] MinT-shr:   $\Omegavet_{ct} = \hat{\lambda}\widehat{\Omegavet}_{ct,D} + (1-\hat{\lambda})\widehat{\Omegavet}_{ct}$,
	where $\hat{\lambda}$ is an estimated shrinkage coefficient (\citealp{ledoit2004a}), $\widehat{\Omegavet}_{ct,D} = \Ivet_{n(k^\ast + m)} \odot \widehat{\Omegavet}_{ct}$ with $\odot$ denoting the Hadamard product, and $\widehat{\Omegavet}_{ct}$ is the covariance matrix of the cross-temporal one-step ahead in-sample forecast errors.
	\item[oct$(sam)$ -] MinT-sam:  $\Omegavet_{ct} = \widehat{\Omegavet}_{ct}$.
\end{itemize}

 \begin{table}[!h]
 	\centering
 	\footnotesize
 	\begin{tabular}{>{\raggedleft\arraybackslash}m{0.15\linewidth}|>{\centering\arraybackslash}m{0.35\linewidth}|>{\centering\arraybackslash}m{0.35\linewidth}}
 		\toprule
 		                & \textbf{Cross-sectional framework}                                                     & \textbf{Temporal framework}                                                                        \\
 		\midrule
 		identity        & cs$(ols)$: $\Wvet = \Ivet_n$                                                           & te$(ols)$: $\Omegavet = \Ivet_{k^\ast + m}$                                                        \\[0.1cm]
 		structural      & cs$(struc)$: $\Wvet = \mathrm{diag}(\Svet_{cs} \mathbf{1}_{nb})$                       & te$(struc)$: $\Omegavet = \mathrm{diag}(\Svet_{te} \mathbf{1}_{m})$                                \\[0.1cm]
 		series variance & cs$(wls)$: $\Wvet = \widehat{\Wvet}_D = \Ivet_n \odot \widehat{\Wvet}$                 & te$(wlsv)$: $\Omegavet = \widehat{\Omegavet}_{wlsv}$                                               \\[0.1cm]
 		MinT-shr        & cs$(shr)$: $\Wvet = \hat{\lambda}\widehat{\Wvet}_D + (1-\hat{\lambda})\widehat{\Wvet}$ & te$(shr)$: $\Omegavet = \hat{\lambda}\widehat{\Omegavet}_D + (1-\hat{\lambda})\widehat{\Omegavet}$ \\[0.1cm]
 		MinT-sam        & cs$(sam)$: $\Wvet = \widehat{\Wvet}$                                                   & te$(sam)$: $\Omegavet = \widehat{\Omegavet}$                                                       \\
 		\bottomrule \addlinespace[0.1cm]
 		\multicolumn{3}{p{0.9\linewidth}}{\footnotesize \textbf{Note:} $\widehat{\Wvet}$ ($\widehat{\Omegavet}$) is the covariance matrix of the cross-sectional (temporal) one-step ahead in-sample forecast errors, $\widehat{\Omegavet}_{wlsv}$ is a diagonal matrix presented by \cite{athanasopoulos2017}, and $\widehat{\Omegavet}_D = \Ivet_{k^\ast + m} \odot \widehat{\Omegavet}$, where $\odot$ denotes the Hadamard product.}
 	\end{tabular}
 	\caption{Approximations for cross-sectional ($\Wvet$) and temporal ($\Omegavet$) covariance matrices.}
 	\label{tab:cov_app}
\end{table}

\clearpage
\section{Alternative forms of the cross-temporal covariance matrix}\label{app:shr}
\setcounter{table}{0} 
\setcounter{figure}{0} 

In this appendix, some derivations of the solutions proposed in Section 4 to obtain an estimator of the cross-temporal covariance matrix are reported.
Starting from the the definition of cross-temporal covariance matrix we obtain the first equivalence in (10). Therefore, we have that
\begin{align*}
	\lambda \widehat{\Omegavet}_{\textit{hf-bts}, D} &+ (1-\lambda) \widehat{\Omegavet}_{\textit{hf-bts}}\\
	&\Downarrow\\
	\widehat{\Omegavet}_{HB} & = \Svet_{ct}\left[\lambda \widehat{\Omegavet}_{\textit{hf-bts}, D} + (1-\lambda) \widehat{\Omegavet}_{\textit{hf-bts}}\right]\Svet_{ct}'                                                                        \\
	                         & = \lambda \Svet_{ct}\widehat{\Omegavet}_{\textit{hf-bts}, D}\Svet_{ct}'+ (1-\lambda) \Svet_{ct}\widehat{\Omegavet}_{\textit{hf-bts}}\Svet_{ct}'.
\end{align*}
The high-frequency time series representation (the second equivalence) can be derived in the following manner:
\begin{align*}
	\widetilde{\Omegavet} & = \Svet_{ct}\Omegavet_{\textit{hf-bts}}\Svet_{ct}'                                                                                                                                                            \\
	          & = \left(\Svet_{cs} \otimes \Svet_{te}\right)\Omegavet_{\textit{hf-bts}}\left(\Svet_{cs} \otimes \Svet_{te}\right)'                                                                                            \\
	          & = \left(\Ivet_n \otimes \Svet_{te}\right)\left(\Svet_{cs} \otimes \Ivet_{m+k^\ast}\right)\Omegavet_{\textit{hf-bts}}\left(\Svet_{cs} \otimes \Ivet_{m+k^\ast}\right)'\left(\Ivet_n \otimes \Svet_{te}\right)' \\
	          & = \left(\Ivet_n \otimes \Svet_{te}\right)\Omegavet_{\textit{hf}}\left(\Ivet_n \otimes \Svet_{te}\right)'
\end{align*}
where $\Omegavet_{\textit{hf}} = \left(\Svet_{cs} \otimes \Ivet_{m+k^\ast}\right)\Omegavet_{\textit{hf-bts}}\left(\Svet_{cs} \otimes \Ivet_{m+k^\ast}\right)'$ and $\Svet_{ct} = \Svet_{cs} \otimes \Svet_{te} = \left(\Ivet_n \otimes \Svet_{te}\right)\left(\Svet_{cs} \otimes \Ivet_{m+k^\ast}\right)$. We can apply the shrinkage estimator as
\begin{align*}
	\lambda \widehat{\Omegavet}_{hf, D} &+ (1-\lambda) \widehat{\Omegavet}_{\textit{hf}}\\
	&\Downarrow\\
	\widehat{\Omegavet}_{H} & = (\Ivet_{n} \otimes \Svet_{te})\left[\lambda \widehat{\Omegavet}_{hf, D} + (1-\lambda) \widehat{\Omegavet}_{\textit{hf}}\right] (\Ivet_{n} \otimes \Svet_{te})'                                                                                                            \\
	                        & = \lambda (\Ivet_{n} \otimes \Svet_{te})\widehat{\Omegavet}_{hf, D}(\Ivet_{n} \otimes \Svet_{te})' + (1-\lambda) (\Ivet_{n} \otimes \Svet_{te})\widehat{\Omegavet}_{\textit{hf}}(\Ivet_{n} \otimes \Svet_{te})'.
\end{align*}
The bottom time series representation (the third equivalence) follows by
\begin{align*}
	\widetilde{\Omegavet} & = \Svet_{ct}\Omegavet_{\textit{hf-bts}}\Svet_{ct}'                                                                                                                                                   \\
	          & = \left(\Svet_{cs} \otimes \Svet_{te}\right)\Omegavet_{\textit{hf-bts}}\left(\Svet_{cs} \otimes \Svet_{te}\right)'                                                                                   \\
	          & = \left(\Svet_{cs} \otimes \Ivet_{m+k^\ast}\right)\left(\Ivet_n \otimes \Svet_{te}\right)\Omegavet_{\textit{hf-bts}}\left(\Ivet_n \otimes \Svet_{te}\right)'\left(\Ivet_n \otimes \Svet_{te}\right)' \\
	          & = \left(\Svet_{cs} \otimes \Ivet_{m+k^\ast}\right)\Omegavet_{bts}\left(\Svet_{cs} \otimes \Ivet_{m+k^\ast}\right)',
\end{align*}
where $\Omegavet_{bts} = \left(\Ivet_n \otimes \Svet_{te}\right)\Omegavet_{\textit{hf-bts}}\left(\Ivet_n \otimes \Svet_{te}\right)'$ and $\Svet_{ct} = \Svet_{cs} \otimes \Svet_{te} = \left(\Svet_{cs} \otimes \Ivet_{m+k^\ast}\right)\left(\Ivet_n \otimes \Svet_{te}\right)$. Finally we have that
\begin{align*}
	\lambda \widehat{\Omegavet}_{bts, D} &+ (1-\lambda) \widehat{\Omegavet}_{bts}\\
	&\Downarrow\\
	\widehat{\Omegavet}_{B} & = \left(\Svet_{cs} \otimes \Ivet_{m+k^\ast}\right)\left[\lambda \widehat{\Omegavet}_{bts, D} + (1-\lambda) \widehat{\Omegavet}_{bts}\right]\left(\Svet_{cs} \otimes \Ivet_{m+k^\ast}\right)'                       \\
	                        & = \lambda \left(\Svet_{cs} \otimes \Ivet_{m+k^\ast}\right)\widehat{\Omegavet}_{bts, D}\left(\Svet_{cs} \otimes \Ivet_{m+k^\ast}\right)' +             \\
	                        & \qquad \qquad (1-\lambda) \left(\Svet_{cs} \otimes \Ivet_{m+k^\ast}\right)\widehat{\Omegavet}_{bts}\left(\Svet_{cs} \otimes \Ivet_{m+k^\ast}\right)'.
\end{align*}

In general, the covariance matrix of the reconciled forecasts is equal to $\Mvet \widehat{\Omegavet} \Mvet'$ where $\Mvet = \Svet_{ct}\Gvet$ is the projection matrix. When using the HB approach, the covariance matrix of the reconciliation and the base forecasts will be identical. Indeed, it can be shown (see \citealp{panagiotelis2021} for more details) that if $\Mvet$ is a projection matrix (6) then $\Mvet\Svet_{ct} = \Svet_{ct}\Gvet\Svet_{ct} = \Svet_{ct}$, and we obtain that
\begin{align*}
	\Mvet \widehat{\Omegavet}_{HB} \Mvet' & = \Mvet\Svet_{ct}\widehat{\Omegavet}_{\textit{hf-bts}, HB}\Svet_{ct}'\Mvet'                      \\
	& = \Svet_{ct}\Gvet\Svet_{ct}\widehat{\Omegavet}_{\textit{hf-bts}, HB}\Svet_{ct}'\Gvet'\Svet_{ct}' \\
	& = \Svet_{ct}\widehat{\Omegavet}_{\textit{hf-bts}, HB}\Svet_{ct}' = \widehat{\Omegavet}_{HB}.
\end{align*}

\autoref{fig:num_param} shows the number of parameters for different values of $m$ and $n$, with $n_b$ fixed to approximately $60\%$ of $n$. The right panel reports the boxplot of the percentage reductions in the number of parameters compared to the global approach.
\autoref{fig:shr_grid} gives some visual insights on the covariance matrices obtainable with $\lambda=0$ and $\lambda=1$, respectively, for a simple cross-temporal hierarchical structure with 3 time series and $\mathcal{K}=\{2,1\}$ (e.g, cross-temporal semi-annual, see the Monte Carlo simulation in \autoref{sec:mcsim}).

\begin{figure}[!t]
	\centering
	\includegraphics[width = \linewidth]{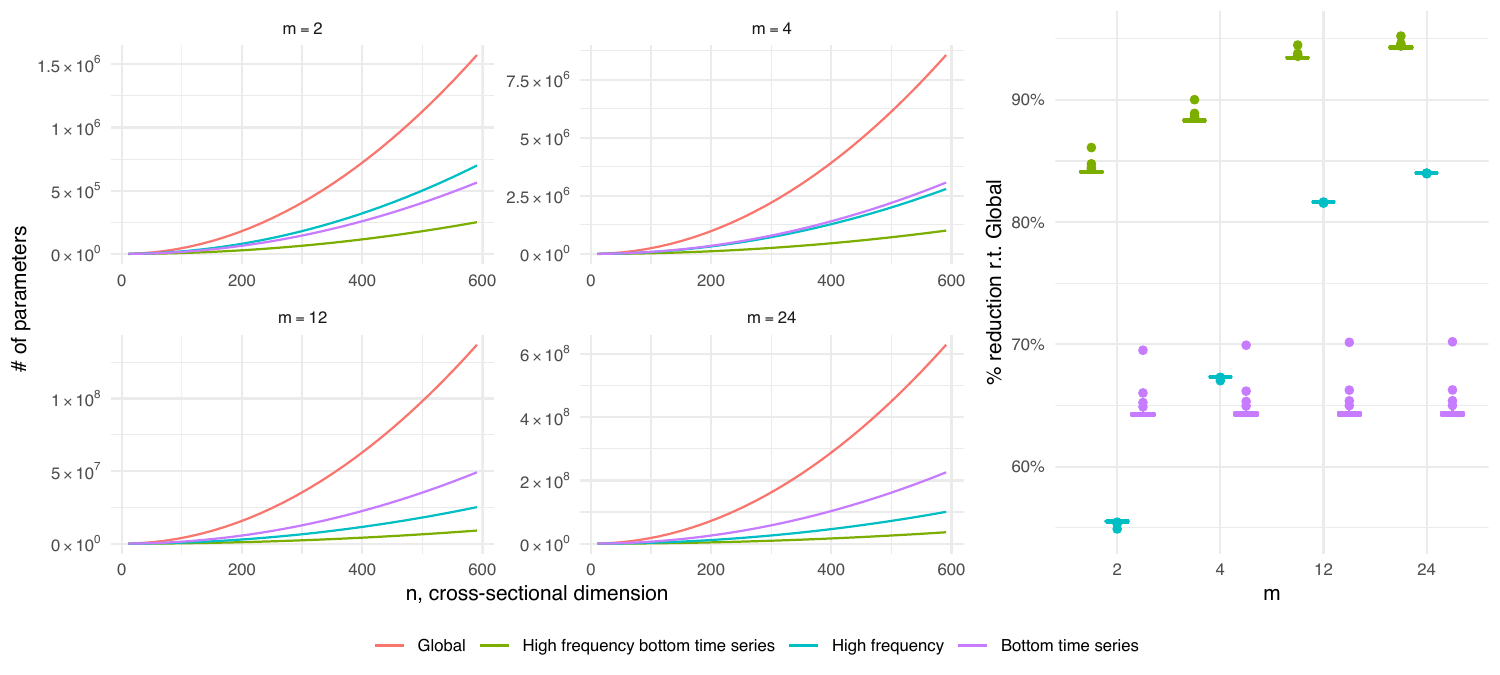}
	\caption{The four graphs on the left represent the number of different parameters in the covariance matrix for the various approaches presented for different values of $m$ and $n$ ($n_b$, the number of bottom time series, is about $60\%$ of the total). On the right, we have the boxplot of the percentage reduction in the number of parameters compared to the global approach.}
	\label{fig:num_param}
\end{figure}

\clearpage
\begin{figure}[!t]
	\centering
	\includegraphics[width = \linewidth]{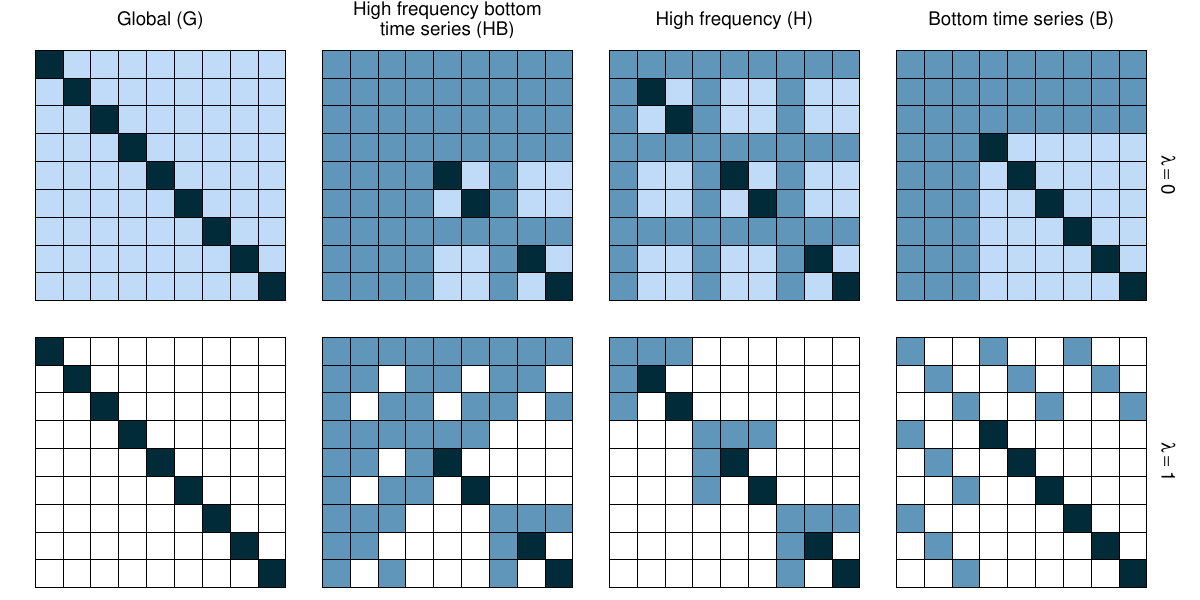}
	\caption{Representation of four types of covariance matrices that can be obtained from the cross-temporal hierarchical structure ($3$ time series and $m = 2$) for two different values of $\lambda\in\{0,1\}$, the shrinkage parameter. The cells that are not modified by shrinkage are colored black, those actively involved in the shrinkage phase are colored light blue, and those derived from and not estimated by the base forecasts errors are colored blue. Additionally, for $\lambda = 1$, the cells corresponding to a zero value are colored white.}
	\label{fig:shr_grid}
\end{figure}

\subsection{Proof Theorem 4.1}
%\subsection*{Using the $HB$ covariance matrix in the reconciliation phase is equivalent to perform an $ols$ reconciliation}
%\begin{theorem}
%	Let $\widehat{\Omega}_{hf-bts}$ a $[(n_bm)\times (n_bm)]$ p.d. matrix. Then, using $\Omegavet_{ct} = \Svet_{ct}\widehat{\Omega}_{hfbts}\Svet_{ct}'$ in the reconciliation formulaw (5) and (6) is equivalent to  $\Omegavet_{ct} = \Ivet_{n(m+k^\ast)}$ ($ols$ approach).
%\end{theorem}
\begin{proof}
Let $\Svet_{ct}^+ = (\Svet_{ct}'\Svet_{ct})^{-1}\Svet_{ct}'$ be the generalized inverse of $\Svet_{ct}$ ($\Svet_{ct}$ has linearly independent columns by construction). Applying (6), we obtain
$$
\widetilde{\bvet}_{ols}^{[1]} = (\Svet_{ct}'\Svet_{ct})^{-1}\Svet_{ct}' \widehat{\xvet}_{h}
$$
and
$$
\widetilde{\bvet}_{hb}^{[1]} = [\Svet_{ct}'(\Svet_{ct}\widehat{\Omega}_{hf-bts}\Svet_{ct}')^{+}\Svet_{ct}]^{-1}\Svet_{ct}' (\Svet_{ct}\widehat{\Omega}_{hf-bts}\Svet_{ct}')^{+}\widehat{\xvet}_{h} .
$$
Since $(\Avet\Bvet)^+ = \Bvet^+\Avet^+$ \citep{greville1966}, then
\begin{equation}\label{eq:part1}\tag{B.1}
\begin{aligned}
	\Svet_{ct}'(\Svet_{ct}\widehat{\Omega}_{hf-bts}\Svet_{ct}')^{+}\Svet_{ct} &= \Svet_{ct}'(\widehat{\Omega}_{hf-bts}^{1/2}\Svet_{ct}')^{+}(\Svet_{ct}\widehat{\Omega}_{hf-bts}^{1/2})^{+}\Svet_{ct} \\
	&= \Svet_{ct}'(\Svet_{ct}')^{+}\widehat{\Omega}_{hf-bts}^{+}\Svet_{ct}^{+}\Svet_{ct} = \widehat{\Omega}_{hf-bts}^{+} = \widehat{\Omega}_{hf-bts}^{-1}
\end{aligned}
\end{equation}
and
\begin{equation}\label{eq:part2}\tag{B.2}
\begin{aligned}
	\Svet_{ct}' (\Svet_{ct}\widehat{\Omega}_{hf-bts}\Svet_{ct}')^{+} &= \Svet_{ct}'(\widehat{\Omega}_{hf-bts}^{1/2}\Svet_{ct}')^{+}(\Svet_{ct}\widehat{\Omega}_{hf-bts}^{1/2})^{+} \\& = \widehat{\Omega}_{hf-bts}^{-1}\Svet_{ct}^+ = \widehat{\Omega}_{hf-bts}^{-1}(\Svet_{ct}'\Svet_{ct})^{-1}\Svet_{ct}'.
\end{aligned}
\end{equation}
Therefore,
\begin{align*}
	\widetilde{\bvet}_{hb}^{[1]} &= [\Svet_{ct}'(\Svet_{ct}\widehat{\Omega}_{hf-bts}\Svet_{ct}')^{+}\Svet_{ct}]^{-1}\Svet_{ct}' (\Svet_{ct}\widehat{\Omega}_{hf-bts}\Svet_{ct}')^{+}\widehat{\xvet}_{h} \\ & \overset{(\ref{eq:part1})}{=} \widehat{\Omega}_{hf-bts}\Svet_{ct}' (\Svet_{ct}\widehat{\Omega}_{hf-bts}\Svet_{ct}')^{+}\widehat{\xvet}_{h}\\
	& \overset{(\ref{eq:part2})}{=} \widehat{\Omega}_{hf-bts}\widehat{\Omega}_{hf-bts}^{-1}(\Svet_{ct}'\Svet_{ct})^{-1}\Svet_{ct}' \widehat{\xvet}_{h} \\
	& = (\Svet_{ct}'\Svet_{ct})^{-1}\Svet_{ct}' \widehat{\xvet}_{h} = \widetilde{\bvet}_{ols}^{[1]} 
\end{align*}	
\end{proof}

\newpage
\section{Monte Carlo simulation}\label{sec:mcsim}

We study the effect of combining cross-sectional and temporal aggregations, using a simple hierarchy that allows us to effectively visualize the quantities involved, such as the covariance matrix. Additionally, the small size and nature of the data generating process make it possible to exactly calculate the true cross-temporal covariance structure, thus providing insights into the nature of the time series data involved in the forecast reconciliation process.

Consider a $2$-level hierarchical structure with three time series (one upper series, $A$, and two bottom series, $B$ and $C$) such that the cross-sectional aggregation matrix is $\Avet_{cs} = \left[ 1 \quad 1 \right]$ ($A = B+C$). The temporal structure we are considering is equivalent to using semi-annual data with $\mathcal{K} = \{2,1\}$ and $m = 2$. The assumed Data-Generating Processes (DPG) for the semi-annual bottom level series are two AR(2) given by
$$
\begin{aligned}
	y_{B,t} &= \phi_{B, 1} y_{B,t-1} + \phi_{B, 2} y_{B,t-2} + \varepsilon_{B, t}\\
	y_{C,t} &= \phi_{C, 1} y_{C,t-1} + \phi_{C, 2} y_{C,t-2} + \varepsilon_{C, t}
\end{aligned}
$$
with parameters\footnote{The $\phivet_B$ and $\phivet_C$ parameters are estimated from the “Lynx" and “Hare" time series contained in the \texttt{pelt} dataset of the \texttt{tsibbledata} package for R \citep{ohara-wild2022}.} $\phivet_B = [\phi_{B,1}\; \phi_{B,2}]' = [1.34\; -0.74]'$ and $\phivet_C  = [\phi_{C,1}\; \phi_{C,2}]' = [0.95\;~-~0.42]'$. The error $\epsvet_t = \left[\varepsilon_{B, t}\quad \varepsilon_{C, t}\right]'$ driving the process is drawn from a multivariate normal distribution with standard deviations simulated from a uniform distribution between 0.5 and 2 and a fixed correlation of -0.8. The cross-sectional error covariance matrix is thus given by
$$
	\Omegavet_{cs} = \begin{bmatrix}
		0.9 & 0   \\
		0   & 1.8
	\end{bmatrix} \begin{bmatrix}
		1    & \rho \\
		\rho & 1
	\end{bmatrix} \begin{bmatrix}
		0.9 & 0   \\
		0   & 1.8
	\end{bmatrix} = \begin{bmatrix}
		\sigma_B^2  & \sigma_{BC} \\
		\sigma_{BC} & \sigma_C^2
	\end{bmatrix}.
$$
To obtain the remaining series, the bottom series are then cross-temporally aggregated.

For the forecast experiment, the base forecasts are computing using AR models where the order is automatically determined by the algorithm proposed by \cite{hyndman2008a} and implemented in the R package \texttt{forecast} \citep{Rforecast}, thus allowing for possible mis-specification in the models. The training window length is 500 years, consisting of 1000 high frequency observations. The experiment is replicated 500 times, with a forecast horizon of 1 year.

Since the AR(2) models used as DPG for the bottom series $B$ and $C$ at the most disaggregated temporal level are known, we may compute the true covariance matrix for one-step ahead forecasts at the annual level $\Omegavet_{ct} = \Svet_{ct}\Omegavet_{\textit{hf-bts}}\Svet_{ct}'$, where
$$
	\Omegavet_{\textit{hf-bts}} = \begin{bmatrix}
		\sigma^2_B            &                                                 &                      &                                        \\
		\phi_{B,1}\sigma_B^2  & \sigma_B^2\left(1+\phi_{B,1}^2\right)           &                      &                                        \\
		\sigma_{BC}           & \phi_{B,1}\sigma_{BC}                           & \sigma_C^2           &                                        \\
		\phi_{C,1}\sigma_{BC} & \sigma_{BC}\left(1+\phi_{B,1}\phi_{C,1} \right) & \phi_{C,1}\sigma_C^2 & \sigma_C^2\left(1+\phi_{C,1}^2\right)\
	\end{bmatrix}.
$$
The detailed calculations can be found in \autoref{app:ar2}.
\autoref{fig:covcorMC} shows both the covariance matrix and the correlation matrix for fixed parameters.

\begin{figure}[!t]
	\centering
	\includegraphics[width = \linewidth]{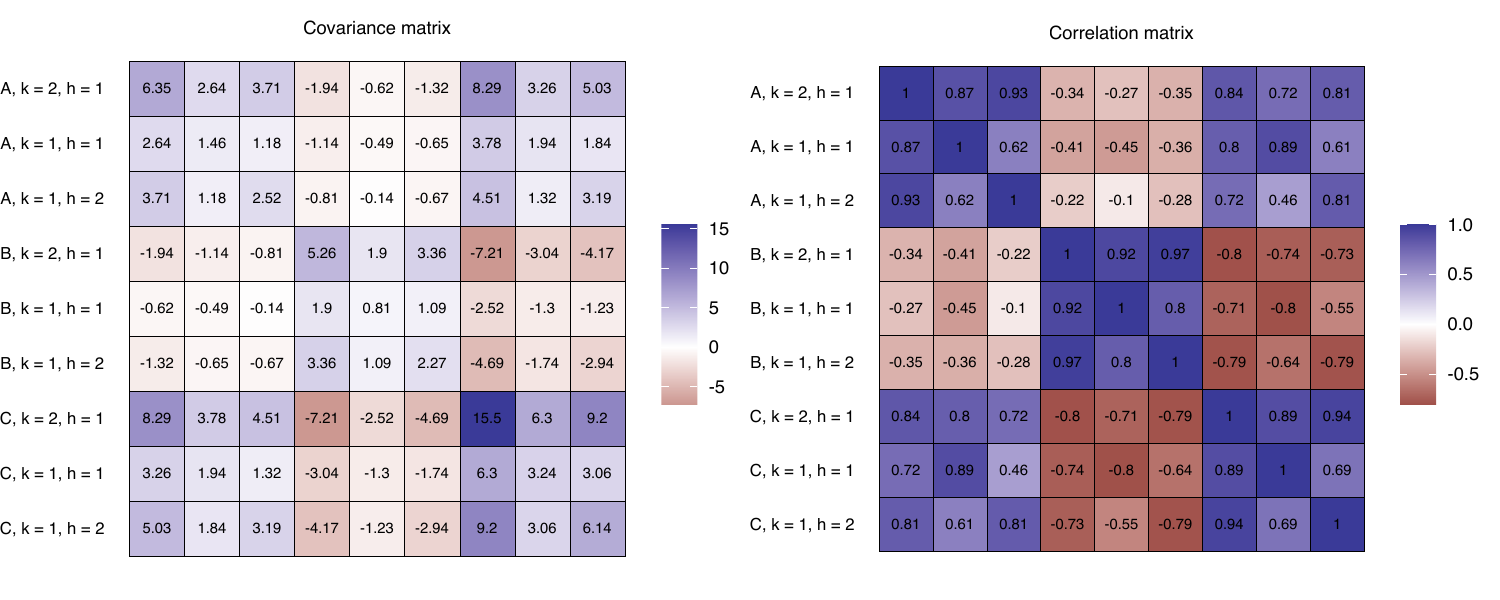}
	\caption{Simulation experiment. True cross-temporal covariance (left) and correlation (right) error matrix of the reconciled forecasts with $\sigma_B = 0.9$, $\sigma_C = 1.8$, $\phivet_B = [1.34\; -0.74]'$, $\phivet_C = [0.95\; -0.42]'$ and $\rho = -0.8$.}
	\label{fig:covcorMC}
\end{figure}

To construct cross-temporal samples of the base forecasts, we use the Gaussian and bootstrap approaches discussed in Sections 3.1 and 3.2, respectively. For the parametric approach we use multi-step residuals with the different covariance matrix structures analyzed in 4.1, while for the non-parametric approach, we use regular one-step residuals. We do not use overlapping residuals in our analysis as we have the advantage of generating a large number of observation. Ten different reconciliation approaches have been adopted (see Table 2): ct$(bu)$, ct$(shr_{cs}, bu_{te})$, ct$(wlsv_{te}, bu_{cs})$, oct$(wlsv)$, oct$(bdshr)$, oct$_h(shr)$, oct$_h(bshr)$, and oct$_h(hshr)$.

\subsection{Covariance matrix comparison and forecast accuracy scores}\label{ssec:acc_scores}
\setcounter{table}{0} 
\setcounter{figure}{0} 

To compare the true covariance matrix $\Omegavet_{ct}$ with the estimated covariance matrix $\widehat{\Omegavet}$, we use the Frobenius norm to quantify the difference between two matrices: 
$$\lVert \Dvet \rVert_F = \displaystyle\sqrt{\sum_{i = 1}^{n(k^\ast + m)}\sum_{j = 1}^{n(k^\ast + m)}|d_{i,j}|^2}$$ where $\Dvet = \widehat{\Omegavet} - \Omegavet_{ct}$. The true covariance matrix, shown in \autoref{fig:covcorMC}, was compared to the estimated covariance matrices obtained using various reconciliation approaches and techniques for generating sample paths of the base forecasts. Thus, we should be able to determine which reconciliation approach and simulation technique produce an accurate estimate of the covariance matrix. Other types of matrix norms were also considered with similar results.

From \autoref{tab:ar2norm}, it appears that the reconciled covariance matrices are always closer to the true matrix than the base forecast matrix when using both the Gaussian and the bootstrap  approach. Overall, there are no major differences in the findings when using either one-step or multi-step residuals in cross-temporal forecast reconciliation. In fact, using approaches like oct$(bdshr)$, we obtain results that are consistent with approaches such as oct$_h(shr)$, where no temporal and/or cross-sectional correlation assumptions are imposed. It is worth noting that the $HB$ covariance matrix when used to calculate the base forecasts samples, is not changed by the reconciliation step (see \autoref{app:shr}). In conclusion, our results suggest that using multi-step residuals or bootstrap techniques may help find a “good" estimate of the covariance matrix, which can be further improved by the reconciliation.

\begin{table}[!t]
	\centering
	\begingroup
	\spacingset{1}
	\fontsize{9}{11}\selectfont
	
\begin{tabular}[t]{l|>{}ccccc}
\toprule
\multicolumn{1}{c}{\textbf{}} & \multicolumn{5}{c}{\textbf{Generation of the base forecasts paths}} \\
\cmidrule(l{0pt}r{0pt}){2-6}
\multicolumn{1}{c}{\makecell[c]{\bfseries Reconciliation\\\bfseries approach}} & \multicolumn{1}{c}{ctjb} & \multicolumn{4}{c}{\makecell[c]{Gaussian approach\textsuperscript{*}}} \\
\multicolumn{1}{c}{} &  & G & B & H & HB\\
\midrule
base & \textcolor{black}{8.260} & \textcolor{black}{7.748} & \textcolor{black}{6.549} & \textcolor{black}{3.409} & \textcolor{black}{2.215}\\
ct$(bu)$ & \textcolor{black}{3.195} & \textcolor{black}{2.215} & \textcolor{black}{2.215} & \textcolor{black}{\textbf{2.215}} & \textcolor{black}{2.215}\\
ct$(shr_{cs}, bu_{te})$ & \textcolor{black}{3.202} & \textcolor{black}{2.224} & \textcolor{black}{2.215} & \textcolor{black}{2.224} & \textcolor{black}{2.215}\\
ct$(wlsv_{te}, bu_{cs})$ & \textcolor{black}{\textbf{3.183}} & \textcolor{black}{\textbf{2.188}} & \textcolor{black}{2.188} & \textcolor{black}{\textbf{2.215}} & \textcolor{black}{2.215}\\
oct$(wlsv)$ & \textcolor{black}{3.766} & \textcolor{black}{3.082} & \textcolor{black}{2.191} & \textcolor{black}{2.910} & \textcolor{black}{2.215}\\
oct$(bdshr)$ & \textcolor{black}{3.203} & \textcolor{black}{2.195} & \textcolor{blue}{\textbf{2.184}} & \textcolor{black}{2.224} & \textcolor{black}{\textbf{2.215}}\\
oct$_h(shr)$ & \textcolor{black}{3.251} & \textcolor{black}{2.260} & \textcolor{black}{2.202} & \textcolor{black}{2.226} & \textcolor{black}{2.215}\\
oct$_h(bshr)$ & \textcolor{black}{3.602} & \textcolor{black}{2.720} & \textcolor{black}{2.220} & \textcolor{black}{2.756} & \textcolor{black}{2.215}\\
oct$_h(hshr)$ & \textcolor{black}{4.869} & \textcolor{black}{4.138} & \textcolor{black}{4.167} & \textcolor{black}{2.225} & \textcolor{black}{2.215}\\
\bottomrule
\multicolumn{6}{l}{\rule{0pt}{1em}\rule{0pt}{1.75em}\makecell[l]{$^\ast$The Gaussian method employs a sample covariance\\ with multi-step residuals.}}\\
\end{tabular}

	\endgroup
	\caption{Simulation experiment. Frobenius norm between the true and the estimated covariance matrix for different reconciliation approaches and different techniques for simulating the base forecasts. Entries in bold represent the lowest value for each column, while the blue entry represent the global minimum. The reconciliation approaches are described in Table 2.}
	\label{tab:ar2norm}
\end{table}

A limitation of this simulation setting is that we are using a high number of residuals, which may result in undervaluing the benefit from using the parameterization form of the covariance matrix such as $HB$, $H$, or $B$ (see Section 4). Additionally, shrinkage techniques often yield very similar results when we use the corresponding matrix with $\lambda = 0$ (full covariance matrix). 

In Tables \ref{tab:ar2crps} and \ref{tab:ar2es} are reported the $\operatorname{\overline{RelCRPS}}$ and ES ratio indices introduced in Sections 5 where low values indicate better quality of the forecasts. The good performance of the ct$(bu)$ approach can be explained by a good quality of the base forecasts at the bottom level for $k=1$, and therefore it is difficult for the other approaches to correctly adjust them using the somewhat less good forecasts of the higher temporal and cross-sectional levels. This also explains the good performance of ct$(shr_{cs}, bu_{te})$, which by definition only takes into account the information provided by the most temporally disaggregated base forecasts.
Looking at the optimal cross-temporal reconciliation approaches, it does not seem to be any advantage in using multi-step residuals to calculate the covariance matrix in the reconciliation step.

\begin{table}[!t]
	\centering
	\begingroup
	\spacingset{1}
	\fontsize{9}{11}\selectfont
	
\begin{tabular}[t]{l|>{}cccc>{}c|ccccc}
\toprule
\multicolumn{1}{c}{\textbf{}} & \multicolumn{10}{c}{\textbf{Generation of the base forecasts sample paths}} \\
\cmidrule(l{0pt}r{0pt}){2-11}
\multicolumn{1}{c}{\makecell[c]{\bfseries Reconciliation\\\bfseries approach}} & \multicolumn{1}{c}{ctjb} & \multicolumn{4}{c}{\makecell[c]{Gaussian approach\textsuperscript{*}}} & \multicolumn{1}{c}{ctjb} & \multicolumn{4}{c}{\makecell[c]{Gaussian approach\textsuperscript{*}}} \\
\multicolumn{1}{c}{} &  & G & B & H & \multicolumn{1}{c}{HB} &  & G & B & H & HB\\
\midrule
\addlinespace[0.3em]
\multicolumn{1}{c}{} & \multicolumn{5}{c}{\textbf{$\forall k \in \{2,1\}$}} & \multicolumn{5}{c}{\textbf{$k = 1$}}\\
base & \textcolor{black}{1.000} & \textcolor{black}{0.998} & \textcolor{black}{0.999} & \textcolor{red}{1.002} & \textcolor{red}{1.004} & \textcolor{black}{1.000} & \textcolor{black}{0.998} & \textcolor{black}{0.999} & \textcolor{black}{0.999} & \textcolor{black}{1.000}\\
ct$(bu)$ & \textcolor{black}{0.901} & \textcolor{black}{0.900} & \textcolor{black}{0.900} & \textcolor{black}{0.900} & \textcolor{black}{0.900} & \textcolor{black}{0.978} & \textcolor{black}{0.976} & \textcolor{black}{0.976} & \textcolor{black}{0.977} & \textcolor{black}{0.977}\\
ct$(shr_{cs}, bu_{te})$ & \textcolor{black}{\textbf{0.901}} & \textcolor{black}{\textbf{0.900}} & \textcolor{blue}{\textbf{0.899}} & \textcolor{black}{\textbf{0.900}} & \textcolor{black}{\textbf{0.900}} & \textcolor{black}{\textbf{0.977}} & \textcolor{black}{\textbf{0.976}} & \textcolor{blue}{\textbf{0.976}} & \textcolor{black}{\textbf{0.976}} & \textcolor{black}{\textbf{0.976}}\\
ct$(wlsv_{te}, bu_{cs})$ & \textcolor{black}{0.910} & \textcolor{black}{0.916} & \textcolor{black}{0.916} & \textcolor{black}{0.916} & \textcolor{black}{0.917} & \textcolor{black}{0.986} & \textcolor{black}{0.993} & \textcolor{black}{0.993} & \textcolor{black}{0.993} & \textcolor{black}{0.993}\\
oct$(wlsv)$ & \textcolor{black}{0.922} & \textcolor{black}{0.930} & \textcolor{black}{0.930} & \textcolor{black}{0.930} & \textcolor{black}{0.931} & \textcolor{black}{0.998} & \textcolor{red}{1.006} & \textcolor{red}{1.006} & \textcolor{red}{1.007} & \textcolor{red}{1.007}\\
oct$(bdshr)$ & \textcolor{black}{0.910} & \textcolor{black}{0.916} & \textcolor{black}{0.915} & \textcolor{black}{0.916} & \textcolor{black}{0.916} & \textcolor{black}{0.986} & \textcolor{black}{0.992} & \textcolor{black}{0.992} & \textcolor{black}{0.993} & \textcolor{black}{0.993}\\
oct$_h(shr)$ & \textcolor{black}{0.904} & \textcolor{black}{0.903} & \textcolor{black}{0.902} & \textcolor{black}{0.902} & \textcolor{black}{0.903} & \textcolor{black}{0.980} & \textcolor{black}{0.979} & \textcolor{black}{0.978} & \textcolor{black}{0.979} & \textcolor{black}{0.979}\\
oct$_h(bshr)$ & \textcolor{black}{0.923} & \textcolor{black}{0.922} & \textcolor{black}{0.922} & \textcolor{black}{0.921} & \textcolor{black}{0.922} & \textcolor{black}{0.999} & \textcolor{black}{0.998} & \textcolor{black}{0.998} & \textcolor{black}{0.998} & \textcolor{black}{0.998}\\
oct$_h(hshr)$ & \textcolor{black}{0.974} & \textcolor{black}{0.972} & \textcolor{black}{0.972} & \textcolor{black}{0.974} & \textcolor{black}{0.975} & \textcolor{red}{1.052} & \textcolor{red}{1.050} & \textcolor{red}{1.050} & \textcolor{red}{1.053} & \textcolor{red}{1.053}\\
\addlinespace[0.3em]
\multicolumn{1}{c}{} & \multicolumn{5}{c}{\textbf{$k = 2$}} & \multicolumn{5}{c}{}\\
base & \textcolor{black}{1.000} & \textcolor{black}{0.998} & \textcolor{black}{0.999} & \textcolor{red}{1.005} & \textcolor{red}{1.008} &  &  &  &  & \\
ct$(bu)$ & \textcolor{black}{0.831} & \textcolor{black}{0.830} & \textcolor{black}{0.829} & \textcolor{black}{0.829} & \textcolor{black}{0.830} &  &  &  &  & \\
ct$(shr_{cs}, bu_{te})$ & \textcolor{black}{\textbf{0.830}} & \textcolor{black}{\textbf{0.830}} & \textcolor{blue}{\textbf{0.829}} & \textcolor{black}{\textbf{0.829}} & \textcolor{black}{\textbf{0.830}} &  &  &  &  & \\
ct$(wlsv_{te}, bu_{cs})$ & \textcolor{black}{0.840} & \textcolor{black}{0.846} & \textcolor{black}{0.844} & \textcolor{black}{0.845} & \textcolor{black}{0.846} &  &  &  &  & \\
oct$(wlsv)$ & \textcolor{black}{0.851} & \textcolor{black}{0.859} & \textcolor{black}{0.859} & \textcolor{black}{0.859} & \textcolor{black}{0.861} &  &  &  &  & \\
oct$(bdshr)$ & \textcolor{black}{0.839} & \textcolor{black}{0.845} & \textcolor{black}{0.844} & \textcolor{black}{0.845} & \textcolor{black}{0.846} &  &  &  &  & \\
oct$_h(shr)$ & \textcolor{black}{0.834} & \textcolor{black}{0.833} & \textcolor{black}{0.831} & \textcolor{black}{0.832} & \textcolor{black}{0.832} &  &  &  &  & \\
oct$_h(bshr)$ & \textcolor{black}{0.852} & \textcolor{black}{0.851} & \textcolor{black}{0.851} & \textcolor{black}{0.851} & \textcolor{black}{0.852} &  &  &  &  & \\
oct$_h(hshr)$ & \textcolor{black}{0.902} & \textcolor{black}{0.900} & \textcolor{black}{0.899} & \textcolor{black}{0.901} & \textcolor{black}{0.902} &  &  &  &  & \\
\bottomrule
\multicolumn{11}{l}{\rule{0pt}{1em}\rule{0pt}{1.75em}\makecell[l]{$^\ast$The Gaussian method employs a sample covariance matrix and includes four techniques\\ (G, B, H, HB) with multi-step residuals.}}\\
\end{tabular}

	\endgroup
	\caption{Simulation experiment. $\overline{RelCRPS}$ defined in Section 5. Approaches performing worse than the benchmark (bootstrap base forecasts, ctjb) are highlighted in red, the best for each column is marked in bold, and the overall lowest value is highlighted in blue. The reconciliation approaches are described in Table 2.}
	\label{tab:ar2crps}
\end{table}

\begin{table}[!t]
	\centering
	\begingroup
	\spacingset{1}
	\fontsize{9}{11}\selectfont
	
\begin{tabular}[t]{l|>{}cccc>{}c|ccccc}
\toprule
\multicolumn{1}{c}{\textbf{}} & \multicolumn{10}{c}{\textbf{Generation of the base forecasts sample paths}} \\
\cmidrule(l{0pt}r{0pt}){2-11}
\multicolumn{1}{c}{\makecell[c]{\bfseries Reconciliation\\\bfseries approach}} & \multicolumn{1}{c}{ctjb} & \multicolumn{4}{c}{\makecell[c]{Gaussian approach\textsuperscript{*}}} & \multicolumn{1}{c}{ctjb} & \multicolumn{4}{c}{\makecell[c]{Gaussian approach\textsuperscript{*}}} \\
\multicolumn{1}{c}{} &  & G & B & H & \multicolumn{1}{c}{HB} &  & G & B & H & HB\\
\midrule
\addlinespace[0.3em]
\multicolumn{1}{c}{} & \multicolumn{5}{c}{\textbf{$\forall k \in \{2,1\}$}} & \multicolumn{5}{c}{\textbf{$k = 1$}}\\
base & \textcolor{black}{1.000} & \textcolor{black}{0.996} & \textcolor{black}{0.999} & \textcolor{black}{1.000} & \textcolor{red}{1.004} & \textcolor{black}{1.000} & \textcolor{black}{0.997} & \textcolor{red}{1.000} & \textcolor{black}{0.997} & \textcolor{red}{1.000}\\
ct$(bu)$ & \textcolor{black}{0.897} & \textcolor{black}{0.895} & \textcolor{black}{0.896} & \textcolor{black}{0.897} & \textcolor{blue}{\textbf{0.895}} & \textcolor{black}{0.969} & \textcolor{black}{\textbf{0.967}} & \textcolor{blue}{\textbf{0.967}} & \textcolor{black}{0.968} & \textcolor{black}{\textbf{0.968}}\\
ct$(shr_{cs}, bu_{te})$ & \textcolor{black}{\textbf{0.896}} & \textcolor{black}{\textbf{0.895}} & \textcolor{black}{\textbf{0.895}} & \textcolor{black}{\textbf{0.896}} & \textcolor{black}{0.896} & \textcolor{black}{\textbf{0.968}} & \textcolor{black}{0.968} & \textcolor{black}{0.967} & \textcolor{black}{\textbf{0.968}} & \textcolor{black}{0.968}\\
ct$(wlsv_{te}, bu_{cs})$ & \textcolor{black}{0.906} & \textcolor{black}{0.912} & \textcolor{black}{0.911} & \textcolor{black}{0.910} & \textcolor{black}{0.912} & \textcolor{black}{0.977} & \textcolor{black}{0.984} & \textcolor{black}{0.983} & \textcolor{black}{0.981} & \textcolor{black}{0.984}\\
oct$(wlsv)$ & \textcolor{black}{0.916} & \textcolor{black}{0.923} & \textcolor{black}{0.923} & \textcolor{black}{0.923} & \textcolor{black}{0.924} & \textcolor{black}{0.989} & \textcolor{black}{0.994} & \textcolor{black}{0.995} & \textcolor{black}{0.995} & \textcolor{black}{0.997}\\
oct$(bdshr)$ & \textcolor{black}{0.906} & \textcolor{black}{0.910} & \textcolor{black}{0.910} & \textcolor{black}{0.911} & \textcolor{black}{0.912} & \textcolor{black}{0.977} & \textcolor{black}{0.981} & \textcolor{black}{0.982} & \textcolor{black}{0.983} & \textcolor{black}{0.985}\\
oct$_h(shr)$ & \textcolor{black}{0.900} & \textcolor{black}{0.898} & \textcolor{black}{0.898} & \textcolor{black}{0.897} & \textcolor{black}{0.898} & \textcolor{black}{0.971} & \textcolor{black}{0.969} & \textcolor{black}{0.969} & \textcolor{black}{0.969} & \textcolor{black}{0.969}\\
oct$_h(bshr)$ & \textcolor{black}{0.916} & \textcolor{black}{0.914} & \textcolor{black}{0.916} & \textcolor{black}{0.915} & \textcolor{black}{0.916} & \textcolor{black}{0.987} & \textcolor{black}{0.986} & \textcolor{black}{0.987} & \textcolor{black}{0.987} & \textcolor{black}{0.988}\\
oct$_h(hshr)$ & \textcolor{black}{0.967} & \textcolor{black}{0.964} & \textcolor{black}{0.964} & \textcolor{black}{0.966} & \textcolor{black}{0.967} & \textcolor{red}{1.040} & \textcolor{red}{1.036} & \textcolor{red}{1.036} & \textcolor{red}{1.040} & \textcolor{red}{1.040}\\
\addlinespace[0.3em]
\multicolumn{1}{c}{} & \multicolumn{5}{c}{\textbf{$k = 2$}} & \multicolumn{5}{c}{}\\
base & \textcolor{black}{1.000} & \textcolor{black}{0.996} & \textcolor{black}{0.998} & \textcolor{red}{1.003} & \textcolor{red}{1.008} &  &  &  &  & \\
ct$(bu)$ & \textcolor{black}{0.831} & \textcolor{black}{0.829} & \textcolor{black}{0.829} & \textcolor{black}{0.830} & \textcolor{blue}{\textbf{0.828}} &  &  &  &  & \\
ct$(shr_{cs}, bu_{te})$ & \textcolor{black}{\textbf{0.829}} & \textcolor{black}{\textbf{0.828}} & \textcolor{black}{\textbf{0.829}} & \textcolor{black}{\textbf{0.829}} & \textcolor{black}{0.829} &  &  &  &  & \\
ct$(wlsv_{te}, bu_{cs})$ & \textcolor{black}{0.839} & \textcolor{black}{0.844} & \textcolor{black}{0.844} & \textcolor{black}{0.844} & \textcolor{black}{0.845} &  &  &  &  & \\
oct$(wlsv)$ & \textcolor{black}{0.849} & \textcolor{black}{0.858} & \textcolor{black}{0.856} & \textcolor{black}{0.856} & \textcolor{black}{0.857} &  &  &  &  & \\
oct$(bdshr)$ & \textcolor{black}{0.839} & \textcolor{black}{0.845} & \textcolor{black}{0.843} & \textcolor{black}{0.845} & \textcolor{black}{0.844} &  &  &  &  & \\
oct$_h(shr)$ & \textcolor{black}{0.835} & \textcolor{black}{0.833} & \textcolor{black}{0.833} & \textcolor{black}{0.831} & \textcolor{black}{0.832} &  &  &  &  & \\
oct$_h(bshr)$ & \textcolor{black}{0.850} & \textcolor{black}{0.847} & \textcolor{black}{0.849} & \textcolor{black}{0.849} & \textcolor{black}{0.850} &  &  &  &  & \\
oct$_h(hshr)$ & \textcolor{black}{0.900} & \textcolor{black}{0.897} & \textcolor{black}{0.896} & \textcolor{black}{0.897} & \textcolor{black}{0.899} &  &  &  &  & \\
\bottomrule
\multicolumn{11}{l}{\rule{0pt}{1em}\rule{0pt}{1.75em}\makecell[l]{$^\ast$The Gaussian method employs a sample covariance matrix and includes four techniques\\ (G, B, H, HB) with multi-step residuals.}}\\
\end{tabular}

	\endgroup
	\caption{Simulation experiment. ES ratio indices defined in Section 5. Approaches performing worse than the benchmark (bootstrap base forecasts, ctjb) are highlighted in red, the best for each column is marked in bold, and the overall lowest value is highlighted in blue. The reconciliation approaches are described in Table 2.}
	\label{tab:ar2es}
\end{table}

In conclusion, we found that simulating base forecasts from multi-step residuals allows us to estimate a covariance matrix close to the true one. Additionally, we observed that reconciliation could be used to further improve the accuracy of these estimates: accurate base forecasts for $k=1$ assist the good performance for bottom-up and optimal cross-temporal reconciliation approaches, such as oct$(wlsv)$ and oct$(bdshr)$, which perform well in terms of both CRPS and ES.

\newpage
\subsection{Cross-temporal covariance matrix}\label{app:ar2}

We assume two AR(2) processes with correlated errors such that
$$
	y_{i,t} = \phi_{i,1}y_{i,t-1} + \phi_{i,2}y_{i,t-2} + \varepsilon_{i,t}
$$
where $\epsvet_t \sim \mathcal{N}_{2}\left(\Zerovet_{(2\times 1)}, \Omegavet_{cs}\right)$ and $i \in \{B, C\}$. Let $y_{i,T+h}$ be the true observation for the $i^{th}$ series and $\widetilde{y}_{i,T+h}$ the corresponding forecasts such that
$$
	\begin{array}{rl}
		y_{i,T+1} & = \phi_{i,1}y_{i,T} + \phi_{i,2}y_{i,T-1} + \varepsilon_{i,T+1} \\
		y_{i,T+2} & = \phi_{i,1}y_{i,T+1} + \phi_{i,2}y_{i,T} + \varepsilon_{i,T+2}
	\end{array}
	\quad\text{and}\quad
	\begin{array}{rl}
		\widetilde{y}_{i,T+1} & = \phi_{i,1}y_{i,T} + \phi_{i,2}y_{i,T-1}             \\
		\widetilde{y}_{i,T+2} & = \phi_{i,1}\widetilde{y}_{i,T+1} + \phi_{i,2}y_{i,T}
	\end{array}\;,
$$
then
\begin{align*}
	y_{i,T+1} - \widetilde{y}_{i,T+1} & = \varepsilon_{i,T+1}                                   \\
	y_{i,T+2} - \widetilde{y}_{i,T+2} & = \varepsilon_{i,T+2} + \phi_{i,1} \varepsilon_{i,T+1}.
\end{align*}
Finally, we can compute each element of the high frequency bottom time series covariance matrix
\begin{align*}\allowdisplaybreaks[2]
	Var\left(y_{i,T+1}-\widetilde{y}_{i,T+1}\right) & = \sigma_i^2, \quad \forall i \in \{B, C\}                                                                    \\
	Var\left(y_{i,T+2}-\widetilde{y}_{i,T+2}\right) & = \sigma_i^2\left(1+\phi_{i,1}^2\right), \quad \forall i \in \{B, C\}\\
	Cov\left[\left(y_{i,T+2}-\widetilde{y}_{i,T+2}\right), \left(y_{i,T+1}-\widetilde{y}_{i,T+1}\right)\right] & = 	Cov\left[\left(y_{i,T+1}-\widetilde{y}_{i,T+1}\right), \left(y_{i,T+2}-\widetilde{y}_{i,T+2}\right)\right] \\
		& = \phi_{i, 1}\sigma_{i}^2, \quad \forall i \in \{B, C\} \\
	Cov\left[\left(y_{i,T+1}-\widetilde{y}_{i,T+1}\right), \left(y_{j,T+1}-\widetilde{y}_{j,T+1}\right)\right] & = 	Cov\left[\left(y_{j,T+1}-\widetilde{y}_{j,T+1}\right), \left(y_{i,T+1}-\widetilde{y}_{i,T+1}\right)\right] \\
	& = \sigma_{i, j}, \quad \forall i,j \in \{B, C\}, \quad i\neq j\\
	Cov\left[\left(y_{i,T+2}-\widetilde{y}_{i,T+2}\right), \left(y_{j,T+1}-\widetilde{y}_{j,T+1}\right)\right] & = 	Cov\left[\left(y_{j,T+1}-\widetilde{y}_{j,T+1}\right), \left(y_{i,T+2}-\widetilde{y}_{i,T+2}\right)\right] \\
	& = \phi_{i,1}\sigma_{i, j}, \quad \forall i,j \in \{B, C\}, \quad i\neq j                                      \\
	Cov\left[\left(y_{i,T+2}-\widetilde{y}_{i,T+2}\right), \left(y_{j,T+2}-\widetilde{y}_{j,T+2}\right)\right] & = 	Cov\left[\left(y_{j,T+2}-\widetilde{y}_{j,T+2}\right), \left(y_{i,T+2}-\widetilde{y}_{i,T+2}\right)\right] \\
	& = \sigma_{i, j}\left(1+\phi_{i,1}\phi_{j,1}\right), \quad \forall i,j \in \{B, C\}, \quad i\neq j.
\end{align*}
In conclusion,
$$
	\Omegavet_{\textit{hf-bts}} = \begin{bmatrix}
		\sigma^2_B & & & \\
		\phi_{B,1}\sigma_B^2  & \sigma_B^2\left(1+\phi_{B,1}^2\right) & & \\
		\sigma_{BC} & \phi_{B,1}\sigma_{BC} & \sigma_C^2 & \\
		\phi_{C,1}\sigma_{BC} & \sigma_{BC}\left(1+\phi_{B,1}\phi_{C,1} \right) & \phi_{C,1}\sigma_C^2 & \sigma_C^2\left(1+\phi_{C,1}^2\right) \\
	\end{bmatrix}
$$
and
$$
	\Omegavet_{ct} = \Svet_{ct}\Omegavet_{\textit{hf-bts}}\Svet_{ct}'.
$$

\newpage
\subsection{One-step residuals and shrinkage covariance matrix}

In Section 4.1, we discussed the use of one-step residuals in estimating the covariance matrix. In particular we point out that one-step residuals lead to a biased estimate of the covariance matrix where some correlation are zeros by definition (see \autoref{fig:ar2covcor_app}). 
In addition, Tables \ref{tab:ar2norm_app}, \ref{tab:ar2crps_app} and \ref{tab:ar2es_app} show the Frobenius norm, CRPS, and ES skill scores as explained in the paper to investigate the effectiveness of one-step residuals. 
Moreover, in Tables \ref{tab:ar2crps_app_shr} and \ref{tab:ar2es_app_shr}, we have utilized a shrinkage matrix rather than the sample covariance matrix to assess the performance of our approach.

\begin{figure}[p]
	\centering
	\includegraphics[width = \linewidth]{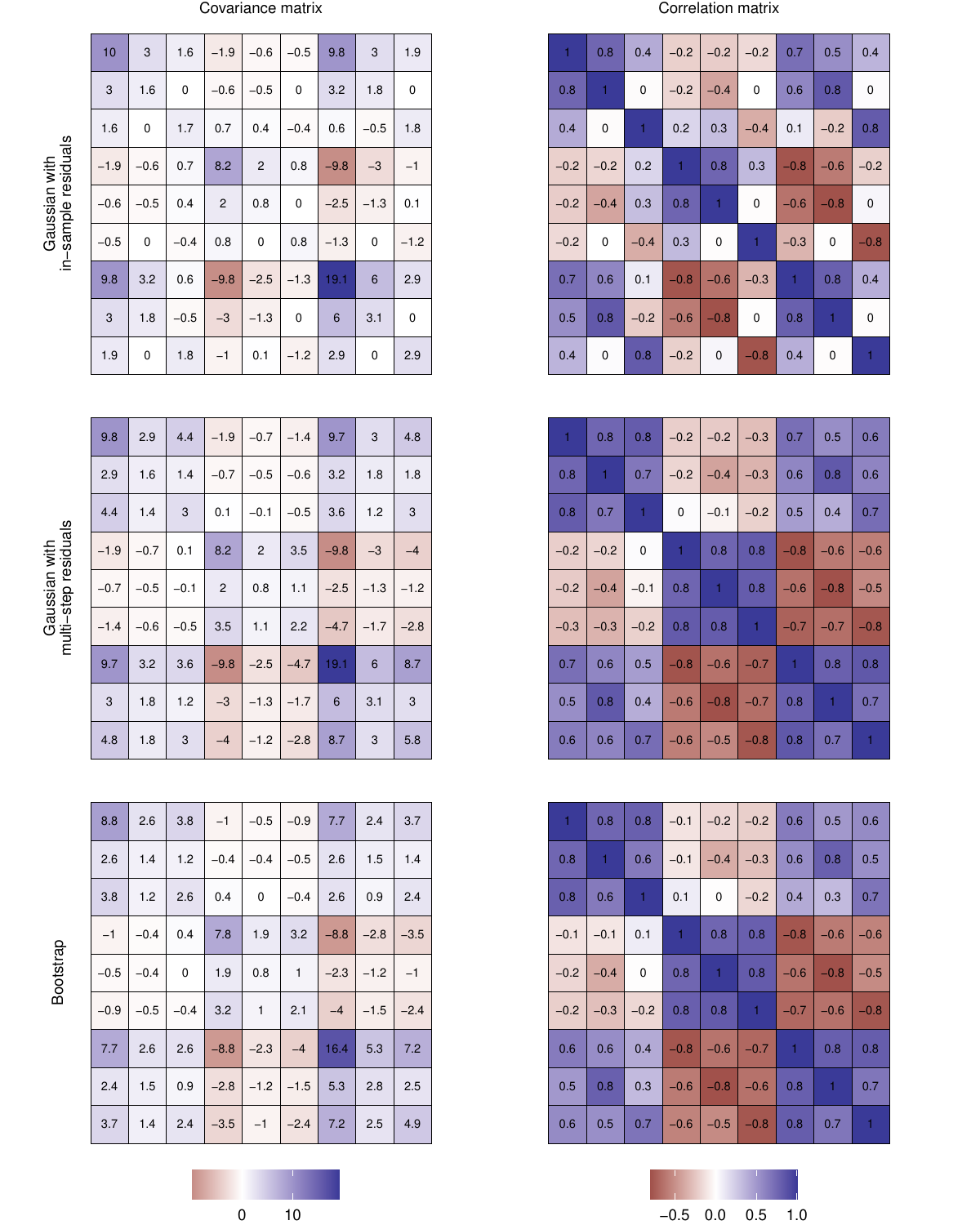}
	\caption{Comparison of estimated covariance and correlation matrices (first simulation) for base forecasts using a parametric Gaussian (with one-step residuals) approach. The true covariance and correlation matrices are shown in Figure 7.}
	\label{fig:ar2covcor_app}
\end{figure}

\begin{table}[H]
	\centering
	\begingroup
	\spacingset{1}
	\fontsize{9}{11}\selectfont
	
\begin{tabular}[t]{l|>{}ccccccccc}
\toprule
\multicolumn{1}{c}{\textbf{}} & \multicolumn{9}{c}{\textbf{Generation of the base forecasts paths}} \\
\cmidrule(l{0pt}r{0pt}){2-10}
\multicolumn{1}{c}{} & \multicolumn{1}{c}{} & \multicolumn{8}{c}{\makecell[c]{Gaussian approach: sample covariance matrix}} \\
\multicolumn{1}{c}{\makecell[c]{\bfseries Reconciliation\\\bfseries approach}} & \multicolumn{1}{c}{ctjb} & \multicolumn{4}{c}{In-sample residuals} & \multicolumn{4}{c}{Multi-step residuals} \\
 &  & G & B & H & HB & G & B & H & HB\\
\midrule
base & \textcolor{black}{8.260} & \textcolor{black}{17.638} & \textcolor{black}{16.733} & \textcolor{black}{22.178} & \textcolor{black}{21.789} & \textcolor{black}{7.748} & \textcolor{black}{6.549} & \textcolor{black}{3.409} & \textcolor{black}{2.215}\\
ct$(bu)$ & \textcolor{black}{3.195} & \textcolor{black}{21.789} & \textcolor{black}{21.789} & \textcolor{black}{\textbf{21.789}} & \textcolor{black}{21.789} & \textcolor{black}{2.215} & \textcolor{black}{2.215} & \textcolor{black}{\textbf{2.215}} & \textcolor{black}{2.215}\\
ct$(shr_{cs}, bu_{te})$ & \textcolor{black}{3.202} & \textcolor{black}{21.942} & \textcolor{black}{21.789} & \textcolor{black}{21.942} & \textcolor{black}{21.789} & \textcolor{black}{2.224} & \textcolor{black}{2.215} & \textcolor{black}{2.224} & \textcolor{black}{2.215}\\
ct$(wlsv_{te}, bu_{cs})$ & \textcolor{black}{\textbf{3.183}} & \textcolor{black}{18.237} & \textcolor{black}{18.237} & \textcolor{black}{21.789} & \textcolor{black}{21.789} & \textcolor{black}{\textbf{2.188}} & \textcolor{black}{2.188} & \textcolor{black}{\textbf{2.215}} & \textcolor{black}{2.215}\\
oct$(wlsv)$ & \textcolor{black}{3.766} & \textcolor{black}{19.174} & \textcolor{black}{18.611} & \textcolor{black}{22.304} & \textcolor{black}{21.789} & \textcolor{black}{3.082} & \textcolor{black}{2.191} & \textcolor{black}{2.910} & \textcolor{black}{2.215}\\
oct$(bdshr)$ & \textcolor{black}{3.203} & \textcolor{black}{18.559} & \textcolor{black}{18.416} & \textcolor{black}{21.937} & \textcolor{black}{21.789} & \textcolor{black}{2.195} & \textcolor{blue}{\textbf{2.184}} & \textcolor{black}{2.224} & \textcolor{black}{\textbf{2.215}}\\
oct$(shr)$ & \textcolor{black}{5.217} & \textcolor{black}{25.015} & \textcolor{black}{23.457} & \textcolor{black}{23.413} & \textcolor{black}{\textbf{21.789}} & \textcolor{black}{2.260} & \textcolor{black}{2.202} & \textcolor{black}{2.226} & \textcolor{black}{2.215}\\
oct$(bshr)$ & \textcolor{black}{5.282} & \textcolor{black}{23.772} & \textcolor{black}{23.997} & \textcolor{black}{22.146} & \textcolor{black}{21.789} & \textcolor{black}{2.720} & \textcolor{black}{2.220} & \textcolor{black}{2.756} & \textcolor{black}{2.215}\\
oct$(hshr)$ & \textcolor{black}{6.161} & \textcolor{black}{\textbf{11.336}} & \textcolor{black}{\textbf{10.940}} & \textcolor{black}{23.598} & \textcolor{black}{\textbf{21.789}} & \textcolor{black}{4.138} & \textcolor{black}{4.167} & \textcolor{black}{2.225} & \textcolor{black}{2.215}\\
oct$(hbshr)$ & \textcolor{black}{5.731} & \textcolor{black}{11.379} & \textcolor{black}{10.940} & \textcolor{black}{22.146} & \textcolor{black}{21.789} & \textcolor{black}{5.085} & \textcolor{black}{4.167} & \textcolor{black}{2.756} & \textcolor{black}{2.215}\\
oct$_h(shr)$ & \textcolor{black}{3.251} & \textcolor{black}{20.965} & \textcolor{black}{19.992} & \textcolor{black}{22.079} & \textcolor{black}{\textbf{21.789}} & \textcolor{black}{2.260} & \textcolor{black}{2.202} & \textcolor{black}{2.226} & \textcolor{black}{2.215}\\
oct$_h(bshr)$ & \textcolor{black}{3.602} & \textcolor{black}{21.306} & \textcolor{black}{21.022} & \textcolor{black}{22.146} & \textcolor{black}{21.789} & \textcolor{black}{2.720} & \textcolor{black}{2.220} & \textcolor{black}{2.756} & \textcolor{black}{2.215}\\
oct$_h(hshr)$ & \textcolor{black}{4.869} & \textcolor{black}{11.405} & \textcolor{black}{10.940} & \textcolor{black}{22.037} & \textcolor{black}{21.789} & \textcolor{black}{4.138} & \textcolor{black}{4.167} & \textcolor{black}{2.225} & \textcolor{black}{2.215}\\
\bottomrule
\end{tabular}

	\endgroup
	\caption{Frobenius norm between the true and the estimated covariance matrix for different reconciliation approaches and different techniques for simulating the base forecasts. Entries in bold represent the lowest value for each column, while the blue entry represent the global minimum. The reconciliation approaches are described in Table 2.}
	\label{tab:ar2norm_app}
\end{table}

\begin{table}[H]
	\centering
	\begingroup
	\spacingset{1}
	\fontsize{9}{11}\selectfont
	
\begin{tabular}[t]{l|ccccccccc}
\toprule
\multicolumn{1}{c}{\textbf{}} & \multicolumn{9}{c}{\textbf{Generation of the base forecasts paths}} \\
\cmidrule(l{0pt}r{0pt}){2-10}
\multicolumn{1}{c}{} & \multicolumn{1}{c}{} & \multicolumn{8}{c}{\makecell[c]{Gaussian approach: sample covariance matrix}} \\
\multicolumn{1}{c}{\makecell[c]{\bfseries Reconciliation\\\bfseries approach}} & \multicolumn{1}{c}{ctjb} & \multicolumn{4}{c}{In-sample residuals} & \multicolumn{4}{c}{Multi-step residuals} \\
 &  & G & B & H & HB & G & B & H & HB\\
\midrule
\addlinespace[0.3em]
\multicolumn{10}{c}{\textbf{$\forall k \in \{2,1\}$}}\\
base & \textcolor{black}{1.000} & \textcolor{red}{1.008} & \textcolor{red}{1.009} & \textcolor{red}{1.044} & \textcolor{red}{1.047} & \textcolor{black}{0.998} & \textcolor{black}{0.999} & \textcolor{red}{1.002} & \textcolor{red}{1.004}\\
ct$(bu)$ & \textcolor{black}{0.901} & \textcolor{black}{0.930} & \textcolor{black}{0.929} & \textcolor{black}{0.929} & \textcolor{black}{0.929} & \textcolor{black}{0.900} & \textcolor{black}{0.900} & \textcolor{black}{0.900} & \textcolor{black}{0.900}\\
ct$(shr_{cs}, bu_{te})$ & \textcolor{black}{\textbf{0.901}} & \textcolor{black}{\textbf{0.929}} & \textcolor{black}{0.928} & \textcolor{black}{\textbf{0.929}} & \textcolor{black}{\textbf{0.928}} & \textcolor{black}{\textbf{0.900}} & \textcolor{blue}{\textbf{0.899}} & \textcolor{black}{\textbf{0.900}} & \textcolor{black}{\textbf{0.900}}\\
ct$(wlsv_{te}, bu_{cs})$ & \textcolor{black}{0.910} & \textcolor{black}{0.930} & \textcolor{black}{0.929} & \textcolor{black}{0.939} & \textcolor{black}{0.939} & \textcolor{black}{0.916} & \textcolor{black}{0.916} & \textcolor{black}{0.916} & \textcolor{black}{0.917}\\
oct$(wlsv)$ & \textcolor{black}{0.922} & \textcolor{black}{0.942} & \textcolor{black}{0.944} & \textcolor{black}{0.951} & \textcolor{black}{0.953} & \textcolor{black}{0.930} & \textcolor{black}{0.930} & \textcolor{black}{0.930} & \textcolor{black}{0.931}\\
oct$(bdshr)$ & \textcolor{black}{0.910} & \textcolor{black}{0.930} & \textcolor{black}{0.930} & \textcolor{black}{0.939} & \textcolor{black}{0.938} & \textcolor{black}{0.916} & \textcolor{black}{0.915} & \textcolor{black}{0.916} & \textcolor{black}{0.916}\\
oct$(shr)$ & \textcolor{black}{0.941} & \textcolor{black}{0.999} & \textcolor{black}{0.985} & \textcolor{black}{0.983} & \textcolor{black}{0.973} & \textcolor{black}{0.903} & \textcolor{black}{0.902} & \textcolor{black}{0.902} & \textcolor{black}{0.903}\\
oct$(bshr)$ & \textcolor{black}{0.951} & \textcolor{black}{0.995} & \textcolor{red}{1.000} & \textcolor{black}{0.983} & \textcolor{black}{0.986} & \textcolor{black}{0.922} & \textcolor{black}{0.922} & \textcolor{black}{0.921} & \textcolor{black}{0.922}\\
oct$(hshr)$ & \textcolor{black}{0.987} & \textcolor{black}{0.995} & \textcolor{black}{0.993} & \textcolor{red}{1.039} & \textcolor{red}{1.026} & \textcolor{black}{0.972} & \textcolor{black}{0.972} & \textcolor{black}{0.974} & \textcolor{black}{0.975}\\
oct$(hbshr)$ & \textcolor{black}{0.987} & \textcolor{black}{0.995} & \textcolor{black}{0.996} & \textcolor{red}{1.024} & \textcolor{red}{1.028} & \textcolor{black}{0.985} & \textcolor{black}{0.985} & \textcolor{black}{0.987} & \textcolor{black}{0.989}\\
oct$_h(shr)$ & \textcolor{black}{0.904} & \textcolor{black}{0.929} & \textcolor{black}{\textbf{0.928}} & \textcolor{black}{0.932} & \textcolor{black}{0.932} & \textcolor{black}{0.903} & \textcolor{black}{0.902} & \textcolor{black}{0.902} & \textcolor{black}{0.903}\\
oct$_h(bshr)$ & \textcolor{black}{0.923} & \textcolor{black}{0.948} & \textcolor{black}{0.952} & \textcolor{black}{0.951} & \textcolor{black}{0.954} & \textcolor{black}{0.922} & \textcolor{black}{0.922} & \textcolor{black}{0.921} & \textcolor{black}{0.922}\\
oct$_h(hshr)$ & \textcolor{black}{0.974} & \textcolor{black}{0.982} & \textcolor{black}{0.982} & \textcolor{red}{1.012} & \textcolor{red}{1.012} & \textcolor{black}{0.972} & \textcolor{black}{0.972} & \textcolor{black}{0.974} & \textcolor{black}{0.975}\\
\addlinespace[0.3em]
\multicolumn{10}{c}{\textbf{$k = 1$}}\\
base & \textcolor{black}{1.000} & \textcolor{red}{1.017} & \textcolor{red}{1.019} & \textcolor{red}{1.017} & \textcolor{red}{1.019} & \textcolor{black}{0.998} & \textcolor{black}{0.999} & \textcolor{black}{0.999} & \textcolor{black}{1.000}\\
ct$(bu)$ & \textcolor{black}{0.978} & \textcolor{black}{0.994} & \textcolor{black}{0.994} & \textcolor{black}{0.994} & \textcolor{black}{0.994} & \textcolor{black}{0.976} & \textcolor{black}{0.976} & \textcolor{black}{0.977} & \textcolor{black}{0.977}\\
ct$(shr_{cs}, bu_{te})$ & \textcolor{black}{\textbf{0.977}} & \textcolor{black}{\textbf{0.993}} & \textcolor{black}{\textbf{0.993}} & \textcolor{black}{\textbf{0.994}} & \textcolor{black}{\textbf{0.993}} & \textcolor{black}{\textbf{0.976}} & \textcolor{blue}{\textbf{0.976}} & \textcolor{black}{\textbf{0.976}} & \textcolor{black}{\textbf{0.976}}\\
ct$(wlsv_{te}, bu_{cs})$ & \textcolor{black}{0.986} & \textcolor{red}{1.002} & \textcolor{red}{1.002} & \textcolor{red}{1.003} & \textcolor{red}{1.003} & \textcolor{black}{0.993} & \textcolor{black}{0.993} & \textcolor{black}{0.993} & \textcolor{black}{0.993}\\
oct$(wlsv)$ & \textcolor{black}{0.998} & \textcolor{red}{1.014} & \textcolor{red}{1.015} & \textcolor{red}{1.015} & \textcolor{red}{1.016} & \textcolor{red}{1.006} & \textcolor{red}{1.006} & \textcolor{red}{1.007} & \textcolor{red}{1.007}\\
oct$(bdshr)$ & \textcolor{black}{0.986} & \textcolor{red}{1.002} & \textcolor{red}{1.002} & \textcolor{red}{1.003} & \textcolor{red}{1.003} & \textcolor{black}{0.992} & \textcolor{black}{0.992} & \textcolor{black}{0.993} & \textcolor{black}{0.993}\\
oct$(shr)$ & \textcolor{red}{1.037} & \textcolor{red}{1.082} & \textcolor{red}{1.067} & \textcolor{red}{1.064} & \textcolor{red}{1.056} & \textcolor{black}{0.979} & \textcolor{black}{0.978} & \textcolor{black}{0.979} & \textcolor{black}{0.979}\\
oct$(bshr)$ & \textcolor{red}{1.041} & \textcolor{red}{1.071} & \textcolor{red}{1.074} & \textcolor{red}{1.060} & \textcolor{red}{1.062} & \textcolor{black}{0.998} & \textcolor{black}{0.998} & \textcolor{black}{0.998} & \textcolor{black}{0.998}\\
oct$(hshr)$ & \textcolor{red}{1.080} & \textcolor{red}{1.090} & \textcolor{red}{1.091} & \textcolor{red}{1.119} & \textcolor{red}{1.105} & \textcolor{red}{1.050} & \textcolor{red}{1.050} & \textcolor{red}{1.053} & \textcolor{red}{1.053}\\
oct$(hbshr)$ & \textcolor{red}{1.065} & \textcolor{red}{1.080} & \textcolor{red}{1.081} & \textcolor{red}{1.088} & \textcolor{red}{1.090} & \textcolor{red}{1.063} & \textcolor{red}{1.064} & \textcolor{red}{1.066} & \textcolor{red}{1.068}\\
oct$_h(shr)$ & \textcolor{black}{0.980} & \textcolor{black}{0.996} & \textcolor{black}{0.995} & \textcolor{black}{0.996} & \textcolor{black}{0.996} & \textcolor{black}{0.979} & \textcolor{black}{0.978} & \textcolor{black}{0.979} & \textcolor{black}{0.979}\\
oct$_h(bshr)$ & \textcolor{black}{0.999} & \textcolor{red}{1.016} & \textcolor{red}{1.018} & \textcolor{red}{1.016} & \textcolor{red}{1.018} & \textcolor{black}{0.998} & \textcolor{black}{0.998} & \textcolor{black}{0.998} & \textcolor{black}{0.998}\\
oct$_h(hshr)$ & \textcolor{red}{1.052} & \textcolor{red}{1.067} & \textcolor{red}{1.066} & \textcolor{red}{1.074} & \textcolor{red}{1.075} & \textcolor{red}{1.050} & \textcolor{red}{1.050} & \textcolor{red}{1.053} & \textcolor{red}{1.053}\\
\addlinespace[0.3em]
\multicolumn{10}{c}{\textbf{$k = 2$}}\\
base & \textcolor{black}{1.000} & \textcolor{black}{0.998} & \textcolor{black}{0.999} & \textcolor{red}{1.071} & \textcolor{red}{1.075} & \textcolor{black}{0.998} & \textcolor{black}{0.999} & \textcolor{red}{1.005} & \textcolor{red}{1.008}\\
ct$(bu)$ & \textcolor{black}{0.831} & \textcolor{black}{0.869} & \textcolor{black}{0.869} & \textcolor{black}{0.869} & \textcolor{black}{0.869} & \textcolor{black}{0.830} & \textcolor{black}{0.829} & \textcolor{black}{0.829} & \textcolor{black}{0.830}\\
ct$(shr_{cs}, bu_{te})$ & \textcolor{black}{\textbf{0.830}} & \textcolor{black}{0.869} & \textcolor{black}{0.868} & \textcolor{black}{\textbf{0.868}} & \textcolor{black}{\textbf{0.868}} & \textcolor{black}{\textbf{0.830}} & \textcolor{blue}{\textbf{0.829}} & \textcolor{black}{\textbf{0.829}} & \textcolor{black}{\textbf{0.830}}\\
ct$(wlsv_{te}, bu_{cs})$ & \textcolor{black}{0.840} & \textcolor{black}{\textbf{0.863}} & \textcolor{black}{\textbf{0.862}} & \textcolor{black}{0.879} & \textcolor{black}{0.878} & \textcolor{black}{0.846} & \textcolor{black}{0.844} & \textcolor{black}{0.845} & \textcolor{black}{0.846}\\
oct$(wlsv)$ & \textcolor{black}{0.851} & \textcolor{black}{0.875} & \textcolor{black}{0.877} & \textcolor{black}{0.891} & \textcolor{black}{0.893} & \textcolor{black}{0.859} & \textcolor{black}{0.859} & \textcolor{black}{0.859} & \textcolor{black}{0.861}\\
oct$(bdshr)$ & \textcolor{black}{0.839} & \textcolor{black}{0.863} & \textcolor{black}{0.863} & \textcolor{black}{0.879} & \textcolor{black}{0.878} & \textcolor{black}{0.845} & \textcolor{black}{0.844} & \textcolor{black}{0.845} & \textcolor{black}{0.846}\\
oct$(shr)$ & \textcolor{black}{0.854} & \textcolor{black}{0.922} & \textcolor{black}{0.909} & \textcolor{black}{0.908} & \textcolor{black}{0.897} & \textcolor{black}{0.833} & \textcolor{black}{0.831} & \textcolor{black}{0.832} & \textcolor{black}{0.832}\\
oct$(bshr)$ & \textcolor{black}{0.869} & \textcolor{black}{0.925} & \textcolor{black}{0.931} & \textcolor{black}{0.911} & \textcolor{black}{0.915} & \textcolor{black}{0.851} & \textcolor{black}{0.851} & \textcolor{black}{0.851} & \textcolor{black}{0.852}\\
oct$(hshr)$ & \textcolor{black}{0.901} & \textcolor{black}{0.908} & \textcolor{black}{0.904} & \textcolor{black}{0.966} & \textcolor{black}{0.952} & \textcolor{black}{0.900} & \textcolor{black}{0.899} & \textcolor{black}{0.901} & \textcolor{black}{0.902}\\
oct$(hbshr)$ & \textcolor{black}{0.915} & \textcolor{black}{0.917} & \textcolor{black}{0.919} & \textcolor{black}{0.964} & \textcolor{black}{0.969} & \textcolor{black}{0.913} & \textcolor{black}{0.913} & \textcolor{black}{0.914} & \textcolor{black}{0.917}\\
oct$_h(shr)$ & \textcolor{black}{0.834} & \textcolor{black}{0.868} & \textcolor{black}{0.865} & \textcolor{black}{0.872} & \textcolor{black}{0.872} & \textcolor{black}{0.833} & \textcolor{black}{0.831} & \textcolor{black}{0.832} & \textcolor{black}{0.832}\\
oct$_h(bshr)$ & \textcolor{black}{0.852} & \textcolor{black}{0.886} & \textcolor{black}{0.890} & \textcolor{black}{0.890} & \textcolor{black}{0.894} & \textcolor{black}{0.851} & \textcolor{black}{0.851} & \textcolor{black}{0.851} & \textcolor{black}{0.852}\\
oct$_h(hshr)$ & \textcolor{black}{0.902} & \textcolor{black}{0.904} & \textcolor{black}{0.904} & \textcolor{black}{0.953} & \textcolor{black}{0.952} & \textcolor{black}{0.900} & \textcolor{black}{0.899} & \textcolor{black}{0.901} & \textcolor{black}{0.902}\\
\bottomrule
\end{tabular}

	\endgroup
	\caption{Simulation experiment. $\overline{RelCRPS}$ defined in Section 5. %A lower value, indicates a more accurate forecast. 
	Approaches performing worse than the benchmark (bootstrap base forecasts, ctjb) are highlighted in red, the best for each column is marked in bold, and the overall lowest value is highlighted in blue. The reconciliation approaches are described in Table 2.}
	\label{tab:ar2crps_app}
\end{table}

\begin{table}[H]
	\centering
	\begingroup
	\spacingset{1}
	\fontsize{9}{11}\selectfont
	
\begin{tabular}[t]{l|ccccccccc}
\toprule
\multicolumn{1}{c}{\textbf{}} & \multicolumn{9}{c}{\textbf{Generation of the base forecasts paths}} \\
\cmidrule(l{0pt}r{0pt}){2-10}
\multicolumn{1}{c}{} & \multicolumn{1}{c}{} & \multicolumn{8}{c}{\makecell[c]{Gaussian approach: sample covariance matrix}} \\
\multicolumn{1}{c}{\makecell[c]{\bfseries Reconciliation\\\bfseries approach}} & \multicolumn{1}{c}{ctjb} & \multicolumn{4}{c}{In-sample residuals} & \multicolumn{4}{c}{Multi-step residuals} \\
 &  & G & B & H & HB & G & B & H & HB\\
\midrule
\addlinespace[0.3em]
\multicolumn{10}{c}{\textbf{$\forall k \in \{2,1\}$}}\\
base & \textcolor{black}{1.000} & \textcolor{red}{1.005} & \textcolor{red}{1.009} & \textcolor{red}{1.039} & \textcolor{red}{1.046} & \textcolor{black}{0.996} & \textcolor{black}{0.999} & \textcolor{black}{1.000} & \textcolor{red}{1.004}\\
ct$(bu)$ & \textcolor{black}{0.897} & \textcolor{black}{0.924} & \textcolor{black}{0.923} & \textcolor{black}{0.924} & \textcolor{black}{0.923} & \textcolor{black}{0.895} & \textcolor{black}{0.896} & \textcolor{black}{0.897} & \textcolor{blue}{\textbf{0.895}}\\
ct$(shr_{cs}, bu_{te})$ & \textcolor{black}{\textbf{0.896}} & \textcolor{black}{0.924} & \textcolor{black}{0.923} & \textcolor{black}{\textbf{0.923}} & \textcolor{black}{\textbf{0.922}} & \textcolor{black}{\textbf{0.895}} & \textcolor{black}{\textbf{0.895}} & \textcolor{black}{\textbf{0.896}} & \textcolor{black}{0.896}\\
ct$(wlsv_{te}, bu_{cs})$ & \textcolor{black}{0.906} & \textcolor{black}{0.924} & \textcolor{black}{0.923} & \textcolor{black}{0.933} & \textcolor{black}{0.932} & \textcolor{black}{0.912} & \textcolor{black}{0.911} & \textcolor{black}{0.910} & \textcolor{black}{0.912}\\
oct$(wlsv)$ & \textcolor{black}{0.916} & \textcolor{black}{0.935} & \textcolor{black}{0.937} & \textcolor{black}{0.944} & \textcolor{black}{0.945} & \textcolor{black}{0.923} & \textcolor{black}{0.923} & \textcolor{black}{0.923} & \textcolor{black}{0.924}\\
oct$(bdshr)$ & \textcolor{black}{0.906} & \textcolor{black}{\textbf{0.923}} & \textcolor{black}{0.923} & \textcolor{black}{0.932} & \textcolor{black}{0.932} & \textcolor{black}{0.910} & \textcolor{black}{0.910} & \textcolor{black}{0.911} & \textcolor{black}{0.912}\\
oct$(shr)$ & \textcolor{black}{0.938} & \textcolor{black}{0.993} & \textcolor{black}{0.980} & \textcolor{black}{0.977} & \textcolor{black}{0.969} & \textcolor{black}{0.898} & \textcolor{black}{0.898} & \textcolor{black}{0.898} & \textcolor{black}{0.897}\\
oct$(bshr)$ & \textcolor{black}{0.947} & \textcolor{black}{0.990} & \textcolor{black}{0.995} & \textcolor{black}{0.979} & \textcolor{black}{0.981} & \textcolor{black}{0.915} & \textcolor{black}{0.915} & \textcolor{black}{0.915} & \textcolor{black}{0.915}\\
oct$(hshr)$ & \textcolor{black}{0.978} & \textcolor{black}{0.987} & \textcolor{black}{0.985} & \textcolor{red}{1.027} & \textcolor{red}{1.016} & \textcolor{black}{0.963} & \textcolor{black}{0.964} & \textcolor{black}{0.966} & \textcolor{black}{0.967}\\
oct$(hbshr)$ & \textcolor{black}{0.977} & \textcolor{black}{0.986} & \textcolor{black}{0.985} & \textcolor{red}{1.012} & \textcolor{red}{1.016} & \textcolor{black}{0.974} & \textcolor{black}{0.976} & \textcolor{black}{0.977} & \textcolor{black}{0.978}\\
oct$_h(shr)$ & \textcolor{black}{0.900} & \textcolor{black}{0.923} & \textcolor{black}{\textbf{0.922}} & \textcolor{black}{0.926} & \textcolor{black}{0.925} & \textcolor{black}{0.898} & \textcolor{black}{0.898} & \textcolor{black}{0.897} & \textcolor{black}{0.898}\\
oct$_h(bshr)$ & \textcolor{black}{0.916} & \textcolor{black}{0.940} & \textcolor{black}{0.943} & \textcolor{black}{0.942} & \textcolor{black}{0.945} & \textcolor{black}{0.914} & \textcolor{black}{0.916} & \textcolor{black}{0.915} & \textcolor{black}{0.916}\\
oct$_h(hshr)$ & \textcolor{black}{0.967} & \textcolor{black}{0.974} & \textcolor{black}{0.974} & \textcolor{red}{1.002} & \textcolor{red}{1.002} & \textcolor{black}{0.964} & \textcolor{black}{0.964} & \textcolor{black}{0.966} & \textcolor{black}{0.967}\\
\addlinespace[0.3em]
\multicolumn{10}{c}{\textbf{$k = 1$}}\\
base & \textcolor{black}{1.000} & \textcolor{red}{1.014} & \textcolor{red}{1.020} & \textcolor{red}{1.015} & \textcolor{red}{1.019} & \textcolor{black}{0.997} & \textcolor{red}{1.000} & \textcolor{black}{0.997} & \textcolor{red}{1.000}\\
ct$(bu)$ & \textcolor{black}{0.969} & \textcolor{black}{0.985} & \textcolor{black}{0.983} & \textcolor{black}{0.985} & \textcolor{black}{0.984} & \textcolor{black}{\textbf{0.967}} & \textcolor{blue}{\textbf{0.967}} & \textcolor{black}{0.968} & \textcolor{black}{\textbf{0.968}}\\
ct$(shr_{cs}, bu_{te})$ & \textcolor{black}{\textbf{0.968}} & \textcolor{black}{\textbf{0.984}} & \textcolor{black}{\textbf{0.983}} & \textcolor{black}{\textbf{0.984}} & \textcolor{black}{\textbf{0.983}} & \textcolor{black}{0.968} & \textcolor{black}{0.967} & \textcolor{black}{\textbf{0.968}} & \textcolor{black}{0.968}\\
ct$(wlsv_{te}, bu_{cs})$ & \textcolor{black}{0.977} & \textcolor{black}{0.991} & \textcolor{black}{0.991} & \textcolor{black}{0.992} & \textcolor{black}{0.992} & \textcolor{black}{0.984} & \textcolor{black}{0.983} & \textcolor{black}{0.981} & \textcolor{black}{0.984}\\
oct$(wlsv)$ & \textcolor{black}{0.989} & \textcolor{red}{1.002} & \textcolor{red}{1.004} & \textcolor{red}{1.003} & \textcolor{red}{1.004} & \textcolor{black}{0.994} & \textcolor{black}{0.995} & \textcolor{black}{0.995} & \textcolor{black}{0.997}\\
oct$(bdshr)$ & \textcolor{black}{0.977} & \textcolor{black}{0.989} & \textcolor{black}{0.991} & \textcolor{black}{0.992} & \textcolor{black}{0.992} & \textcolor{black}{0.981} & \textcolor{black}{0.982} & \textcolor{black}{0.983} & \textcolor{black}{0.985}\\
oct$(shr)$ & \textcolor{red}{1.028} & \textcolor{red}{1.070} & \textcolor{red}{1.056} & \textcolor{red}{1.053} & \textcolor{red}{1.046} & \textcolor{black}{0.969} & \textcolor{black}{0.969} & \textcolor{black}{0.970} & \textcolor{black}{0.969}\\
oct$(bshr)$ & \textcolor{red}{1.034} & \textcolor{red}{1.061} & \textcolor{red}{1.065} & \textcolor{red}{1.051} & \textcolor{red}{1.053} & \textcolor{black}{0.985} & \textcolor{black}{0.987} & \textcolor{black}{0.986} & \textcolor{black}{0.987}\\
oct$(hshr)$ & \textcolor{red}{1.066} & \textcolor{red}{1.075} & \textcolor{red}{1.076} & \textcolor{red}{1.099} & \textcolor{red}{1.090} & \textcolor{red}{1.037} & \textcolor{red}{1.037} & \textcolor{red}{1.039} & \textcolor{red}{1.039}\\
oct$(hbshr)$ & \textcolor{red}{1.050} & \textcolor{red}{1.065} & \textcolor{red}{1.065} & \textcolor{red}{1.070} & \textcolor{red}{1.073} & \textcolor{red}{1.048} & \textcolor{red}{1.049} & \textcolor{red}{1.049} & \textcolor{red}{1.052}\\
oct$_h(shr)$ & \textcolor{black}{0.971} & \textcolor{black}{0.985} & \textcolor{black}{0.985} & \textcolor{black}{0.986} & \textcolor{black}{0.986} & \textcolor{black}{0.969} & \textcolor{black}{0.969} & \textcolor{black}{0.969} & \textcolor{black}{0.969}\\
oct$_h(bshr)$ & \textcolor{black}{0.987} & \textcolor{red}{1.002} & \textcolor{red}{1.005} & \textcolor{red}{1.002} & \textcolor{red}{1.005} & \textcolor{black}{0.986} & \textcolor{black}{0.987} & \textcolor{black}{0.987} & \textcolor{black}{0.988}\\
oct$_h(hshr)$ & \textcolor{red}{1.040} & \textcolor{red}{1.053} & \textcolor{red}{1.053} & \textcolor{red}{1.059} & \textcolor{red}{1.058} & \textcolor{red}{1.036} & \textcolor{red}{1.036} & \textcolor{red}{1.040} & \textcolor{red}{1.040}\\
\addlinespace[0.3em]
\multicolumn{10}{c}{\textbf{$k = 2$}}\\
base & \textcolor{black}{1.000} & \textcolor{black}{0.997} & \textcolor{black}{0.999} & \textcolor{red}{1.063} & \textcolor{red}{1.073} & \textcolor{black}{0.996} & \textcolor{black}{0.998} & \textcolor{red}{1.003} & \textcolor{red}{1.008}\\
ct$(bu)$ & \textcolor{black}{0.831} & \textcolor{black}{0.867} & \textcolor{black}{0.867} & \textcolor{black}{0.867} & \textcolor{black}{0.867} & \textcolor{black}{0.829} & \textcolor{black}{0.829} & \textcolor{black}{0.830} & \textcolor{blue}{\textbf{0.828}}\\
ct$(shr_{cs}, bu_{te})$ & \textcolor{black}{\textbf{0.829}} & \textcolor{black}{0.867} & \textcolor{black}{0.866} & \textcolor{black}{\textbf{0.866}} & \textcolor{black}{\textbf{0.865}} & \textcolor{black}{\textbf{0.828}} & \textcolor{black}{\textbf{0.829}} & \textcolor{black}{\textbf{0.829}} & \textcolor{black}{0.829}\\
ct$(wlsv_{te}, bu_{cs})$ & \textcolor{black}{0.839} & \textcolor{black}{\textbf{0.860}} & \textcolor{black}{\textbf{0.860}} & \textcolor{black}{0.877} & \textcolor{black}{0.876} & \textcolor{black}{0.844} & \textcolor{black}{0.844} & \textcolor{black}{0.844} & \textcolor{black}{0.845}\\
oct$(wlsv)$ & \textcolor{black}{0.849} & \textcolor{black}{0.872} & \textcolor{black}{0.875} & \textcolor{black}{0.887} & \textcolor{black}{0.890} & \textcolor{black}{0.858} & \textcolor{black}{0.856} & \textcolor{black}{0.856} & \textcolor{black}{0.857}\\
oct$(bdshr)$ & \textcolor{black}{0.839} & \textcolor{black}{0.861} & \textcolor{black}{0.861} & \textcolor{black}{0.876} & \textcolor{black}{0.875} & \textcolor{black}{0.845} & \textcolor{black}{0.843} & \textcolor{black}{0.845} & \textcolor{black}{0.844}\\
oct$(shr)$ & \textcolor{black}{0.856} & \textcolor{black}{0.921} & \textcolor{black}{0.909} & \textcolor{black}{0.907} & \textcolor{black}{0.898} & \textcolor{black}{0.832} & \textcolor{black}{0.831} & \textcolor{black}{0.832} & \textcolor{black}{0.831}\\
oct$(bshr)$ & \textcolor{black}{0.868} & \textcolor{black}{0.924} & \textcolor{black}{0.930} & \textcolor{black}{0.911} & \textcolor{black}{0.915} & \textcolor{black}{0.849} & \textcolor{black}{0.848} & \textcolor{black}{0.849} & \textcolor{black}{0.848}\\
oct$(hshr)$ & \textcolor{black}{0.897} & \textcolor{black}{0.905} & \textcolor{black}{0.901} & \textcolor{black}{0.959} & \textcolor{black}{0.947} & \textcolor{black}{0.895} & \textcolor{black}{0.896} & \textcolor{black}{0.898} & \textcolor{black}{0.899}\\
oct$(hbshr)$ & \textcolor{black}{0.910} & \textcolor{black}{0.912} & \textcolor{black}{0.912} & \textcolor{black}{0.957} & \textcolor{black}{0.961} & \textcolor{black}{0.906} & \textcolor{black}{0.909} & \textcolor{black}{0.909} & \textcolor{black}{0.910}\\
oct$_h(shr)$ & \textcolor{black}{0.835} & \textcolor{black}{0.865} & \textcolor{black}{0.862} & \textcolor{black}{0.870} & \textcolor{black}{0.868} & \textcolor{black}{0.833} & \textcolor{black}{0.833} & \textcolor{black}{0.831} & \textcolor{black}{0.832}\\
oct$_h(bshr)$ & \textcolor{black}{0.850} & \textcolor{black}{0.881} & \textcolor{black}{0.885} & \textcolor{black}{0.886} & \textcolor{black}{0.889} & \textcolor{black}{0.847} & \textcolor{black}{0.849} & \textcolor{black}{0.849} & \textcolor{black}{0.850}\\
oct$_h(hshr)$ & \textcolor{black}{0.900} & \textcolor{black}{0.902} & \textcolor{black}{0.901} & \textcolor{black}{0.947} & \textcolor{black}{0.948} & \textcolor{black}{0.897} & \textcolor{black}{0.896} & \textcolor{black}{0.897} & \textcolor{black}{0.899}\\
\bottomrule
\end{tabular}

	\endgroup
	\caption{Simulation experiment. ES ratio indices defined in Section 5. %A lower value, indicates amore accurate forecast. 
	Approaches performing worse than the benchmark (bootstrap base forecasts, ctjb) are highlighted in red, the best for each column is marked in bold, and the overall lowest value is highlighted in blue. The reconciliation approaches are described in Table 2.}
	\label{tab:ar2es_app}
\end{table}

\begin{table}[H]
	\centering
	\begingroup
	\spacingset{1}
	\fontsize{9}{11}\selectfont
	
\begin{tabular}[t]{l|ccccccccc}
\toprule
\multicolumn{1}{c}{\textbf{}} & \multicolumn{9}{c}{\textbf{Generation of the base forecasts paths}} \\
\cmidrule(l{0pt}r{0pt}){2-10}
\multicolumn{1}{c}{} & \multicolumn{1}{c}{} & \multicolumn{8}{c}{\makecell[c]{Gaussian approach: shrinkage covariance matrix}} \\
\multicolumn{1}{c}{\makecell[c]{\bfseries Reconciliation\\\bfseries approach}} & \multicolumn{1}{c}{ctjb} & \multicolumn{4}{c}{In-sample residuals} & \multicolumn{4}{c}{Multi-step residuals} \\
 &  & G & B & H & HB & G & B & H & HB\\
\midrule
\addlinespace[0.3em]
\multicolumn{10}{c}{\textbf{$\forall k \in \{2,1\}$}}\\
base & \textcolor{red}{1.007} & \textcolor{red}{1.009} & \textcolor{red}{1.044} & \textcolor{red}{1.046} & \textcolor{black}{0.997} & \textcolor{black}{0.999} & \textcolor{red}{1.002} & \textcolor{red}{1.003} & \textcolor{black}{1.000}\\
ct$(bu)$ & \textcolor{black}{0.929} & \textcolor{black}{0.929} & \textcolor{black}{0.929} & \textcolor{black}{0.929} & \textcolor{black}{0.899} & \textcolor{black}{0.900} & \textcolor{black}{0.900} & \textcolor{black}{0.900} & \textcolor{black}{0.901}\\
ct$(shr_{cs}, bu_{te})$ & \textcolor{black}{\textbf{0.929}} & \textcolor{black}{0.928} & \textcolor{black}{\textbf{0.929}} & \textcolor{black}{\textbf{0.928}} & \textcolor{blue}{\textbf{0.899}} & \textcolor{black}{\textbf{0.899}} & \textcolor{black}{\textbf{0.900}} & \textcolor{black}{\textbf{0.900}} & \textcolor{black}{\textbf{0.901}}\\
ct$(wlsv_{te}, bu_{cs})$ & \textcolor{black}{0.930} & \textcolor{black}{0.930} & \textcolor{black}{0.939} & \textcolor{black}{0.938} & \textcolor{black}{0.915} & \textcolor{black}{0.916} & \textcolor{black}{0.917} & \textcolor{black}{0.916} & \textcolor{black}{0.910}\\
oct$(wlsv)$ & \textcolor{black}{0.943} & \textcolor{black}{0.944} & \textcolor{black}{0.951} & \textcolor{black}{0.952} & \textcolor{black}{0.929} & \textcolor{black}{0.930} & \textcolor{black}{0.931} & \textcolor{black}{0.930} & \textcolor{black}{0.922}\\
oct$(bdshr)$ & \textcolor{black}{0.930} & \textcolor{black}{0.930} & \textcolor{black}{0.938} & \textcolor{black}{0.938} & \textcolor{black}{0.915} & \textcolor{black}{0.916} & \textcolor{black}{0.916} & \textcolor{black}{0.916} & \textcolor{black}{0.910}\\
oct$(shr)$ & \textcolor{black}{0.994} & \textcolor{black}{0.982} & \textcolor{black}{0.980} & \textcolor{black}{0.973} & \textcolor{black}{0.902} & \textcolor{black}{0.902} & \textcolor{black}{0.903} & \textcolor{black}{0.902} & \textcolor{black}{0.941}\\
oct$(bshr)$ & \textcolor{black}{0.995} & \textcolor{black}{0.998} & \textcolor{black}{0.983} & \textcolor{black}{0.986} & \textcolor{black}{0.921} & \textcolor{black}{0.922} & \textcolor{black}{0.922} & \textcolor{black}{0.922} & \textcolor{black}{0.951}\\
oct$(hshr)$ & \textcolor{black}{0.994} & \textcolor{black}{0.994} & \textcolor{red}{1.035} & \textcolor{red}{1.025} & \textcolor{black}{0.971} & \textcolor{black}{0.972} & \textcolor{black}{0.974} & \textcolor{black}{0.974} & \textcolor{black}{0.987}\\
oct$(hbshr)$ & \textcolor{black}{0.995} & \textcolor{black}{0.997} & \textcolor{red}{1.025} & \textcolor{red}{1.027} & \textcolor{black}{0.984} & \textcolor{black}{0.986} & \textcolor{black}{0.988} & \textcolor{black}{0.988} & \textcolor{black}{0.987}\\
oct$_h(shr)$ & \textcolor{black}{0.929} & \textcolor{black}{\textbf{0.928}} & \textcolor{black}{0.932} & \textcolor{black}{0.932} & \textcolor{black}{0.902} & \textcolor{black}{0.902} & \textcolor{black}{0.903} & \textcolor{black}{0.902} & \textcolor{black}{0.904}\\
oct$_h(bshr)$ & \textcolor{black}{0.948} & \textcolor{black}{0.951} & \textcolor{black}{0.951} & \textcolor{black}{0.953} & \textcolor{black}{0.921} & \textcolor{black}{0.922} & \textcolor{black}{0.922} & \textcolor{black}{0.922} & \textcolor{black}{0.923}\\
oct$_h(hshr)$ & \textcolor{black}{0.982} & \textcolor{black}{0.982} & \textcolor{red}{1.011} & \textcolor{red}{1.011} & \textcolor{black}{0.971} & \textcolor{black}{0.972} & \textcolor{black}{0.974} & \textcolor{black}{0.974} & \textcolor{black}{0.974}\\
\addlinespace[0.3em]
\multicolumn{10}{c}{\textbf{$k = 1$}}\\
base & \textcolor{red}{1.017} & \textcolor{red}{1.019} & \textcolor{red}{1.017} & \textcolor{red}{1.019} & \textcolor{black}{0.998} & \textcolor{black}{0.999} & \textcolor{black}{0.999} & \textcolor{black}{0.999} & \textcolor{black}{1.000}\\
ct$(bu)$ & \textcolor{black}{0.994} & \textcolor{black}{0.994} & \textcolor{black}{0.994} & \textcolor{black}{0.994} & \textcolor{black}{0.976} & \textcolor{black}{0.976} & \textcolor{black}{0.977} & \textcolor{black}{0.976} & \textcolor{black}{0.978}\\
ct$(shr_{cs}, bu_{te})$ & \textcolor{black}{\textbf{0.993}} & \textcolor{black}{\textbf{0.993}} & \textcolor{black}{\textbf{0.993}} & \textcolor{black}{\textbf{0.993}} & \textcolor{blue}{\textbf{0.975}} & \textcolor{black}{\textbf{0.976}} & \textcolor{black}{\textbf{0.976}} & \textcolor{black}{\textbf{0.976}} & \textcolor{black}{\textbf{0.977}}\\
ct$(wlsv_{te}, bu_{cs})$ & \textcolor{red}{1.002} & \textcolor{red}{1.002} & \textcolor{red}{1.003} & \textcolor{red}{1.003} & \textcolor{black}{0.992} & \textcolor{black}{0.993} & \textcolor{black}{0.993} & \textcolor{black}{0.993} & \textcolor{black}{0.986}\\
oct$(wlsv)$ & \textcolor{red}{1.015} & \textcolor{red}{1.015} & \textcolor{red}{1.015} & \textcolor{red}{1.016} & \textcolor{red}{1.005} & \textcolor{red}{1.007} & \textcolor{red}{1.007} & \textcolor{red}{1.007} & \textcolor{black}{0.998}\\
oct$(bdshr)$ & \textcolor{red}{1.002} & \textcolor{red}{1.002} & \textcolor{red}{1.003} & \textcolor{red}{1.002} & \textcolor{black}{0.992} & \textcolor{black}{0.992} & \textcolor{black}{0.993} & \textcolor{black}{0.992} & \textcolor{black}{0.986}\\
oct$(shr)$ & \textcolor{red}{1.076} & \textcolor{red}{1.065} & \textcolor{red}{1.061} & \textcolor{red}{1.056} & \textcolor{black}{0.978} & \textcolor{black}{0.978} & \textcolor{black}{0.979} & \textcolor{black}{0.978} & \textcolor{red}{1.037}\\
oct$(bshr)$ & \textcolor{red}{1.070} & \textcolor{red}{1.072} & \textcolor{red}{1.060} & \textcolor{red}{1.062} & \textcolor{black}{0.997} & \textcolor{black}{0.998} & \textcolor{black}{0.998} & \textcolor{black}{0.998} & \textcolor{red}{1.041}\\
oct$(hshr)$ & \textcolor{red}{1.090} & \textcolor{red}{1.092} & \textcolor{red}{1.114} & \textcolor{red}{1.105} & \textcolor{red}{1.049} & \textcolor{red}{1.050} & \textcolor{red}{1.053} & \textcolor{red}{1.052} & \textcolor{red}{1.080}\\
oct$(hbshr)$ & \textcolor{red}{1.080} & \textcolor{red}{1.081} & \textcolor{red}{1.089} & \textcolor{red}{1.090} & \textcolor{red}{1.062} & \textcolor{red}{1.064} & \textcolor{red}{1.066} & \textcolor{red}{1.066} & \textcolor{red}{1.065}\\
oct$_h(shr)$ & \textcolor{black}{0.996} & \textcolor{black}{0.995} & \textcolor{black}{0.996} & \textcolor{black}{0.996} & \textcolor{black}{0.978} & \textcolor{black}{0.978} & \textcolor{black}{0.979} & \textcolor{black}{0.978} & \textcolor{black}{0.980}\\
oct$_h(bshr)$ & \textcolor{red}{1.016} & \textcolor{red}{1.018} & \textcolor{red}{1.016} & \textcolor{red}{1.018} & \textcolor{black}{0.997} & \textcolor{black}{0.998} & \textcolor{black}{0.998} & \textcolor{black}{0.998} & \textcolor{black}{0.999}\\
oct$_h(hshr)$ & \textcolor{red}{1.066} & \textcolor{red}{1.067} & \textcolor{red}{1.075} & \textcolor{red}{1.075} & \textcolor{red}{1.049} & \textcolor{red}{1.050} & \textcolor{red}{1.053} & \textcolor{red}{1.052} & \textcolor{red}{1.052}\\
\addlinespace[0.3em]
\multicolumn{10}{c}{\textbf{$k = 2$}}\\
base & \textcolor{black}{0.997} & \textcolor{black}{0.999} & \textcolor{red}{1.071} & \textcolor{red}{1.074} & \textcolor{black}{0.997} & \textcolor{black}{0.999} & \textcolor{red}{1.005} & \textcolor{red}{1.008} & \textcolor{black}{1.000}\\
ct$(bu)$ & \textcolor{black}{0.869} & \textcolor{black}{0.868} & \textcolor{black}{0.868} & \textcolor{black}{0.868} & \textcolor{black}{0.829} & \textcolor{black}{0.829} & \textcolor{black}{0.830} & \textcolor{black}{0.830} & \textcolor{black}{0.831}\\
ct$(shr_{cs}, bu_{te})$ & \textcolor{black}{0.868} & \textcolor{black}{0.867} & \textcolor{black}{\textbf{0.868}} & \textcolor{black}{\textbf{0.867}} & \textcolor{blue}{\textbf{0.829}} & \textcolor{black}{\textbf{0.829}} & \textcolor{black}{\textbf{0.830}} & \textcolor{black}{\textbf{0.829}} & \textcolor{black}{\textbf{0.830}}\\
ct$(wlsv_{te}, bu_{cs})$ & \textcolor{black}{\textbf{0.863}} & \textcolor{black}{\textbf{0.862}} & \textcolor{black}{0.878} & \textcolor{black}{0.878} & \textcolor{black}{0.845} & \textcolor{black}{0.845} & \textcolor{black}{0.846} & \textcolor{black}{0.846} & \textcolor{black}{0.840}\\
oct$(wlsv)$ & \textcolor{black}{0.876} & \textcolor{black}{0.877} & \textcolor{black}{0.891} & \textcolor{black}{0.892} & \textcolor{black}{0.859} & \textcolor{black}{0.860} & \textcolor{black}{0.860} & \textcolor{black}{0.860} & \textcolor{black}{0.851}\\
oct$(bdshr)$ & \textcolor{black}{0.863} & \textcolor{black}{0.863} & \textcolor{black}{0.878} & \textcolor{black}{0.877} & \textcolor{black}{0.844} & \textcolor{black}{0.845} & \textcolor{black}{0.846} & \textcolor{black}{0.845} & \textcolor{black}{0.839}\\
oct$(shr)$ & \textcolor{black}{0.918} & \textcolor{black}{0.906} & \textcolor{black}{0.906} & \textcolor{black}{0.897} & \textcolor{black}{0.832} & \textcolor{black}{0.832} & \textcolor{black}{0.833} & \textcolor{black}{0.832} & \textcolor{black}{0.854}\\
oct$(bshr)$ & \textcolor{black}{0.924} & \textcolor{black}{0.928} & \textcolor{black}{0.911} & \textcolor{black}{0.915} & \textcolor{black}{0.850} & \textcolor{black}{0.851} & \textcolor{black}{0.852} & \textcolor{black}{0.851} & \textcolor{black}{0.869}\\
oct$(hshr)$ & \textcolor{black}{0.907} & \textcolor{black}{0.905} & \textcolor{black}{0.962} & \textcolor{black}{0.951} & \textcolor{black}{0.898} & \textcolor{black}{0.899} & \textcolor{black}{0.902} & \textcolor{black}{0.902} & \textcolor{black}{0.901}\\
oct$(hbshr)$ & \textcolor{black}{0.917} & \textcolor{black}{0.919} & \textcolor{black}{0.964} & \textcolor{black}{0.968} & \textcolor{black}{0.912} & \textcolor{black}{0.913} & \textcolor{black}{0.915} & \textcolor{black}{0.916} & \textcolor{black}{0.915}\\
oct$_h(shr)$ & \textcolor{black}{0.867} & \textcolor{black}{0.864} & \textcolor{black}{0.872} & \textcolor{black}{0.871} & \textcolor{black}{0.832} & \textcolor{black}{0.832} & \textcolor{black}{0.833} & \textcolor{black}{0.832} & \textcolor{black}{0.834}\\
oct$_h(bshr)$ & \textcolor{black}{0.886} & \textcolor{black}{0.890} & \textcolor{black}{0.890} & \textcolor{black}{0.893} & \textcolor{black}{0.850} & \textcolor{black}{0.851} & \textcolor{black}{0.852} & \textcolor{black}{0.851} & \textcolor{black}{0.852}\\
oct$_h(hshr)$ & \textcolor{black}{0.904} & \textcolor{black}{0.905} & \textcolor{black}{0.952} & \textcolor{black}{0.952} & \textcolor{black}{0.898} & \textcolor{black}{0.899} & \textcolor{black}{0.902} & \textcolor{black}{0.902} & \textcolor{black}{0.902}\\
\bottomrule
\end{tabular}

	\endgroup
	\caption{Simulation experiment. $\overline{RelCRPS}$ defined in Section 5. %A lower value, indicates a more accurate forecast. 
	Approaches performing worse than the benchmark (bootstrap base forecasts, ctjb) are highlighted in red, the best for each column is marked in bold, and the overall lowest value is highlighted in blue. The reconciliation approaches are described in Table 2.}
	\label{tab:ar2crps_app_shr}
\end{table}

\begin{table}[H]
	\centering
	\begingroup
	\spacingset{1}
	\fontsize{9}{11}\selectfont
	
\begin{tabular}[t]{l|ccccccccc}
\toprule
\multicolumn{1}{c}{\textbf{}} & \multicolumn{9}{c}{\textbf{Generation of the base forecasts paths}} \\
\cmidrule(l{0pt}r{0pt}){2-10}
\multicolumn{1}{c}{} & \multicolumn{1}{c}{} & \multicolumn{8}{c}{\makecell[c]{Gaussian approach: shrinkage covariance matrix}} \\
\multicolumn{1}{c}{\makecell[c]{\bfseries Reconciliation\\\bfseries approach}} & \multicolumn{1}{c}{ctjb} & \multicolumn{4}{c}{In-sample residuals} & \multicolumn{4}{c}{Multi-step residuals} \\
 &  & G & B & H & HB & G & B & H & HB\\
\midrule
\addlinespace[0.3em]
\multicolumn{10}{c}{\textbf{$\forall k \in \{2,1\}$}}\\
base & \textcolor{red}{1.005} & \textcolor{red}{1.008} & \textcolor{red}{1.039} & \textcolor{red}{1.045} & \textcolor{black}{0.996} & \textcolor{black}{0.999} & \textcolor{black}{1.000} & \textcolor{red}{1.003} & \textcolor{black}{1.000}\\
ct$(bu)$ & \textcolor{black}{\textbf{0.923}} & \textcolor{black}{0.923} & \textcolor{black}{0.923} & \textcolor{black}{0.923} & \textcolor{black}{\textbf{0.895}} & \textcolor{black}{0.896} & \textcolor{black}{0.897} & \textcolor{black}{0.897} & \textcolor{black}{0.897}\\
ct$(shr_{cs}, bu_{te})$ & \textcolor{black}{0.923} & \textcolor{black}{0.922} & \textcolor{black}{\textbf{0.922}} & \textcolor{black}{\textbf{0.922}} & \textcolor{black}{0.896} & \textcolor{blue}{\textbf{0.895}} & \textcolor{black}{\textbf{0.895}} & \textcolor{black}{\textbf{0.895}} & \textcolor{black}{\textbf{0.896}}\\
ct$(wlsv_{te}, bu_{cs})$ & \textcolor{black}{0.924} & \textcolor{black}{0.924} & \textcolor{black}{0.932} & \textcolor{black}{0.932} & \textcolor{black}{0.910} & \textcolor{black}{0.911} & \textcolor{black}{0.911} & \textcolor{black}{0.911} & \textcolor{black}{0.906}\\
oct$(wlsv)$ & \textcolor{black}{0.935} & \textcolor{black}{0.937} & \textcolor{black}{0.944} & \textcolor{black}{0.945} & \textcolor{black}{0.922} & \textcolor{black}{0.924} & \textcolor{black}{0.923} & \textcolor{black}{0.923} & \textcolor{black}{0.916}\\
oct$(bdshr)$ & \textcolor{black}{0.924} & \textcolor{black}{0.924} & \textcolor{black}{0.932} & \textcolor{black}{0.931} & \textcolor{black}{0.909} & \textcolor{black}{0.911} & \textcolor{black}{0.911} & \textcolor{black}{0.910} & \textcolor{black}{0.906}\\
oct$(shr)$ & \textcolor{black}{0.989} & \textcolor{black}{0.978} & \textcolor{black}{0.975} & \textcolor{black}{0.968} & \textcolor{black}{0.897} & \textcolor{black}{0.898} & \textcolor{black}{0.898} & \textcolor{black}{0.898} & \textcolor{black}{0.938}\\
oct$(bshr)$ & \textcolor{black}{0.990} & \textcolor{black}{0.993} & \textcolor{black}{0.978} & \textcolor{black}{0.981} & \textcolor{black}{0.915} & \textcolor{black}{0.915} & \textcolor{black}{0.915} & \textcolor{black}{0.915} & \textcolor{black}{0.947}\\
oct$(hshr)$ & \textcolor{black}{0.986} & \textcolor{black}{0.985} & \textcolor{red}{1.024} & \textcolor{red}{1.015} & \textcolor{black}{0.963} & \textcolor{black}{0.964} & \textcolor{black}{0.966} & \textcolor{black}{0.967} & \textcolor{black}{0.978}\\
oct$(hbshr)$ & \textcolor{black}{0.985} & \textcolor{black}{0.986} & \textcolor{red}{1.012} & \textcolor{red}{1.015} & \textcolor{black}{0.973} & \textcolor{black}{0.976} & \textcolor{black}{0.977} & \textcolor{black}{0.978} & \textcolor{black}{0.977}\\
oct$_h(shr)$ & \textcolor{black}{0.923} & \textcolor{black}{\textbf{0.922}} & \textcolor{black}{0.925} & \textcolor{black}{0.925} & \textcolor{black}{0.897} & \textcolor{black}{0.898} & \textcolor{black}{0.898} & \textcolor{black}{0.898} & \textcolor{black}{0.900}\\
oct$_h(bshr)$ & \textcolor{black}{0.941} & \textcolor{black}{0.943} & \textcolor{black}{0.942} & \textcolor{black}{0.945} & \textcolor{black}{0.913} & \textcolor{black}{0.915} & \textcolor{black}{0.915} & \textcolor{black}{0.915} & \textcolor{black}{0.916}\\
oct$_h(hshr)$ & \textcolor{black}{0.974} & \textcolor{black}{0.975} & \textcolor{red}{1.001} & \textcolor{red}{1.001} & \textcolor{black}{0.964} & \textcolor{black}{0.964} & \textcolor{black}{0.966} & \textcolor{black}{0.966} & \textcolor{black}{0.967}\\
\addlinespace[0.3em]
\multicolumn{10}{c}{\textbf{$k = 1$}}\\
base & \textcolor{red}{1.014} & \textcolor{red}{1.018} & \textcolor{red}{1.015} & \textcolor{red}{1.019} & \textcolor{black}{0.997} & \textcolor{black}{0.999} & \textcolor{black}{0.997} & \textcolor{black}{0.998} & \textcolor{black}{1.000}\\
ct$(bu)$ & \textcolor{black}{0.983} & \textcolor{black}{0.984} & \textcolor{black}{0.984} & \textcolor{black}{0.984} & \textcolor{black}{0.967} & \textcolor{black}{0.967} & \textcolor{black}{0.969} & \textcolor{black}{0.969} & \textcolor{black}{0.969}\\
ct$(shr_{cs}, bu_{te})$ & \textcolor{black}{\textbf{0.983}} & \textcolor{black}{\textbf{0.982}} & \textcolor{black}{\textbf{0.982}} & \textcolor{black}{\textbf{0.983}} & \textcolor{black}{\textbf{0.966}} & \textcolor{black}{\textbf{0.967}} & \textcolor{black}{\textbf{0.966}} & \textcolor{blue}{\textbf{0.966}} & \textcolor{black}{\textbf{0.968}}\\
ct$(wlsv_{te}, bu_{cs})$ & \textcolor{black}{0.991} & \textcolor{black}{0.992} & \textcolor{black}{0.993} & \textcolor{black}{0.992} & \textcolor{black}{0.983} & \textcolor{black}{0.983} & \textcolor{black}{0.983} & \textcolor{black}{0.983} & \textcolor{black}{0.977}\\
oct$(wlsv)$ & \textcolor{red}{1.002} & \textcolor{red}{1.004} & \textcolor{red}{1.004} & \textcolor{red}{1.004} & \textcolor{black}{0.994} & \textcolor{black}{0.995} & \textcolor{black}{0.994} & \textcolor{black}{0.996} & \textcolor{black}{0.989}\\
oct$(bdshr)$ & \textcolor{black}{0.990} & \textcolor{black}{0.991} & \textcolor{black}{0.992} & \textcolor{black}{0.991} & \textcolor{black}{0.981} & \textcolor{black}{0.983} & \textcolor{black}{0.984} & \textcolor{black}{0.982} & \textcolor{black}{0.977}\\
oct$(shr)$ & \textcolor{red}{1.065} & \textcolor{red}{1.054} & \textcolor{red}{1.051} & \textcolor{red}{1.045} & \textcolor{black}{0.969} & \textcolor{black}{0.970} & \textcolor{black}{0.970} & \textcolor{black}{0.969} & \textcolor{red}{1.028}\\
oct$(bshr)$ & \textcolor{red}{1.061} & \textcolor{red}{1.063} & \textcolor{red}{1.050} & \textcolor{red}{1.052} & \textcolor{black}{0.986} & \textcolor{black}{0.986} & \textcolor{black}{0.987} & \textcolor{black}{0.985} & \textcolor{red}{1.034}\\
oct$(hshr)$ & \textcolor{red}{1.076} & \textcolor{red}{1.077} & \textcolor{red}{1.095} & \textcolor{red}{1.088} & \textcolor{red}{1.036} & \textcolor{red}{1.036} & \textcolor{red}{1.040} & \textcolor{red}{1.038} & \textcolor{red}{1.066}\\
oct$(hbshr)$ & \textcolor{red}{1.064} & \textcolor{red}{1.065} & \textcolor{red}{1.071} & \textcolor{red}{1.073} & \textcolor{red}{1.047} & \textcolor{red}{1.048} & \textcolor{red}{1.050} & \textcolor{red}{1.050} & \textcolor{red}{1.050}\\
oct$_h(shr)$ & \textcolor{black}{0.984} & \textcolor{black}{0.985} & \textcolor{black}{0.986} & \textcolor{black}{0.986} & \textcolor{black}{0.969} & \textcolor{black}{0.969} & \textcolor{black}{0.969} & \textcolor{black}{0.968} & \textcolor{black}{0.971}\\
oct$_h(bshr)$ & \textcolor{red}{1.003} & \textcolor{red}{1.005} & \textcolor{red}{1.003} & \textcolor{red}{1.005} & \textcolor{black}{0.985} & \textcolor{black}{0.987} & \textcolor{black}{0.987} & \textcolor{black}{0.986} & \textcolor{black}{0.987}\\
oct$_h(hshr)$ & \textcolor{red}{1.054} & \textcolor{red}{1.054} & \textcolor{red}{1.059} & \textcolor{red}{1.059} & \textcolor{red}{1.036} & \textcolor{red}{1.037} & \textcolor{red}{1.038} & \textcolor{red}{1.039} & \textcolor{red}{1.040}\\
\addlinespace[0.3em]
\multicolumn{10}{c}{\textbf{$k = 2$}}\\
base & \textcolor{black}{0.996} & \textcolor{black}{0.998} & \textcolor{red}{1.064} & \textcolor{red}{1.073} & \textcolor{black}{0.995} & \textcolor{black}{0.999} & \textcolor{red}{1.003} & \textcolor{red}{1.007} & \textcolor{black}{1.000}\\
ct$(bu)$ & \textcolor{black}{0.867} & \textcolor{black}{0.866} & \textcolor{black}{0.867} & \textcolor{black}{0.866} & \textcolor{blue}{\textbf{0.829}} & \textcolor{black}{0.829} & \textcolor{black}{0.830} & \textcolor{black}{0.830} & \textcolor{black}{0.831}\\
ct$(shr_{cs}, bu_{te})$ & \textcolor{black}{0.867} & \textcolor{black}{0.866} & \textcolor{black}{\textbf{0.866}} & \textcolor{black}{\textbf{0.866}} & \textcolor{black}{0.830} & \textcolor{black}{\textbf{0.829}} & \textcolor{black}{\textbf{0.830}} & \textcolor{black}{\textbf{0.830}} & \textcolor{black}{\textbf{0.829}}\\
ct$(wlsv_{te}, bu_{cs})$ & \textcolor{black}{\textbf{0.861}} & \textcolor{black}{\textbf{0.861}} & \textcolor{black}{0.875} & \textcolor{black}{0.875} & \textcolor{black}{0.843} & \textcolor{black}{0.845} & \textcolor{black}{0.845} & \textcolor{black}{0.845} & \textcolor{black}{0.839}\\
oct$(wlsv)$ & \textcolor{black}{0.873} & \textcolor{black}{0.874} & \textcolor{black}{0.888} & \textcolor{black}{0.889} & \textcolor{black}{0.856} & \textcolor{black}{0.857} & \textcolor{black}{0.857} & \textcolor{black}{0.856} & \textcolor{black}{0.849}\\
oct$(bdshr)$ & \textcolor{black}{0.862} & \textcolor{black}{0.861} & \textcolor{black}{0.876} & \textcolor{black}{0.874} & \textcolor{black}{0.843} & \textcolor{black}{0.844} & \textcolor{black}{0.844} & \textcolor{black}{0.844} & \textcolor{black}{0.839}\\
oct$(shr)$ & \textcolor{black}{0.918} & \textcolor{black}{0.907} & \textcolor{black}{0.905} & \textcolor{black}{0.898} & \textcolor{black}{0.831} & \textcolor{black}{0.832} & \textcolor{black}{0.832} & \textcolor{black}{0.832} & \textcolor{black}{0.856}\\
oct$(bshr)$ & \textcolor{black}{0.924} & \textcolor{black}{0.928} & \textcolor{black}{0.911} & \textcolor{black}{0.915} & \textcolor{black}{0.849} & \textcolor{black}{0.849} & \textcolor{black}{0.849} & \textcolor{black}{0.849} & \textcolor{black}{0.868}\\
oct$(hshr)$ & \textcolor{black}{0.904} & \textcolor{black}{0.901} & \textcolor{black}{0.957} & \textcolor{black}{0.946} & \textcolor{black}{0.895} & \textcolor{black}{0.896} & \textcolor{black}{0.898} & \textcolor{black}{0.900} & \textcolor{black}{0.897}\\
oct$(hbshr)$ & \textcolor{black}{0.912} & \textcolor{black}{0.913} & \textcolor{black}{0.956} & \textcolor{black}{0.961} & \textcolor{black}{0.905} & \textcolor{black}{0.909} & \textcolor{black}{0.909} & \textcolor{black}{0.911} & \textcolor{black}{0.910}\\
oct$_h(shr)$ & \textcolor{black}{0.866} & \textcolor{black}{0.863} & \textcolor{black}{0.869} & \textcolor{black}{0.869} & \textcolor{black}{0.830} & \textcolor{black}{0.831} & \textcolor{black}{0.832} & \textcolor{black}{0.832} & \textcolor{black}{0.835}\\
oct$_h(bshr)$ & \textcolor{black}{0.882} & \textcolor{black}{0.886} & \textcolor{black}{0.886} & \textcolor{black}{0.889} & \textcolor{black}{0.846} & \textcolor{black}{0.848} & \textcolor{black}{0.849} & \textcolor{black}{0.848} & \textcolor{black}{0.850}\\
oct$_h(hshr)$ & \textcolor{black}{0.901} & \textcolor{black}{0.902} & \textcolor{black}{0.947} & \textcolor{black}{0.946} & \textcolor{black}{0.896} & \textcolor{black}{0.896} & \textcolor{black}{0.898} & \textcolor{black}{0.899} & \textcolor{black}{0.900}\\
\bottomrule
\end{tabular}

	\endgroup
	\caption{Simulation experiment. ES ratio indices defined in Section 5. %A lower value, indicates amore accurate forecast. 
	Approaches performing worse than the benchmark (bootstrap base forecasts, ctjb) are highlighted in red, the best for each column is marked in bold, and the overall lowest value is highlighted in blue. The reconciliation approaches are described in Table 2.}
	\label{tab:ar2es_app_shr}
\end{table}

\newpage
\section{Forecast reconciliation of the Australian GDP dataset}
\setcounter{table}{0} 

\subsection{The dataset}
\cite{athanasopoulos2020} proposed using state-of-the-art forecast reconciliation methods to improve the accuracy of macroeconomic forecasts and facilitate aligned decision-making. 
In their empirical analysis, they applied cross-sectional forecast reconciliation to 95 Australian QNA time series that represent the Gross Domestic Product (GDP) calculated using both the income and expenditure approaches. These two approaches correspond to two distinct hierarchical structures, with GDP at the top and 15 lower-level aggregates in the income approach, and GDP as the top-level aggregate in a hierarchy of 79 time series in the expenditure approach (for more information, see \citealp{athanasopoulos2020}, pp. 702--705 and figures 21.4--21.7).
\cite{bisaglia2020} showed how to obtain a ``one-number'' forecast where the GDP reconciled forecasts are coherent for both the expenditure and income sides.
\cite{difonzo2022c, giro2022} extended the one number forecasts idea to obtain fully reconciled probabilistic forecasts, and \cite{difonzo2023} computed cross-temporally reconciled point forecasts. 

\subsection{One-step residuals and shrinkage covariance matrix}
\begin{table}[H]
	\centering
	\begingroup
	\spacingset{1}
	\fontsize{8}{10}\selectfont
	
\begin{tabular}[t]{l|>{}cccc>{}c|ccccc}
\toprule
\multicolumn{1}{c}{\textbf{}} & \multicolumn{10}{c}{\textbf{Generation of the base forecasts paths}} \\
\cmidrule(l{0pt}r{0pt}){2-11}
\multicolumn{1}{c}{\makecell[c]{\bfseries Reconciliation\\\bfseries approach}} & \multicolumn{1}{c}{ctjb} & \multicolumn{4}{c}{\makecell[c]{Gaussian approach\textsuperscript{*}}} & \multicolumn{1}{c}{ctjb} & \multicolumn{4}{c}{\makecell[c]{Gaussian approach\textsuperscript{*}}} \\
\multicolumn{1}{c}{} &  & G$_{h}$ & H$_{h}$ & G$_{oh}$ & \multicolumn{1}{c}{H$_{oh}$} &  & G$_{h}$ & H$_{h}$ & G$_{oh}$ & \multicolumn{1}{c}{H$_{oh}$}\\
\midrule
\addlinespace[0.3em]
\multicolumn{1}{c}{} & \multicolumn{5}{c}{\textbf{$\forall k \in \{4,2,1\}$}} & \multicolumn{5}{c}{\textbf{$k = 1$}}\\
base & \textcolor{black}{1.000} & \textcolor{black}{0.979} & \textcolor{black}{0.995} & \textcolor{black}{0.968} & \textcolor{black}{0.976} & \textcolor{black}{1.000} & \textcolor{black}{0.988} & \textcolor{black}{0.988} & \textcolor{black}{0.971} & \textcolor{black}{0.971}\\
ct$(shr_{cs}, bu_{te})$ & \textcolor{black}{0.937} & \textcolor{black}{0.956} & \textcolor{black}{0.956} & \textcolor{black}{0.976} & \textcolor{black}{0.976} & \textcolor{black}{0.992} & \textcolor{red}{1.008} & \textcolor{red}{1.008} & \textcolor{red}{1.029} & \textcolor{red}{1.029}\\
ct$(wls_{cs}, bu_{te})$ & \textcolor{black}{0.930} & \textcolor{black}{0.917} & \textcolor{black}{0.917} & \textcolor{black}{0.898} & \textcolor{black}{0.898} & \textcolor{black}{0.986} & \textcolor{black}{0.974} & \textcolor{black}{0.975} & \textcolor{black}{0.956} & \textcolor{black}{0.956}\\
oct$(wlsv)$ & \textcolor{black}{0.926} & \textcolor{black}{0.919} & \textcolor{black}{0.920} & \textcolor{black}{0.900} & \textcolor{black}{0.900} & \textcolor{black}{0.984} & \textcolor{black}{0.981} & \textcolor{black}{0.979} & \textcolor{black}{0.959} & \textcolor{black}{0.959}\\
oct$(bdshr)$ & \textcolor{black}{0.940} & \textcolor{black}{0.965} & \textcolor{black}{0.945} & \textcolor{black}{0.992} & \textcolor{black}{0.957} & \textcolor{black}{0.997} & \textcolor{red}{1.019} & \textcolor{red}{1.003} & \textcolor{red}{1.044} & \textcolor{red}{1.018}\\
oct$(shr)$ & \textcolor{black}{0.944} & \textcolor{red}{1.020} & \textcolor{black}{0.940} & \textcolor{red}{1.094} & \textcolor{black}{0.988} & \textcolor{red}{1.015} & \textcolor{red}{1.095} & \textcolor{red}{1.010} & \textcolor{red}{1.160} & \textcolor{red}{1.059}\\
oct$(hshr)$ & \textcolor{black}{0.988} & \textcolor{black}{0.972} & \textcolor{red}{1.002} & \textcolor{black}{0.974} & \textcolor{red}{1.001} & \textcolor{red}{1.048} & \textcolor{red}{1.037} & \textcolor{red}{1.060} & \textcolor{red}{1.034} & \textcolor{red}{1.061}\\
oct$_o(wlsv)$ & \textcolor{black}{\textbf{0.926}} & \textcolor{black}{\textbf{0.911}} & \textcolor{black}{\textbf{0.912}} & \textcolor{black}{\textbf{0.896}} & \textcolor{blue}{\textbf{0.895}} & \textcolor{black}{\textbf{0.984}} & \textcolor{black}{\textbf{0.971}} & \textcolor{black}{\textbf{0.970}} & \textcolor{black}{\textbf{0.954}} & \textcolor{blue}{\textbf{0.954}}\\
oct$_o(bdshr)$ & \textcolor{black}{0.978} & \textcolor{black}{0.964} & \textcolor{black}{0.946} & \textcolor{black}{0.952} & \textcolor{black}{0.930} & \textcolor{red}{1.034} & \textcolor{red}{1.016} & \textcolor{red}{1.003} & \textcolor{red}{1.005} & \textcolor{black}{0.989}\\
oct$_o(shr)$ & \textcolor{black}{0.950} & \textcolor{black}{0.946} & \textcolor{black}{0.922} & \textcolor{black}{0.925} & \textcolor{black}{0.903} & \textcolor{red}{1.014} & \textcolor{red}{1.003} & \textcolor{black}{0.985} & \textcolor{black}{0.987} & \textcolor{black}{0.968}\\
oct$_o(hshr)$ & \textcolor{black}{0.989} & \textcolor{black}{0.966} & \textcolor{black}{0.984} & \textcolor{black}{0.954} & \textcolor{black}{0.965} & \textcolor{red}{1.047} & \textcolor{red}{1.028} & \textcolor{red}{1.038} & \textcolor{red}{1.012} & \textcolor{red}{1.023}\\
oct$_{oh}(shr)$ & \textcolor{red}{1.102} & \textcolor{red}{1.059} & \textcolor{red}{1.001} & \textcolor{red}{1.094} & \textcolor{black}{0.988} & \textcolor{red}{1.172} & \textcolor{red}{1.109} & \textcolor{red}{1.066} & \textcolor{red}{1.160} & \textcolor{red}{1.059}\\
oct$_{oh}(hshr)$ & \textcolor{red}{1.006} & \textcolor{black}{0.983} & \textcolor{red}{1.009} & \textcolor{black}{0.974} & \textcolor{red}{1.001} & \textcolor{red}{1.068} & \textcolor{red}{1.046} & \textcolor{red}{1.059} & \textcolor{red}{1.034} & \textcolor{red}{1.061}\\
\addlinespace[0.3em]
\multicolumn{1}{c}{} & \multicolumn{5}{c}{\textbf{$k = 2$}} & \multicolumn{5}{c}{\textbf{$k = 4$}}\\
base & \textcolor{black}{1.000} & \textcolor{black}{0.984} & \textcolor{black}{0.993} & \textcolor{black}{0.968} & \textcolor{black}{0.976} & \textcolor{black}{1.000} & \textcolor{black}{0.966} & \textcolor{red}{1.004} & \textcolor{black}{0.964} & \textcolor{black}{0.981}\\
ct$(shr_{cs}, bu_{te})$ & \textcolor{black}{0.949} & \textcolor{black}{0.966} & \textcolor{black}{0.966} & \textcolor{black}{0.987} & \textcolor{black}{0.987} & \textcolor{black}{0.874} & \textcolor{black}{0.896} & \textcolor{black}{0.896} & \textcolor{black}{0.914} & \textcolor{black}{0.914}\\
ct$(wls_{cs}, bu_{te})$ & \textcolor{black}{0.942} & \textcolor{black}{0.928} & \textcolor{black}{0.928} & \textcolor{black}{0.909} & \textcolor{black}{0.909} & \textcolor{black}{0.866} & \textcolor{black}{0.853} & \textcolor{black}{0.853} & \textcolor{black}{0.834} & \textcolor{black}{0.834}\\
oct$(wlsv)$ & \textcolor{black}{0.938} & \textcolor{black}{0.929} & \textcolor{black}{0.931} & \textcolor{black}{0.911} & \textcolor{black}{0.911} & \textcolor{black}{0.860} & \textcolor{black}{0.853} & \textcolor{black}{0.855} & \textcolor{black}{0.835} & \textcolor{black}{0.834}\\
oct$(bdshr)$ & \textcolor{black}{0.953} & \textcolor{black}{0.976} & \textcolor{black}{0.956} & \textcolor{red}{1.003} & \textcolor{black}{0.969} & \textcolor{black}{0.874} & \textcolor{black}{0.904} & \textcolor{black}{0.880} & \textcolor{black}{0.931} & \textcolor{black}{0.889}\\
oct$(shr)$ & \textcolor{black}{0.955} & \textcolor{red}{1.031} & \textcolor{black}{0.951} & \textcolor{red}{1.113} & \textcolor{red}{1.002} & \textcolor{black}{0.866} & \textcolor{black}{0.940} & \textcolor{black}{0.864} & \textcolor{red}{1.015} & \textcolor{black}{0.909}\\
oct$(hshr)$ & \textcolor{red}{1.001} & \textcolor{black}{0.985} & \textcolor{red}{1.014} & \textcolor{black}{0.987} & \textcolor{red}{1.016} & \textcolor{black}{0.919} & \textcolor{black}{0.900} & \textcolor{black}{0.935} & \textcolor{black}{0.904} & \textcolor{black}{0.931}\\
oct$_o(wlsv)$ & \textcolor{black}{\textbf{0.938}} & \textcolor{black}{\textbf{0.921}} & \textcolor{black}{\textbf{0.923}} & \textcolor{black}{\textbf{0.907}} & \textcolor{blue}{\textbf{0.906}} & \textcolor{black}{\textbf{0.860}} & \textcolor{black}{\textbf{0.847}} & \textcolor{black}{\textbf{0.848}} & \textcolor{black}{\textbf{0.832}} & \textcolor{blue}{\textbf{0.830}}\\
oct$_o(bdshr)$ & \textcolor{black}{0.991} & \textcolor{black}{0.974} & \textcolor{black}{0.957} & \textcolor{black}{0.964} & \textcolor{black}{0.942} & \textcolor{black}{0.914} & \textcolor{black}{0.905} & \textcolor{black}{0.883} & \textcolor{black}{0.892} & \textcolor{black}{0.865}\\
oct$_o(shr)$ & \textcolor{black}{0.965} & \textcolor{black}{0.958} & \textcolor{black}{0.934} & \textcolor{black}{0.938} & \textcolor{black}{0.916} & \textcolor{black}{0.877} & \textcolor{black}{0.882} & \textcolor{black}{0.852} & \textcolor{black}{0.854} & \textcolor{black}{0.831}\\
oct$_o(hshr)$ & \textcolor{red}{1.002} & \textcolor{black}{0.979} & \textcolor{black}{0.996} & \textcolor{black}{0.967} & \textcolor{black}{0.978} & \textcolor{black}{0.922} & \textcolor{black}{0.898} & \textcolor{black}{0.923} & \textcolor{black}{0.888} & \textcolor{black}{0.898}\\
oct$_{oh}(shr)$ & \textcolor{red}{1.120} & \textcolor{red}{1.069} & \textcolor{red}{1.013} & \textcolor{red}{1.113} & \textcolor{red}{1.002} & \textcolor{red}{1.020} & \textcolor{red}{1.002} & \textcolor{black}{0.928} & \textcolor{red}{1.015} & \textcolor{black}{0.909}\\
oct$_{oh}(hshr)$ & \textcolor{red}{1.021} & \textcolor{black}{0.996} & \textcolor{red}{1.021} & \textcolor{black}{0.987} & \textcolor{red}{1.016} & \textcolor{black}{0.934} & \textcolor{black}{0.912} & \textcolor{black}{0.951} & \textcolor{black}{0.904} & \textcolor{black}{0.931}\\
\bottomrule
\multicolumn{11}{l}{\rule{0pt}{1em}\rule{0pt}{1.75em}\makecell[l]{$^\ast$The Gaussian method employs a sample covariance matrix:\\G$_{h}$ and H$_{h}$ use multi-step residuals and G$_{oh}$ and H$_{oh}$ use overlapping and multi-step residuals.}}\\
\end{tabular}

	\endgroup
	\caption{$\overline{RelCRPS}$ indices defined in Section 5 for the Australian QNA dataset. %A lower value, indicates a more accurate forecast. 
	Approaches performing worse than the benchmark (bootstrap base forecasts, ctjb) are highlighted in red, the best for each column is marked in bold, and the overall lowest value is highlighted in blue. The reconciliation approaches are described in Table 2.}
	\label{tab:auscrps}
\end{table}

\begin{table}[H]
	\centering
	\begingroup
	\spacingset{1}
	\fontsize{8}{10}\selectfont
	
\begin{tabular}[t]{l|>{}cccc>{}c|ccccc}
\toprule
\multicolumn{1}{c}{\textbf{}} & \multicolumn{10}{c}{\textbf{Generation of the base forecasts paths}} \\
\cmidrule(l{0pt}r{0pt}){2-11}
\multicolumn{1}{c}{\makecell[c]{\bfseries Reconciliation\\\bfseries approach}} & \multicolumn{1}{c}{ctjb} & \multicolumn{4}{c}{\makecell[c]{Gaussian approach\textsuperscript{*}}} & \multicolumn{1}{c}{ctjb} & \multicolumn{4}{c}{\makecell[c]{Gaussian approach\textsuperscript{*}}} \\
\multicolumn{1}{c}{} &  & G$_{h}$ & H$_{h}$ & G$_{oh}$ & \multicolumn{1}{c}{H$_{oh}$} &  & G$_{h}$ & H$_{h}$ & G$_{oh}$ & \multicolumn{1}{c}{H$_{oh}$}\\
\midrule
\addlinespace[0.3em]
\multicolumn{1}{c}{} & \multicolumn{5}{c}{\textbf{$\forall k \in \{4,2,1\}$}} & \multicolumn{5}{c}{\textbf{$k = 1$}}\\
base & \textcolor{black}{1.000} & \textcolor{black}{0.970} & \textcolor{black}{0.988} & \textcolor{black}{0.960} & \textcolor{black}{0.970} & \textcolor{black}{1.000} & \textcolor{black}{0.977} & \textcolor{black}{0.977} & \textcolor{black}{0.965} & \textcolor{black}{0.965}\\
ct$(shr_{cs}, bu_{te})$ & \textcolor{black}{0.897} & \textcolor{black}{0.944} & \textcolor{black}{0.944} & \textcolor{black}{0.973} & \textcolor{black}{0.973} & \textcolor{black}{0.964} & \textcolor{red}{1.001} & \textcolor{red}{1.001} & \textcolor{red}{1.033} & \textcolor{red}{1.033}\\
ct$(wls_{cs}, bu_{te})$ & \textcolor{black}{\textbf{0.886}} & \textcolor{black}{0.880} & \textcolor{black}{0.880} & \textcolor{black}{\textbf{0.860}} & \textcolor{black}{0.860} & \textcolor{black}{\textbf{0.954}} & \textcolor{black}{\textbf{0.944}} & \textcolor{black}{0.945} & \textcolor{blue}{\textbf{0.928}} & \textcolor{black}{\textbf{0.928}}\\
oct$(wlsv)$ & \textcolor{black}{0.890} & \textcolor{black}{0.890} & \textcolor{black}{0.894} & \textcolor{black}{0.872} & \textcolor{black}{0.872} & \textcolor{black}{0.958} & \textcolor{black}{0.957} & \textcolor{black}{0.957} & \textcolor{black}{0.938} & \textcolor{black}{0.939}\\
oct$(bdshr)$ & \textcolor{black}{0.905} & \textcolor{black}{0.956} & \textcolor{black}{0.934} & \textcolor{black}{0.992} & \textcolor{black}{0.954} & \textcolor{black}{0.972} & \textcolor{red}{1.014} & \textcolor{black}{0.994} & \textcolor{red}{1.048} & \textcolor{red}{1.018}\\
oct$(shr)$ & \textcolor{black}{0.895} & \textcolor{black}{0.979} & \textcolor{black}{0.895} & \textcolor{red}{1.053} & \textcolor{black}{0.944} & \textcolor{black}{0.973} & \textcolor{red}{1.060} & \textcolor{black}{0.969} & \textcolor{red}{1.121} & \textcolor{red}{1.015}\\
oct$(hshr)$ & \textcolor{black}{0.951} & \textcolor{black}{0.940} & \textcolor{black}{0.973} & \textcolor{black}{0.959} & \textcolor{black}{0.992} & \textcolor{red}{1.017} & \textcolor{red}{1.010} & \textcolor{red}{1.034} & \textcolor{red}{1.023} & \textcolor{red}{1.055}\\
oct$_o(wlsv)$ & \textcolor{black}{0.891} & \textcolor{black}{\textbf{0.879}} & \textcolor{black}{0.881} & \textcolor{black}{0.864} & \textcolor{black}{0.864} & \textcolor{black}{0.958} & \textcolor{black}{0.945} & \textcolor{black}{0.945} & \textcolor{black}{0.931} & \textcolor{black}{0.931}\\
oct$_o(bdshr)$ & \textcolor{black}{0.940} & \textcolor{black}{0.928} & \textcolor{black}{0.910} & \textcolor{black}{0.918} & \textcolor{black}{0.895} & \textcolor{red}{1.004} & \textcolor{black}{0.986} & \textcolor{black}{0.971} & \textcolor{black}{0.980} & \textcolor{black}{0.961}\\
oct$_o(shr)$ & \textcolor{black}{0.900} & \textcolor{black}{0.899} & \textcolor{black}{\textbf{0.876}} & \textcolor{black}{0.878} & \textcolor{blue}{\textbf{0.858}} & \textcolor{black}{0.973} & \textcolor{black}{0.963} & \textcolor{black}{\textbf{0.944}} & \textcolor{black}{0.949} & \textcolor{black}{0.930}\\
oct$_o(hshr)$ & \textcolor{black}{0.956} & \textcolor{black}{0.936} & \textcolor{black}{0.955} & \textcolor{black}{0.922} & \textcolor{black}{0.936} & \textcolor{red}{1.021} & \textcolor{red}{1.004} & \textcolor{red}{1.012} & \textcolor{black}{0.987} & \textcolor{black}{1.000}\\
oct$_{oh}(shr)$ & \textcolor{red}{1.059} & \textcolor{red}{1.015} & \textcolor{black}{0.956} & \textcolor{red}{1.053} & \textcolor{black}{0.945} & \textcolor{red}{1.130} & \textcolor{red}{1.063} & \textcolor{red}{1.019} & \textcolor{red}{1.121} & \textcolor{red}{1.016}\\
oct$_{oh}(hshr)$ & \textcolor{black}{0.986} & \textcolor{black}{0.968} & \textcolor{black}{0.999} & \textcolor{black}{0.959} & \textcolor{black}{0.992} & \textcolor{red}{1.053} & \textcolor{red}{1.034} & \textcolor{red}{1.049} & \textcolor{red}{1.024} & \textcolor{red}{1.055}\\
\addlinespace[0.3em]
\multicolumn{1}{c}{} & \multicolumn{5}{c}{\textbf{$k = 2$}} & \multicolumn{5}{c}{\textbf{$k = 4$}}\\
base & \textcolor{black}{1.000} & \textcolor{black}{0.972} & \textcolor{black}{0.985} & \textcolor{black}{0.959} & \textcolor{black}{0.969} & \textcolor{black}{1.000} & \textcolor{black}{0.959} & \textcolor{red}{1.000} & \textcolor{black}{0.957} & \textcolor{black}{0.976}\\
ct$(shr_{cs}, bu_{te})$ & \textcolor{black}{0.915} & \textcolor{black}{0.961} & \textcolor{black}{0.960} & \textcolor{black}{0.991} & \textcolor{black}{0.991} & \textcolor{black}{0.818} & \textcolor{black}{0.874} & \textcolor{black}{0.874} & \textcolor{black}{0.899} & \textcolor{black}{0.900}\\
ct$(wls_{cs}, bu_{te})$ & \textcolor{black}{\textbf{0.904}} & \textcolor{black}{0.896} & \textcolor{black}{\textbf{0.896}} & \textcolor{blue}{\textbf{0.877}} & \textcolor{black}{\textbf{0.877}} & \textcolor{black}{\textbf{0.807}} & \textcolor{black}{0.805} & \textcolor{black}{0.805} & \textcolor{black}{\textbf{0.782}} & \textcolor{black}{0.783}\\
oct$(wlsv)$ & \textcolor{black}{0.909} & \textcolor{black}{0.907} & \textcolor{black}{0.912} & \textcolor{black}{0.889} & \textcolor{black}{0.889} & \textcolor{black}{0.811} & \textcolor{black}{0.813} & \textcolor{black}{0.819} & \textcolor{black}{0.794} & \textcolor{black}{0.794}\\
oct$(bdshr)$ & \textcolor{black}{0.925} & \textcolor{black}{0.976} & \textcolor{black}{0.953} & \textcolor{red}{1.013} & \textcolor{black}{0.974} & \textcolor{black}{0.825} & \textcolor{black}{0.883} & \textcolor{black}{0.860} & \textcolor{black}{0.920} & \textcolor{black}{0.876}\\
oct$(shr)$ & \textcolor{black}{0.913} & \textcolor{red}{1.000} & \textcolor{black}{0.914} & \textcolor{red}{1.076} & \textcolor{black}{0.963} & \textcolor{black}{0.807} & \textcolor{black}{0.885} & \textcolor{black}{0.808} & \textcolor{black}{0.967} & \textcolor{black}{0.861}\\
oct$(hshr)$ & \textcolor{black}{0.973} & \textcolor{black}{0.960} & \textcolor{black}{0.993} & \textcolor{black}{0.978} & \textcolor{red}{1.014} & \textcolor{black}{0.871} & \textcolor{black}{0.856} & \textcolor{black}{0.897} & \textcolor{black}{0.881} & \textcolor{black}{0.913}\\
oct$_o(wlsv)$ & \textcolor{black}{0.908} & \textcolor{black}{\textbf{0.895}} & \textcolor{black}{0.898} & \textcolor{black}{0.881} & \textcolor{black}{0.882} & \textcolor{black}{0.812} & \textcolor{black}{\textbf{0.802}} & \textcolor{black}{0.806} & \textcolor{black}{0.786} & \textcolor{black}{0.786}\\
oct$_o(bdshr)$ & \textcolor{black}{0.960} & \textcolor{black}{0.947} & \textcolor{black}{0.929} & \textcolor{black}{0.938} & \textcolor{black}{0.915} & \textcolor{black}{0.860} & \textcolor{black}{0.856} & \textcolor{black}{0.836} & \textcolor{black}{0.841} & \textcolor{black}{0.816}\\
oct$_o(shr)$ & \textcolor{black}{0.921} & \textcolor{black}{0.919} & \textcolor{black}{0.896} & \textcolor{black}{0.898} & \textcolor{black}{0.878} & \textcolor{black}{0.814} & \textcolor{black}{0.821} & \textcolor{black}{\textbf{0.796}} & \textcolor{black}{0.794} & \textcolor{blue}{\textbf{0.775}}\\
oct$_o(hshr)$ & \textcolor{black}{0.977} & \textcolor{black}{0.956} & \textcolor{black}{0.976} & \textcolor{black}{0.942} & \textcolor{black}{0.957} & \textcolor{black}{0.876} & \textcolor{black}{0.854} & \textcolor{black}{0.882} & \textcolor{black}{0.844} & \textcolor{black}{0.856}\\
oct$_{oh}(shr)$ & \textcolor{red}{1.082} & \textcolor{red}{1.029} & \textcolor{black}{0.973} & \textcolor{red}{1.076} & \textcolor{black}{0.963} & \textcolor{black}{0.971} & \textcolor{black}{0.954} & \textcolor{black}{0.882} & \textcolor{black}{0.967} & \textcolor{black}{0.861}\\
oct$_{oh}(hshr)$ & \textcolor{red}{1.007} & \textcolor{black}{0.988} & \textcolor{red}{1.017} & \textcolor{black}{0.979} & \textcolor{red}{1.014} & \textcolor{black}{0.904} & \textcolor{black}{0.888} & \textcolor{black}{0.934} & \textcolor{black}{0.881} & \textcolor{black}{0.913}\\
\bottomrule
\multicolumn{11}{l}{\rule{0pt}{1em}\rule{0pt}{1.75em}\makecell[l]{$^\ast$The Gaussian method employs a sample covariance matrix:\\G$_{h}$ and H$_{h}$ use multi-step residuals and G$_{oh}$ and H$_{oh}$ use overlapping and multi-step residuals.}}\\
\end{tabular}

	\endgroup
	\caption{ES ratio indices defined in Section 5 for the Australian QNA dataset. %A lower value, indicates a more accurate forecast. 
	Approaches performing worse than the benchmark (bootstrap base forecasts, ctjb) are highlighted in red, the best for each column is marked in bold, and the overall lowest value is highlighted in blue. The reconciliation approaches are described in Table 2.}
	\label{tab:auses}
\end{table}

\begin{table}[H]
	\centering
	\begingroup
	\spacingset{1}
	\fontsize{8}{10}\selectfont
	
\begin{tabular}[t]{l|>{}cccc>{}c|ccccc}
\toprule
\multicolumn{1}{c}{\textbf{}} & \multicolumn{10}{c}{\textbf{Generation of the base forecasts paths}} \\
\cmidrule(l{0pt}r{0pt}){2-11}
\multicolumn{1}{c}{\makecell[c]{\bfseries Reconciliation\\\bfseries approach}} & \multicolumn{1}{c}{ctjb} & \multicolumn{4}{c}{\makecell[c]{Gaussian approach\textsuperscript{*}}} & \multicolumn{1}{c}{ctjb} & \multicolumn{4}{c}{\makecell[c]{Gaussian approach\textsuperscript{*}}} \\
\multicolumn{1}{c}{} &  & G$_{h}$ & H$_{h}$ & G$_{oh}$ & \multicolumn{1}{c}{H$_{oh}$} &  & G$_{h}$ & H$_{h}$ & G$_{oh}$ & \multicolumn{1}{c}{H$_{oh}$}\\
\midrule
\addlinespace[0.3em]
\multicolumn{1}{c}{} & \multicolumn{5}{c}{\textbf{$\forall k \in \{4,2,1\}$}} & \multicolumn{5}{c}{\textbf{$k = 1$}}\\
base & \textcolor{black}{1.000} & \textcolor{black}{0.979} & \textcolor{red}{1.011} & \textcolor{black}{0.968} & \textcolor{black}{0.987} & \textcolor{black}{1.000} & \textcolor{black}{0.988} & \textcolor{black}{0.988} & \textcolor{black}{0.971} & \textcolor{black}{0.971}\\
ct$(shr_{cs}, bu_{te})$ & \textcolor{black}{0.937} & \textcolor{black}{0.960} & \textcolor{black}{0.961} & \textcolor{black}{0.962} & \textcolor{black}{0.960} & \textcolor{black}{0.992} & \textcolor{red}{1.001} & \textcolor{red}{1.001} & \textcolor{red}{1.004} & \textcolor{black}{1.000}\\
ct$(wls_{cs}, bu_{te})$ & \textcolor{black}{0.930} & \textcolor{black}{\textbf{0.951}} & \textcolor{black}{0.953} & \textcolor{blue}{\textbf{0.911}} & \textcolor{black}{0.915} & \textcolor{black}{0.986} & \textcolor{black}{0.997} & \textcolor{black}{0.998} & \textcolor{blue}{\textbf{0.964}} & \textcolor{black}{0.967}\\
oct$(wlsv)$ & \textcolor{black}{0.926} & \textcolor{black}{0.972} & \textcolor{black}{0.957} & \textcolor{black}{0.918} & \textcolor{black}{0.917} & \textcolor{black}{0.984} & \textcolor{red}{1.010} & \textcolor{red}{1.003} & \textcolor{black}{0.971} & \textcolor{black}{0.970}\\
oct$(bdshr)$ & \textcolor{black}{0.940} & \textcolor{black}{0.986} & \textcolor{black}{0.966} & \textcolor{black}{0.981} & \textcolor{black}{0.956} & \textcolor{black}{0.997} & \textcolor{red}{1.015} & \textcolor{red}{1.006} & \textcolor{red}{1.016} & \textcolor{black}{1.000}\\
oct$(shr)$ & \textcolor{black}{0.944} & \textcolor{black}{0.999} & \textcolor{black}{0.962} & \textcolor{red}{1.051} & \textcolor{black}{0.995} & \textcolor{red}{1.015} & \textcolor{red}{1.047} & \textcolor{red}{1.021} & \textcolor{red}{1.105} & \textcolor{red}{1.058}\\
oct$(hshr)$ & \textcolor{black}{0.988} & \textcolor{black}{1.000} & \textcolor{red}{1.021} & \textcolor{black}{0.979} & \textcolor{red}{1.002} & \textcolor{red}{1.048} & \textcolor{red}{1.045} & \textcolor{red}{1.066} & \textcolor{red}{1.034} & \textcolor{red}{1.053}\\
oct$_o(wlsv)$ & \textcolor{black}{\textbf{0.926}} & \textcolor{black}{0.961} & \textcolor{black}{0.948} & \textcolor{black}{0.914} & \textcolor{black}{\textbf{0.912}} & \textcolor{black}{\textbf{0.984}} & \textcolor{black}{1.000} & \textcolor{black}{0.993} & \textcolor{black}{0.966} & \textcolor{black}{\textbf{0.965}}\\
oct$_o(bdshr)$ & \textcolor{black}{0.978} & \textcolor{black}{0.956} & \textcolor{black}{0.949} & \textcolor{black}{0.949} & \textcolor{black}{0.934} & \textcolor{red}{1.034} & \textcolor{black}{\textbf{0.984}} & \textcolor{black}{\textbf{0.983}} & \textcolor{black}{0.988} & \textcolor{black}{0.977}\\
oct$_o(shr)$ & \textcolor{black}{0.950} & \textcolor{black}{0.957} & \textcolor{black}{\textbf{0.946}} & \textcolor{black}{0.933} & \textcolor{black}{0.917} & \textcolor{red}{1.014} & \textcolor{black}{0.998} & \textcolor{black}{0.995} & \textcolor{black}{0.986} & \textcolor{black}{0.974}\\
oct$_o(hshr)$ & \textcolor{black}{0.989} & \textcolor{black}{0.997} & \textcolor{red}{1.013} & \textcolor{black}{0.967} & \textcolor{black}{0.982} & \textcolor{red}{1.047} & \textcolor{red}{1.039} & \textcolor{red}{1.054} & \textcolor{red}{1.019} & \textcolor{red}{1.032}\\
oct$_{oh}(shr)$ & \textcolor{red}{1.102} & \textcolor{red}{1.010} & \textcolor{red}{1.006} & \textcolor{red}{1.051} & \textcolor{black}{0.995} & \textcolor{red}{1.172} & \textcolor{red}{1.059} & \textcolor{red}{1.063} & \textcolor{red}{1.105} & \textcolor{red}{1.058}\\
oct$_{oh}(hshr)$ & \textcolor{red}{1.006} & \textcolor{black}{0.989} & \textcolor{red}{1.004} & \textcolor{black}{0.979} & \textcolor{red}{1.002} & \textcolor{red}{1.068} & \textcolor{red}{1.037} & \textcolor{red}{1.050} & \textcolor{red}{1.034} & \textcolor{red}{1.053}\\
\addlinespace[0.3em]
\multicolumn{1}{c}{} & \multicolumn{5}{c}{\textbf{$k = 2$}} & \multicolumn{5}{c}{\textbf{$k = 4$}}\\
base & \textcolor{black}{1.000} & \textcolor{black}{0.984} & \textcolor{red}{1.009} & \textcolor{black}{0.968} & \textcolor{black}{0.987} & \textcolor{black}{1.000} & \textcolor{black}{0.966} & \textcolor{red}{1.037} & \textcolor{black}{0.964} & \textcolor{red}{1.002}\\
ct$(shr_{cs}, bu_{te})$ & \textcolor{black}{0.949} & \textcolor{black}{0.972} & \textcolor{black}{0.972} & \textcolor{black}{0.974} & \textcolor{black}{0.971} & \textcolor{black}{0.874} & \textcolor{black}{0.910} & \textcolor{black}{0.911} & \textcolor{black}{0.910} & \textcolor{black}{0.910}\\
ct$(wls_{cs}, bu_{te})$ & \textcolor{black}{0.942} & \textcolor{black}{\textbf{0.962}} & \textcolor{black}{0.964} & \textcolor{blue}{\textbf{0.923}} & \textcolor{black}{0.927} & \textcolor{black}{0.866} & \textcolor{black}{\textbf{0.897}} & \textcolor{black}{0.900} & \textcolor{black}{\textbf{0.851}} & \textcolor{black}{0.855}\\
oct$(wlsv)$ & \textcolor{black}{0.938} & \textcolor{black}{0.988} & \textcolor{black}{0.968} & \textcolor{black}{0.931} & \textcolor{black}{0.929} & \textcolor{black}{0.860} & \textcolor{black}{0.921} & \textcolor{black}{0.903} & \textcolor{black}{0.856} & \textcolor{black}{0.856}\\
oct$(bdshr)$ & \textcolor{black}{0.953} & \textcolor{red}{1.004} & \textcolor{black}{0.979} & \textcolor{black}{0.996} & \textcolor{black}{0.970} & \textcolor{black}{0.874} & \textcolor{black}{0.942} & \textcolor{black}{0.914} & \textcolor{black}{0.932} & \textcolor{black}{0.900}\\
oct$(shr)$ & \textcolor{black}{0.955} & \textcolor{red}{1.016} & \textcolor{black}{0.973} & \textcolor{red}{1.070} & \textcolor{red}{1.010} & \textcolor{black}{0.866} & \textcolor{black}{0.937} & \textcolor{black}{0.895} & \textcolor{black}{0.981} & \textcolor{black}{0.922}\\
oct$(hshr)$ & \textcolor{red}{1.001} & \textcolor{red}{1.015} & \textcolor{red}{1.034} & \textcolor{black}{0.993} & \textcolor{red}{1.017} & \textcolor{black}{0.919} & \textcolor{black}{0.942} & \textcolor{black}{0.965} & \textcolor{black}{0.913} & \textcolor{black}{0.937}\\
oct$_o(wlsv)$ & \textcolor{black}{\textbf{0.938}} & \textcolor{black}{0.976} & \textcolor{black}{\textbf{0.959}} & \textcolor{black}{0.927} & \textcolor{black}{\textbf{0.925}} & \textcolor{black}{\textbf{0.860}} & \textcolor{black}{0.910} & \textcolor{black}{0.894} & \textcolor{black}{0.853} & \textcolor{black}{0.852}\\
oct$_o(bdshr)$ & \textcolor{black}{0.991} & \textcolor{black}{0.970} & \textcolor{black}{0.963} & \textcolor{black}{0.963} & \textcolor{black}{0.948} & \textcolor{black}{0.914} & \textcolor{black}{0.917} & \textcolor{black}{0.905} & \textcolor{black}{0.899} & \textcolor{black}{0.880}\\
oct$_o(shr)$ & \textcolor{black}{0.965} & \textcolor{black}{0.973} & \textcolor{black}{0.959} & \textcolor{black}{0.948} & \textcolor{black}{0.931} & \textcolor{black}{0.877} & \textcolor{black}{0.903} & \textcolor{black}{\textbf{0.886}} & \textcolor{black}{0.868} & \textcolor{blue}{\textbf{0.850}}\\
oct$_o(hshr)$ & \textcolor{red}{1.002} & \textcolor{red}{1.013} & \textcolor{red}{1.026} & \textcolor{black}{0.980} & \textcolor{black}{0.996} & \textcolor{black}{0.922} & \textcolor{black}{0.943} & \textcolor{black}{0.962} & \textcolor{black}{0.905} & \textcolor{black}{0.921}\\
oct$_{oh}(shr)$ & \textcolor{red}{1.120} & \textcolor{red}{1.026} & \textcolor{red}{1.019} & \textcolor{red}{1.070} & \textcolor{red}{1.010} & \textcolor{red}{1.020} & \textcolor{black}{0.947} & \textcolor{black}{0.939} & \textcolor{black}{0.981} & \textcolor{black}{0.922}\\
oct$_{oh}(hshr)$ & \textcolor{red}{1.021} & \textcolor{red}{1.005} & \textcolor{red}{1.017} & \textcolor{black}{0.993} & \textcolor{red}{1.017} & \textcolor{black}{0.934} & \textcolor{black}{0.929} & \textcolor{black}{0.946} & \textcolor{black}{0.913} & \textcolor{black}{0.937}\\
\bottomrule
\multicolumn{11}{l}{\rule{0pt}{1em}\rule{0pt}{1.75em}\makecell[l]{$^\ast$The Gaussian method employs a shrinkage covariance matrix:\\G$_{h}$ and H$_{h}$ use multi-step residuals and G$_{oh}$ and H$_{oh}$ use overlapping and multi-step residuals.}}\\
\end{tabular}

	\endgroup
	\caption{$\overline{RelCRPS}$ indices defined in Section 5 for the Australian QNA dataset. %A lower value, indicates a more accurate forecast. 
	Approaches performing worse than the benchmark (bootstrap base forecasts, ctjb) are highlighted in red, the best for each column is marked in bold, and the overall lowest value is highlighted in blue. The reconciliation approaches are described in Table 2.}
	\label{tab:auscrps}
\end{table}

\begin{table}[H]
	\centering
	\begingroup
	\spacingset{1}
	\fontsize{8}{10}\selectfont
	
\begin{tabular}[t]{l|>{}cccc>{}c|ccccc}
\toprule
\multicolumn{1}{c}{\textbf{}} & \multicolumn{10}{c}{\textbf{Generation of the base forecasts paths}} \\
\cmidrule(l{0pt}r{0pt}){2-11}
\multicolumn{1}{c}{\makecell[c]{\bfseries Reconciliation\\\bfseries approach}} & \multicolumn{1}{c}{ctjb} & \multicolumn{4}{c}{\makecell[c]{Gaussian approach\textsuperscript{*}}} & \multicolumn{1}{c}{ctjb} & \multicolumn{4}{c}{\makecell[c]{Gaussian approach\textsuperscript{*}}} \\
\multicolumn{1}{c}{} &  & G$_{h}$ & H$_{h}$ & G$_{oh}$ & \multicolumn{1}{c}{H$_{oh}$} &  & G$_{h}$ & H$_{h}$ & G$_{oh}$ & \multicolumn{1}{c}{H$_{oh}$}\\
\midrule
\addlinespace[0.3em]
\multicolumn{1}{c}{} & \multicolumn{5}{c}{\textbf{$\forall k \in \{4,2,1\}$}} & \multicolumn{5}{c}{\textbf{$k = 1$}}\\
base & \textcolor{black}{1.000} & \textcolor{black}{0.967} & \textcolor{red}{1.002} & \textcolor{black}{0.957} & \textcolor{black}{0.980} & \textcolor{black}{1.000} & \textcolor{black}{0.973} & \textcolor{black}{0.973} & \textcolor{black}{0.961} & \textcolor{black}{0.962}\\
ct$(shr_{cs}, bu_{te})$ & \textcolor{black}{0.897} & \textcolor{black}{0.968} & \textcolor{black}{0.969} & \textcolor{black}{0.963} & \textcolor{black}{0.962} & \textcolor{black}{0.964} & \textcolor{red}{1.012} & \textcolor{red}{1.012} & \textcolor{red}{1.009} & \textcolor{red}{1.004}\\
ct$(wls_{cs}, bu_{te})$ & \textcolor{black}{\textbf{0.886}} & \textcolor{black}{0.939} & \textcolor{black}{0.944} & \textcolor{blue}{\textbf{0.882}} & \textcolor{black}{0.888} & \textcolor{black}{\textbf{0.954}} & \textcolor{black}{0.994} & \textcolor{black}{0.998} & \textcolor{blue}{\textbf{0.947}} & \textcolor{black}{0.952}\\
oct$(wlsv)$ & \textcolor{black}{0.890} & \textcolor{black}{0.966} & \textcolor{black}{0.959} & \textcolor{black}{0.897} & \textcolor{black}{0.901} & \textcolor{black}{0.958} & \textcolor{red}{1.017} & \textcolor{red}{1.012} & \textcolor{black}{0.960} & \textcolor{black}{0.965}\\
oct$(bdshr)$ & \textcolor{black}{0.905} & \textcolor{black}{0.997} & \textcolor{black}{0.981} & \textcolor{black}{0.986} & \textcolor{black}{0.960} & \textcolor{black}{0.972} & \textcolor{red}{1.031} & \textcolor{red}{1.021} & \textcolor{red}{1.024} & \textcolor{red}{1.005}\\
oct$(shr)$ & \textcolor{black}{0.895} & \textcolor{black}{0.979} & \textcolor{black}{0.945} & \textcolor{red}{1.021} & \textcolor{black}{0.962} & \textcolor{black}{0.973} & \textcolor{red}{1.041} & \textcolor{red}{1.011} & \textcolor{red}{1.083} & \textcolor{red}{1.028}\\
oct$(hshr)$ & \textcolor{black}{0.951} & \textcolor{black}{0.997} & \textcolor{red}{1.023} & \textcolor{black}{0.973} & \textcolor{red}{1.005} & \textcolor{red}{1.017} & \textcolor{red}{1.051} & \textcolor{red}{1.073} & \textcolor{red}{1.034} & \textcolor{red}{1.063}\\
oct$_o(wlsv)$ & \textcolor{black}{0.891} & \textcolor{black}{0.950} & \textcolor{black}{0.945} & \textcolor{black}{0.889} & \textcolor{black}{0.892} & \textcolor{black}{0.958} & \textcolor{red}{1.002} & \textcolor{black}{0.997} & \textcolor{black}{0.953} & \textcolor{black}{0.956}\\
oct$_o(bdshr)$ & \textcolor{black}{0.940} & \textcolor{black}{0.935} & \textcolor{black}{0.933} & \textcolor{black}{0.922} & \textcolor{black}{0.909} & \textcolor{red}{1.004} & \textcolor{black}{\textbf{0.965}} & \textcolor{black}{\textbf{0.964}} & \textcolor{black}{0.969} & \textcolor{black}{0.959}\\
oct$_o(shr)$ & \textcolor{black}{0.900} & \textcolor{black}{\textbf{0.935}} & \textcolor{black}{\textbf{0.928}} & \textcolor{black}{0.895} & \textcolor{black}{\textbf{0.884}} & \textcolor{black}{0.973} & \textcolor{black}{0.984} & \textcolor{black}{0.982} & \textcolor{black}{0.960} & \textcolor{black}{\textbf{0.950}}\\
oct$_o(hshr)$ & \textcolor{black}{0.956} & \textcolor{black}{0.997} & \textcolor{red}{1.015} & \textcolor{black}{0.945} & \textcolor{black}{0.965} & \textcolor{red}{1.021} & \textcolor{red}{1.049} & \textcolor{red}{1.062} & \textcolor{red}{1.007} & \textcolor{red}{1.024}\\
oct$_{oh}(shr)$ & \textcolor{red}{1.059} & \textcolor{black}{0.981} & \textcolor{black}{0.983} & \textcolor{red}{1.021} & \textcolor{black}{0.962} & \textcolor{red}{1.130} & \textcolor{red}{1.034} & \textcolor{red}{1.041} & \textcolor{red}{1.083} & \textcolor{red}{1.029}\\
oct$_{oh}(hshr)$ & \textcolor{black}{0.986} & \textcolor{black}{0.996} & \textcolor{red}{1.014} & \textcolor{black}{0.973} & \textcolor{red}{1.005} & \textcolor{red}{1.053} & \textcolor{red}{1.050} & \textcolor{red}{1.064} & \textcolor{red}{1.034} & \textcolor{red}{1.063}\\
\addlinespace[0.3em]
\multicolumn{1}{c}{} & \multicolumn{5}{c}{\textbf{$k = 2$}} & \multicolumn{5}{c}{\textbf{$k = 4$}}\\
base & \textcolor{black}{1.000} & \textcolor{black}{0.970} & \textcolor{black}{0.999} & \textcolor{black}{0.955} & \textcolor{black}{0.980} & \textcolor{black}{1.000} & \textcolor{black}{0.958} & \textcolor{red}{1.033} & \textcolor{black}{0.953} & \textcolor{black}{1.000}\\
ct$(shr_{cs}, bu_{te})$ & \textcolor{black}{0.915} & \textcolor{black}{0.987} & \textcolor{black}{0.988} & \textcolor{black}{0.983} & \textcolor{black}{0.982} & \textcolor{black}{0.818} & \textcolor{black}{0.909} & \textcolor{black}{0.910} & \textcolor{black}{0.902} & \textcolor{black}{0.902}\\
ct$(wls_{cs}, bu_{te})$ & \textcolor{black}{\textbf{0.904}} & \textcolor{black}{\textbf{0.958}} & \textcolor{black}{0.962} & \textcolor{blue}{\textbf{0.900}} & \textcolor{black}{0.906} & \textcolor{black}{\textbf{0.807}} & \textcolor{black}{0.871} & \textcolor{black}{0.876} & \textcolor{black}{\textbf{0.805}} & \textcolor{black}{0.812}\\
oct$(wlsv)$ & \textcolor{black}{0.909} & \textcolor{black}{0.988} & \textcolor{black}{0.979} & \textcolor{black}{0.916} & \textcolor{black}{0.920} & \textcolor{black}{0.811} & \textcolor{black}{0.896} & \textcolor{black}{0.891} & \textcolor{black}{0.820} & \textcolor{black}{0.825}\\
oct$(bdshr)$ & \textcolor{black}{0.925} & \textcolor{red}{1.024} & \textcolor{red}{1.005} & \textcolor{red}{1.010} & \textcolor{black}{0.984} & \textcolor{black}{0.825} & \textcolor{black}{0.938} & \textcolor{black}{0.919} & \textcolor{black}{0.926} & \textcolor{black}{0.895}\\
oct$(shr)$ & \textcolor{black}{0.913} & \textcolor{red}{1.006} & \textcolor{black}{0.967} & \textcolor{red}{1.045} & \textcolor{black}{0.982} & \textcolor{black}{0.807} & \textcolor{black}{0.898} & \textcolor{black}{0.864} & \textcolor{black}{0.940} & \textcolor{black}{0.881}\\
oct$(hshr)$ & \textcolor{black}{0.973} & \textcolor{red}{1.020} & \textcolor{red}{1.046} & \textcolor{black}{0.994} & \textcolor{red}{1.028} & \textcolor{black}{0.871} & \textcolor{black}{0.924} & \textcolor{black}{0.954} & \textcolor{black}{0.897} & \textcolor{black}{0.929}\\
oct$_o(wlsv)$ & \textcolor{black}{0.908} & \textcolor{black}{0.972} & \textcolor{black}{0.964} & \textcolor{black}{0.908} & \textcolor{black}{0.911} & \textcolor{black}{0.812} & \textcolor{black}{0.882} & \textcolor{black}{0.876} & \textcolor{black}{0.812} & \textcolor{black}{0.816}\\
oct$_o(bdshr)$ & \textcolor{black}{0.960} & \textcolor{black}{0.959} & \textcolor{black}{0.957} & \textcolor{black}{0.945} & \textcolor{black}{0.932} & \textcolor{black}{0.860} & \textcolor{black}{0.884} & \textcolor{black}{0.879} & \textcolor{black}{0.857} & \textcolor{black}{0.841}\\
oct$_o(shr)$ & \textcolor{black}{0.921} & \textcolor{black}{0.958} & \textcolor{black}{\textbf{0.950}} & \textcolor{black}{0.917} & \textcolor{black}{\textbf{0.905}} & \textcolor{black}{0.814} & \textcolor{black}{\textbf{0.867}} & \textcolor{black}{\textbf{0.857}} & \textcolor{black}{0.815} & \textcolor{blue}{\textbf{0.803}}\\
oct$_o(hshr)$ & \textcolor{black}{0.977} & \textcolor{red}{1.021} & \textcolor{red}{1.038} & \textcolor{black}{0.966} & \textcolor{black}{0.987} & \textcolor{black}{0.876} & \textcolor{black}{0.926} & \textcolor{black}{0.949} & \textcolor{black}{0.868} & \textcolor{black}{0.889}\\
oct$_{oh}(shr)$ & \textcolor{red}{1.082} & \textcolor{red}{1.002} & \textcolor{red}{1.003} & \textcolor{red}{1.045} & \textcolor{black}{0.982} & \textcolor{black}{0.971} & \textcolor{black}{0.910} & \textcolor{black}{0.911} & \textcolor{black}{0.941} & \textcolor{black}{0.882}\\
oct$_{oh}(hshr)$ & \textcolor{red}{1.007} & \textcolor{red}{1.017} & \textcolor{red}{1.036} & \textcolor{black}{0.994} & \textcolor{red}{1.028} & \textcolor{black}{0.904} & \textcolor{black}{0.924} & \textcolor{black}{0.947} & \textcolor{black}{0.896} & \textcolor{black}{0.929}\\
\bottomrule
\multicolumn{11}{l}{\rule{0pt}{1em}\rule{0pt}{1.75em}\makecell[l]{$^\ast$The Gaussian method employs a shrinkage covariance matrix:\\G$_{h}$ and H$_{h}$ use multi-step residuals and G$_{oh}$ and H$_{oh}$ use overlapping and multi-step residuals.}}\\
\end{tabular}

	\endgroup
	\caption{ES ratio indices defined in Section 5 for the Australian QNA dataset. %A lower value, indicates a more accurate forecast. 
	Approaches performing worse than the benchmark (bootstrap base forecasts, ctjb) are highlighted in red, the best for each column is marked in bold, and the overall lowest value is highlighted in blue. The reconciliation approaches are described in Table 2.}
	\label{tab:auses}
\end{table}

\newpage
\section{Australian Tourism Demand dataset}
\setcounter{table}{0} 

\begin{table}[H]
	\caption{Geographic divisions of Australia in States, Zones e Regions. Zones formed by a single region are highlighted in italics and not numbered.}
	\spacingset{1}
	\label{tab:australia}
	\fontsize{9}{10}\selectfont
	\centering
	\begin{tabular}{r l l|r l l}
		\toprule
		\textbf{Series}                      & \textbf{Name} & \textbf{Label} & \textbf{Series}   & \textbf{Name}         & \textbf{Label} \\
		\midrule
		\multicolumn{1}{l}{\textit{Total}}   &     &      & \multicolumn{3}{l}{\textit{continues Regions}}  \\
		1      & Australia     & Total          & 49      & Gippsland   & BCB  \\
		\cline{1-3}
		\multicolumn{1}{l}{\textit{States}}  &     &      & 50      & Phillip Island        & BCC  \\
		2      & New South Wales (NSW)   & A    & 51      & Central Murray        & BDA  \\
		3      & Victoria (VIC)          & B    & 52      & Goulburn    & BDB  \\
		4      & Queensland (QLD)        & C    & 53      & High Country          & BDC  \\
		5      & South Australia (SA)    & D    & 54      & Melbourne East        & BDD  \\
		6      & Western Australia (WA)  & E    & 55      & Upper Yarra & BDE  \\
		7      & Tasmania (TAS)          & F    & 56      & MurrayEast  & BDF  \\
		8      & Northern Territory (NT) & G    & 57      & Mallee      & BEA  \\
		\cline{1-3}
		\multicolumn{1}{l}{\textit{Zones}}   &     &      & 58      & Wimmera     & BEB  \\
		9      & Metro NSW     & AA   & 59      & Western Grampians     & BEC  \\
		10     & Nth Coast NSW & AB   & 60      & Bendigo Loddon        & BED  \\
		       & \textit{Sth Coast NSW}  & \textit{AC}    & 61      & Macedon     & BEE  \\
		11     & Sth NSW       & AD   & 62      & Spa Country & BEF  \\
		12     & Nth NSW       & AE   & 63      & Ballarat    & BEG  \\
		       & \textit{ACT}  & \textit{AF}    & 64      & Central Highlands     & BEG  \\
		13     & Metro VIC     & BA   & 65      & Gold Coast  & CAA  \\
		       & \textit{West Coast VIC} & \textit{BB}    & 66      & Brisbane    & CAB  \\
		14     & East Coast VIC          & BC   & 67      & Sunshine Coast        & CAC  \\
		15     & Nth East VIC  & BD   & 68      & Central Queensland    & CBA  \\
		16     & Nth West VIC  & BE   & 69      & Bundaberg   & CBB  \\
		17     & Metro QLD     & CA   & 70      & Fraser Coast          & CBC  \\
		18     & Central Coast QLD       & CB   & 71      & Mackay      & CBD  \\
		19     & Nth Coast QLD & CC   & 72      & Whitsundays & CCA  \\
		20     & Inland QLD    & CD   & 73      & Northern    & CCB  \\
		21     & Metro SA      & DA   & 74      & Tropical North Queensland       & CCC  \\
		22     & Sth Coast SA  & DB   & 75      & Darling Downs         & CDA  \\
		23     & Inland SA     & DC   & 76      & Outback     & CDB  \\
		24     & West Coast SA & DD   & 77      & Adelaide    & DAA  \\
		25     & West CoastWA  & EA   & 78      & Barossa     & DAB  \\
		       & \textit{Nth WA}         & \textit{EB}    & 79      & Adelaide Hills        & DAC  \\
		       & \textit{SthWA}          & \textit{EC}    & 80      & Limestone Coast       & DBA  \\
		       & \textit{Sth TAS}        & \textit{FA}    & 81      & Fleurieu Peninsula    & DBB  \\
		26     & Nth East TAS  & FB   & 82      & Kangaroo Island       & DBC  \\
		27     & Nth West TAS  & FC   & 83      & Murraylands & DCA  \\
		28     & Nth Coast NT  & GA   & 84      & Riverland   & DCB  \\
		29     & Central NT    & GB   & 85      & Clare Valley          & DCC  \\
		\cline{1-3}
		\multicolumn{1}{l}{\textit{Regions}} &     &      & 86      & Flinders Range and Outback      & DCD  \\
		30     & Sydney        & AAA  & 87      & Eyre Peninsula        & DDA  \\
		31     & Central Coast & AAB  & 88      & Yorke Peninsula       & DDB  \\
		32     & Hunter        & ABA  & 89      & Australia’s Coral Coast         & EAA  \\
		33     & North Coast NSW         & ABB  & 90      & Experience Perth      & EAB  \\
		34     & South Coast   & ACA  & 91      & Australia’s SouthWest & EAC  \\
		35     & Snowy Mountains         & ADA  & 92      & Australia’s North West          & EBA  \\
		36     & Capital Country         & ADB  & 93      & Australia’s Golden Outback      & ECA  \\
		37     & The Murray    & ADC  & 94      & Hobart and the South  & FAA  \\
		38     & Riverina      & ADD  & 95      & East Coast  & FBA  \\
		39     & Central NSW   & AEA  & 96      & Launceston, Tamar and the North & FBB  \\
		40     & New England North West  & AEB  & 97      & North West  & FCA  \\
		41     & Outback NSW   & AEC  & 98      & WildernessWest        & FCB  \\
		42     & Blue Mountains          & AED  & 99      & Darwin      & GAA  \\
		43     & Canberra      & AFA  & 100     & Kakadu Arnhem         & GAB  \\
		44     & Melbourne     & BAA  & 101     & Katherine Daly        & GAC  \\
		45     & Peninsula     & BAB  & 102     & Barkly      & GBA  \\
		46     & Geelong       & BAC  & 103     & Lasseter    & GBB  \\
		47     & Western       & BBA  & 104     & Alice Springs         & GBC  \\
		48     & Lakes         & BCA  & 105     & MacDonnell & GBD\\
		\bottomrule
	\end{tabular}
	\begin{flushleft}
		\begin{footnotesize}
			Source: \cite{wickramasuriya2019, difonzo2022a}
		\end{footnotesize}
	\end{flushleft}
\end{table}

\subsection{Dealing with negative reconciled forecasts}
One issue in working with time series data is the presence of negative values, which can cause difficulties for certain types of models or analyses.
For the base forecasts, using the bootstrap approach produces forecasts naturally non negative (ETS model with the log-transformation), while this is not true for the Gaussian approach. In this case, any negative forecast is set equal to zero. For the cross-temporal reconciliation, \citet{difonzo2022b, difonzo2023a} propose two solutions: either a state-of-the-art numerical optimization procedure (\texttt{osqp}, \citealp{stellato2020, stellato2019}), or a simple heuristic strategy called set-negative-to-zero (sntz). With sntz, any negative high frequency bottom time series reconciled forecasts are set to zero, and then a cross-temporal reconciliation bottom-up is used to obtain the complete set of fully coherent forecasts. \cite{difonzo2023a} found that both methods produce similar quality forecasts, but the optimization method required much more time and computational effort compared to the sntz heuristic. To reduce computational demands, we used the less time-intensive heuristic approach for reconciliation. 

\subsection{Tables for all the temporal aggregation orders} 

\begin{table}[H]
	\centering
	\begingroup
	\spacingset{1}
	\fontsize{9}{10}\selectfont
	
\begin{tabular}[t]{l|>{}cccc>{}c|ccccc}
\toprule
\multicolumn{1}{c}{\textbf{}} & \multicolumn{10}{c}{\textbf{Generation of the base forecasts paths}} \\
\cmidrule(l{0pt}r{0pt}){2-11}
\multicolumn{1}{c}{\makecell[c]{\bfseries Reconciliation\\\bfseries approach}} & \multicolumn{1}{c}{ctjb} & \multicolumn{4}{c}{\makecell[c]{Gaussian approach\textsuperscript{*}}} & \multicolumn{1}{c}{ctjb} & \multicolumn{4}{c}{\makecell[c]{Gaussian approach\textsuperscript{*}}} \\
\multicolumn{1}{c}{} &  & G & B & H & \multicolumn{1}{c}{HB} &  & G & B & H & HB\\
\midrule
\addlinespace[0.3em]
\multicolumn{1}{c}{} & \multicolumn{5}{c}{\textbf{$\forall k \in \{12,6,4,3,2,1\}$}} & \multicolumn{5}{c}{\textbf{$k = 1$}}\\
base & \textcolor{black}{1.000} & \textcolor{black}{0.971} & \textcolor{black}{0.971} & \textcolor{black}{0.973} & \textcolor{black}{0.973} & \textcolor{black}{1.000} & \textcolor{black}{0.972} & \textcolor{black}{0.972} & \textcolor{black}{0.972} & \textcolor{black}{0.972}\\
ct$(bu)$ & \textcolor{red}{1.321} & \textcolor{red}{1.011} & \textcolor{red}{1.011} & \textcolor{red}{1.011} & \textcolor{red}{1.011} & \textcolor{red}{1.077} & \textcolor{black}{0.983} & \textcolor{black}{0.982} & \textcolor{black}{0.982} & \textcolor{black}{0.982}\\
ct$(shr_{cs}, bu_{te})$ & \textcolor{red}{1.057} & \textcolor{black}{0.974} & \textcolor{black}{0.969} & \textcolor{black}{0.974} & \textcolor{black}{0.969} & \textcolor{black}{0.976} & \textcolor{black}{0.963} & \textcolor{black}{0.962} & \textcolor{black}{0.963} & \textcolor{black}{0.962}\\
ct$(wlsv_{te}, bu_{cs})$ & \textcolor{red}{1.062} & \textcolor{black}{0.974} & \textcolor{black}{0.974} & \textcolor{black}{0.972} & \textcolor{black}{0.972} & \textcolor{black}{0.976} & \textcolor{black}{0.965} & \textcolor{black}{0.965} & \textcolor{black}{0.966} & \textcolor{black}{0.966}\\
oct$(ols)$ & \textcolor{black}{0.989} & \textcolor{black}{0.989} & \textcolor{black}{0.989} & \textcolor{black}{0.987} & \textcolor{black}{0.987} & \textcolor{black}{0.982} & \textcolor{black}{0.986} & \textcolor{black}{0.988} & \textcolor{black}{0.986} & \textcolor{black}{0.989}\\
oct$(struc)$ & \textcolor{black}{0.982} & \textcolor{black}{0.962} & \textcolor{black}{0.961} & \textcolor{black}{0.961} & \textcolor{black}{0.959} & \textcolor{black}{0.970} & \textcolor{black}{0.963} & \textcolor{black}{0.963} & \textcolor{black}{0.963} & \textcolor{black}{0.963}\\
oct$(wlsv)$ & \textcolor{black}{0.987} & \textcolor{black}{0.959} & \textcolor{black}{0.959} & \textcolor{black}{0.958} & \textcolor{black}{0.957} & \textcolor{black}{0.952} & \textcolor{black}{0.957} & \textcolor{black}{0.957} & \textcolor{black}{0.957} & \textcolor{black}{0.957}\\
oct$(bdshr)$ & \textcolor{black}{0.975} & \textcolor{black}{\textbf{0.956}} & \textcolor{black}{\textbf{0.953}} & \textcolor{black}{\textbf{0.952}} & \textcolor{blue}{\textbf{0.951}} & \textcolor{blue}{\textbf{0.949}} & \textcolor{black}{\textbf{0.955}} & \textcolor{black}{\textbf{0.953}} & \textcolor{black}{\textbf{0.954}} & \textcolor{black}{\textbf{0.954}}\\
oct$_h(bshr)$ & \textcolor{black}{0.994} & \textcolor{red}{1.018} & \textcolor{red}{1.020} & \textcolor{red}{1.016} & \textcolor{red}{1.019} & \textcolor{black}{0.988} & \textcolor{red}{1.007} & \textcolor{red}{1.013} & \textcolor{red}{1.006} & \textcolor{red}{1.012}\\
oct$_h(hshr)$ & \textcolor{black}{\textbf{0.969}} & \textcolor{black}{0.993} & \textcolor{black}{0.993} & \textcolor{black}{0.990} & \textcolor{black}{0.991} & \textcolor{black}{0.953} & \textcolor{black}{0.977} & \textcolor{black}{0.977} & \textcolor{black}{0.979} & \textcolor{black}{0.979}\\
oct$_h(shr)$ & \textcolor{red}{1.007} & \textcolor{black}{0.980} & \textcolor{black}{0.972} & \textcolor{black}{0.970} & \textcolor{black}{0.970} & \textcolor{red}{1.000} & \textcolor{black}{0.986} & \textcolor{black}{0.977} & \textcolor{black}{0.976} & \textcolor{black}{0.974}\\
\addlinespace[0.3em]
\multicolumn{1}{c}{} & \multicolumn{5}{c}{\textbf{$k = 2$}} & \multicolumn{5}{c}{\textbf{$k = 3$}}\\
base & \textcolor{black}{1.000} & \textcolor{black}{0.970} & \textcolor{black}{0.969} & \textcolor{black}{0.970} & \textcolor{black}{0.971} & \textcolor{black}{1.000} & \textcolor{black}{0.971} & \textcolor{black}{0.971} & \textcolor{black}{0.972} & \textcolor{black}{0.973}\\
ct$(bu)$ & \textcolor{red}{1.189} & \textcolor{black}{0.999} & \textcolor{black}{0.999} & \textcolor{black}{0.999} & \textcolor{black}{0.999} & \textcolor{red}{1.273} & \textcolor{red}{1.010} & \textcolor{red}{1.010} & \textcolor{red}{1.010} & \textcolor{red}{1.010}\\
ct$(shr_{cs}, bu_{te})$ & \textcolor{red}{1.015} & \textcolor{black}{0.972} & \textcolor{black}{0.970} & \textcolor{black}{0.972} & \textcolor{black}{0.970} & \textcolor{red}{1.041} & \textcolor{black}{0.977} & \textcolor{black}{0.974} & \textcolor{black}{0.977} & \textcolor{black}{0.974}\\
ct$(wlsv_{te}, bu_{cs})$ & \textcolor{red}{1.016} & \textcolor{black}{0.971} & \textcolor{black}{0.971} & \textcolor{black}{0.970} & \textcolor{black}{0.970} & \textcolor{red}{1.046} & \textcolor{black}{0.976} & \textcolor{black}{0.976} & \textcolor{black}{0.974} & \textcolor{black}{0.974}\\
oct$(ols)$ & \textcolor{black}{0.992} & \textcolor{black}{0.991} & \textcolor{black}{0.991} & \textcolor{black}{0.990} & \textcolor{black}{0.991} & \textcolor{black}{0.994} & \textcolor{black}{0.992} & \textcolor{black}{0.993} & \textcolor{black}{0.991} & \textcolor{black}{0.992}\\
oct$(struc)$ & \textcolor{black}{0.982} & \textcolor{black}{0.966} & \textcolor{black}{0.965} & \textcolor{black}{0.965} & \textcolor{black}{0.965} & \textcolor{black}{0.986} & \textcolor{black}{0.967} & \textcolor{black}{0.966} & \textcolor{black}{0.966} & \textcolor{black}{0.965}\\
oct$(wlsv)$ & \textcolor{black}{0.972} & \textcolor{black}{0.961} & \textcolor{black}{0.960} & \textcolor{black}{0.960} & \textcolor{black}{0.960} & \textcolor{black}{0.983} & \textcolor{black}{0.963} & \textcolor{black}{0.962} & \textcolor{black}{0.962} & \textcolor{black}{0.962}\\
oct$(bdshr)$ & \textcolor{black}{\textbf{0.964}} & \textcolor{black}{\textbf{0.958}} & \textcolor{black}{\textbf{0.957}} & \textcolor{black}{\textbf{0.956}} & \textcolor{blue}{\textbf{0.956}} & \textcolor{black}{0.972} & \textcolor{black}{\textbf{0.960}} & \textcolor{black}{\textbf{0.958}} & \textcolor{black}{\textbf{0.957}} & \textcolor{blue}{\textbf{0.957}}\\
oct$_h(bshr)$ & \textcolor{black}{0.997} & \textcolor{red}{1.015} & \textcolor{red}{1.018} & \textcolor{red}{1.013} & \textcolor{red}{1.017} & \textcolor{black}{0.999} & \textcolor{red}{1.021} & \textcolor{red}{1.022} & \textcolor{red}{1.018} & \textcolor{red}{1.022}\\
oct$_h(hshr)$ & \textcolor{black}{0.965} & \textcolor{black}{0.987} & \textcolor{black}{0.987} & \textcolor{black}{0.986} & \textcolor{black}{0.987} & \textcolor{black}{\textbf{0.971}} & \textcolor{black}{0.994} & \textcolor{black}{0.994} & \textcolor{black}{0.992} & \textcolor{black}{0.993}\\
oct$_h(shr)$ & \textcolor{red}{1.005} & \textcolor{black}{0.986} & \textcolor{black}{0.978} & \textcolor{black}{0.976} & \textcolor{black}{0.975} & \textcolor{red}{1.009} & \textcolor{black}{0.986} & \textcolor{black}{0.978} & \textcolor{black}{0.976} & \textcolor{black}{0.976}\\
\addlinespace[0.3em]
\multicolumn{1}{c}{} & \multicolumn{5}{c}{\textbf{$k = 4$}} & \multicolumn{5}{c}{\textbf{$k = 6$}}\\
base & \textcolor{black}{1.000} & \textcolor{black}{0.973} & \textcolor{black}{0.973} & \textcolor{black}{0.974} & \textcolor{black}{0.975} & \textcolor{black}{1.000} & \textcolor{black}{0.976} & \textcolor{black}{0.976} & \textcolor{black}{0.978} & \textcolor{black}{0.978}\\
ct$(bu)$ & \textcolor{red}{1.340} & \textcolor{red}{1.016} & \textcolor{red}{1.015} & \textcolor{red}{1.015} & \textcolor{red}{1.015} & \textcolor{red}{1.450} & \textcolor{red}{1.023} & \textcolor{red}{1.023} & \textcolor{red}{1.023} & \textcolor{red}{1.023}\\
ct$(shr_{cs}, bu_{te})$ & \textcolor{red}{1.061} & \textcolor{black}{0.978} & \textcolor{black}{0.973} & \textcolor{black}{0.978} & \textcolor{black}{0.973} & \textcolor{red}{1.094} & \textcolor{black}{0.978} & \textcolor{black}{0.972} & \textcolor{black}{0.978} & \textcolor{black}{0.972}\\
ct$(wlsv_{te}, bu_{cs})$ & \textcolor{red}{1.068} & \textcolor{black}{0.977} & \textcolor{black}{0.977} & \textcolor{black}{0.974} & \textcolor{black}{0.974} & \textcolor{red}{1.103} & \textcolor{black}{0.977} & \textcolor{black}{0.977} & \textcolor{black}{0.974} & \textcolor{black}{0.974}\\
oct$(ols)$ & \textcolor{black}{0.993} & \textcolor{black}{0.991} & \textcolor{black}{0.992} & \textcolor{black}{0.990} & \textcolor{black}{0.990} & \textcolor{black}{0.989} & \textcolor{black}{0.989} & \textcolor{black}{0.989} & \textcolor{black}{0.987} & \textcolor{black}{0.986}\\
oct$(struc)$ & \textcolor{black}{0.986} & \textcolor{black}{0.965} & \textcolor{black}{0.964} & \textcolor{black}{0.964} & \textcolor{black}{0.963} & \textcolor{black}{0.986} & \textcolor{black}{0.961} & \textcolor{black}{0.960} & \textcolor{black}{0.959} & \textcolor{black}{0.957}\\
oct$(wlsv)$ & \textcolor{black}{0.990} & \textcolor{black}{0.962} & \textcolor{black}{0.961} & \textcolor{black}{0.961} & \textcolor{black}{0.960} & \textcolor{red}{1.001} & \textcolor{black}{0.960} & \textcolor{black}{0.959} & \textcolor{black}{0.958} & \textcolor{black}{0.957}\\
oct$(bdshr)$ & \textcolor{black}{0.977} & \textcolor{black}{\textbf{0.959}} & \textcolor{black}{\textbf{0.956}} & \textcolor{black}{\textbf{0.955}} & \textcolor{blue}{\textbf{0.954}} & \textcolor{black}{0.985} & \textcolor{black}{\textbf{0.956}} & \textcolor{black}{\textbf{0.953}} & \textcolor{black}{\textbf{0.950}} & \textcolor{blue}{\textbf{0.948}}\\
oct$_h(bshr)$ & \textcolor{black}{0.997} & \textcolor{red}{1.022} & \textcolor{red}{1.022} & \textcolor{red}{1.019} & \textcolor{red}{1.022} & \textcolor{black}{0.994} & \textcolor{red}{1.022} & \textcolor{red}{1.022} & \textcolor{red}{1.020} & \textcolor{red}{1.022}\\
oct$_h(hshr)$ & \textcolor{black}{\textbf{0.973}} & \textcolor{black}{0.996} & \textcolor{black}{0.997} & \textcolor{black}{0.994} & \textcolor{black}{0.995} & \textcolor{black}{\textbf{0.976}} & \textcolor{black}{1.000} & \textcolor{red}{1.001} & \textcolor{black}{0.996} & \textcolor{black}{0.997}\\
oct$_h(shr)$ & \textcolor{red}{1.009} & \textcolor{black}{0.984} & \textcolor{black}{0.976} & \textcolor{black}{0.973} & \textcolor{black}{0.973} & \textcolor{red}{1.010} & \textcolor{black}{0.978} & \textcolor{black}{0.970} & \textcolor{black}{0.967} & \textcolor{black}{0.967}\\
\addlinespace[0.3em]
\multicolumn{1}{c}{} & \multicolumn{5}{c}{\textbf{$k = 12$}} & \multicolumn{5}{c}{}\\
base & \textcolor{black}{1.000} & \textcolor{black}{0.968} & \textcolor{black}{0.967} & \textcolor{black}{0.969} & \textcolor{black}{0.969} &  &  &  &  & \\
ct$(bu)$ & \textcolor{red}{1.675} & \textcolor{red}{1.038} & \textcolor{red}{1.037} & \textcolor{red}{1.037} & \textcolor{red}{1.038} &  &  &  &  & \\
ct$(shr_{cs}, bu_{te})$ & \textcolor{red}{1.163} & \textcolor{black}{0.977} & \textcolor{black}{0.965} & \textcolor{black}{0.977} & \textcolor{black}{0.965} &  &  &  &  & \\
ct$(wlsv_{te}, bu_{cs})$ & \textcolor{red}{1.174} & \textcolor{black}{0.978} & \textcolor{black}{0.978} & \textcolor{black}{0.971} & \textcolor{black}{0.971} &  &  &  &  & \\
oct$(ols)$ & \textcolor{black}{0.982} & \textcolor{black}{0.982} & \textcolor{black}{0.983} & \textcolor{black}{0.980} & \textcolor{black}{0.975} &  &  &  &  & \\
oct$(struc)$ & \textcolor{black}{0.982} & \textcolor{black}{0.951} & \textcolor{black}{0.949} & \textcolor{black}{0.947} & \textcolor{black}{0.943} &  &  &  &  & \\
oct$(wlsv)$ & \textcolor{red}{1.025} & \textcolor{black}{0.954} & \textcolor{black}{0.953} & \textcolor{black}{0.949} & \textcolor{black}{0.947} &  &  &  &  & \\
oct$(bdshr)$ & \textcolor{red}{1.002} & \textcolor{black}{\textbf{0.950}} & \textcolor{black}{\textbf{0.944}} & \textcolor{black}{\textbf{0.939}} & \textcolor{blue}{\textbf{0.935}} &  &  &  &  & \\
oct$_h(bshr)$ & \textcolor{black}{0.987} & \textcolor{red}{1.024} & \textcolor{red}{1.021} & \textcolor{red}{1.021} & \textcolor{red}{1.019} &  &  &  &  & \\
oct$_h(hshr)$ & \textcolor{black}{\textbf{0.978}} & \textcolor{red}{1.003} & \textcolor{red}{1.005} & \textcolor{black}{0.996} & \textcolor{black}{0.997} &  &  &  &  & \\
oct$_h(shr)$ & \textcolor{red}{1.010} & \textcolor{black}{0.963} & \textcolor{black}{0.956} & \textcolor{black}{0.952} & \textcolor{black}{0.952} &  &  &  &  & \\
\bottomrule
\multicolumn{11}{l}{\rule{0pt}{1em}\rule{0pt}{1.75em}\makecell[l]{$^\ast$The Gaussian method employs a sample covariance matrix and includes four techniques\\ (G, B, H, HB) with multi-step residuals.}}\\
\end{tabular}

	\endgroup
	\caption{$\overline{RelCRPS}$ defined in Section 5 for the Australian Tourism Demand dataset. %A lower value, indicates a more accurate forecast. 
	Approaches performing worse than the benchmark (bootstrap base forecasts, ctjb) are highlighted in red, the best for each column is marked in bold, and the overall lowest value is highlighted in blue. The reconciliation approaches are described in Table 2.}
	\label{tab:vncrps_sam}
\end{table}

\begin{table}[H]
	\centering
	\begingroup
	\spacingset{1}
	\fontsize{9}{10}\selectfont
	
\begin{tabular}[t]{l|>{}cccc>{}c|ccccc}
\toprule
\multicolumn{1}{c}{\textbf{}} & \multicolumn{10}{c}{\textbf{Generation of the base forecasts paths}} \\
\cmidrule(l{0pt}r{0pt}){2-11}
\multicolumn{1}{c}{\makecell[c]{\bfseries Reconciliation\\\bfseries approach}} & \multicolumn{1}{c}{ctjb} & \multicolumn{4}{c}{\makecell[c]{Gaussian approach\textsuperscript{*}}} & \multicolumn{1}{c}{ctjb} & \multicolumn{4}{c}{\makecell[c]{Gaussian approach\textsuperscript{*}}} \\
\multicolumn{1}{c}{} &  & G & B & H & \multicolumn{1}{c}{HB} &  & G & B & H & HB\\
\midrule
\addlinespace[0.3em]
\multicolumn{1}{c}{} & \multicolumn{5}{c}{\textbf{$\forall k \in \{12,6,4,3,2,1\}$}} & \multicolumn{5}{c}{\textbf{$k = 1$}}\\
base & \textcolor{black}{1.000} & \textcolor{black}{0.956} & \textcolor{black}{0.955} & \textcolor{black}{0.958} & \textcolor{black}{0.951} & \textcolor{black}{1.000} & \textcolor{black}{0.952} & \textcolor{black}{0.950} & \textcolor{black}{0.952} & \textcolor{black}{0.950}\\
ct$(bu)$ & \textcolor{red}{2.427} & \textcolor{black}{0.983} & \textcolor{black}{0.983} & \textcolor{black}{0.983} & \textcolor{black}{0.983} & \textcolor{red}{1.759} & \textcolor{black}{0.982} & \textcolor{black}{0.982} & \textcolor{black}{0.982} & \textcolor{black}{0.982}\\
ct$(shr_{cs}, bu_{te})$ & \textcolor{red}{1.243} & \textcolor{black}{0.886} & \textcolor{black}{0.879} & \textcolor{black}{0.886} & \textcolor{black}{0.879} & \textcolor{red}{1.098} & \textcolor{black}{0.929} & \textcolor{black}{0.928} & \textcolor{black}{0.930} & \textcolor{black}{0.927}\\
ct$(wlsv_{te}, bu_{cs})$ & \textcolor{red}{1.499} & \textcolor{black}{0.977} & \textcolor{black}{0.977} & \textcolor{black}{0.971} & \textcolor{black}{0.972} & \textcolor{red}{1.241} & \textcolor{black}{0.975} & \textcolor{black}{0.975} & \textcolor{black}{0.973} & \textcolor{black}{0.974}\\
oct$(ols)$ & \textcolor{black}{0.955} & \textcolor{black}{0.893} & \textcolor{black}{0.891} & \textcolor{black}{0.893} & \textcolor{black}{0.888} & \textcolor{black}{0.975} & \textcolor{black}{0.937} & \textcolor{black}{0.936} & \textcolor{black}{0.936} & \textcolor{black}{0.935}\\
oct$(struc)$ & \textcolor{red}{1.085} & \textcolor{black}{0.917} & \textcolor{black}{0.915} & \textcolor{black}{0.916} & \textcolor{black}{0.912} & \textcolor{red}{1.027} & \textcolor{black}{0.943} & \textcolor{black}{0.942} & \textcolor{black}{0.943} & \textcolor{black}{0.942}\\
oct$(wlsv)$ & \textcolor{red}{1.132} & \textcolor{black}{0.933} & \textcolor{black}{0.929} & \textcolor{black}{0.931} & \textcolor{black}{0.927} & \textcolor{red}{1.050} & \textcolor{black}{0.951} & \textcolor{black}{0.949} & \textcolor{black}{0.950} & \textcolor{black}{0.949}\\
oct$(bdshr)$ & \textcolor{red}{1.047} & \textcolor{black}{0.904} & \textcolor{black}{0.897} & \textcolor{black}{0.897} & \textcolor{black}{0.891} & \textcolor{red}{1.009} & \textcolor{black}{0.936} & \textcolor{black}{0.933} & \textcolor{black}{0.934} & \textcolor{black}{0.931}\\
oct$_h(bshr)$ & \textcolor{black}{\textbf{0.931}} & \textcolor{black}{\textbf{0.867}} & \textcolor{black}{\textbf{0.866}} & \textcolor{black}{\textbf{0.863}} & \textcolor{blue}{\textbf{0.860}} & \textcolor{black}{\textbf{0.965}} & \textcolor{black}{\textbf{0.927}} & \textcolor{black}{0.927} & \textcolor{black}{0.925} & \textcolor{black}{0.923}\\
oct$_h(hshr)$ & \textcolor{red}{1.081} & \textcolor{black}{0.935} & \textcolor{black}{0.931} & \textcolor{black}{0.935} & \textcolor{black}{0.927} & \textcolor{red}{1.028} & \textcolor{black}{0.952} & \textcolor{black}{0.951} & \textcolor{black}{0.952} & \textcolor{black}{0.950}\\
oct$_h(shr)$ & \textcolor{red}{1.068} & \textcolor{black}{0.899} & \textcolor{black}{0.878} & \textcolor{black}{0.875} & \textcolor{black}{0.864} & \textcolor{red}{1.023} & \textcolor{black}{0.935} & \textcolor{black}{\textbf{0.923}} & \textcolor{black}{\textbf{0.921}} & \textcolor{blue}{\textbf{0.916}}\\
\addlinespace[0.3em]
\multicolumn{1}{c}{} & \multicolumn{5}{c}{\textbf{$k = 2$}} & \multicolumn{5}{c}{\textbf{$k = 3$}}\\
base & \textcolor{black}{1.000} & \textcolor{black}{0.958} & \textcolor{black}{0.954} & \textcolor{black}{0.956} & \textcolor{black}{0.953} & \textcolor{black}{1.000} & \textcolor{black}{0.961} & \textcolor{black}{0.958} & \textcolor{black}{0.960} & \textcolor{black}{0.955}\\
ct$(bu)$ & \textcolor{red}{2.176} & \textcolor{red}{1.001} & \textcolor{red}{1.001} & \textcolor{red}{1.001} & \textcolor{red}{1.001} & \textcolor{red}{2.428} & \textcolor{black}{0.998} & \textcolor{black}{0.997} & \textcolor{black}{0.997} & \textcolor{black}{0.997}\\
ct$(shr_{cs}, bu_{te})$ & \textcolor{red}{1.192} & \textcolor{black}{0.927} & \textcolor{black}{0.921} & \textcolor{black}{0.927} & \textcolor{black}{0.921} & \textcolor{red}{1.245} & \textcolor{black}{0.911} & \textcolor{black}{0.904} & \textcolor{black}{0.911} & \textcolor{black}{0.904}\\
ct$(wlsv_{te}, bu_{cs})$ & \textcolor{red}{1.400} & \textcolor{black}{0.992} & \textcolor{black}{0.992} & \textcolor{black}{0.988} & \textcolor{black}{0.988} & \textcolor{red}{1.500} & \textcolor{black}{0.991} & \textcolor{black}{0.991} & \textcolor{black}{0.986} & \textcolor{black}{0.987}\\
oct$(ols)$ & \textcolor{black}{0.985} & \textcolor{black}{0.935} & \textcolor{black}{0.932} & \textcolor{black}{0.934} & \textcolor{black}{0.930} & \textcolor{black}{0.976} & \textcolor{black}{0.918} & \textcolor{black}{0.915} & \textcolor{black}{0.917} & \textcolor{black}{0.912}\\
oct$(struc)$ & \textcolor{red}{1.075} & \textcolor{black}{0.949} & \textcolor{black}{0.947} & \textcolor{black}{0.948} & \textcolor{black}{0.944} & \textcolor{red}{1.096} & \textcolor{black}{0.939} & \textcolor{black}{0.936} & \textcolor{black}{0.938} & \textcolor{black}{0.933}\\
oct$(wlsv)$ & \textcolor{red}{1.110} & \textcolor{black}{0.960} & \textcolor{black}{0.958} & \textcolor{black}{0.958} & \textcolor{black}{0.955} & \textcolor{red}{1.142} & \textcolor{black}{0.953} & \textcolor{black}{0.949} & \textcolor{black}{0.951} & \textcolor{black}{0.946}\\
oct$(bdshr)$ & \textcolor{red}{1.045} & \textcolor{black}{0.938} & \textcolor{black}{0.933} & \textcolor{black}{0.933} & \textcolor{black}{0.929} & \textcolor{red}{1.060} & \textcolor{black}{0.926} & \textcolor{black}{0.920} & \textcolor{black}{0.921} & \textcolor{black}{0.915}\\
oct$_h(bshr)$ & \textcolor{black}{\textbf{0.967}} & \textcolor{black}{\textbf{0.917}} & \textcolor{black}{\textbf{0.916}} & \textcolor{black}{0.913} & \textcolor{black}{0.908} & \textcolor{black}{\textbf{0.954}} & \textcolor{black}{\textbf{0.895}} & \textcolor{black}{\textbf{0.895}} & \textcolor{black}{\textbf{0.892}} & \textcolor{blue}{\textbf{0.887}}\\
oct$_h(hshr)$ & \textcolor{red}{1.073} & \textcolor{black}{0.962} & \textcolor{black}{0.959} & \textcolor{black}{0.963} & \textcolor{black}{0.956} & \textcolor{red}{1.093} & \textcolor{black}{0.955} & \textcolor{black}{0.951} & \textcolor{black}{0.956} & \textcolor{black}{0.949}\\
oct$_h(shr)$ & \textcolor{red}{1.064} & \textcolor{black}{0.933} & \textcolor{black}{0.916} & \textcolor{black}{\textbf{0.913}} & \textcolor{blue}{\textbf{0.904}} & \textcolor{red}{1.082} & \textcolor{black}{0.923} & \textcolor{black}{0.903} & \textcolor{black}{0.900} & \textcolor{black}{0.890}\\
\addlinespace[0.3em]
\multicolumn{1}{c}{} & \multicolumn{5}{c}{\textbf{$k = 4$}} & \multicolumn{5}{c}{\textbf{$k = 6$}}\\
base & \textcolor{black}{1.000} & \textcolor{black}{0.960} & \textcolor{black}{0.960} & \textcolor{black}{0.962} & \textcolor{black}{0.956} & \textcolor{black}{1.000} & \textcolor{black}{0.961} & \textcolor{black}{0.959} & \textcolor{black}{0.964} & \textcolor{black}{0.956}\\
ct$(bu)$ & \textcolor{red}{2.585} & \textcolor{black}{0.996} & \textcolor{black}{0.996} & \textcolor{black}{0.995} & \textcolor{black}{0.996} & \textcolor{red}{2.849} & \textcolor{red}{1.004} & \textcolor{red}{1.003} & \textcolor{red}{1.003} & \textcolor{red}{1.004}\\
ct$(shr_{cs}, bu_{te})$ & \textcolor{red}{1.277} & \textcolor{black}{0.898} & \textcolor{black}{0.890} & \textcolor{black}{0.899} & \textcolor{black}{0.891} & \textcolor{red}{1.339} & \textcolor{black}{0.882} & \textcolor{black}{0.873} & \textcolor{black}{0.883} & \textcolor{black}{0.874}\\
ct$(wlsv_{te}, bu_{cs})$ & \textcolor{red}{1.559} & \textcolor{black}{0.990} & \textcolor{black}{0.990} & \textcolor{black}{0.984} & \textcolor{black}{0.985} & \textcolor{red}{1.662} & \textcolor{black}{0.997} & \textcolor{black}{0.997} & \textcolor{black}{0.991} & \textcolor{black}{0.992}\\
oct$(ols)$ & \textcolor{black}{0.966} & \textcolor{black}{0.905} & \textcolor{black}{0.902} & \textcolor{black}{0.904} & \textcolor{black}{0.899} & \textcolor{black}{0.962} & \textcolor{black}{0.889} & \textcolor{black}{0.887} & \textcolor{black}{0.890} & \textcolor{black}{0.885}\\
oct$(struc)$ & \textcolor{red}{1.106} & \textcolor{black}{0.930} & \textcolor{black}{0.927} & \textcolor{black}{0.928} & \textcolor{black}{0.924} & \textcolor{red}{1.132} & \textcolor{black}{0.923} & \textcolor{black}{0.919} & \textcolor{black}{0.922} & \textcolor{black}{0.916}\\
oct$(wlsv)$ & \textcolor{red}{1.157} & \textcolor{black}{0.947} & \textcolor{black}{0.943} & \textcolor{black}{0.945} & \textcolor{black}{0.939} & \textcolor{red}{1.192} & \textcolor{black}{0.942} & \textcolor{black}{0.937} & \textcolor{black}{0.941} & \textcolor{black}{0.934}\\
oct$(bdshr)$ & \textcolor{red}{1.065} & \textcolor{black}{0.917} & \textcolor{black}{0.909} & \textcolor{black}{0.910} & \textcolor{black}{0.903} & \textcolor{red}{1.084} & \textcolor{black}{0.907} & \textcolor{black}{0.897} & \textcolor{black}{0.898} & \textcolor{black}{0.890}\\
oct$_h(bshr)$ & \textcolor{black}{\textbf{0.943}} & \textcolor{black}{\textbf{0.879}} & \textcolor{black}{\textbf{0.878}} & \textcolor{black}{\textbf{0.876}} & \textcolor{blue}{\textbf{0.871}} & \textcolor{black}{\textbf{0.932}} & \textcolor{black}{\textbf{0.856}} & \textcolor{black}{\textbf{0.855}} & \textcolor{black}{\textbf{0.851}} & \textcolor{blue}{\textbf{0.848}}\\
oct$_h(hshr)$ & \textcolor{red}{1.101} & \textcolor{black}{0.949} & \textcolor{black}{0.944} & \textcolor{black}{0.949} & \textcolor{black}{0.941} & \textcolor{red}{1.126} & \textcolor{black}{0.945} & \textcolor{black}{0.939} & \textcolor{black}{0.945} & \textcolor{black}{0.936}\\
oct$_h(shr)$ & \textcolor{red}{1.089} & \textcolor{black}{0.915} & \textcolor{black}{0.893} & \textcolor{black}{0.890} & \textcolor{black}{0.878} & \textcolor{red}{1.107} & \textcolor{black}{0.899} & \textcolor{black}{0.875} & \textcolor{black}{0.871} & \textcolor{black}{0.858}\\
\addlinespace[0.3em]
\multicolumn{1}{c}{} & \multicolumn{5}{c}{\textbf{$k = 12$}} & \multicolumn{5}{c}{}\\
base & \textcolor{black}{1.000} & \textcolor{black}{0.942} & \textcolor{black}{0.947} & \textcolor{black}{0.951} & \textcolor{black}{0.937} &  &  &  &  & \\
ct$(bu)$ & \textcolor{red}{2.990} & \textcolor{black}{0.922} & \textcolor{black}{0.921} & \textcolor{black}{0.923} & \textcolor{black}{0.923} &  &  &  &  & \\
ct$(shr_{cs}, bu_{te})$ & \textcolor{red}{1.326} & \textcolor{black}{0.779} & \textcolor{black}{0.767} & \textcolor{black}{0.777} & \textcolor{black}{0.766} &  &  &  &  & \\
ct$(wlsv_{te}, bu_{cs})$ & \textcolor{red}{1.679} & \textcolor{black}{0.917} & \textcolor{black}{0.917} & \textcolor{black}{0.906} & \textcolor{black}{0.908} &  &  &  &  & \\
oct$(ols)$ & \textcolor{black}{0.872} & \textcolor{black}{0.783} & \textcolor{black}{0.784} & \textcolor{black}{0.783} & \textcolor{black}{0.779} &  &  &  &  & \\
oct$(struc)$ & \textcolor{red}{1.077} & \textcolor{black}{0.826} & \textcolor{black}{0.822} & \textcolor{black}{0.823} & \textcolor{black}{0.818} &  &  &  &  & \\
oct$(wlsv)$ & \textcolor{red}{1.149} & \textcolor{black}{0.851} & \textcolor{black}{0.845} & \textcolor{black}{0.847} & \textcolor{black}{0.840} &  &  &  &  & \\
oct$(bdshr)$ & \textcolor{red}{1.021} & \textcolor{black}{0.808} & \textcolor{black}{0.796} & \textcolor{black}{0.796} & \textcolor{black}{0.787} &  &  &  &  & \\
oct$_h(bshr)$ & \textcolor{black}{\textbf{0.833}} & \textcolor{black}{\textbf{0.741}} & \textcolor{black}{\textbf{0.741}} & \textcolor{black}{\textbf{0.737}} & \textcolor{blue}{\textbf{0.735}} &  &  &  &  & \\
oct$_h(hshr)$ & \textcolor{red}{1.066} & \textcolor{black}{0.851} & \textcolor{black}{0.846} & \textcolor{black}{0.848} & \textcolor{black}{0.838} &  &  &  &  & \\
oct$_h(shr)$ & \textcolor{red}{1.043} & \textcolor{black}{0.797} & \textcolor{black}{0.768} & \textcolor{black}{0.764} & \textcolor{black}{0.750} &  &  &  &  & \\
\bottomrule
\multicolumn{11}{l}{\rule{0pt}{1em}\rule{0pt}{1.75em}\makecell[l]{$^\ast$The Gaussian method employs a sample covariance matrix and includes four techniques\\ (G, B, H, HB) with multi-step residuals.}}\\
\end{tabular}

	\endgroup
	\caption{ES ratio indices defined in Section 5 for the Australian Tourism Demand dataset. %A lower value, indicates a more accurate forecast. 
	Approaches performing worse than the benchmark (bootstrap base forecasts, ctjb) are highlighted in red, the best for each column is marked in bold, and the overall lowest value is highlighted in blue. The reconciliation approaches are described in Table 2.}
	\label{tab:vnes_sam}
\end{table}

\begin{table}[H]
	\centering
	\begingroup
	\spacingset{1}
	\fontsize{9}{10}\selectfont
	
\begin{tabular}[t]{l|>{}cccc>{}c|ccccc}
\toprule
\multicolumn{1}{c}{\textbf{}} & \multicolumn{10}{c}{\textbf{Generation of the base forecasts paths}} \\
\cmidrule(l{0pt}r{0pt}){2-11}
\multicolumn{1}{c}{\makecell[c]{\bfseries Reconciliation\\\bfseries approach}} & \multicolumn{1}{c}{ctjb} & \multicolumn{4}{c}{\makecell[c]{Gaussian approach\textsuperscript{*}}} & \multicolumn{1}{c}{ctjb} & \multicolumn{4}{c}{\makecell[c]{Gaussian approach\textsuperscript{*}}} \\
\multicolumn{1}{c}{} &  & G & B & H & \multicolumn{1}{c}{HB} &  & G & B & H & HB\\
\midrule
\addlinespace[0.3em]
\multicolumn{1}{c}{} & \multicolumn{5}{c}{\textbf{$\forall k \in \{12,6,4,3,2,1\}$}} & \multicolumn{5}{c}{\textbf{$k = 1$}}\\
base & \textcolor{black}{1.000} & \textcolor{black}{\textbf{0.971}} & \textcolor{black}{0.972} & \textcolor{black}{\textbf{0.971}} & \textcolor{black}{0.972} & \textcolor{black}{1.000} & \textcolor{black}{\textbf{0.972}} & \textcolor{black}{0.971} & \textcolor{black}{0.972} & \textcolor{black}{0.971}\\
ct$(bu)$ & \textcolor{red}{1.321} & \textcolor{red}{1.017} & \textcolor{red}{1.018} & \textcolor{red}{1.017} & \textcolor{red}{1.017} & \textcolor{red}{1.077} & \textcolor{black}{0.983} & \textcolor{black}{0.983} & \textcolor{black}{0.983} & \textcolor{black}{0.983}\\
ct$(shr_{cs}, bu_{te})$ & \textcolor{red}{1.057} & \textcolor{red}{1.013} & \textcolor{black}{\textbf{0.971}} & \textcolor{red}{1.013} & \textcolor{black}{0.971} & \textcolor{black}{0.976} & \textcolor{black}{0.987} & \textcolor{black}{\textbf{0.961}} & \textcolor{black}{0.988} & \textcolor{black}{0.961}\\
ct$(wlsv_{te}, bu_{cs})$ & \textcolor{red}{1.062} & \textcolor{red}{1.069} & \textcolor{red}{1.070} & \textcolor{black}{0.974} & \textcolor{black}{0.974} & \textcolor{black}{0.976} & \textcolor{black}{0.986} & \textcolor{black}{0.986} & \textcolor{black}{\textbf{0.965}} & \textcolor{black}{0.965}\\
oct$(ols)$ & \textcolor{black}{0.989} & \textcolor{red}{1.163} & \textcolor{red}{1.052} & \textcolor{red}{1.139} & \textcolor{black}{0.987} & \textcolor{black}{0.982} & \textcolor{red}{1.038} & \textcolor{black}{0.992} & \textcolor{red}{1.047} & \textcolor{black}{0.987}\\
oct$(struc)$ & \textcolor{black}{0.982} & \textcolor{red}{1.099} & \textcolor{red}{1.039} & \textcolor{red}{1.037} & \textcolor{black}{0.960} & \textcolor{black}{0.970} & \textcolor{red}{1.007} & \textcolor{black}{0.971} & \textcolor{black}{0.999} & \textcolor{black}{0.962}\\
oct$(wlsv)$ & \textcolor{black}{0.987} & \textcolor{red}{1.080} & \textcolor{red}{1.041} & \textcolor{black}{0.992} & \textcolor{black}{0.958} & \textcolor{black}{0.952} & \textcolor{red}{1.004} & \textcolor{black}{0.969} & \textcolor{black}{0.978} & \textcolor{black}{0.956}\\
oct$(bdshr)$ & \textcolor{black}{0.975} & \textcolor{red}{1.072} & \textcolor{red}{1.032} & \textcolor{black}{0.985} & \textcolor{blue}{\textbf{0.950}} & \textcolor{blue}{\textbf{0.949}} & \textcolor{black}{0.999} & \textcolor{black}{0.965} & \textcolor{black}{0.975} & \textcolor{black}{\textbf{0.952}}\\
oct$_h(bshr)$ & \textcolor{black}{0.994} & \textcolor{red}{1.202} & \textcolor{red}{1.073} & \textcolor{red}{1.168} & \textcolor{red}{1.021} & \textcolor{black}{0.988} & \textcolor{red}{1.046} & \textcolor{red}{1.012} & \textcolor{red}{1.063} & \textcolor{red}{1.012}\\
oct$_h(hshr)$ & \textcolor{black}{\textbf{0.969}} & \textcolor{red}{1.066} & \textcolor{red}{1.052} & \textcolor{red}{1.008} & \textcolor{black}{0.994} & \textcolor{black}{0.953} & \textcolor{black}{0.994} & \textcolor{black}{0.972} & \textcolor{black}{0.991} & \textcolor{black}{0.979}\\
oct$_h(shr)$ & \textcolor{red}{1.007} & \textcolor{red}{1.090} & \textcolor{red}{1.046} & \textcolor{black}{1.000} & \textcolor{black}{0.970} & \textcolor{red}{1.000} & \textcolor{red}{1.035} & \textcolor{black}{0.992} & \textcolor{black}{0.998} & \textcolor{black}{0.973}\\
\addlinespace[0.3em]
\multicolumn{1}{c}{} & \multicolumn{5}{c}{\textbf{$k = 2$}} & \multicolumn{5}{c}{\textbf{$k = 3$}}\\
base & \textcolor{black}{1.000} & \textcolor{black}{\textbf{0.969}} & \textcolor{black}{0.969} & \textcolor{black}{\textbf{0.968}} & \textcolor{black}{0.968} & \textcolor{black}{1.000} & \textcolor{black}{\textbf{0.971}} & \textcolor{black}{\textbf{0.970}} & \textcolor{black}{\textbf{0.969}} & \textcolor{black}{0.970}\\
ct$(bu)$ & \textcolor{red}{1.189} & \textcolor{black}{1.000} & \textcolor{red}{1.000} & \textcolor{red}{1.000} & \textcolor{red}{1.000} & \textcolor{red}{1.273} & \textcolor{red}{1.013} & \textcolor{red}{1.013} & \textcolor{red}{1.013} & \textcolor{red}{1.013}\\
ct$(shr_{cs}, bu_{te})$ & \textcolor{red}{1.015} & \textcolor{red}{1.004} & \textcolor{black}{\textbf{0.968}} & \textcolor{red}{1.004} & \textcolor{black}{0.968} & \textcolor{red}{1.041} & \textcolor{red}{1.013} & \textcolor{black}{0.973} & \textcolor{red}{1.014} & \textcolor{black}{0.973}\\
ct$(wlsv_{te}, bu_{cs})$ & \textcolor{red}{1.016} & \textcolor{red}{1.043} & \textcolor{red}{1.044} & \textcolor{black}{0.969} & \textcolor{black}{0.969} & \textcolor{red}{1.046} & \textcolor{red}{1.067} & \textcolor{red}{1.068} & \textcolor{black}{0.974} & \textcolor{black}{0.974}\\
oct$(ols)$ & \textcolor{black}{0.992} & \textcolor{red}{1.118} & \textcolor{red}{1.037} & \textcolor{red}{1.092} & \textcolor{black}{0.989} & \textcolor{black}{0.994} & \textcolor{red}{1.153} & \textcolor{red}{1.053} & \textcolor{red}{1.124} & \textcolor{black}{0.990}\\
oct$(struc)$ & \textcolor{black}{0.982} & \textcolor{red}{1.075} & \textcolor{red}{1.022} & \textcolor{red}{1.020} & \textcolor{black}{0.963} & \textcolor{black}{0.986} & \textcolor{red}{1.099} & \textcolor{red}{1.041} & \textcolor{red}{1.033} & \textcolor{black}{0.964}\\
oct$(wlsv)$ & \textcolor{black}{0.972} & \textcolor{red}{1.064} & \textcolor{red}{1.021} & \textcolor{black}{0.987} & \textcolor{black}{0.958} & \textcolor{black}{0.983} & \textcolor{red}{1.083} & \textcolor{red}{1.041} & \textcolor{black}{0.993} & \textcolor{black}{0.960}\\
oct$(bdshr)$ & \textcolor{black}{\textbf{0.964}} & \textcolor{red}{1.057} & \textcolor{red}{1.015} & \textcolor{black}{0.983} & \textcolor{blue}{\textbf{0.953}} & \textcolor{black}{0.972} & \textcolor{red}{1.075} & \textcolor{red}{1.033} & \textcolor{black}{0.988} & \textcolor{blue}{\textbf{0.955}}\\
oct$_h(bshr)$ & \textcolor{black}{0.997} & \textcolor{red}{1.145} & \textcolor{red}{1.059} & \textcolor{red}{1.114} & \textcolor{red}{1.016} & \textcolor{black}{0.999} & \textcolor{red}{1.190} & \textcolor{red}{1.075} & \textcolor{red}{1.151} & \textcolor{red}{1.021}\\
oct$_h(hshr)$ & \textcolor{black}{0.965} & \textcolor{red}{1.050} & \textcolor{red}{1.029} & \textcolor{red}{1.001} & \textcolor{black}{0.986} & \textcolor{black}{\textbf{0.971}} & \textcolor{red}{1.067} & \textcolor{red}{1.051} & \textcolor{red}{1.009} & \textcolor{black}{0.994}\\
oct$_h(shr)$ & \textcolor{red}{1.005} & \textcolor{red}{1.083} & \textcolor{red}{1.035} & \textcolor{red}{1.001} & \textcolor{black}{0.973} & \textcolor{red}{1.009} & \textcolor{red}{1.097} & \textcolor{red}{1.050} & \textcolor{red}{1.004} & \textcolor{black}{0.974}\\
\addlinespace[0.3em]
\multicolumn{1}{c}{} & \multicolumn{5}{c}{\textbf{$k = 4$}} & \multicolumn{5}{c}{\textbf{$k = 6$}}\\
base & \textcolor{black}{1.000} & \textcolor{black}{\textbf{0.973}} & \textcolor{black}{\textbf{0.973}} & \textcolor{black}{\textbf{0.971}} & \textcolor{black}{0.973} & \textcolor{black}{1.000} & \textcolor{black}{\textbf{0.976}} & \textcolor{black}{0.977} & \textcolor{black}{\textbf{0.975}} & \textcolor{black}{0.977}\\
ct$(bu)$ & \textcolor{red}{1.340} & \textcolor{red}{1.021} & \textcolor{red}{1.021} & \textcolor{red}{1.021} & \textcolor{red}{1.021} & \textcolor{red}{1.450} & \textcolor{red}{1.032} & \textcolor{red}{1.033} & \textcolor{red}{1.032} & \textcolor{red}{1.033}\\
ct$(shr_{cs}, bu_{te})$ & \textcolor{red}{1.061} & \textcolor{red}{1.018} & \textcolor{black}{0.974} & \textcolor{red}{1.018} & \textcolor{black}{0.974} & \textcolor{red}{1.094} & \textcolor{red}{1.023} & \textcolor{black}{\textbf{0.974}} & \textcolor{red}{1.024} & \textcolor{black}{0.974}\\
ct$(wlsv_{te}, bu_{cs})$ & \textcolor{red}{1.068} & \textcolor{red}{1.087} & \textcolor{red}{1.089} & \textcolor{black}{0.976} & \textcolor{black}{0.976} & \textcolor{red}{1.103} & \textcolor{red}{1.108} & \textcolor{red}{1.110} & \textcolor{black}{0.978} & \textcolor{black}{0.978}\\
oct$(ols)$ & \textcolor{black}{0.993} & \textcolor{red}{1.186} & \textcolor{red}{1.068} & \textcolor{red}{1.148} & \textcolor{black}{0.989} & \textcolor{black}{0.989} & \textcolor{red}{1.223} & \textcolor{red}{1.080} & \textcolor{red}{1.184} & \textcolor{black}{0.987}\\
oct$(struc)$ & \textcolor{black}{0.986} & \textcolor{red}{1.120} & \textcolor{red}{1.057} & \textcolor{red}{1.042} & \textcolor{black}{0.962} & \textcolor{black}{0.986} & \textcolor{red}{1.141} & \textcolor{red}{1.071} & \textcolor{red}{1.054} & \textcolor{black}{0.959}\\
oct$(wlsv)$ & \textcolor{black}{0.990} & \textcolor{red}{1.100} & \textcolor{red}{1.059} & \textcolor{black}{0.996} & \textcolor{black}{0.959} & \textcolor{red}{1.001} & \textcolor{red}{1.115} & \textcolor{red}{1.076} & \textcolor{black}{0.998} & \textcolor{black}{0.958}\\
oct$(bdshr)$ & \textcolor{black}{0.977} & \textcolor{red}{1.091} & \textcolor{red}{1.049} & \textcolor{black}{0.989} & \textcolor{blue}{\textbf{0.952}} & \textcolor{black}{0.985} & \textcolor{red}{1.103} & \textcolor{red}{1.064} & \textcolor{black}{0.989} & \textcolor{blue}{\textbf{0.949}}\\
oct$_h(bshr)$ & \textcolor{black}{0.997} & \textcolor{red}{1.230} & \textcolor{red}{1.089} & \textcolor{red}{1.178} & \textcolor{red}{1.023} & \textcolor{black}{0.994} & \textcolor{red}{1.278} & \textcolor{red}{1.101} & \textcolor{red}{1.219} & \textcolor{red}{1.025}\\
oct$_h(hshr)$ & \textcolor{black}{\textbf{0.973}} & \textcolor{red}{1.084} & \textcolor{red}{1.071} & \textcolor{red}{1.012} & \textcolor{black}{0.996} & \textcolor{black}{\textbf{0.976}} & \textcolor{red}{1.097} & \textcolor{red}{1.091} & \textcolor{red}{1.017} & \textcolor{red}{1.002}\\
oct$_h(shr)$ & \textcolor{red}{1.009} & \textcolor{red}{1.108} & \textcolor{red}{1.062} & \textcolor{red}{1.003} & \textcolor{black}{0.972} & \textcolor{red}{1.010} & \textcolor{red}{1.113} & \textcolor{red}{1.070} & \textcolor{red}{1.000} & \textcolor{black}{0.968}\\
\addlinespace[0.3em]
\multicolumn{1}{c}{} & \multicolumn{5}{c}{\textbf{$k = 12$}} & \multicolumn{5}{c}{}\\
base & \textcolor{black}{1.000} & \textcolor{black}{\textbf{0.968}} & \textcolor{black}{\textbf{0.969}} & \textcolor{black}{\textbf{0.969}} & \textcolor{black}{0.971} &  &  &  &  & \\
ct$(bu)$ & \textcolor{red}{1.675} & \textcolor{red}{1.056} & \textcolor{red}{1.057} & \textcolor{red}{1.057} & \textcolor{red}{1.057} &  &  &  &  & \\
ct$(shr_{cs}, bu_{te})$ & \textcolor{red}{1.163} & \textcolor{red}{1.032} & \textcolor{black}{0.974} & \textcolor{red}{1.033} & \textcolor{black}{0.974} &  &  &  &  & \\
ct$(wlsv_{te}, bu_{cs})$ & \textcolor{red}{1.174} & \textcolor{red}{1.128} & \textcolor{red}{1.130} & \textcolor{black}{0.982} & \textcolor{black}{0.982} &  &  &  &  & \\
oct$(ols)$ & \textcolor{black}{0.982} & \textcolor{red}{1.277} & \textcolor{red}{1.085} & \textcolor{red}{1.252} & \textcolor{black}{0.982} &  &  &  &  & \\
oct$(struc)$ & \textcolor{black}{0.982} & \textcolor{red}{1.158} & \textcolor{red}{1.074} & \textcolor{red}{1.075} & \textcolor{black}{0.950} &  &  &  &  & \\
oct$(wlsv)$ & \textcolor{red}{1.025} & \textcolor{red}{1.122} & \textcolor{red}{1.085} & \textcolor{red}{1.001} & \textcolor{black}{0.954} &  &  &  &  & \\
oct$(bdshr)$ & \textcolor{red}{1.002} & \textcolor{red}{1.110} & \textcolor{red}{1.071} & \textcolor{black}{0.989} & \textcolor{blue}{\textbf{0.941}} &  &  &  &  & \\
oct$_h(bshr)$ & \textcolor{black}{0.987} & \textcolor{red}{1.347} & \textcolor{red}{1.107} & \textcolor{red}{1.297} & \textcolor{red}{1.031} &  &  &  &  & \\
oct$_h(hshr)$ & \textcolor{black}{\textbf{0.978}} & \textcolor{red}{1.106} & \textcolor{red}{1.107} & \textcolor{red}{1.021} & \textcolor{red}{1.010} &  &  &  &  & \\
oct$_h(shr)$ & \textcolor{red}{1.010} & \textcolor{red}{1.107} & \textcolor{red}{1.067} & \textcolor{black}{0.991} & \textcolor{black}{0.959} &  &  &  &  & \\
\bottomrule
\multicolumn{11}{l}{\rule{0pt}{1em}\rule{0pt}{1.75em}\makecell[l]{$^\ast$The Gaussian method employs a shrikage covariance matrix and includes four techniques\\ (G, B, H, HB) with multi-step residuals..}}\\
\end{tabular}

	\endgroup
	\caption{$\overline{RelCRPS}$ defined in Section 5 for the Australian Tourism Demand dataset. %A lower value, indicates a more accurate forecast. 
	Approaches performing worse than the benchmark (bootstrap base forecasts, ctjb) are highlighted in red, the best for each column is marked in bold, and the overall lowest value is highlighted in blue. The reconciliation approaches are described in Table 2.}
	\label{tab:vncrps}
\end{table}

\begin{table}[H]
	\centering
	\begingroup
	\spacingset{1}
	\fontsize{9}{10}\selectfont
	
\begin{tabular}[t]{l|>{}cccc>{}c|ccccc}
\toprule
\multicolumn{1}{c}{\textbf{}} & \multicolumn{10}{c}{\textbf{Generation of the base forecasts paths}} \\
\cmidrule(l{0pt}r{0pt}){2-11}
\multicolumn{1}{c}{\makecell[c]{\bfseries Reconciliation\\\bfseries approach}} & \multicolumn{1}{c}{ctjb} & \multicolumn{4}{c}{\makecell[c]{Gaussian approach\textsuperscript{*}}} & \multicolumn{1}{c}{ctjb} & \multicolumn{4}{c}{\makecell[c]{Gaussian approach\textsuperscript{*}}} \\
\multicolumn{1}{c}{} &  & G & B & H & \multicolumn{1}{c}{HB} &  & G & B & H & HB\\
\midrule
\addlinespace[0.3em]
\multicolumn{1}{c}{} & \multicolumn{5}{c}{\textbf{$\forall k \in \{12,6,4,3,2,1\}$}} & \multicolumn{5}{c}{\textbf{$k = 1$}}\\
base & \textcolor{black}{1.000} & \textcolor{black}{\textbf{0.958}} & \textcolor{black}{0.984} & \textcolor{black}{\textbf{0.972}} & \textcolor{black}{0.992} & \textcolor{black}{1.000} & \textcolor{black}{\textbf{0.954}} & \textcolor{black}{0.958} & \textcolor{black}{\textbf{0.954}} & \textcolor{black}{0.958}\\
ct$(bu)$ & \textcolor{red}{2.427} & \textcolor{red}{1.040} & \textcolor{red}{1.042} & \textcolor{red}{1.040} & \textcolor{red}{1.041} & \textcolor{red}{1.759} & \textcolor{red}{1.001} & \textcolor{red}{1.002} & \textcolor{red}{1.002} & \textcolor{red}{1.002}\\
ct$(shr_{cs}, bu_{te})$ & \textcolor{red}{1.243} & \textcolor{black}{0.988} & \textcolor{black}{\textbf{0.913}} & \textcolor{black}{0.990} & \textcolor{black}{0.913} & \textcolor{red}{1.098} & \textcolor{red}{1.011} & \textcolor{black}{\textbf{0.938}} & \textcolor{red}{1.013} & \textcolor{black}{0.938}\\
ct$(wlsv_{te}, bu_{cs})$ & \textcolor{red}{1.499} & \textcolor{red}{1.117} & \textcolor{red}{1.120} & \textcolor{red}{1.025} & \textcolor{red}{1.025} & \textcolor{red}{1.241} & \textcolor{red}{1.019} & \textcolor{red}{1.020} & \textcolor{black}{0.990} & \textcolor{black}{0.990}\\
oct$(ols)$ & \textcolor{black}{0.955} & \textcolor{black}{1.000} & \textcolor{black}{0.984} & \textcolor{black}{0.985} & \textcolor{black}{0.922} & \textcolor{black}{0.975} & \textcolor{black}{0.983} & \textcolor{black}{0.961} & \textcolor{black}{0.987} & \textcolor{black}{0.945}\\
oct$(struc)$ & \textcolor{red}{1.085} & \textcolor{red}{1.094} & \textcolor{red}{1.047} & \textcolor{red}{1.018} & \textcolor{black}{0.952} & \textcolor{red}{1.027} & \textcolor{red}{1.054} & \textcolor{black}{0.981} & \textcolor{red}{1.022} & \textcolor{black}{0.953}\\
oct$(wlsv)$ & \textcolor{red}{1.132} & \textcolor{red}{1.137} & \textcolor{red}{1.065} & \textcolor{red}{1.059} & \textcolor{black}{0.969} & \textcolor{red}{1.050} & \textcolor{red}{1.078} & \textcolor{black}{0.989} & \textcolor{red}{1.043} & \textcolor{black}{0.960}\\
oct$(bdshr)$ & \textcolor{red}{1.047} & \textcolor{red}{1.085} & \textcolor{red}{1.013} & \textcolor{red}{1.011} & \textcolor{black}{0.927} & \textcolor{red}{1.009} & \textcolor{red}{1.050} & \textcolor{black}{0.966} & \textcolor{red}{1.019} & \textcolor{black}{0.942}\\
oct$_h(bshr)$ & \textcolor{black}{\textbf{0.931}} & \textcolor{red}{1.002} & \textcolor{red}{1.001} & \textcolor{black}{0.982} & \textcolor{blue}{\textbf{0.889}} & \textcolor{black}{\textbf{0.965}} & \textcolor{black}{0.980} & \textcolor{black}{0.975} & \textcolor{black}{0.985} & \textcolor{black}{0.933}\\
oct$_h(hshr)$ & \textcolor{red}{1.081} & \textcolor{red}{1.109} & \textcolor{red}{1.039} & \textcolor{red}{1.076} & \textcolor{black}{0.973} & \textcolor{red}{1.028} & \textcolor{red}{1.061} & \textcolor{black}{0.978} & \textcolor{red}{1.052} & \textcolor{black}{0.963}\\
oct$_h(shr)$ & \textcolor{red}{1.068} & \textcolor{red}{1.088} & \textcolor{red}{1.008} & \textcolor{black}{0.995} & \textcolor{black}{0.896} & \textcolor{red}{1.023} & \textcolor{red}{1.061} & \textcolor{black}{0.966} & \textcolor{red}{1.011} & \textcolor{blue}{\textbf{0.924}}\\
\addlinespace[0.3em]
\multicolumn{1}{c}{} & \multicolumn{5}{c}{\textbf{$k = 2$}} & \multicolumn{5}{c}{\textbf{$k = 3$}}\\
base & \textcolor{black}{1.000} & \textcolor{black}{\textbf{0.960}} & \textcolor{black}{0.971} & \textcolor{black}{\textbf{0.958}} & \textcolor{black}{0.972} & \textcolor{black}{1.000} & \textcolor{black}{\textbf{0.963}} & \textcolor{black}{0.981} & \textcolor{black}{\textbf{0.966}} & \textcolor{black}{0.986}\\
ct$(bu)$ & \textcolor{red}{2.176} & \textcolor{red}{1.035} & \textcolor{red}{1.036} & \textcolor{red}{1.035} & \textcolor{red}{1.035} & \textcolor{red}{2.428} & \textcolor{red}{1.042} & \textcolor{red}{1.044} & \textcolor{red}{1.042} & \textcolor{red}{1.043}\\
ct$(shr_{cs}, bu_{te})$ & \textcolor{red}{1.192} & \textcolor{red}{1.020} & \textcolor{black}{\textbf{0.942}} & \textcolor{red}{1.021} & \textcolor{black}{0.942} & \textcolor{red}{1.245} & \textcolor{red}{1.009} & \textcolor{black}{\textbf{0.931}} & \textcolor{red}{1.011} & \textcolor{black}{0.931}\\
ct$(wlsv_{te}, bu_{cs})$ & \textcolor{red}{1.400} & \textcolor{red}{1.104} & \textcolor{red}{1.106} & \textcolor{red}{1.018} & \textcolor{red}{1.019} & \textcolor{red}{1.500} & \textcolor{red}{1.127} & \textcolor{red}{1.130} & \textcolor{red}{1.029} & \textcolor{red}{1.029}\\
oct$(ols)$ & \textcolor{black}{0.985} & \textcolor{red}{1.028} & \textcolor{red}{1.008} & \textcolor{red}{1.002} & \textcolor{black}{0.950} & \textcolor{black}{0.976} & \textcolor{red}{1.020} & \textcolor{red}{1.004} & \textcolor{black}{0.994} & \textcolor{black}{0.938}\\
oct$(struc)$ & \textcolor{red}{1.075} & \textcolor{red}{1.115} & \textcolor{red}{1.051} & \textcolor{red}{1.039} & \textcolor{black}{0.967} & \textcolor{red}{1.096} & \textcolor{red}{1.117} & \textcolor{red}{1.064} & \textcolor{red}{1.033} & \textcolor{black}{0.965}\\
oct$(wlsv)$ & \textcolor{red}{1.110} & \textcolor{red}{1.149} & \textcolor{red}{1.065} & \textcolor{red}{1.070} & \textcolor{black}{0.979} & \textcolor{red}{1.142} & \textcolor{red}{1.160} & \textcolor{red}{1.082} & \textcolor{red}{1.073} & \textcolor{black}{0.981}\\
oct$(bdshr)$ & \textcolor{red}{1.045} & \textcolor{red}{1.105} & \textcolor{red}{1.024} & \textcolor{red}{1.033} & \textcolor{black}{0.949} & \textcolor{red}{1.060} & \textcolor{red}{1.109} & \textcolor{red}{1.032} & \textcolor{red}{1.029} & \textcolor{black}{0.943}\\
oct$_h(bshr)$ & \textcolor{black}{\textbf{0.967}} & \textcolor{red}{1.029} & \textcolor{red}{1.025} & \textcolor{black}{0.998} & \textcolor{black}{0.928} & \textcolor{black}{\textbf{0.954}} & \textcolor{red}{1.024} & \textcolor{red}{1.025} & \textcolor{black}{0.993} & \textcolor{blue}{\textbf{0.911}}\\
oct$_h(hshr)$ & \textcolor{red}{1.073} & \textcolor{red}{1.122} & \textcolor{red}{1.042} & \textcolor{red}{1.083} & \textcolor{black}{0.983} & \textcolor{red}{1.093} & \textcolor{red}{1.129} & \textcolor{red}{1.054} & \textcolor{red}{1.090} & \textcolor{black}{0.984}\\
oct$_h(shr)$ & \textcolor{red}{1.064} & \textcolor{red}{1.110} & \textcolor{red}{1.019} & \textcolor{red}{1.018} & \textcolor{blue}{\textbf{0.922}} & \textcolor{red}{1.082} & \textcolor{red}{1.116} & \textcolor{red}{1.030} & \textcolor{red}{1.015} & \textcolor{black}{0.915}\\
\addlinespace[0.3em]
\multicolumn{1}{c}{} & \multicolumn{5}{c}{\textbf{$k = 4$}} & \multicolumn{5}{c}{\textbf{$k = 6$}}\\
base & \textcolor{black}{1.000} & \textcolor{black}{\textbf{0.962}} & \textcolor{black}{0.987} & \textcolor{black}{\textbf{0.973}} & \textcolor{black}{0.996} & \textcolor{black}{1.000} & \textcolor{black}{\textbf{0.963}} & \textcolor{black}{0.998} & \textcolor{black}{\textbf{0.984}} & \textcolor{red}{1.011}\\
ct$(bu)$ & \textcolor{red}{2.585} & \textcolor{red}{1.052} & \textcolor{red}{1.054} & \textcolor{red}{1.053} & \textcolor{red}{1.053} & \textcolor{red}{2.849} & \textcolor{red}{1.083} & \textcolor{red}{1.085} & \textcolor{red}{1.083} & \textcolor{red}{1.084}\\
ct$(shr_{cs}, bu_{te})$ & \textcolor{red}{1.277} & \textcolor{red}{1.000} & \textcolor{black}{\textbf{0.923}} & \textcolor{red}{1.002} & \textcolor{black}{0.923} & \textcolor{red}{1.339} & \textcolor{black}{0.999} & \textcolor{black}{\textbf{0.921}} & \textcolor{red}{1.000} & \textcolor{black}{0.920}\\
ct$(wlsv_{te}, bu_{cs})$ & \textcolor{red}{1.559} & \textcolor{red}{1.150} & \textcolor{red}{1.153} & \textcolor{red}{1.037} & \textcolor{red}{1.037} & \textcolor{red}{1.662} & \textcolor{red}{1.189} & \textcolor{red}{1.193} & \textcolor{red}{1.066} & \textcolor{red}{1.066}\\
oct$(ols)$ & \textcolor{black}{0.966} & \textcolor{red}{1.022} & \textcolor{red}{1.008} & \textcolor{black}{0.994} & \textcolor{black}{0.931} & \textcolor{black}{0.962} & \textcolor{red}{1.023} & \textcolor{red}{1.014} & \textcolor{red}{1.003} & \textcolor{black}{0.930}\\
oct$(struc)$ & \textcolor{red}{1.106} & \textcolor{red}{1.120} & \textcolor{red}{1.076} & \textcolor{red}{1.031} & \textcolor{black}{0.963} & \textcolor{red}{1.132} & \textcolor{red}{1.132} & \textcolor{red}{1.100} & \textcolor{red}{1.039} & \textcolor{black}{0.972}\\
oct$(wlsv)$ & \textcolor{red}{1.157} & \textcolor{red}{1.167} & \textcolor{red}{1.097} & \textcolor{red}{1.075} & \textcolor{black}{0.982} & \textcolor{red}{1.192} & \textcolor{red}{1.187} & \textcolor{red}{1.124} & \textcolor{red}{1.090} & \textcolor{black}{0.995}\\
oct$(bdshr)$ & \textcolor{red}{1.065} & \textcolor{red}{1.112} & \textcolor{red}{1.041} & \textcolor{red}{1.025} & \textcolor{black}{0.939} & \textcolor{red}{1.084} & \textcolor{red}{1.121} & \textcolor{red}{1.058} & \textcolor{red}{1.029} & \textcolor{black}{0.940}\\
oct$_h(bshr)$ & \textcolor{black}{\textbf{0.943}} & \textcolor{red}{1.028} & \textcolor{red}{1.028} & \textcolor{black}{0.994} & \textcolor{blue}{\textbf{0.900}} & \textcolor{black}{\textbf{0.932}} & \textcolor{red}{1.029} & \textcolor{red}{1.032} & \textcolor{black}{1.000} & \textcolor{blue}{\textbf{0.887}}\\
oct$_h(hshr)$ & \textcolor{red}{1.101} & \textcolor{red}{1.137} & \textcolor{red}{1.068} & \textcolor{red}{1.093} & \textcolor{black}{0.986} & \textcolor{red}{1.126} & \textcolor{red}{1.153} & \textcolor{red}{1.089} & \textcolor{red}{1.110} & \textcolor{black}{0.999}\\
oct$_h(shr)$ & \textcolor{red}{1.089} & \textcolor{red}{1.118} & \textcolor{red}{1.039} & \textcolor{red}{1.012} & \textcolor{black}{0.910} & \textcolor{red}{1.107} & \textcolor{red}{1.118} & \textcolor{red}{1.045} & \textcolor{red}{1.006} & \textcolor{black}{0.902}\\
\addlinespace[0.3em]
\multicolumn{1}{c}{} & \multicolumn{5}{c}{\textbf{$k = 12$}} & \multicolumn{5}{c}{}\\
base & \textcolor{black}{1.000} & \textcolor{black}{0.948} & \textcolor{red}{1.010} & \textcolor{red}{1.002} & \textcolor{red}{1.033} &  &  &  &  & \\
ct$(bu)$ & \textcolor{red}{2.990} & \textcolor{red}{1.028} & \textcolor{red}{1.031} & \textcolor{red}{1.029} & \textcolor{red}{1.029} &  &  &  &  & \\
ct$(shr_{cs}, bu_{te})$ & \textcolor{red}{1.326} & \textcolor{black}{\textbf{0.897}} & \textcolor{black}{\textbf{0.830}} & \textcolor{black}{\textbf{0.899}} & \textcolor{black}{0.830} &  &  &  &  & \\
ct$(wlsv_{te}, bu_{cs})$ & \textcolor{red}{1.679} & \textcolor{red}{1.119} & \textcolor{red}{1.123} & \textcolor{red}{1.009} & \textcolor{red}{1.009} &  &  &  &  & \\
oct$(ols)$ & \textcolor{black}{0.872} & \textcolor{black}{0.927} & \textcolor{black}{0.914} & \textcolor{black}{0.930} & \textcolor{black}{0.840} &  &  &  &  & \\
oct$(struc)$ & \textcolor{red}{1.077} & \textcolor{red}{1.028} & \textcolor{red}{1.012} & \textcolor{black}{0.950} & \textcolor{black}{0.894} &  &  &  &  & \\
oct$(wlsv)$ & \textcolor{red}{1.149} & \textcolor{red}{1.089} & \textcolor{red}{1.041} & \textcolor{red}{1.006} & \textcolor{black}{0.922} &  &  &  &  & \\
oct$(bdshr)$ & \textcolor{red}{1.021} & \textcolor{red}{1.015} & \textcolor{black}{0.964} & \textcolor{black}{0.935} & \textcolor{black}{0.855} &  &  &  &  & \\
oct$_h(bshr)$ & \textcolor{black}{\textbf{0.833}} & \textcolor{black}{0.927} & \textcolor{black}{0.927} & \textcolor{black}{0.927} & \textcolor{blue}{\textbf{0.784}} &  &  &  &  & \\
oct$_h(hshr)$ & \textcolor{red}{1.066} & \textcolor{red}{1.056} & \textcolor{red}{1.005} & \textcolor{red}{1.026} & \textcolor{black}{0.926} &  &  &  &  & \\
oct$_h(shr)$ & \textcolor{red}{1.043} & \textcolor{red}{1.011} & \textcolor{black}{0.952} & \textcolor{black}{0.909} & \textcolor{black}{0.809} &  &  &  &  & \\
\bottomrule
\multicolumn{11}{l}{\rule{0pt}{1em}\rule{0pt}{1.75em}\makecell[l]{$^\ast$The Gaussian method employs a shrikage covariance matrix and includes four techniques\\ (G, B, H, HB) with multi-step residuals.}}\\
\end{tabular}

	\endgroup
	\caption{ES ratio indices defined in Section 5 for the Australian Tourism Demand dataset. %A lower value, indicates a more accurate forecast. 
	Approaches performing worse than the benchmark (bootstrap base forecasts, ctjb) are highlighted in red, the best for each column is marked in bold, and the overall lowest value is highlighted in blue. The reconciliation approaches are described in Table 2.}
	\label{tab:vnes}
\end{table}

\newpage
\section{Computation time report}
\setcounter{table}{0}  

In this section, we provide a computational time analysis for the two forecasting experiments in the paper. Tables \ref{tab:runGDP} and \ref{tab:runATD} show the runtime (in seconds) required for simulating 1000 samples (first row, base) and the additional time needed for reconciliation with various approaches in the first iteration of the experiment. The first table is refers to the Australian QNA dataset, while the second to the Australian Tourism Demand dataset. The system's hardware and software specifications are 
\begin{itemize}[nosep]
	\item CPU: \texttt{Intel(R) Core(TM) i7-10700 CPU \@ 2.90GHz 2.90 GHz}
	\item RAM size: \texttt{64 GB}
	\item R version: \texttt{R-4.2.1\_2022-06-23\_ucrt}
	\item R packages: \texttt{forecast} \citep{Rforecast}, \texttt{MASS} 	\citep{mass2002}, \texttt{Rfast} \citep{rfast2022}, and \texttt{FoReco} \citep{foreco2023}
\end{itemize}
\vskip1cm
\begin{table}[H]
\centering
	\begingroup
	\spacingset{1}
	\fontsize{9}{11}\selectfont
\begin{tabular}[t]{l|>{}crcrcrcrcr}
\toprule
\multicolumn{1}{c}{\makecell[c]{\bfseries Reconciliation\\\bfseries approach}} & \multicolumn{2}{c}{ctjb} & \multicolumn{2}{c}{G$_{h}$} & \multicolumn{2}{c}{H$_{h}$} & \multicolumn{2}{c}{G$_{oh}$} & \multicolumn{2}{c}{H$_{oh}$} \\
\midrule
base &  & 29.01 &  & 2.21 &  & 2.17 &  & 2.14 &  & 2.18\\[0.15cm]
ct$(shr_{cs}, bu_{te})$ & + & 0.19 & + & 0.13 & + & 0.12 & + & 0.30 & + & 0.30\\
ct$(wls_{cs}, bu_{te})$ & + & 0.21 & + & 0.31 & + & 0.31 & + & 0.33 & + & 0.35\\
oct$_o(wlsv)$ & + & 0.25 & + & 0.24 & + & 0.22 & + & 0.22 & + & 0.22\\
oct$_o(bdshr)$ & + & 0.48 & + & 0.44 & + & 0.45 & + & 0.45 & + & 0.45\\
%oct$_{oh}(shr)$ & + & 0.69 & + & 0.68 & + & 0.67 & + & 0.68 & + & 0.69\\
oct$_{oh}(hshr)$ & + & 0.64 & + & 0.65 & + & 0.64 & + & 0.65 & + & 0.64\\
\bottomrule
\end{tabular}
	\endgroup
\caption{Computational time (in seconds) for the first iteration of the Australian QNA forecasting experiment. The first row (base) reports the time to simulate 1000 samples, and the remaining rows the additional time to reconcile them with different approaches.}
\label{tab:runGDP}
\end{table}

\begin{table}[H]
	\centering
	\begingroup
	\spacingset{1}
	\fontsize{9}{11}\selectfont
\begin{tabular}[t]{l|>{}crcrcrcrcr}
\toprule
\multicolumn{1}{c}{\makecell[c]{\bfseries Reconciliation\\\bfseries approach}} & \multicolumn{2}{c}{ctjb} & \multicolumn{2}{c}{G} & \multicolumn{2}{c}{B} & \multicolumn{2}{c}{H} & \multicolumn{2}{c}{HB} \\
\midrule
base &  & 61.21 &  & 660.43 &  & 643.54 &  & 641.40 &  & 692.63\\[0.15cm]
%ct$(bu)$ & + & 2.14 & + & 2.25 & + & 2.05 & + & 2.12 & + & 2.54\\
ct$(shr_{cs}, bu_{te})$ & + & 3.79 & + & 4.02 & + & 3.79 & + & 3.54 & + & 4.18\\
%ct$(wlsv_{te}, bu_{cs})$ & + & 13.29 & + & 14.14 & + & 13.02 & + & 13.12 & + & 14.42\\
%oct$(ols)$ & + & 8.67 & + & 7.34 & + & 7.40 & + & 7.64 & + & 6.85\\
oct$(struc)$ & + & 7.73 & + & 7.10 & + & 7.19 & + & 7.11 & + & 6.56\\
oct$(wlsv)$ & + & 8.24 & + & 6.97 & + & 6.99 & + & 7.04 & + & 6.46\\
oct$(bdshr)$ & + & 60.27 & + & 52.65 & + & 52.13 & + & 51.92 & + & 49.51\\
oct$_h(bshr)$ & + & 503.52 & + & 426.20 & + & 419.99 & + & 418.94 & + & 422.94\\
oct$_h(hshr)$ & + & 485.66 & + & 482.13 & + & 418.64 & + & 418.82 & + & 463.18\\
%oct$_h(shr)$ & + & 359.78 & + & 444.40 & + & 373.84 & + & 372.72 & + & 448.20\\
\bottomrule
\end{tabular}
	\endgroup
	\caption{Computational time (in seconds) for the first iteration of the Australian Tourism Demand forecasting experiment. The first row (base) reports the time to simulate 1000 samples, and the remaining rows the additional time to reconcile them with different approaches.}
\label{tab:runATD}
\end{table}

\newpage
\phantomsection\addcontentsline{toc}{section}{References}

\bibliographystyle{agsm}
\bibliography{mybibfile}